\documentclass[12pt]{article}
\usepackage[height=8.5in,width=6.4in]{geometry}
\usepackage{color}
\usepackage{appendix}
\usepackage{empheq}
\usepackage{braket}
\usepackage{ascmac}
\usepackage{amsmath,amsthm,amssymb}
\theoremstyle{definition}
\newtheorem{dfn}{Definition}[section]
\newtheorem{prop}[dfn]{Proposition}
\newtheorem{lem}[dfn]{Lemma}
\newtheorem{thm}[dfn]{Theorem}
\newtheorem{cor}[dfn]{Corollary}

\usepackage{type1cm}
\usepackage{tikz}
\usetikzlibrary{intersections,calc,arrows.meta}
\usepackage[
 setpagesize=false,%
 bookmarks=true,%
 bookmarksdepth=tocdepth,%
 bookmarksnumbered=true,%
 colorlinks=true,%
 pdftitle={},%
 pdfsubject={},%
 pdfauthor={},%
 pdfkeywords={}%
]{hyperref}

\numberwithin{equation}{section}

\title{}
\author{}
\date{\empty}

\begin{document}

\begin{center}

{\LARGE \fontseries{mx}\selectfont
Generators of a Bosonic VOA and \\[1.5mm]
Connections to a Boundary VOA\par}

\vskip 1.2cm

Hikaru Sasaki,

\vskip .8cm

{\it
Department of Physics, Graduate School of Science\\
Osaka Metropolitan University, Osaka 558-8585, Japan
}

\end{center}

\vskip.7cm

\begin{abstract}
    We show how to identify generators of bosonic VOAs associated with $T_{[n-1,1]}^{[1^n]}(SU(n))$ and $T_{[n-1,1^2]}^{[2,1^{n-1}]}(SU(n+1))$,
    and conjecture that the algebraic structure of these VOAs can be constructed by these generators.
    We also find out that the boundary VOA associated with $T_{[n-1,1]}^{[1^n]}(SU(n))$ naturally includes the bosonic VOA.
\end{abstract}
\tableofcontents
\section{Introduction}
    \label{Introduction}
    It has been shown that a VOA can be constructed from a $3D$ $\mathcal{N}=4$ gauge theory on $\mathbb{R}_{\geq 0} \times \mathbb{C}$ \cite{Costello:2018fnz}.
    When using this method, we need an operation called topological twist. 
    In the case of $\mathcal{N}=4$, two kind of topological twists exist; H-twist and C-twist, and we study the VOA constructed by using the H-twist here.

    In the previous paper \cite{Nishinaka:2025nbe},
    T. Nishinaka and I studied bosonic VOAs of $T_{[n-1,1]}^{[1^n]}(SU(n))$ and $T^{[2,1^{n-1}]}_{[n-1,1^2]}(SU(n+1))$, where $U(1)$ gauge anomalies are canceled by using Heisenberg currents \cite{10.1093/imrn/rnaa031}.
    We also conjectured generators\footnote{In some examples \cite{Yoshida:2023wyt,Beem:2023dub,Coman:2023xcq}, the generators include the Higgs branch chiral operators of theories.} and compared these bosonic VOAs with VOAs associated with Argyres-Douglas theories of $(A_1,A_{2n-1})$ and $(A_1,D_{2n})$ \cite{Beem:2013sza,Creutzig:2017qyf}.
    However, we remain some problems whose results are only given:
    \begin{itemize}
        \item In \cite{Nishinaka:2025nbe}, we found the non-trivial operators which are different from the Higgs branch chiral operators. How do we check that these non-trivial operators are closed under the BRST cohomology?
        \item How do we derive the number of generators composed of these non-trivial operators?
    \end{itemize}
    Then, these procedures are given in this paper. 

    In \cite{Yoshida:2023wyt}, Y. Yoshida explained the idea that we could extend a bosonic VOA to a boundary VOA constructed by \cite{Costello:2018fnz}.
    Since the non-trivial operators which discovered in \cite{Nishinaka:2025nbe} were not included in \cite{Yoshida:2023wyt}, it is necessary for us to improve this idea.
    This improvement is so useful that we understand algebraic structures because of the following relationships:
    \begin{itemize}
        \item When a boundary VOA includes a bosonic VOA as a sub-VOA, the boundary VOA naturally inherits the properties of the bosonic VOA. 
        \item We observe that finding generators related of a boundary VOA is often easier than deriving those related of a bosonic VOA.
    \end{itemize}

    Thus, in this paper, we solve the above problems and supplement the results of the previous paper \cite{Nishinaka:2025nbe}.
    Then, we discover algebraic relationships between a bosonic VOA and a boundary VOA associated with an abelian quiver gauge theory based on Yoshida's idea \cite{Yoshida:2023wyt}.
    
    This paper is organized into three parts as follows.
    The first part covers the bosonic VOA associated with $T^{[1^n]}_{[n-1,1]} (SU(n))$ in Section.\ref{Preliminary of the bosonic VOA associated with T_{[n-1,1]}^{[1^n]}(SU(n))} to \ref{The OPEs with mathcal X I and mathcal Y I} and Appendix.\ref{OPE 1}.
    In Sec.\ref{Preliminary of the bosonic VOA associated with T_{[n-1,1]}^{[1^n]}(SU(n))}, we describe the BRST reduction and the idea about constructing operators except for Higgs branch chiral operators. In Sec.\ref{To derive non-trivial operators},
    we implement the idea given in Sec.\ref{Preliminary of the bosonic VOA associated with T_{[n-1,1]}^{[1^n]}(SU(n))}, and derive operators except for Higgs branch chiral operators. In Sec.\ref{How to construct primary fields from T i, mathcal X I, mathcal Y I} to \ref{Linearly independent operators consisting of X I and Y I},
    we identify generators of this bosonic VOA in generic case. In particular, we explain how to construct the stress tensor and intuitive primary fields in Sec.\ref{How to construct primary fields from T i, mathcal X I, mathcal Y I}. Then, in Sec.\ref{How to count independent operators} and \ref{Linearly independent operators consisting of X I and Y I}, we discuss the remaining problems that are not immediately clear in Sec.\ref{How to construct primary fields from T i, mathcal X I, mathcal Y I}.
    In Sec.\ref{The OPEs with mathcal X I and mathcal Y I} and Appendix.\ref{OPE 1}, the OPEs with this bosonic VOA are given. In particular, while some of OPEs given in Appendix.\ref{OPE 1} are highly nontrivial and intricate to calculate, two of them can nevertheless be derived directly through calculation as shown in Sec.\ref{The OPEs with mathcal X I and mathcal Y I}.
    
    The second part covers the bosonic VOA associated with $T^{[2,1^{n-1}]}_{[n-1,1^2]} (SU(n+1))$ in Sec.\ref{The bosonic algebra associated with T^{[2,1^{n-1}]}_{[n-1,1^2]}(SU(n+1))}, \ref{Generators of the bosonic VOA of T^{[2,1^{n-1}]}_{[n-1,1^2]}(SU(n+1))} and Appendix.\ref{OPE 2}. Sec.\ref{The bosonic algebra associated with T^{[2,1^{n-1}]}_{[n-1,1^2]}(SU(n+1))} is similar to Sec.\ref{Preliminary of the bosonic VOA associated with T_{[n-1,1]}^{[1^n]}(SU(n))} and \ref{To derive non-trivial operators};
    we construct the BRST cohomology and derive the candidates of generators regarding with the bosonic VOA associated with $T^{[2,1^{n-1}]}_{[n-1,1^2]}(SU(n+1))$. 
    Sec.\ref{Generators of the bosonic VOA of T^{[2,1^{n-1}]}_{[n-1,1^2]}(SU(n+1))} is also similar to Sec.{\ref{How to construct primary fields from T i, mathcal X I, mathcal Y I}} to \ref{Linearly independent operators consisting of X I and Y I}; we identify generators of this bosonic VOA. The OPEs with this bosonic VOA are given in Appendix.\ref{OPE 2}.

    The third part covers the relationship between a bosonic VOA and a boundary VOA in Sec.\ref{Fermionic extension} and Appendix.\ref{To generalize fermionic extension for abelian quiver gauge theories}. In Sec.\ref{Fermionic extension}, we consider the relationship between the bosonic VOA and the boundary VOA associated with $T_{[n-1,1]}^{[1^n]} (SU(n))$. 
    In Appendix.\ref{To generalize fermionic extension for abelian quiver gauge theories}, we propose how to generalize fermionic extension for other abelian quiver gauge theories based on an example in Sec.\ref{Fermionic extension}.
\section{Preliminary of the bosonic VOA associated with \texorpdfstring{$T_{[n-1,1]}^{[1^n]}(SU(n))$}{}}
    \label{Preliminary of the bosonic VOA associated with T_{[n-1,1]}^{[1^n]}(SU(n))}
    From this section to Sec.\ref{The OPEs with mathcal X I and mathcal Y I}, we study the bosonic VOA associated with $T^{[1^n]}_{[n-1,1]}\left(SU(n)\right)$ which is $3D$ $\mathcal{N}=4$ abelian quiver gauge theory.
    In order to analyze generators of the bosonic VOA associated with $T^{[1^n]}_{[n-1,1]}(SU(n))$, let us construct the BRST cohomology and introduce how Ideas on the construction of non-trivial operators.
    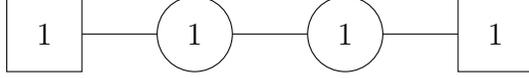
\begin{figure}[t]
        \center
        \begin{tikzpicture}
            \draw (0,-0.5) rectangle (1,0.5);
            \draw (0.5,0) node{1};
            \draw (1,0)--(2,0);
            \draw (2.5,0) circle[radius=0.5cm];
            \draw (2.5,0) node{1};
            \draw (3,0)--(4,0);
            \draw (4.5,0) circle[radius=0.5cm];
            \draw (4.5,0) node{1};
            \draw (5,0)--(6,0);
            \draw (6,-0.5) rectangle (7,0.5);
            \draw (6.5,0) node{1};
        \end{tikzpicture}
        \caption{The quiver diagram for $T_{[2,1]}^{[1^3]}(SU(3))$.}
    \end{figure}
    \subsection{BRST reduction}
        \label{BRST reduction 1}
        In this subsection, we construct the BRST cohomology by using the fields on the boundary.
        To do so, let us explain the fields on the boundary derived from the fields in the bulk, that is hypermultiplts and vector multiplets, by $(0,4)$-deformed boundary condition\cite{Costello:2018fnz,Gaiotto:2017euk}.
        
        By $(0,4)$-deformed boundary condition, symplectic bosons $(X_i,Y_i)$ satisfying the following OPEs can be derived from hypermultiplts $(q_i,\tilde{q}_i)$ in the bulk,
        \begin{align}
            X_i(z) Y_j(0) \sim \frac{\delta_{ij}}{z}~.
        \end{align}
        Similarly, vector multiplets in the bulk become $bc$-ghosts $(b_a,c^a)$ satisfying the following OPEs on the boundary,
        \begin{align}
            b_a(z) c^b(0) \sim \frac{\delta_{a}^{b}}{z}~.
        \end{align}
        
        The gauge group of $T_{[n-1,1]}^{[1^n]}(SU(n))$ is $\prod_{i=1}^{n-1} U(1)_i$, and $(q_i,\tilde{q}_i)$ has the gauge symmetry $U(1)_{i-1} \times U(1)_{i}$, where $q_i$ (resp. $\tilde{q}_i$) is the fundamental (resp. anti-fundamental) representation for $U(1)_{i-1}$
        and the anti-fundamental (resp. fundamental) representation for $U(1)_{i}$. By inheriting these properties to symplectic bosons, we can define the BRST current $\tilde{J}_{BRST}$ as follows \footnote{
            We use the normal order product defined by (\ref{product}) for $p=0$ and the following normal order product of operators $A$, $B$, $C$, ``$(A(BC)_0)_0$'' is abbreviated as ``$ABC$''.
        }
        \begin{align}
            \tilde{J}_{BRST} \coloneqq \sum_{a=1}^{n-1}c^a \left(X_a Y_a -X_{a+1} Y_{a+1}\right)~.
        \end{align}
        We can perform the BRST transformation regarding with $(X_i,Y_i)$ by using $\tilde{J}_{BRST}$.

        This definition of $\tilde{J}_{BRST}$, however, is satisfying $\tilde{J}_{BRST}(z) \tilde{J}_{BRST} \nsim 0$. In other words, the $U(1)$ gauge anomalies emerge on the boundary.
        In order to remove these gauge anomalies, we introduce the Heisenberg currents $h_i$ satisfying the following OPEs
        \begin{align}
            h_i(z) h_j(0) \sim \frac{C_{ij}}{z^2}~,
        \end{align}
        where $C_{ij}$ is the Cartan matrix of $\mathfrak{su}(n)$. By using these currents, we refine the BRST current as follows
        \begin{align}
            J_{BRST} \coloneqq \sum_{a=1}^{n-1}c^{a}\left(X_a Y_a -X_{a+1} Y_{a+1} +h_a\right)~.
        \end{align}
        Then, the BRST charge is defined by
        \begin{align}
            Q \coloneqq \oint\frac{dz}{2\pi i}J_{BRST}(z)~.
        \end{align}
        Since satisfying $J_{BRST}(z) J_{BRST}(0) \sim 0$, the BRST charge is nilpotent, that is $Q^2=0$.
        Thus, we can construct the BRST cohomology by using $Q$. 

        After this subsection, we will derive closed operators only consisting symplectic bosons under the BRST cohomology. In particular, we will construct non-trivial operators \cite{Nishinaka:2025nbe}.
        To do so, in the next subsection, we will study how to represent the stress tensor only consisting of symplectic bosons.
    \subsection{Ideas on the construction of non-trivial operators}
        \label{Ideas on the construction of non-trivial operators}
        In \cite{Nishinaka:2025nbe}, T.Nishinaka and I proposed generators of the algebraic construction under the BRST cohomology:
        $U$, $\mathcal{X}$ and $\mathcal{Y}$ which come from the Higgs branch chiral operators are defined by
        \begin{align}
            U \coloneqq -\frac{1}{n} \sum_{i=1}^{n} X_i Y_i~,\quad \mathcal{X} \coloneqq \prod_{i=1}^{n} X_i~,\quad \mathcal{Y} \coloneqq \prod_{i=1}^{n} Y_i~,
        \end{align}
        and the others $T_i$, $W_i$, $\mathcal{X}_{I}$ and $\mathcal{Y}_{I}$. \footnote{In this paper, $T_i$, $W_i$, $\mathcal{X}_{I}$ and $\mathcal{Y}_{I}$ will be defined in Sec.\ref{The non-trivial operators}.}
        However, we did not explain how we discovered $T_i$, $W_i$, $\mathcal{X}_{I}$ and $\mathcal{Y}_{I}$ before.
        Then, let us explain how we arrived at our notation in this subsection.

        The idea of our notation can be obtained by deforming the stress tensor $T$ under the BRST cohomology. The stress tensor $T$ in the bosonic VOA is defined by
        \begin{align}
            \label{Total T}
            T \coloneqq T_{sb} + T_h + T_{bc}~, \\
            \label{T SB}
            T_{sb} \coloneqq \frac{1}{2}\sum_{i=1}^{n}\left(X_i \partial Y_i - \partial X_i Y_i\right)~,\\
            \label{T h and T bc}
            T_h \coloneqq \frac{1}{2}\sum_{i,j=1}^{n-1} C^{ij}h_i h_j~,\quad T_{bc} \coloneqq -\sum_{a=1}^{n-1} b_a \partial c^a~.
        \end{align}
        Rewriting $T_h + T_{bc}$ using symplectic bosons $(X_i,Y_i)$ under the BRST cohomology, we can obtain the following expression of the stress tensor $T$ only using $X_i$, $Y_i$ \cite{Nishinaka:2025nbe};
        \footnote{Here, we define ``$A \equiv B$'' as ``$A-B = \left(Q\text{-exact}\right)$''.}
        \begin{align}
            \label{definition of T}
            T \equiv \frac{1}{2} \sum_{i=1}^{n} \left\{\left(X_i \partial Y_i -\partial X_i Y_i\right) +\mathcal{A}_i \mathcal{A}_i\right\} -\frac{n}{2}U^2~,
        \end{align}
        where we set $\mathcal{A}_i \coloneqq -X_i Y_i$.
        Since $T$ and $-\frac{n}{2} U^2$ are closed operators, the following operator is also closed;
        \begin{align}
            \sum_{i=1}^{n} \left\{\left(X_i \partial Y_i -\partial X_i Y_i\right) +\mathcal{A}_i \mathcal{A}_i\right\}~.
        \end{align}
        In particular, we can check that $\left(X_i \partial Y_i -\partial X_i Y_i\right) +\mathcal{A}_i \mathcal{A}_i$ is closed;
        \begin{align}
            \label{from mathcal D and hat A}
            Q  \left\{\left(X_i \partial Y_i -\partial X_i Y_i\right) +\mathcal{A}_i \mathcal{A}_i\right\} = 0~.            
        \end{align}

        Based on these examples, we would like to assume the existence of $\mathcal{D}_i$ and $\hat{\mathcal{A}}_i$ that perform the following actions;
        \begin{itemize}
            \item $\mathcal{D}_i$ performs the following actions;
            \begin{align*}
                \left\{
                    \begin{aligned}
                        X_j &\mapsto \delta_{ij} \partial X_j \\
                        Y_j &\mapsto - \delta_{ij} \partial Y_j
                    \end{aligned}
                \right.~.
            \end{align*}
            In particular, $\mathcal{D}_i$ are satisfied the Leibniz rule as follows;
            \begin{align*}
                \mathcal{D}_i \left(X_i Y_i \right) &= \left(\mathcal{D}_i X_i\right) Y_i + X_i \mathcal{D}_i Y_i \\
                &= \partial X_i Y_i - X_i \partial Y_i~.
            \end{align*}
            \item For any operator $\mathcal{O}$, $\hat{\mathcal{A}}_i$ performs the action;
            \begin{align*}
                \hat{\mathcal{A}}_i \mathcal{O} = \mathcal{A}_i \mathcal{O}~.
            \end{align*}
        \end{itemize}
        Then, $X_i \partial Y_i - \partial X_i Y_i + \mathcal{A}_i \mathcal{A}_i$ are equal to $\mathcal{D}_i \mathcal{A}_i + \hat{\mathcal{A}}_i \mathcal{A}_i$.
        These operators look very similar to ``field strength $F = dA + A \wedge A$''.
        
        Therefore, assuming now that $\mathcal{D}_i$ and $\hat{\mathcal{A}}_i$ can be defined,
        we find the following interesting possibilities. From some operator $\mathcal{O}$, we may construct a closed operator $\left(\mathcal{D}_i +e \hat{\mathcal{A}}_{i} \right)\mathcal{O}$,
        \begin{align}
            \label{Can this equation be solved?}
            Q \left(\mathcal{D}_i+ e \hat{\mathcal{A}}_i\right) \mathcal{O} = 0~,
        \end{align}
        where $e$ is a proper constant, and $\mathcal{D}_i + e \hat{\mathcal{A}}_i$ closely resemble ``covariant derivatives $d + e A$''. However, the difficulty between solving $Q \mathcal{O} = 0$ and solving (\ref{Can this equation be solved?}) is almost equal.
        
        In order to simplify this problem, we consider the commutator $\left[Q,\partial \right] = 0$. By using $\left[Q, \partial \right] = 0$, we know that $Q \left(\partial \mathcal{O} \right) = \partial \left(Q \mathcal{O}\right)$ always holds.
        In particular, when $Q \mathcal{O} = 0$, we can easily find out that $Q \left(\partial \mathcal{O} \right) =0$ holds. Then, based on a closed operator $\mathcal{O}$, we can easily construct closed operators $\partial \mathcal{O}$, $\partial^2 \mathcal{O}$ and so on.
        
        Similarly, we except that considering the commutator $\left[Q, \mathcal{D}_i +e \hat{\mathcal{A}}_i\right]$ can make solving (\ref{Can this equation be solved?}) easier.
        After computing $\left[Q,\mathcal{D}_i +e \hat{\mathcal{A}}_i\right] \mathcal{O}$ in general, it suffices to apply the condition (\ref{Can this equation be solved?}).
        In other words, we would like to study how to construct a closed operator by deforming $\mathcal{O}$. Since we deform the original Higgs branch to construct the bosonic VOA associated with $T^{[1^n]}_{[n-1,1]}(SU(n))$,
        we naturally predict that the candidates of $\mathcal{O}$ are related with the Higgs branch chiral operators
        and the F-term conditions before deforming the original Higgs branch.
        Then, in the next section, we will define $\mathcal{D}_{i}$ and $\hat{\mathcal{A}}_i$ and study a non-trivial closed operators using $\mathcal{D}_{i}$ and $\hat{\mathcal{A}}_i$.
\section{To derive non-trivial operators}
    \label{To derive non-trivial operators}
    In the previous subsection, we expected that an analogy of a gauge theory exists for the bosonic VOA. 
    Then, we define ``partial derivative'' $\mathcal{D}_{i}$, correction operator $\hat{\mathcal{A}}_{i}$, and ``covariant derivative'' $\mathcal{D}_i +e \hat{\mathcal{A}}$, where $e$ is a proper constant, in this section.
    Moreover, we study the existence of closed operators $\mathcal{D}_i \mathcal{O} + e \hat{\mathcal{A}}_i \mathcal{O}$.

    \subsection{The definition of ``covariant derivatives'' \texorpdfstring{$\mathcal{D}_i +e \hat{\mathcal{A}}_i$}{}}
        \label{The definition of ``covariant derivatives'' mathcalD i + e hat A_i}
        Before we define $\mathcal{D}_i$ and $\hat{\mathcal{A}}_i$, let us define the following product $(AB)_p$:\footnote{See for instances \cite{kac1998vertex,BA68850995}. We follow the notation of \cite{doi:10.1142/S0129183191001001}}
        \begin{align}
            \label{product}
            (AB)_p(w) \coloneqq \oint \frac{dz}{2\pi i}(z-w)^{p-1} A(z) B(w)~,
        \end{align}
        where $p$ is an integer. In particular, when $p=0$, (\ref{product}) is called for the normal order product.
        Using this product, the following relationships hold.
        \begin{align}
            \label{OPEdefs1}
            \left(B A\right)_q &= (-1)^{|A||B|}\sum_{l \geq q} \frac{(-1)^l}{(l-q)!} \partial^{l-q} (AB)_{l}~,\\
            \label{OPEdefs2}
            \left(A \left(B C\right)_p \right)_q &= (-1)^{|A||B|} \left(B \left(AC\right)_q\right)_p 
            +\sum_{l>0} \binom{q-1}{l-1} \left(\left(AB\right)_{l} C\right)_{p+q-l}~,\\
            \label{OPEdefs3}
            (\partial A B)_p &= -(p-1) \left(A B\right)_{p-1}~,\\
            \label{OPEdefs4}
            (A \partial B)_p &= (p-1) (A B)_{p-1} +\partial (A B)_p~,
        \end{align}
        where $p$ and $q$ are integers and $|A|$ (resp. $|B|$) is the parity of $A$ (resp. $B$).\footnote{$\binom{a}{b}$ is binomial coefficient.
        When we use a matrix in later sections, a matrix is written by ``$\left[*\right]$''.}
        
        Then, let us define the linear operators $\mathcal{D}_i$ which are similar to the ordinary partial derivatives,
        and $\hat{\mathcal{A}}_i$ which are similar to gauge fields.
        The latter is easily defined by 
        \begin{align}
            \hat{\mathcal{A}}_i \mathcal{O} \coloneqq \left(\mathcal{A}_i \mathcal{O}\right)_0~.
        \end{align}
        Since
        \begin{align}
            \left\{
                \begin{aligned}
                    \left(\mathcal{A}_i X_j\right)_0 &= \delta_{ij} \partial X_j -\left(X_i \left(Y_i X_j\right)_0 \right)_0\\
                    \left(\mathcal{A}_i Y_j\right)_0 &= -\delta_{ij} \partial Y_j -\left(X_i \left(Y_i Y_j\right)_0 \right)_0
                \end{aligned}
            \right.~,
        \end{align}
        we can define $\mathcal{D}_i$ as follows
        \begin{align}
            \label{definition of mathcal D}
            \mathcal{D}_i \mathcal{O} \coloneqq \left(\mathcal{A}_i \mathcal{O}\right)_0 + \left(X_i \left(Y_i \mathcal{O}\right)_0 \right)_0~,
        \end{align}
        where $\mathcal{O}$ is an operator.

        Based on these calculations, let us compute $\mathcal{D}_i \mathcal{O}$ and check whether this definition (\ref{definition of mathcal D}) is well-defined.
        When $\mathcal{O}= \partial^d X_j$ or $\mathcal{O} = \partial^d Y_j$,
        \begin{align}
            \mathcal{D}_i \left(\partial^d X_j\right) 
            &= \frac{\delta_{ij}}{d+1}\partial^{d+1} X_i~,\\
            \mathcal{D}_i \left(\partial^d Y_j\right)
            &= -\frac{\delta_{ij}}{d+1} \partial^{d+1} Y_i~.
        \end{align}
        More generally, when given inductively as $\mathcal{O} = \left(\partial^d X_j \tilde{O}\right)_0$ or $\mathcal{O}=\left(\partial^d Y_j \tilde{O}\right)_0$,
        \begin{align}
            \mathcal{D}_i\left(\partial^d X_j \tilde{O}\right)_0
            &= \frac{\delta_{ij}}{d+1}\left(\partial^{d+1} X_j \tilde{\mathcal{O}}\right)_0 + \left(\partial^d X_j \mathcal{D}_i \tilde{O}\right)_0~,\\
            \mathcal{D}_i \left(Y_j^{(d)} \tilde{O}\right)_0 &= -\frac{\delta_{ij}}{d+1} \left(\partial^{d+1} Y_j \tilde{O}\right)_0 +\left(\partial^d Y_j \mathcal{D}_i \tilde{O}\right)_0~.
        \end{align}
        Thus, ``partial derivative'' $\mathcal{D}_i$ performs the desired action.

        We note the properties of $\mathcal{D}_i$. 
        \begin{enumerate}
            \item Since the above calculations regarding with $\mathcal{D}_i$, these derivatives are clearly commutative; $[\mathcal{D}_i~,\mathcal{D}_j]=0~$.
            However, we note $\left[\partial, \mathcal{D}_i \right] \neq 0$ in general.
            \item The Leibniz rule for $\mathcal{D}_i$ may not hold in general. 
            For example, let us compute $\mathcal{D}_i \left(\left(X_iY_i\right)_0 X_i\right)_0$.        
            By performing following algebraic manipulation,
            \begin{align*}
                \left(\left(X_iY_i\right)_0 X_i\right)_0 = \left(X_i \left(X_i Y_i\right)_0 \right)_0 -\partial X_i~,
            \end{align*}
            we can derive the following operation that $\mathcal{D}_i$ performs on $\left(\left(X_i Y_i \right)_0 X_i\right)_0$.
            \begin{align*}
                \mathcal{D}_i\left(\left(X_iY_i\right)_0 X_i\right)_0 &= \mathcal{D}_i \left\{\left(X_i \left(X_i Y_i\right)_0 \right)_0 -\partial X_i\right\} \\
                &= 2\left(X_i \left(\partial X_i Y_i\right)_0 \right)_0-\left(X_i \left(X_i \partial Y_i\right)_0 \right)_0 -\frac{1}{2} \partial^2 X_i~.
            \end{align*}
            On the other hand,
            \begin{align*}
                \left(\mathcal{D}_i\left(X_iY_i\right)_0 X_i\right)_0 + \left(\left(X_iY_i\right)_0 \mathcal{D}_iX_i\right)_0 
                &=  2\left(X_i \left(\partial X_i Y_i\right)_0 \right)_0-\left(X_i \left(X_i \partial Y_i\right)_0 \right)_0 -2 \partial^2 X_i~.
            \end{align*}
        \end{enumerate}
        
        From the above results, it should be noted that
        \begin{align}
            \mathcal{D}_i\left(\mathcal{O}_1 \mathcal{O}_2\right)_0 \neq \left(\left\{\mathcal{D}_i\mathcal{O}_1\right\} \mathcal{O}_2\right)_0
            + \left(\mathcal{O}_1 \left\{\mathcal{D}_i\mathcal{O}_2\right\}\right)_0
        \end{align}
        in general. Thus, $\mathcal{D}_i$ is satisfied with the Leibniz rule if a operator $\mathcal{O}$ is a nest structure for normal order product.

        The operator $\hat{\mathcal{A}}_i$ have two following properties.
        \begin{itemize}
            \item The commutation relations between $\hat{\mathcal{A}}_i$ and $\hat{\mathcal{A}}_j$ are trivial.
            \begin{align}
                \left[\hat{\mathcal{A}}_i, \hat{\mathcal{A}}_j\right] = 0~.
            \end{align}
            \item The commutation relations between $\mathcal{D}_i$ and $\hat{\mathcal{A}}_j$ for $i \neq j$ are also trivial
            \begin{align}
                &\left[\mathcal{D}_i, \hat{\mathcal{A}}_{j} \right] = 0 \quad (i \neq j)~.
            \end{align}
        \end{itemize}
        
        Thus, we can naturally define ``covariant derivative'' $D_{i,(e)}$ as follows:
        \begin{align}
            D_{i,(e)} \mathcal{O} \coloneqq \mathcal{D}_i \mathcal{O} +e \hat{\mathcal{A}}_{i} \mathcal{O}~.
        \end{align}
        In the next section, we derive closed operators using ``covariant derivatives'' $D_{i,(e)}$.
    \subsection{Commutation relation for ``covariant derivatives'' and the BRST charge}
        \label{Commutation relation for ``covariant derivatives'' and the BRST charge}
        In this subsection, let us study the commutation relation for the BRST charge $Q$ and ``covariant derivative'' $\mathcal{D}_i + e \hat{\mathcal{A}}_i$.
        Here, we would like to generate new closed operators by using commutator $\left[Q,\mathcal{D}_i + e\hat{\mathcal{A}}_i\right]$.         Considered general cases, we need to use the following operator $\mathcal{O}$ such as
        \begin{align}
            \label{mono}
            \mathcal{O} = \partial^{a_1} A_1 \partial^{a_2} A_2 \cdots \partial^{a_{m-1}} A_{m-1} \partial^{a_m} A_m~,
        \end{align}
        where $A_1,\dots,A_m$ are $X_j$ or $Y_j$ for $j=1,\dots,n$.
        Then,
        \begin{align}
            \left\{
                \begin{aligned}
                    [Q~,\mathcal{D}_i]\mathcal{O} 
                    &= -\sum_{k=1}^{m}\frac{\delta_{A_k=X_i}+\delta_{A_k=Y_i}}{a_k+1}
                    \left\{\partial^{a_k +1} c_i\mathcal{O}_{[k]} -\partial^{a_k +1} c_{i-1}\mathcal{O}_{[k]}\right\} \\
                    [Q,\hat{\mathcal{A}}_i]\mathcal{O} 
                    &= \partial c_i \mathcal{O} - \partial c_{i-1} \mathcal{O}~.
                \end{aligned}
            \right.~,
        \end{align}
        where $\mathcal{O}_{[i]}$ is the operator with $\partial^{a_i} A_i$ replaced with $A_i$.

        Then, from the commutations $\left[Q,\mathcal{D}_i\right] \mathcal{O}$ and $\left[Q,\hat{\mathcal{A}}_i\right] \mathcal{O}$, we can easily derive the following relation
        \begin{align}
            [Q,D_{i,(e)}] \mathcal{O} &= -\sum_{k=1}^{m}\frac{\delta_{A_k=X_i} +\delta_{A_k=Y_i}}{a_k+1}
            \left\{\partial^{a_k +1} c_i \mathcal{O}_{[k]} -\partial^{a_k +1} c_{i-1} \mathcal{O}_{[k]} \right\} \notag\\
            &\quad +e \left\{\partial c_i \mathcal{O}  -\partial c_{i-1} \mathcal{O} \right\}~.
        \end{align}
        Separating the sum ``$\sum_{k=1}^m \cdots$'' into ``$\sum' \cdots $'' with $a_k \geq 1$, and ``$\sum'' \cdots $'' with $a_k=0$ for $k=1,\dots,m$, we can transform the commutation relation as follows:
        \begin{align}
            \label{commutation relation}
            [Q,D_{i,(e)}] \mathcal{O} &= -\sideset{}{'}{\sum}_{k=1}^{m}\frac{\delta_{A_k=X_i} +\delta_{A_k=Y_i}}{a_k+1}
            \left\{\partial^{a_k +1} c_i\mathcal{O}_{[k]} -\partial^{a_k +1} c_{i-1}\mathcal{O}_{[k]}\right\} \notag\\
            &\quad +\left(e -\sideset{}{''}{\sum}_{k=1}^{m}(\delta_{A_k=X_i} +\delta_{A_k=Y_i})\right)\left\{\partial c_i \mathcal{O} -\partial c_{i-1} \mathcal{O} \right\}~.
        \end{align}
        From this relation, we simply find the case that $[Q,D_{i,(e)}]\mathcal{O}=0$ can hold.
        \begin{itemize}
            \item The total number of $X_i$ and $Y_i$ contained on $\mathcal{O}$ is equal to $e$.
            \item $\mathcal{O}$ does not contain neither $\partial^a X_i$ or $\partial^a Y_i$ when $a$ is a positive integers.\footnote{In fact, when $\mathcal{O} = \sum_{h} \mathcal{O}_{h}$ where $\mathcal{O}_{h}$ are linearly independent operators and $h \geq 2$,
            we can include $\partial^a X_i$ and $\partial^a Y_i$.}
        \end{itemize}
        In particular, when $\mathcal{O}$ is satisfying these conditions and $Q \mathcal{O} = 0$, the following equality can also hold
        \begin{align*}
            Q\left(D_{i,(e)}\mathcal{O}\right) = D_{i,(e)}Q\mathcal{O} = 0~.
        \end{align*}
        Then, $D_{i,(e)}\mathcal{O}$ is either a closed operator or equal to $0$.

        It is also necessary for us to consider the non-commutative case $[Q,D_{i,(e)}]\mathcal{O} \neq 0$. In this case, the most important and interested condition for $e$ and $\mathcal{O}$ is $Q\mathcal{O} \neq 0$ and $QD_{i,(e)}\mathcal{O} = 0$ simultaneously.
        This relationship is equivalent of the following equality.
        \begin{align}
            \label{non-commutative case}
            D_{i,(e)}Q\mathcal{O} &= \sideset{}{'}{\sum}_{k=1}^{m}\frac{\delta_{A_k=X_j}\delta_{ij}+\delta_{A_k=Y_j}\delta_{ij}}{a_k+1}
            \left\{\partial^{a_k +1}c_j\mathcal{O}_k  -\partial^{a_k +1} c_{j-1}\mathcal{O}_k \right\} \notag \\
            &\quad +\left(-e +\sideset{}{''}{\sum}_{k=1}^{m}(\delta_{A_k=X_j} +\delta_{A_k=Y_j})\delta_{ij}\right)\left\{\partial c_i\mathcal{O} -\partial c_{i-1} \mathcal{O} \right\}~.
        \end{align}
        We consider the following examples:
        \begin{enumerate}
            \item For $\mathcal{O} = (X_i Y_i)_0$, we can directly derive the following equality. 
            \begin{align*}        
                D_{i,(e)}Q\mathcal{O} = 
                    e \left\{\partial c_{i} \mathcal{O}  -\partial c_{i-1} \mathcal{O} \right\}~.
            \end{align*}
            Compared with the equality (\ref{non-commutative case}), we naturally get the following equation. 
            \begin{align*}
                2(e-1) \left\{\partial c_{i} \mathcal{O} -\partial c_{i-1} \mathcal{O} \right\} = 0~.
            \end{align*}
            Therefore, when we choose $e=1$, $D_{i,(1)} \mathcal{O} = -X_i X_i Y_i Y_i +2 \partial X_i Y_i -2 X_i \partial Y_i$ just becomes closed.
            \item We define the operators $\mathcal{O}_{1}$, $\mathcal{O}_{2}$, $\mathcal{O}_{3}$ and $\mathcal{O}_4$ as follows:
            \begin{align*}
                \mathcal{O}_{1} \coloneqq X_i X_i Y_i Y_i~,\quad \mathcal{O}_{2} \coloneqq \partial X_i Y_i~,
                \quad \mathcal{O}_{3} \coloneqq X_i \partial Y_i~,\quad \mathcal{O}_{4} \coloneqq X_i Y_i~.
            \end{align*}
            Using these operators, we directly get the following relations : 
            \begin{align*}
                D_{i,(e)}Q\mathcal{O}_{1} &= 
                4e\left\{\partial c_{i} \mathcal{O}_{1}
                -\partial c_{i-1} \mathcal{O}_{1} \right\}
                -4(e+1)\left\{\partial c_{i} \mathcal{O}_{2} 
                - \partial c_{i-1} \mathcal{O}_{2} \right\} \\
                &\quad +4(e+1) \left\{\partial c_{i} \mathcal{O}_{3}
                - \partial c_{i-1} \mathcal{O}_{3}\right\}~,
            \end{align*}
            \begin{align*}
                D_{i,(e)}Q \mathcal{O}_{2} 
                &= e \left\{\partial c_i\mathcal{O}_{1} -\partial c_{i-1}\mathcal{O}_{1} \right\}
                -(e+1) \left\{\partial c_{i} \mathcal{O}_{2} -\partial c_{i-1} \mathcal{O}_{2} \right\} \\
                &\quad +(e+1) \left\{\partial c_{i} \mathcal{O}_{3} -\partial c_{i-1} \mathcal{O}_{3} \right\}
                +\frac{e}{2}\left\{\partial^2 c_{i}\mathcal{O}_{4} -\partial^2 c_{i-1} \mathcal{O}_{4} \right\}~,
            \end{align*}
            and
            \begin{align*}
                D_{i,(e)}Q \mathcal{O}_{3}
                &= -e \left\{\partial c_i \mathcal{O}_{1} -\partial c_{i-1} \mathcal{O}_{1} \right\}
                +(e+1) \left\{\partial c_{i} \mathcal{O}_{2} -\partial c_{i-1} \mathcal{O}_{2} \right\} \\
                &\quad -(e+1) \left\{\partial c_{i} \mathcal{O}_{3} -\partial c_{i-1} \mathcal{O}_{3} \right\}
                +\frac{e}{2}\left\{\partial^2 c_{i}\mathcal{O}_{4} -\partial^2 c_{i-1} \mathcal{O}_{4} \right\}~.
            \end{align*}
            Then, when we set $\mathcal{O} = a_1 \mathcal{O}_1 + a_2 \mathcal{O}_2 + a_3 \mathcal{O}_3$, we derive the following result:
            \begin{align*}
                D_{i,(e)} Q \mathcal{O} &=
                \left(4a_1 +a_2 -a_3 \right)e\left\{\partial c_{i} \mathcal{O}_{1}
                -\partial c_{j-1} \mathcal{O}_{1}\right\} \\
                &\quad +\left(4a_1 -a_2 +a_3 \right)(e+1)
                \left\{\partial c_{i} \mathcal{O}_{2} 
                - \partial c_{i-1} \mathcal{O}_{2}  \right\} \\
                &\quad +\left(4a_1 +a_2 -a_3 \right)(e+1) \left\{\partial c_{i} \mathcal{O}_{3}
                - \partial c_{i-1} \mathcal{O}_{3}  \right\} \\
                &\quad +\frac{a_2 +a_3}{2} e \left\{\partial^2 c_{i}\mathcal{O}_{4} -\partial^2 c_{i-1} \mathcal{O}_{4} \right\}~.
            \end{align*}

            On the other hand, from the calculation that  the commutator $\left[Q,D_{i,(e)}\right]$ acts on $\mathcal{O} = a_1 \mathcal{O}_1 + a_2 \mathcal{O}_2 + a_3 \mathcal{O}_3$,
            we can derive the following result.
            \begin{align*}
                [Q,D_{i,(e)}] \mathcal{O} 
                &= a_1 \left(e -4\right) \left\{\partial c_i \mathcal{O}_1 -\partial c_{i-1} \mathcal{O}_1 \right\} \\
                &\quad + a_2 \left(e -1\right) \left\{\partial c_i \mathcal{O}_2 -\partial c_{i-1} \mathcal{O}_2 \right\} \\
                &\quad + a_3 \left(e -1\right) \left\{\partial c_i \mathcal{O}_3 -\partial c_{i-1} \mathcal{O}_3 \right\} \\
                &\quad -\frac{a_2+a_3}{2} \left\{\partial^2 c_i \mathcal{O}_4 -\partial^2 c_{i-1} \mathcal{O}_4 \right\}~.
            \end{align*}
            Thus, when assuming that $QD_{i,(e)} \mathcal{O} = 0$ holds, we need to solve the following equations;
            \begin{align*}
                \left\{
                \begin{aligned}
                    \left(4 a_1 +a_2 -a_3 \right)e &= -a_1 \left(e-4\right) \\
                    \left(4 a_1 +a_2 -a_3 \right)\left(e+1\right) &= a_2 \left(e-1 \right) \\
                    \left(4 a_1 +a_2 -a_3 \right)\left(e+1\right) &= -a_3 \left(e-1\right) \\
                    \left(a_2 +a_3\right) e &= a_2 +a_3
                \end{aligned}
                \right.~.
            \end{align*}
            The solutions except for the trivial case $a_1 =a_2 = a_3 =0$ are 
            \begin{align}
                \label{solution 1}
                (a_1,a_2,a_3,e) &= \left(0,a_2,a_2,1\right)~,\\
                \label{solution 2}
                &\quad \left(a_1, \frac{-9 +\sqrt{15}i}{4}a_1, \frac{9 -\sqrt{15}i}{4}a_1,\frac{1-\sqrt{15}i}{2}\right)~,\\
                \label{solution 3}
                &\quad \left(a_1, \frac{-9-\sqrt{15}i}{4}a_1, \frac{9 +\sqrt{15}i}{4}a_1,\frac{1+\sqrt{15}i}{2}\right)~,
            \end{align}
            where $a_2 \neq 0$ for (\ref{solution 1}), $a_1 \neq 0$ for (\ref{solution 2}) and (\ref{solution 3}).

            For (\ref{solution 1}), we derive the following operator for $a_2 =1$;
            \begin{align}
                D_{i,(1)} \partial \mathcal{O}_{4} = -\frac{1}{2} \partial \left(D_{i,(1)} \mathcal{O}_{4}\right)~.
            \end{align}
            We have already found this operator in the previous example.

            Similarly, for (\ref{solution 2}) and (\ref{solution 3}), we also derive the following operators for $a_1 =1$;
            \begin{align}
                \frac{-1 \pm \sqrt{15}i }{4} &\left\{2 X_i X_i X_i Y_i Y_i Y_i +9 X_i X_i Y_i \partial Y_i -9 X_i \partial X_i Y_i Y_i \right. \notag \\
                &\quad \left. +3 X_i \partial^2 Y_i -12 \partial X_i \partial Y_i +3 \partial^2 X_i Y_i\right\}~.
            \end{align}
            Furthermore, since $a_2 = -a_3$ also holds by (\ref{solution 2}) or (\ref{solution 3}), we can consider the following formula:
            \begin{align*}
                D_{i,(A)}\mathcal{O}_{4} = -A\mathcal{O}_{1} +(A+1)\left\{\mathcal{O}_{2} -\mathcal{O}_{3}\right\}~,
            \end{align*}
            where we set $a_1 = -A$ and $a_2 = - a_3 =A +1$. Then,
            \begin{align}
                \label{solution 4}
                (e,A) = \left(\frac{1 +\sqrt{15}i}{2},\frac{5 -\sqrt{15}i}{10}\right)~,\quad \left(\frac{1-\sqrt{15}i}{2},\frac{5 +\sqrt{15}i}{10}\right)~.
            \end{align}
            Thus, $D_{i,(e)} D_{i,(A)} \mathcal{O}_4$ for (\ref{solution 4}) are closed operators and equal to
            \begin{align*}
                \frac{-5 \pm \sqrt{-15}}{10} &\left\{2 X_i X_i X_i Y_i Y_i Y_i +9 X_i X_i Y_i \partial Y_i -9 X_i \partial X_i Y_i Y_i \right. \notag \\
                    &\quad \left. +3 X_i \partial^2 Y_i -12 \partial X_i \partial Y_i + 3 X_i \partial^2 Y_i\right\}~.
            \end{align*}
        \end{enumerate}
    \subsection{The non-trivial operators}
        \label{The non-trivial operators}
        In the previous subsection, we checked how to construct non-trivial closed operators using $D_{i,(e)}$. From now on, let us reproduce the operators given in \cite{Nishinaka:2025nbe}.
        
        \begin{itemize}
            \item $\prod_{i \in I}D_{i,(1)}\mathcal{X}$ for $I \subset \{1,\dots,n\}$ are recursively constructible closed operators. Since $Q \mathcal{X} = 0$ and $[Q,D_{i,(1)}]\mathcal{X}=0$,
            then $D_{i,(1)} \mathcal{X}$ is a closed operator. Similarly, when $I \subset \{1,\dots,n\}$ and $j \notin I$, $[Q,D_{j,(1)}]\left(\prod_{i \in I}D_{i,(1)}\mathcal{X} \right)=0$ can be also satisfied.
            Therefore, $\prod_{i \in I \cup \{j\}}D_{i,(1)}\mathcal{X}$ is closed too. In particular, because of the commutation relation $[D_{i,(1)},D_{j,(1)}]=0$, it is not necessary for us to worry about the order of $D_{i,(1)}$. 
            \item Similarly, we can construct the closed operators $\prod_{i \in I}D_{i,(1)}\mathcal{Y}$ for $I \subset \{1,\dots,n\}$ recursively.
            \item $T_i \coloneqq \frac{1}{2}D_{i,(1)} \mathcal{A}_i$ are closed operators.
            since they are obtained from the calculation of $\left[Q,D_{i,(1)} \right] \mathcal{A}_{i}$. We consider that these operators are related with the F-term condition ``$q_1 \tilde{q}_1 = \cdots = q_n \tilde{q}_n$''
            and Higgs branch chiral operator ``$u= -\frac{1}{n} \sum_{i=1}^{n} q_i \tilde{q}_i$'' before deforming Higgs branch.
            Since $\mathcal{A}_i \not\equiv \mathcal{A}_j$ for $i \neq j$ and $Q \mathcal{A}_i = \partial c_i -\partial c_{i-1}$, where these BRST transformations are similar to the transformations of gauge fields, hold respectively under the BRST cohomology, 
            we can well construct $T_i$ as ``field strength'' of $\mathcal{A}_i$.
            \item Similarly, $W_i \coloneqq \frac{1}{3}\sqrt{\frac{2}{3}} \frac{5 \mp \sqrt{15}i}{8} \left(\mathcal{D}_i +\frac{1 \pm \sqrt{-15}}{2}\hat{\mathcal{A}}_i\right) \left\{\left(\mathcal{D}_i +\frac{5 \mp \sqrt{-15}}{10}\hat{\mathcal{A}}_{i} \right) \mathcal{A}_{i} \right\}$ are closed operators.\footnote{
            In \cite{Nishinaka:2025nbe}, $(T_i,W_i)$ for $i \in \{1,\dots,n\}$ are $W_3$-algebras with the center charge $c=-2$ independently by (\ref{T-1}), (\ref{T-2}) and (\ref{W-1}). 
            We can also expand $T_i$ and $W_i$ specifically as follows:
            \begin{align}
                \left\{
                    \begin{aligned}
                        T_i &= \frac{1}{2} X_i X_i Y_i Y_i +X_i \partial Y_i - \partial X_i Y_i \\
                        W_i &=-\sqrt{\frac{3}{2}} X_i X_i Y_i \partial Y_i +\sqrt{\frac{3}{2}} X_i \partial X_i Y_i Y_i 
                        +2\sqrt{\frac{2}{3}}\partial X_i \partial Y_i \\
                        &\quad -\frac{X_i \partial^2 Y_i}{\sqrt{6}}-\frac{\partial^2 X_i Y_i}{\sqrt{6}}-\frac{1}{3} \sqrt{\frac{2}{3}} X_i X_i X_i Y_i Y_i Y_i
                    \end{aligned}
                \right.~,
            \end{align}
            where these were known in \cite{10.1093/imrn/rnaa031, LINSHAW2009632}.}
            This definition, however, is very inconvenient and awkward. Therefore, we would like to define this operator in the equivalent representation. For example, we can derive the following representations using $\mathcal{D}_i$, $\hat{\mathcal{A}}_i$;
            \begin{align}
                \label{definition of W_i}
                W_{i} &= \frac{1}{3}\sqrt{\frac{2}{3}} \left\{\mathcal{D}_i \left(\mathcal{D}_i \mathcal{A}_i\right) +\frac{1}{2}\mathcal{D}_{i} \left(\hat{\mathcal{A}}_{i} \mathcal{A}_{i}\right) +\frac{1}{2}\hat{\mathcal{A}}_i \left(\mathcal{D}_i\mathcal{A}_i\right) +\hat{\mathcal{A}}_{i} \left(\hat{\mathcal{A}}_{i} \mathcal{A}_{i}\right) \right\}~,
            \end{align}
            or
            \begin{align}
                \label{definition of W_i ver.2}
                W_{i} &= \frac{1}{3}\sqrt{\frac{2}{3}} \left\{\frac{1}{2}\mathcal{D}_i \left(\mathcal{D}_i \mathcal{A}_i\right) +\frac{3}{2}\hat{\mathcal{A}}_{i} \left(\mathcal{D}_i \mathcal{A}_i\right) +\hat{\mathcal{A}}_{i} \left(\hat{\mathcal{A}}_{i} \mathcal{A}_{i}\right)\right\}.
            \end{align}
        \end{itemize}
        
        Therefore, we can define non-trivial closed operators by using $\mathcal{D}_i$ and $\mathcal{A}_i$.
        In particular, these closed operators except for $W_i$ can be defined using $D_{i,(1)}$. So, let us set $D_i \coloneqq D_{i,(1)} = \mathcal{D}_i + \mathcal{A}_i$.
        By using $D_i$, we can express the following candidate operators.
        \begin{align}
            T_i &\coloneqq \frac{1}{2} D_i \mathcal{A}_i~,\\
            \mathcal{X}_{I} &\coloneqq \prod_{i \in I} D_i \mathcal{X}~,\\
            \mathcal{Y}_{I} &\coloneqq \prod_{i \in I} D_i \mathcal{Y}~.
        \end{align}
        where $i \in \{1,\dots,n\}$ and $I \subset \{1,\dots,n\}$.

        Thus, from $\mathcal{D}_i$ and $\hat{\mathcal{A}}_i$ which can be found by deforming the stress tensor $T$ under the BRST cohomology, we can construct non-trivial closed operators. 
        Then, we conjecture that $U$, $\mathcal{X}$ and $\mathcal{Y}$ derived from the Higgs branch chiral operators, and part of non-trivial operators can generate the bosonic VOA with the BRST cohomology.
        The OPEs of these operators are given in Appendix.\ref{OPE 1}.
\section{How to construct primary fields from \texorpdfstring{$T_i, \mathcal{X}_{I},\mathcal{Y}_{I}$}{}}
    \label{How to construct primary fields from T i, mathcal X I, mathcal Y I}
    In Appendix.\ref{OPE 1}, from the OPEs with the stress tensor $T$, we can find out that $U$, $W_i$, $\mathcal{X}$ and $\mathcal{Y}$ are primary fields, 
    and the other fields $T_i$, $\mathcal{X}_{I}$ and $\mathcal{Y}_{I}$ for $I \neq \phi$ are not.
    We consider that generators in generic cases can be composed of the stress tensor $T$ and primary fields. \footnote{In \cite{Nishinaka:2025nbe}, since the following equalities can hold for $n=2$
    \begin{align*}
        \left\{
            \begin{aligned}
                T_1 + T_2 &=2 U^2 -\frac{1}{2}\mathcal{X} \mathcal{Y} -\frac{1}{2}\mathcal{Y} \mathcal{X}\\
                W_i &= (-1)^{i+1}\frac{1}{3}\sqrt{\frac{2}{3}} \left\{2 U\left(T_1 -T_2\right) -\mathcal{Y} \left(D_1-D_2\right)\mathcal{X} +\frac{1}{4}\partial \left(T_1 -T_2\right) \right\} \\
                &\quad +\sqrt{\frac{2}{3}} \left\{-U \partial U +\frac{1}{4} \mathcal{X} \partial \mathcal{Y}  -\frac{1}{4}\mathcal{Y} \partial \mathcal{X} + \frac{4}{3} U^3 -U \mathcal{X} \mathcal{Y} +\partial^2 U\right\}
            \end{aligned}
        \right.~,
    \end{align*}
    $T$, $W_1$, and $W_2$ are no longer generators.

    Similarly, since we can derive the following expression for $n=3$
    \begin{align*}
        W_1 + W_2 + W_3 = \sqrt{\frac{2}{3}}\left(3UT + 3U\partial U
        -\mathcal{X}\mathcal{Y} -\frac{1}{2}\partial T - \frac{1}{2}\partial^2
        U\right)~,
    \end{align*}
    $W_1 +W_2 +W_3$ is no longer a generator, but $W_1 -W_2$ and $W_2 -W_3$ are generators.}
    In this section, we try to construct primary fields from some summations of $T_i$, $\mathcal{X}_{I}$ and $\mathcal{Y}_{I}$ for $I \neq \phi$.
    \subsection{\texorpdfstring{$T_i$}{}}
        $T_i$ are $n$ linearly independent and quasi primary operators. 
        Since it is natural for us to select the following summation $\sum_{i=1}^{n}T_i$ by (\ref{definition of T}), 
        the others should be 
        \begin{align}
            \label{primary from T_i}
            T_1 -T_2~,\dots,\ T_{n-1}-T_{n}~,
        \end{align}
        where the operators (\ref{primary from T_i}) are primary.
        Then, the above construction (\ref{primary from T_i}) can be realized by the following transformation
        \begin{align}
            \label{complete set T_i}
            \begin{bmatrix}
                1 & 1 & 1 & \cdots & 1 & 1 \\
                1 & -1 & 0 & \cdots & 0 & 0 \\
                0 & 1 & -1 & \cdots & 0 & 0 \\
                \vdots & \vdots & \vdots & & \vdots & \vdots\\
                0 & 0 & 0 & \cdots & 1 & -1 
            \end{bmatrix}
            \begin{bmatrix}
                T_1 \\
                T_2 \\
                \vdots \\
                T_{n-1} \\
                T_n
            \end{bmatrix}
            =
            \begin{bmatrix}
                \sum_{i=1}^{n}T_i \\
                T_1 -T_2\\
                T_2 -T_3 \\
                \vdots \\
                T_{n-1} -T_{n}
            \end{bmatrix}
            \equiv 
            \begin{bmatrix}
                T +\frac{n}{2}U^2 \\
                T_1 -T_2\\
                T_2 -T_3 \\
                \vdots \\
                T_{n-1} -T_{n}
            \end{bmatrix}~,
        \end{align}
        where the matrix on the left-hand side of (\ref{complete set T_i}) is invertible. Thus, we can construct the stress tensor $T$ and $n-1$ primary fields (\ref{primary from T_i}) by (\ref{complete set T_i}).
    \subsection{\texorpdfstring{$\mathcal{X}_{I}$}{} and \texorpdfstring{$\mathcal{Y}_{I}$}{}}
        Let us construct the primary fields from $\mathcal{X}_{I}$ and $\mathcal{Y}_{I}$ for $I \neq \phi$. Since (\ref{T XI}) and (\ref{T YI}),
        the case of $\mathcal{Y}_{I}$ is similar to the case of $\mathcal{X}_{I}$.
        Then, we examine primary fields composed of $\mathcal{X}_{I}$ in this subsection.

        From (\ref{T XI}), we can clearly find out how to construct primary fields;
        we remove the term of $z^{-3}$ from the OPE (\ref{T XI}).
        Then, we introduce the following formula.
        \begin{itemize}
            \item For an positive integer $k$ satisfying $2k \leq n$, we can construct the following primary field
            \begin{align}
                \label{primary from XI}
                \prod_{j=1}^{k}\left(D_{i_{2j-1}} -D_{i_{2j}}\right) \mathcal{X}~,
            \end{align}
            where $i_1,\dots,i_{2k}$ are distinct interger in $\{1,\dots,n\}$.
        \item For an positive integer $k$ satisfying $2k > n$, primary fields cannot be constructed by this formula (\ref{primary from XI}).
    \end{itemize}
    \begin{proof}
        The latter case is trivial. Then we will prove the former case.
        For an positive integer $m$ satisfying $k=m$ and $2m \leq n$, we show that the following field for $\{i_{(1,1)},i_{(1,2)},\dots,i_{(m,1)},i_{(m,2)}\} \subset \{1,\dots,n\}$ is primary;
        \begin{align}
            \prod_{c=1}^{m}\left(D_{i_{(c,1)}} -D_{i_{(c,2)}}\right)\mathcal{X} = \sum_{d_1,\dots,d_m \in \{1,2\}} (-1)^{\sum_{j=1}^{m} d_j -m} \mathcal{X}_{\{i_{(1,d_1)},\dots,i_{(m,d_m)}\}}~.
        \end{align}
        From the $(-3)$-th degree of the OPE between $T$ and the above operator, the following summation can be calculated.
        \begin{align*}
            &\quad \sum_{d_1,\dots,d_{m} \in \{1,2\}} (-1)^{\sum_{j=1}^{m} d_j -m} \mathcal{X}_{\{i_{(1,d_1)},\dots,i_{(m,d_{m})}\}} \notag \\
            &\mapsto 2 \sum_{d_1,\dots,d_{m} \in \{1,2\}} \sum_{l \in \{i_{(1,d_1)},\dots,i_{(m,d_{m})}\}}(-1)^{\sum_{j=1}^{m}d_{j}-(m)} \mathcal{X}_{\{i_{(1,d_1)},\dots,i_{(m,d_{m})}\}\setminus \{l\}} \notag \\
            &= 2\left(D_{i_{(m,1)}} -D_{i_{(m,2)}}\right)\sum_{d_1,\dots,d_{m-1} \in \{1,2\}} \sum_{l \in \{i_{(1,d_1)},\dots,i_{(m-1,d_{m-1})}\}}(-1)^{\sum_{j=1}^{m-1}d_{j}-(m-1)} \mathcal{X}_{\{i_{(1,d_1)},\dots,i_{(m-1,d_{m-1})}\}\setminus \{l\}} \notag \\
            &= 2 \prod_{c=2}^{m} \left(D_{i_{(c,1)}} -D_{i_{(c,2)}} \right) \sum_{d_1 \in \{1,2\}} \sum_{l \in \{i_{(1,d_1)}\}} (-1)^{d_1 -1} \mathcal{X}_{\{i_{(1,d_1)}\} \setminus {l}} \notag \\
            &= 2 \prod_{c=2}^{m} \left(D_{i_{(c,1)}} -D_{i_{(c,2)}} \right) \left(\mathcal{X} -\mathcal{X}\right) \notag \\
            &= 0~.
        \end{align*}

        Thus, for $2k \leq n$, the formulas (\ref{primary from XI}) are primary.
    \end{proof}

    Similarly, the above statement can be adapted for $\mathcal{Y}_{I}$. For $n \geq 2k$, the following operator composed of $\mathcal{Y}_{I}$ for $I \neq \phi$ is primary.
    \begin{align}
        \label{primary from YI}
        \prod_{j=1}^{k}\left(D_{i_{2j-1}} -D_{i_{2j}}\right) \mathcal{Y}~,
    \end{align}
    where $i_1, \dots,i_{2k}$ are distinct integer in $\{1,\dots,n\}$.

    However, we face three problems:
    \begin{itemize}
        \item The number of (\ref{primary from XI}) or (\ref{primary from YI}) are $(2k-1)!!\binom{n}{2k}$, and $(2k-1)!! \binom{n}{2k} > \binom{n}{k}$ holds for all pairs $(n,k)$, with only some exceptions.
        Then, how many linearly independent operators consisting of (\ref{primary from XI}) or (\ref{primary from YI}) exist?
        \item Can primary operators which are difference from linear sums of (\ref{primary from XI}) or (\ref{primary from YI}) exist?
        \item Is the sum of the number of linearly independent primary operators and non-primary operators regarding with $\mathcal{X}_{I}$ or $\mathcal{Y}_{I}$ equal to $\binom{n}{|I|}$?
    \end{itemize}
    The first will be solved in Sec.\ref{How to count independent operators}, and the others will be in Sec.\ref{Linearly independent operators consisting of X I and Y I}.
\section{How to count independent operators (\ref{primary from XI}) and (\ref{primary from YI})}
    \label{How to count independent operators}
    In this section, we study the number of the linearly independent operators consisting of (\ref{primary from XI}) or (\ref{primary from YI}).
    When $n$ is small, independent operators consisting of (\ref{primary from XI}) or (\ref{primary from YI}) can be explicit identified.
    For example, the linearly independent operator for $n=2$ is
    \begin{align*}
        (D_1 -D_2) \mathcal{X}~.
    \end{align*}
    For $n=3$, since $D_1 -D_3 = \left(D_1 -D_2\right) + \left(D_2 -D_3\right)$, we can choose the following linearly independent operators
    \begin{align*}
        (D_1-D_2)\mathcal{X}~,\quad (D_2 -D_3)\mathcal{X}~.
    \end{align*}
    However, when $n \geq 4$, it is more difficult for us to identify linearly independent operators composed of (\ref{primary from XI}) or (\ref{primary from YI}).
    This difficulty arises from the large number of degrees of freedom in the subsets of $\{1,\dots,n\}$. Thus, this problem should be solved by using combinatorics method rather than linear algebra's method.

    \subsection{Combinatorial construction}
        To solve this problem as a combinatorics problem, let us consider how to express $D_i-D_j$, $(D_i -D_j)(D_k -D_l)$, and so on appeared in (\ref{primary from XI}) or (\ref{primary from YI}).

        For example, the case of $D_i-D_j$ for $1 \leq i < j \leq n$ is drawn in Figure \ref{1 i j n}.
        \begin{figure}[t]
            \center
            \begin{tikzpicture}
                \draw (-1.2,0.2) node{$1$};
                \draw (-0.6,0.2) node{$\cdots$};
                \draw (0,0.2) node{$i$};
                \draw (0.6,0.2) node{$\cdots$};
                \draw (1.2,0.2) node{$j$};
                \draw (1.8,0.2) node{$\cdots$};
                \draw (2.4,0.2) node{$n$};
                \draw (0,0)--(0,-0.5)--(1.2,-0.5)--(1.2,0);
                \draw[<->] (2.8,0)--(4.1,0);
                \draw (5,0) node{$D_i -D_j$};
            \end{tikzpicture}
            \caption{The case of $D_i -D_j$ for $1 \leq i < j \leq n$.}
            \label{1 i j n}
        \end{figure}
        Similarly, the cases for $(D_i -D_j)(D_k -D_l)$ are also drawn in Figure.\ref{1 i j k l n}, \ref{1 i k j l n} and \ref{1 i k l j n}.
        \begin{figure}[t]
            \center
            \begin{tikzpicture}
                \draw (-1.2,0.2) node{$1$};
                \draw (-0.6,0.2) node{$\cdots$};
                \draw (0,0.2) node{$i$};
                \draw (0.6,0.2) node{$\cdots$};
                \draw (1.2,0.2) node{$j$};
                \draw (1.8,0.2) node{$\cdots$};
                \draw (2.4,0.2) node{$k$};
                \draw (3,0.2) node{$\cdots$};
                \draw (3.6,0.2) node{$l$};
                \draw (4.2,0.2) node{$\cdots$};
                \draw (4.8,0.2) node{$n$};
                \draw (0,0)--(0,-0.5)--(1.2,-0.5)--(1.2,0);
                \draw (2.4,0)--(2.4,-0.5)--(3.6,-0.5)--(3.6,0);
                \draw (8.7,0) node{$(D_i -D_j)(D_k -D_l)$};
                \draw[<->] (5.2,0)--(6.5,0);
            \end{tikzpicture}
            \caption{The case of $(D_i -D_j)(D_k -D_l)$ for $1 \leq i < j <k < l \leq n$.}
        \label{1 i j k l n}
        \end{figure}
        \begin{figure}[t]
            \center
            \begin{tikzpicture}
                \draw (-1.2,0.2) node{$1$};
                \draw (-0.6,0.2) node{$\cdots$};
                \draw (0,0.2) node{$i$};
                \draw (0.6,0.2) node{$\cdots$};
                \draw (1.2,0.2) node{$k$};
                \draw (1.8,0.2) node{$\cdots$};
                \draw (2.4,0.2) node{$j$};
                \draw (3,0.2) node{$\cdots$};
                \draw (3.6,0.2) node{$l$};
                \draw (4.2,0.2) node{$\cdots$};
                \draw (4.8,0.2) node{$n$};
                \draw (0,0)--(0,-0.5)--(1.1,-0.5);
                \draw (1.3,-0.5)--(2.4,-0.5)--(2.4,0);
                \draw (1.2,0)--(1.2,-1)--(3.6,-1)--(3.6,0);
                \draw (8.7,0) node{$(D_i -D_j)(D_k -D_l)$};
                \draw[<->] (5.2,0)--(6.5,0);
            \end{tikzpicture}
            \caption{The case of $(D_i -D_j)(D_k -D_l)$ for $1 \leq i < k < j < l \leq n$.}
        \label{1 i k j l n}
        \end{figure}
        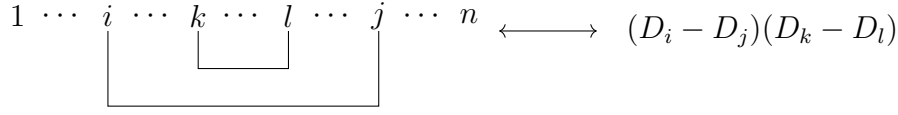
\begin{figure}[t]
            \center
            \begin{tikzpicture}
                \draw (-1.2,0.2) node{$1$};
                \draw (-0.6,0.2) node{$\cdots$};
                \draw (0,0.2) node{$i$};
                \draw (0.6,0.2) node{$\cdots$};
                \draw (1.2,0.2) node{$k$};
                \draw (1.8,0.2) node{$\cdots$};
                \draw (2.4,0.2) node{$l$};
                \draw (3,0.2) node{$\cdots$};
                \draw (3.6,0.2) node{$j$};
                \draw (4.2,0.2) node{$\cdots$};
                \draw (4.8,0.2) node{$n$};
                \draw (0,0)--(0,-1)--(3.6,-1)--(3.6,0);
                \draw (1.2,0)--(1.2,-0.5)--(2.4,-0.5)--(2.4,0);
                \draw (8.7,0) node{$(D_i -D_j)(D_k -D_l)$};
                \draw[<->] (5.2,0)--(6.5,0);
            \end{tikzpicture}
            \caption{The case of $(D_i -D_j)(D_k -D_l)$ for $1 \leq i < k < l < j \leq n$.}
        \label{1 i k l j n}
        \end{figure}
        More generally, we can construct the above similar relations for $(D_{i_1} -D_{i_2}) \dots (D_{i_{2k-1}} -D_{i_{2k}})$ when $k$ is a positive integer satisfying $n \geq 2k$.

        In order to count the number of linearly independent operators such as (\ref{primary from XI}) and (\ref{primary from YI}),
        we need to use the following properties for $D_i$.
        \begin{enumerate}
            \item \label{no blank} When $1 \leq i < j \leq n$,
            \begin{align*}
                D_i -D_j = \sum_{k=0}^{j-i-1}\left(D_{i+k} -D_{i+k+1}\right)~.
            \end{align*}
            Similarly, when $1 \leq i <j < k < l \leq n$,
            \begin{align*}
                &\quad \left(D_{j} -D_{k}\right)\left(D_{i} -D_{l}\right) \\
                &= \delta_{j \geq i+2}\sum_{m=i}^{j-2}\left(D_{j} -D_{k}\right)\left(D_{m} -D_{m+1}\right) 
                +\delta_{j+1 < k}\left(D_{j} -D_{k}\right)\left(D_{j-1} -D_{j+1} \right) \\
                &\quad +\delta_{k \geq j+3} \sum_{m=j+1}^{k-2} \left(D_{j} -D_{k}\right)\left(D_{m} -D_{m+1}\right) 
                +\delta_{j+1<k} \left(D_{j} -D_{k}\right)\left(D_{k-1} -D_{k+1}\right) \\
                &\quad +\delta_{l \geq k+2} \sum_{m=k+1}^{l-1} \left(D_{j} -D_{k}\right) \left(D_{m} -D_{m+1}\right)~.
            \end{align*}
            In other words, when connecting $i$ to $j$ satisfying $1 \leq i < j \leq n$ and $i+1 < j$, we permit conditions that $k$ satisfying $i+1 \leq k \leq j-1$ must pair with $l$ satisfying $i+1 \leq l \leq j-1$ and $l \neq k$. 
            For example, we can permit the case of Figure \ref{permit}.
            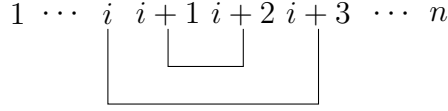
\begin{figure}[t]
                \center
                \begin{tikzpicture}
                    \draw (-1.2,0.2) node{$1$};
                    \draw (-0.6,0.2) node{$\cdots$};
                    \draw (0,0.2) node{$i$};
                    \draw (0.8,0.2) node{$i+1$};
                    \draw (1.8,0.2) node{$i+2$};
                    \draw (2.8,0.2) node{$i+3$};
                    \draw (3.8,0.2) node{$\cdots$};
                    \draw (4.4,0.2) node{$n$};
                    \draw (0,0)--(0,-1)--(2.8,-1)--(2.8,0);
                    \draw (0.8,0)--(0.8,-0.5)--(1.8,-0.5)--(1.8,0);
                \end{tikzpicture}
                \caption{Permitted case.}
                \label{permit}
            \end{figure}
            But Figure \ref{ban 1} must not be permitted.
            \begin{figure}[t]
                \center
                \begin{tikzpicture}
                    \draw (-1.2,0.2) node{$1$};
                    \draw (-0.6,0.2) node{$\cdots$};
                    \draw (0,0.2) node{$i$};
                    \draw (0.8,0.2) node{$i+1$};
                    \draw (1.8,0.2) node{$i+2$};
                    \draw (2.8,0.2) node{$i+3$};
                    \draw (3.8,0.2) node{$\cdots$};
                    \draw (4.4,0.2) node{$n$};
                    \draw (0,0)--(0,-0.5)--(2.8,-0.5)--(2.8,0);
                \end{tikzpicture}
                \caption{Denied case 1.}
                \label{ban 1}
            \end{figure}
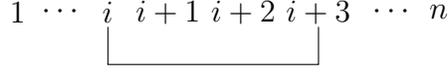
            \item \label{no intersection} Since the commutation relation for $D_i$ holds, the following equality holds for $1 \leq i < k < j < l \leq n$,
            \begin{align*}
                \left(D_{i} -D_{j}\right)\left(D_{k} -D_{l}\right) = \left(D_i -D_k\right)\left(D_j -D_l\right) +\left(D_i -D_l\right)\left(D_k-D_j\right)~.
            \end{align*}
            The lines in the figure cannot intersect each other. For example, the case of Figure \ref{ban 2} must not be permitted.
            \begin{figure}[t]
                \center
                \begin{tikzpicture}
                    \draw (-1.2,0.2) node{$1$};
                    \draw (-0.6,0.2) node{$\cdots$};
                    \draw (0,0.2) node{$i$};
                    \draw (0.6,0.2) node{$\cdots$};
                    \draw (1.2,0.2) node{$k$};
                    \draw (1.8,0.2) node{$\cdots$};
                    \draw (2.4,0.2) node{$j$};
                    \draw (3,0.2) node{$\cdots$};
                    \draw (3.6,0.2) node{$l$};
                    \draw (4.2,0.2) node{$\cdots$};
                    \draw (4.8,0.2) node{$n$};
                    \draw (0,0)--(0,-0.5)--(1.1,-0.5);
                    \draw (1.3,-0.5)--(2.4,-0.5)--(2.4,0);
                    \draw (1.2,0)--(1.2,-1)--(3.6,-1)--(3.6,0);
                \end{tikzpicture}
                \caption{Denied case 2.}
                \label{ban 2}
            \end{figure}
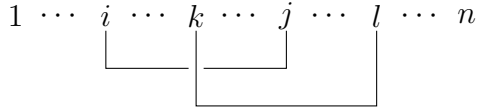
        \end{enumerate}
        Thus, under these conditions, counting the figures is correspondence to counting the independnt operators such as (\ref{primary from XI}) or (\ref{primary from YI}).

        Before solving this counting problem, let us set $n=2L+M$, where $L$ is equal to the number of pairs and $M$ is equal to the number of isolated elements.
        In particular, when $M=0$, the total number of the pairs satisfying the above conditions is equal to the Catalan number $C_L$ \cite{Stanley_Fomin_1999,Stanley_2015}:
        \begin{align}
            \label{Catalan number}
            C_L \coloneqq \frac{1}{L+1} \binom{2L}{L}~.
        \end{align}

        More generally, when $M \geq 1$, we consider the partition of integer $L$. When $L$ is divided into $p$ positive integers $i_1,\dots,i_p$, namely $L = \sum_{j=1}^{p} i_j$, 
        where $p$ is satisfying $1 \leq p \leq M+1$ and no greater or lesser relationship is among $i_1,\dots,i_p$. 
        For example, when $L=2$ and $M=1$, Figure.\ref{The cases for $p_1$ and $i_1=2$} and \ref{The case for $p=2$ and $(i_1,i_2)=(1,1)$} emerge. 
        \begin{figure}[t]
            \center
            \begin{tikzpicture}
                \draw (0,0.2) node{$1$};
                \draw (0.6,0.2) node{$2$};
                \draw (1.2,0.2) node{$3$};
                \draw (1.8,0.2) node{$4$};
                \draw (2.4,0.2) node{$5$};
                \draw (0,0)--(0,-0.3)--(0.6,-0.3)--(0.6,0);
                \draw (1.2,0)--(1.2,-0.3)--(1.8,-0.3)--(1.8,0);
                \draw (3.5,0.2) node{$1$};
                \draw (4.1,0.2) node{$2$};
                \draw (4.7,0.2) node{$3$};
                \draw (5.3,0.2) node{$4$};
                \draw (5.9,0.2) node{$5$};
                \draw (3.5,0)--(3.5,-0.6)--(5.3,-0.6)--(5.3,0);
                \draw (4.1,0)--(4.1,-0.3)--(4.7,-0.3)--(4.7,0);
                \draw (0,-1.8) node{$1$};
                \draw (0.6,-1.8) node{$2$};
                \draw (1.2,-1.8) node{$3$};
                \draw (1.8,-1.8) node{$4$};
                \draw (2.4,-1.8) node{$5$};
                \draw (0.6,-2)--(0.6,-2.3)--(1.2,-2.3)--(1.2,-2);
                \draw (1.8,-2)--(1.8,-2.3)--(2.4,-2.3)--(2.4,-2);
                \draw (3.5,-1.8) node{$1$};
                \draw (4.1,-1.8) node{$2$};
                \draw (4.7,-1.8) node{$3$};
                \draw (5.3,-1.8) node{$4$};
                \draw (5.9,-1.8) node{$5$};
                \draw (4.1,-2)--(4.1,-2.6)--(5.9,-2.6)--(5.9,-2);
                \draw (4.7,-2)--(4.7,-2.3)--(5.3,-2.3)--(5.3,-2);
            \end{tikzpicture}
            \caption{The cases for $p_1$ and $i_1=2$.}
            \label{The cases for $p_1$ and $i_1=2$}
        \end{figure}
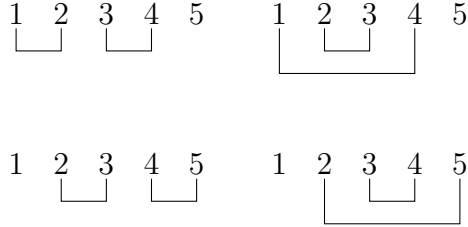
        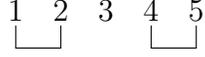
\begin{figure}[t]
            \center
            \begin{tikzpicture}
                \draw (0,0.2) node{$1$};
                \draw (0.6,0.2) node{$2$};
                \draw (1.2,0.2) node{$3$};
                \draw (1.8,0.2) node{$4$};
                \draw (2.4,0.2) node{$5$};
                \draw (0,0)--(0,-0.3)--(0.6,-0.3)--(0.6,0);
                \draw (1.8,0)--(1.8,-0.3)--(2.4,-0.3)--(2.4,0);
            \end{tikzpicture}
            \caption{The case for $p=2$ and $(i_1,i_2)=(1,1)$.}
            \label{The case for $p=2$ and $(i_1,i_2)=(1,1)$}
        \end{figure}
        Then,
        \begin{align*}
            2 \times C_2 +1 \times C_1^2 = 5~.
        \end{align*}

        To generalize this example, we will consider the following cases.
        \begin{itemize}
            \item When $L=0$, that is $n=M$, this case is simple, with only one instance. This situation is equivalent of $\mathcal{X}$ or $\mathcal{Y}$.
            \item When $M=0$, that is $n=2L$, the number of cases satisfying the condition \ref{no intersection} is equal to $C_L$.
            \item When $L,M \geq 1$, the number of cases satisfying condition.\ref{no blank} and \ref{no intersection} is equal to 
            \begin{align}
                \label{independent operators}
                \sum_{p=1}^{\min(L,M+1)} \binom{M+1}{p} g_{L,p}~,
            \end{align}
            where $g_{L,p}$ is defined by 
            \begin{align}
                \label{partition}
                g_{L,p} \coloneqq \sum \prod_{j=1}^{p} C_{i_j}~,
            \end{align}
            where the sum $\sum$ in (\ref{partition}) is taken over all partitions $(i_1,\dots,i_p)$ satisfying $L = \sum_{j=1}^{p}i_j,\ i_1 >0,\dots,i_p >0$.
        \end{itemize}  

        By summarizing these results, we can derive the number of independent operators consisting of (\ref{primary from XI}) or (\ref{primary from YI}) as follows;
        \begin{align}
            \label{the number of the important operators 1}
            \frac{n-2L+1}{n-L+1} \binom{n}{L}~,
        \end{align}
        where $M=n-2L$.
    \subsection{Solution}
        First, we solve the case for $M=0$, namely, derive the Catalan number (\ref{Catalan number}). 
        
        \begin{proof}
            In order to solve this case, we assume that we find out $C_1,\dots,C_L$, and construct the recurrence relation regarding with $C_{L+1}$. 
            When we set the $2(L+1)$ points, ``$1$'', ``$2$'',..., ``$2L+2$'', and can count how to connect under the condition \ref{no intersection} as follows.
            \begin{itemize}
                \item When connecting ``$1$'' and ``$2$'', the number of ways to connect all the rest from ``$3$'' to ``$2L+2$'' under the condition \ref{no intersection} is equal to $C_L$.
                \item When connecting ``$1$'' and ``$2L+2$'', we can also derive $C_L$.
                \item When connecting ``$1$'' and ``$2i$'' for $2 \leq i \leq L$, the original group is divided into two groups; one group is consisting of $2(i-1)$ points from ``$2$'' to ``$2i-1$'', and the other group is $2(L-i+1)$ points from ``$2i+1$'' to ``$2L+2$''. 
                Therefore, the number of ways to connect points is equal to $C_{i-1} C_{L-i+1}$. 
            \end{itemize}
            By summarizing the above situations, we can derive the following recurrence relation.
            \begin{align}
                \label{recurrence relation}
                C_{L+1} &= C_L + C_L + \sum_{k=2}^{L}C_{k-1}C_{L-k+1} \notag \\
                &= \sum_{i=0}^{L}C_i C_{L-i}~,
            \end{align}
            where we can set $C_0 = 1$ for convenience.

            Since $C_0 =1$ and (\ref{recurrence relation}), $C_L$ is correspondence to the Catalan number \cite{Stanley_Fomin_1999}
            \begin{align}
                C_L = \frac{1}{L+1} \binom{2L}{L}~.
            \end{align}
            This result is equal to (\ref{the number of the important operators 1}) for $M=0$.
        \end{proof}
                
        Next, in order to solve the statement for $L \geq 1$ and $M \geq 1$, we will prove two statements;
        \begin{enumerate}
            \item The following equality for $g_{L,p}$ holds.
            \begin{align}
                \label{other expression of g_{L,p}}
                g_{L,p} = \frac{p}{L}\binom{2L}{L-p}~.
            \end{align}
            \item (\ref{independent operators}) is equal to (\ref{the number of the important operators 1}).
        \end{enumerate}
        \begin{proof}
            In order to derive (\ref{other expression of g_{L,p}}), we define the generator function $g_p(x)$ by using the generating function of the Catalan number $f(x) \coloneqq \sum_{n=0}^{\infty}C_n x^n$ as follows:
            \[
                g_p(x) \coloneqq \left(f(x)-1\right)^p~.
            \]
            Since $C_0 = 1$ and $f(x)-1 = \sum_{n=1}^{\infty}C_n x^n$, we can set $g_p(x) = \sum_{n=p}^{\infty} g_{n,p}x^{n}$.

            From other viewpoint for $g_p(x)$, by using functions $A_p(x)$ and $B_p(x)$, we set $g_p(x) = A_p(x)f(x) - B_p(x)$.
            Here, we need the following equality regarding with $f(x)$
            \begin{align}
                \label{equality of fx}
                \left\{f(x)\right\}^2 = \frac{1}{x}f(x) - \frac{1}{x}\quad (x \neq 0)~.
            \end{align}
            \begin{enumerate}
                \item When $p=1$, we set $A_p(x) =1$ and $B_p(x)= 1$.
                \item When $p=2$, by using (\ref{equality of fx}), we can derive the following equality.
                \begin{align*}
                    g_2(x) &= \left(\frac{1}{x}-2\right)f(x) -\left(\frac{1}{x}-1\right)~.
                \end{align*}
                Therefore, we set $A_2(x) = \frac{1}{x} -2$ and $B_2(x) = \frac{1}{x}-1$.
                \item More generally, we consider the recurrent relation between $g_p(x)$ and $g_{p+1}(x)$. 
                \begin{align*}
                    g_{p+1}(x) &= \left\{\left(\frac{1}{x}-1\right)A_p(x) -B_p(x)\right\}f(x) -\left\{\frac{1}{x}A_p(x) -B_p(x)\right\}~.
                \end{align*}
                Then, we set $A_{p+1}(x) = \left(\frac{1}{x}-1\right)A_p(x) -B_p(x)$ and $B_{p+1}(x) = \frac{1}{x}A_p(x) -B_p(x)$.
            \end{enumerate}
            Thus, the sequences of functions $\{A_p(x)\}_{p \in \mathbb{Z}_{>0}}$ and $\{B_p(x)\}_{p \in \mathbb{Z}_{>0}}$ are satisfied with
            \begin{align}
                \label{recurrence relation for functions}
                \left\{
                \begin{aligned}
                    &A_1(x) = 1,\ B_1(x) = 1 \\
                    &A_2(x) = \frac{1}{x}-2,\ B_2(x) = \frac{1}{x}-1 \\
                    &A_{p+1}(x) = \left(\frac{1}{x}-1\right)A_p(x) -B_p(x)\\
                    &B_{p+1}(x) = \frac{1}{x}A_p(x) -B_p(x)
                \end{aligned}
                \right.~.
            \end{align}
            Since both $A_p(x)$ and $B_p(x)$ can be regarded as $p-1$ degree polynomials for $x^{-1}$, it is necessary for us to calculate the $L$-th order coefficients of $A_p(x)f(x)$, where $L \geq p$.

            We try to remove $B_p(x)$ from the recurrence relation regarding with $\{A_{p}(x)\}_{p \in \mathbb{Z}_{>0}}$. Adding $A_{p+2}(x) +A_{p+1}(x) = B_{p+2}(x)$ and $A_{p+1}(x) +A_p(x) = B_{p+1}(x)$ together, we obtain the recurrence relation for $A_p(x)$:
            \[
                A_{p+2}(x) = \left(\frac{1}{x}-2 \right) A_{p+1}(x) -A_p(x)~.
            \]
            Using these relations, let us derive the $L$-th order coefficient of $g_p(x)$.
            \begin{enumerate}
                \item When $p=1$, $A_1(x)f(x) = f(x)$. Then, we can derive 
                \begin{align*}
                    g_{L,1} = C_L = \frac{1}{L}\binom{2L}{L-1}\quad (L \geq 1)~.
                \end{align*}
                \item When $p=2$, $A_2(x)f(x) = \left(\frac{1}{x}-2\right)f(x)$. Then, we can calculate the $L$th order coefficient of $A_2(x) f(x)$ for $L \geq 2$;
                \begin{align*}
                    g_{L,2} &= C_{L+1} -2C_L \\
                    &= \frac{2}{L}\binom{2L}{L-2}~.
                \end{align*}
                \item Let us assume that for $p$ and $p+1$, the following relations hold.
                \begin{align*}
                    g_p(x) = \sum_{L=p}^{\infty}\frac{p}{L}\binom{2L}{L-p}x^L~,\quad g_{p+1}(x) = \sum_{L=p+1}^{\infty}\frac{p+1}{L}\binom{2L}{L-p-1}x^L~.
                \end{align*}
                Then, from the recurrece relations for $\{A_{p}(x)\}$, the $L$-th order coefficient of $g_{p+2}(x)$ for $L\ (\geq p+2)$ can be derived as follows;
                \begin{align*}
                    g_{L,p+2}&= \frac{p+1}{L+1}\binom{2L+2}{L-p} -\frac{2(p+1)}{L}\binom{2L}{L-p-1} -\frac{p}{L}\binom{2L}{L-p} \\
                    &= \frac{p+2}{L}\binom{2L}{L-p-2}~.
                \end{align*}
            \end{enumerate}
            Thus, for $L \geq p$, we can derive the $L$th order coefficient of $g_{p}(x)$:
            \[
                g_{L,p} = \frac{p}{L}\binom{2L}{L-p}~.
            \]
            This result is clearly equal to (\ref{other expression of g_{L,p}}).

            By using this result (\ref{other expression of g_{L,p}}), let us show the following equality:
            \[
                \sum_{p=1}^{\min\left(L,M+1\right)}\binom{M+1}{p}g_{L,p} = \frac{M+1}{L+M+1}\binom{2L+M}{L}~.
            \]
            Substituing $g_{L,p}$ for the left-hand side $g_{L,p}$, we get
            \begin{align*}
                    \sum_{p=1}^{\min\left(L,M+1\right)}\binom{M+1}{p}g_{L,p} 
                    &= \frac{M+1}{L} \left\{ \binom{2L}{L-1} + \delta_{L \geq 2} \sum_{p=1}^{\min \left(L-1,M\right)}\binom{2L}{L-p-1} \binom{M}{p} \right\}~.
            \end{align*}
            \begin{itemize}
                \item When $L=1$, because of $\min \left(1,M+1\right)=1$, we can easily execute the following calculation.
                \[
                    \left.\frac{M+1}{L} \left\{ \binom{2L}{L-1} + \delta_{L \geq 2} \sum_{p=1}^{\min \left(L-1,M\right)}\binom{2L}{L-p-1} \binom{M}{p} \right\}\right|_{L=1} = M+1~.
                \]
                For the right-hand side, we also get
                \[
                    \frac{M+1}{M+2} \binom{M+2}{1} = M+1~.
                \]
                These expressions are clearly equal.
                \item When $L \geq 2$, let us transform the right-hand side of the original relation into
                \begin{align*}
                    \frac{M+1}{L+M+1}\binom{2L+M}{L} = \frac{M+1}{L} \binom{2L+M}{L-1}~.
                \end{align*}
                Then, the following equality which we must prove is 
                \begin{align}
                    \label{combinatorial relation}
                    \sum_{p=1}^{\min \left(L-1,M\right)}\binom{2L}{L-p-1} \binom{M}{p} = \binom{2L+M}{L-1} -\binom{2L}{L-1}~.
                \end{align}
                The right-hand side of (\ref{combinatorial relation}) is equal to the $(L-1)$-th coefficients of $h(x) = (1+x)^{2L+M}-(1+x)^{2L}$.
                
                $h(x)$ can also be interpreted as follows:
                \begin{align*}
                        h(x) &= (1+x)^{2L}\left\{(1+x)^M -1\right\} \\
                        &= \sum_{i=0}^{2L} \sum_{j=1}^{M} \binom{2L}{i}\binom{M}{j} x^{i+j}~.
                \end{align*}
                Then, the $(L-1)$-th order coefficient of $h(x)$ is satisfied with
                \begin{align*}
                    \left.\sum_{i=0}^{2L} \sum_{j=1}^{M} \binom{2L}{i}\binom{M}{j} \right|_{i+j=L-1} 
                    = \left\{
                        \begin{aligned}    
                            &\sum_{p=1}^{M}\binom{2L}{L-p-1}\binom{M}{p} & &(L-1\geq M)\\
                            \\
                            &\sum_{p=1}^{L-1}\binom{2L}{L-p-1}\binom{M}{p} & &(L-1 < M)
                        \end{aligned}
                    \right.~.
                \end{align*}
                Therefore, we can derive the relation (\ref{combinatorial relation}).
            \end{itemize}
            Thus, the two statements are proved.
        \end{proof}

        In the special case for $L=0$ and $M \geq 2$, (\ref{the number of the important operators 1}) is equal to $1$. This result is correspondence to the case of $\mathcal{X}$ or $\mathcal{Y}$ for $n \geq 2$.
        In order to count independent operators such as (\ref{primary from XI}) or (\ref{primary from YI}), $L$ and $M$ is satisfying the following condition.
        \begin{align*}
            \left\{
                \begin{aligned}
                    &L \geq 0 \\
                    &M \geq 0 \\
                    &2L +M \geq 2
                \end{aligned}
            \right.~,
        \end{align*}
        where $2L+M$ is equal to $n$.

        From this result (\ref{the number of the important operators 1}), we find out the following inequality;
        \begin{align*}
            \frac{n -2L +1}{n -L+1} \binom{n}{L} \leq \binom{n}{L}~.
        \end{align*}
        Since the operators (\ref{primary from XI}) or (\ref{primary from YI}) alone for $L \geq 1$ are clearly insufficient, we will construct and add the missing operators from $\mathcal{X}_{I}$ or $\mathcal{Y}_{I}$ to satisfy linear independence in the next section.
\section{Linearly independent operators consisting of \texorpdfstring{$\mathcal{X}_{I}$}{} and \texorpdfstring{$\mathcal{Y}_{I}$}{}}
    \label{Linearly independent operators consisting of X I and Y I}
    In this section, we study the number of linearly independent operators consisting of $\mathcal{X}_{I}$ and $\mathcal{Y}_{I}$. In particular, (\ref{primary from XI}) or (\ref{primary from YI}) satisfying condition.\ref{no blank} and \ref{no intersection} can become generators in this bosonic VOA,
    and others consisting of $\mathcal{X}_{I}$ or $\mathcal{Y}_{I}$ cannot always become generators.
    In order to examine this, it is necessary to construct linearly independent operators regarding with $\mathcal{X}_{I}$ or $\mathcal{Y}_{I}$.
    Then, let us define the following operators $\sigma_{K}^{(l)}$.
    \begin{align}
        \label{symmetric polynomial}
        \left\{
            \begin{aligned}
            \sigma^{(0)}_{K} &= 1\\
            \sigma^{(1)}_{K} &= \sum_{i \in \{1,\dots,n\} \setminus K} D_i\\
            \sigma^{(2)}_{K} &= \sum_{\substack{1 \leq i < j \leq n \\ i,j \in \{1,\dots,n\} \setminus K}} D_i D_j\\
            &\vdots \\
            \sigma^{(n-|K|)}_K &= \prod_{i \in \{1,\dots,n\} \setminus K} D_i \\
            \sigma^{(l)}_K &= 0 &\quad (l < 0~,\quad n- |K| < l)
            \end{aligned}
        \right.~,
    \end{align}
    where we choose $K=\phi$ or $K = \left\{i_1,\dots,i_{2 L}\right\}$ satisfying $1 \leq L \leq \left[ \frac{n}{2} \right]$.
    Then, we can construct the following operators consisting of $\mathcal{X}_{I}$ or $\mathcal{Y}_{I}$
    \begin{align}
        \label{basis operators XI}
        \sigma_{\phi}^{|I|} \mathcal{X}~,\quad \prod_{j=1}^{L}\left(D_{i_{2j-1}} -D_{i_{2j}}\right) \sigma^{(|I|-L)}_{K} \mathcal{X}~,
    \end{align}
    and
    \begin{align}
        \label{basis operators YI}
        \sigma_{\phi}^{|I|}\mathcal{Y}~,\quad \prod_{i=1}^{L}\left(D_{i_{2j-1}} -D_{i_{2j}}\right) \sigma^{(|I|-L)}_{K} \mathcal{Y}~,
    \end{align}
    where $\prod_{j=1}^{L} \left(D_{i_{2j-1}} -D_{i_{2j}}\right)$ must satisfy condition.\ref{no blank} and \ref{no intersection}.
    Therefore, let us prove that the above operators are linearly independent with respect to $|I|$,
    and that (\ref{basis operators XI}) or (\ref{basis operators YI}) for $L < |I| \leq n-L$ are all composite operators.

    \begin{proof}
        Let us construct linear maps $F$ and $G$ from the OPEs regarding with the stress tensor $T$.
        $F$ is derived from the $(-3)$-th degree of (\ref{TX}), (\ref{T XI}), (\ref{TY}) and (\ref{T YI}) as follows.
        \begin{align}
            F: \mathcal{X}_{I} \mapsto \left\{
                \begin{aligned}
                    &2 \sum_{i \in I} \mathcal{X}_{I \setminus \{i\}} &\quad (I \neq \phi)\\
                    &0 &\quad (I = \phi)
                \end{aligned}
            \right.~,
        \end{align}
        and
        \begin{align}
            F: \mathcal{Y}_{I} \mapsto \left\{
                \begin{aligned}
                    &-2 \sum_{i \in I} \mathcal{Y}_{I \setminus \{i\}} &\quad (I \neq \phi) \\
                    &0 &\quad (I = \phi)
                \end{aligned}
            \right.~.
        \end{align}
        Then, we can derive the following actions;
        \begin{align}
            \label{I to I-1 for XI}
            \left\{
                \begin{aligned}
                    \sigma_{\phi}^{\left(|I|\right)} \mathcal{X} &\overset{F}{\mapsto} 2(n-|I|+1) \sigma_{\phi}^{\left(|I|-1\right)} \mathcal{X} \\
                    \prod_{j=1}^{L} \left(D_{i_{2j-1}} -D_{i_{2j}} \right) \sigma_{K}^{\left(|I| -L\right)} \mathcal{X} &\overset{F}{\mapsto} 2(n-|I|+L+1) \prod_{j=1}^{L} \left(D_{i_{2j-1}} -D_{i_{2j}} \right) \sigma_{K}^{\left(|I| -L -1\right)} \mathcal{X}
                \end{aligned}
            \right.~,
        \end{align}
        and
        \begin{align}
            \label{I to I-1 for YI}
            \left\{
                \begin{aligned}
                    \sigma_{\phi}^{\left(|I|\right)} \mathcal{Y} &\overset{F}{\mapsto} -2(n-|I|+1) \sigma_{\phi}^{\left(|I|-1\right)} \mathcal{Y} \\
                    \prod_{j=1}^{L} \left(D_{i_{2j-1}} -D_{i_{2j}} \right) \sigma_{K}^{\left(|I| -L\right)} \mathcal{Y} &\overset{F}{\mapsto} -2 (n-|I|+L+1) \prod_{j=1}^{L} \left(D_{i_{2j-1}} -D_{i_{2j}} \right) \sigma_{K}^{\left(|I| -L -1\right)} \mathcal{Y}
                \end{aligned}
            \right.~,
        \end{align}
        where $I \subset \{1,\dots,n\}$ and $K= \left\{i_1,\dots,i_{2L}\right\}$ satisfying $1 \leq L \leq \min \left(|I|,\left[ \frac{n}{2} \right]\right)$.
        In particular, when $n-|I|+L \geq 0$ holds, the kernel of this linear map $F$ come from $\sigma_{\phi}^{\left(|I|-1\right)}=0 $ or $\sigma_{K}^{\left(|I|-L-1\right)} = 0$, except for the trivial case.
        Since these conditions are equivalent of $|I|=L$, these operators are clearly correspondence to $\mathcal{X}$, $\mathcal{Y}$, (\ref{primary from XI}) and (\ref{primary from YI}).

        The other map $G$ is derived from the $(-1)$-th degree of the OPEs regarding with $T$ or the equivalent formula $\sum_{i=1}^{n} T_i + T_U$.\footnote{We can set $T_U \coloneqq -\frac{n}{2} U^2$ in this theory.}
        The former comes from (\ref{TX}), (\ref{T XI}), (\ref{TY}) and (\ref{T YI}):
        \begin{align*}
            T(z) \mathcal{X}_{I}(0) &\sim \cdots + \frac{\partial \mathcal{X}_{I}}{z}~,\\
            T(z) \mathcal{Y}_{I}(0) &\sim \cdots + \frac{\partial \mathcal{Y}_{I}}{z}~.
        \end{align*}
        The latter comes from (\ref{T-3}), (\ref{T-4}), (\ref{T-5}), (\ref{T-6}) and the OPEs regarding with $T_U$:
        \begin{align}
            T_i(z) \mathcal{X}_{I}(0) &\sim \cdots + \frac{2 T_i \mathcal{X}_{I \setminus \{i\}}}{z} &\quad (i \in I)~,\notag \\
            T_i(z) \mathcal{X}_{I}(0) &\sim \cdots + \frac{\mathcal{X}_{I \cup \{i\}}}{z} &\quad (i \notin I)~,\notag \\
            T_i(z) \mathcal{Y}_{I}(0) &\sim \cdots + \frac{-2 T_i \mathcal{Y}_{I \setminus \{i\}}}{z} &\quad (i \in I)~,\notag \\
            T_i(z) \mathcal{Y}_{I}(0) &\sim \cdots + \frac{-\mathcal{Y}_{I \cup \{i\}}}{z} &\quad (i \notin I)~,\notag
        \end{align}
        and
        \begin{align}
            T_U(z) \mathcal{X}_{I}(0) &\sim -\frac{n \mathcal{X}_{I}}{2z^2} -\frac{n U \mathcal{X}_{I}}{z}~,\\
            T_U(z) \mathcal{Y}_{I}(0) &\sim -\frac{n \mathcal{Y}_{I}}{2z^2} + \frac{n U \mathcal{Y}_{I}}{z}~.
        \end{align}
        Therefore, we can construct the following equivalences from the $(-1)$-th degree of the above OPEs.
        \begin{align}
            \label{composite XI 1}
            \sum_{i \in \{1,\dots,n\} \setminus I} \mathcal{X}_{I \cup \{i\}} &\equiv 
            \left\{
                \begin{aligned}
                    &\partial \mathcal{X} + n U \mathcal{X} &\quad (I = \phi) \\
                    &\partial \mathcal{X}_{I} -2 \sum_{i \in I} T_i \mathcal{X}_{I \setminus \{i\}} + n U \mathcal{X}_{I} &\quad (I \neq \phi, \{1,\dots,n\})
                \end{aligned}
            \right.~,\\
            \label{composite XI 2}
            0 &\equiv \partial \mathcal{X}_{I} -2 \sum_{i \in I} T_i \mathcal{X}_{I \setminus \{i\}} +n U \mathcal{X}_{I} \quad (I = \{1,\dots,n\})~, \\
            \label{composite YI 1}
            \sum_{i \in \{1,\dots,n\} \setminus I} \mathcal{Y}_{I \cup \{i\}} &\equiv 
            \left\{
                \begin{aligned}
                    &-\partial \mathcal{Y} + n U \mathcal{Y} &\quad (I = \phi) \\
                    &-\partial \mathcal{Y}_{I} -2 \sum_{i \in I} T_i \mathcal{Y}_{I \setminus \{i\}} + n U \mathcal{Y}_{I} &\quad (I \neq \phi, \{1,\dots,n\})
                \end{aligned}
            \right.~, \\
            \label{composite YI 2}
            0 &\equiv -\partial \mathcal{Y}_{I} -2 \sum_{i \in I} T_i \mathcal{Y}_{I \setminus \{i\}} +n U \mathcal{Y}_{I} \quad (I = \{1,\dots,n\})~.
        \end{align}
        Thus, $G$ can be defined by using (\ref{composite XI 1}) and (\ref{composite XI 2}), or (\ref{composite YI 1}) and (\ref{composite YI 2}) as follows.
        \begin{align*}
            G: \mathcal{X}_{I} \mapsto \left\{
                \begin{aligned}
                    &\sum_{i \in \{1,\dots,n\} \setminus I} \mathcal{X}_{I \cup \{i\}} &\quad (I \neq \{1,\dots,n\}) \\
                    &0 &\quad (I = \{1,\dots,n\})
                \end{aligned}
            \right.~,
        \end{align*}
        and
        \begin{align*}
            G: \mathcal{Y}_{I} \mapsto \left\{
                \begin{aligned}
                    &\sum_{i \in \{1,\dots,n\} \setminus I} \mathcal{Y}_{I \cup \{i\}} &\quad (I \neq \{1,\dots,n\}) \\
                    &0 &\quad (I = \{1,\dots,n\})
                \end{aligned}
            \right.~,
        \end{align*}
        repectively.
        Then, we can derive the following actions;
        \begin{align}
            \label{I to I+1 for XI}
            \left\{
                \begin{aligned}
                    \sigma_{\phi}^{\left(|I|\right)} \mathcal{X} &\overset{G}{\mapsto} \left(|I|+1\right) \sigma_{\phi}^{\left(|I|+1\right)} \mathcal{X} \\
                    \prod_{j=1}^{L} \left(D_{i_{2j-1}} -D_{i_{2j}} \right) \sigma_{K}^{\left(|I| -L\right)} \mathcal{X} &\overset{G}{\mapsto} \left(|I|-L+1\right) \prod_{j=1}^{L} \left(D_{i_{2j-1}} -D_{i_{2j}} \right) \sigma_{K}^{\left(|I| -L +1\right)} \mathcal{X}
                \end{aligned}
            \right.~,
        \end{align}
        and
        \begin{align}
            \label{I to I+1 for YI}
            \left\{
                \begin{aligned}
                    \sigma_{\phi}^{\left(|I|\right)} \mathcal{Y} &\overset{G}{\mapsto} \left(|I|+1\right) \sigma_{\phi}^{\left(|I|+1\right)} \mathcal{Y} \\
                    \prod_{j=1}^{L} \left(D_{i_{2j-1}} -D_{i_{2j}} \right) \sigma_{K}^{\left(|I| -L\right)} \mathcal{Y} &\overset{G}{\mapsto} (|I|-L+1) \prod_{j=1}^{L} \left(D_{i_{2j-1}} -D_{i_{2j}} \right) \sigma_{K}^{\left(|I| -L +1\right)} \mathcal{Y}
                \end{aligned}
            \right.~,
        \end{align}
        where $I \subset \{1,\dots,n\}$ and $K = \{i_1,\dots,i_{2L}\}$ satisfying $1 \leq L \leq \min \left(|I|, \left[ \frac{n}{2} \right]\right)$. 
        In particular, when $|I| -L\geq 0$ holds for $0 \leq L \leq \min \left(|I|, \left[ \frac{n}{2}\right]\right)$, the kernel of this linear map $G$ partially comes from $\sigma_{\phi}^{\left(|I|+1\right)} = 0$ or $\sigma_{K}^{\left(|I|-L+1\right)} = 0$,
        except for the trivial case.
        These conditions are equivalent of $|I|= n -L$.

        The total number of $\mathcal{X}_{I}$ (resp. $\mathcal{Y}_{I}$) is $\binom{n}{|I|}$ with respect to $|I|$. On the other hand,
        the number of (\ref{basis operators XI}) (resp. (\ref{basis operators YI})) is equal to 
        \begin{align}
            \frac{n-2L+1}{n-L+1} \binom{n}{L}~,
        \end{align}
        with respect to $L$ for $0 \leq L \leq \min \left(|I|,\left[\frac{n}{2}\right] \right)$. Since the following equality for $L \geq 1$ holds 
        \begin{align}
            \frac{n-2L+1}{n-L+1} \binom{n}{L} = \binom{n}{L} - \binom{n}{L-1}~,
        \end{align}
        the following summation with regarding $L$ can derive
        \begin{align}
            \binom{n}{|I|} =
            \left\{
                \begin{aligned}
                    &\sum_{L=0}^{|I|} \frac{n-2L+1}{n-L+1} \binom{n}{L} &\quad \left(0 \leq |I| \leq \left[ \frac{n}{2} \right]\right)\\
                    &\sum_{L=0}^{n-|I|} \frac{n -2L +1}{n-L+1} \binom{n}{L} &\quad \left(\left[ \frac{n}{2} \right]+1 \leq |I| \leq n \right)
                \end{aligned}
            \right.~,
        \end{align}
        with respect to $|I|$.
        Therefore, it satisfies requirement that (\ref{basis operators XI}) or (\ref{basis operators YI}) are basis operators.

        When $0 \leq |I| < \left[ \frac{n}{2} \right]$, the linear map $F \circ G$ acts as follows.
        \begin{align}
            \left\{
                \begin{aligned}
                    \sigma_{\phi}^{\left(|I| \right)} \mathcal{X} &\overset{F \circ G}{\mapsto} 2\left(|I|+1\right) \left(n-|I|\right) \sigma_{\phi}^{\left(|I| \right)}\mathcal{X} \neq 0 \\
                    \prod_{j=1}^{L} \left(D_{i_{2j-1}} -D_{i_{2j}}\right) \sigma_{K}^{\left(|I|-L\right)} \mathcal{X} &\overset{F \circ G}{\mapsto} 2\left(|I|-L+1\right) \left(n-|I|+L\right) \prod_{j=1}^{L} \left(D_{i_{2j-1}} -D_{i_{2j}}\right) \sigma_{K}^{\left(|I|-L\right)} \mathcal{X} \neq 0
                \end{aligned}
            \right.~,
        \end{align}
        and
        \begin{align}
            \left\{
                \begin{aligned}
                    \sigma_{\phi}^{\left(|I| \right)} \mathcal{Y} &\overset{F \circ G}{\mapsto} -2\left(|I|+1\right) \left(n-|I|\right) \sigma_{\phi}^{\left(|I| \right)}\mathcal{Y} \neq 0\\
                    \prod_{j=1}^{L} \left(D_{i_{2j-1}} -D_{i_{2j}}\right) \sigma_{K}^{\left(|I|-L\right)} \mathcal{Y} &\overset{F \circ G}{\mapsto} -2\left(|I|-L+1\right) \left(n-|I|+L\right) \prod_{j=1}^{L} \left(D_{i_{2j-1}} -D_{i_{2j}}\right) \sigma_{K}^{\left(|I|-L\right)} \mathcal{Y} \neq 0
                \end{aligned}
            \right.~,
        \end{align}
        where $I \subset \{1,\dots,n\}$ and $K=\left\{i_1,\dots,i_{2L}\right\}$ satisfying $1 \leq L \leq |I|$.
        Since these actions are equivalent to multiplying by a non-zero constant, $F \circ G$ is isomorphic for $0 \leq |I| < \left[ \frac{n}{2} \right]$.
        In other words, the kernel of $F \circ G$ must be trivial and the image of $G$ cannot be primary.
        Thus, when (\ref{basis operators XI}) or (\ref{basis operators YI}) are basis operators for $0 \leq |I| < \left[ \frac{n}{2} \right]$,
        (\ref{basis operators XI}) or (\ref{basis operators YI}) for $|I|+1$ are also basis.
        Since $\mathcal{X}$ or $\mathcal{Y}$ are always basis, (\ref{basis operators XI}) or (\ref{basis operators YI}) for $1 \leq |I| \leq \left[\frac{n}{2}\right]$ are also basis.

        Similarly, when $\left[ \frac{n}{2} \right] < |I| \leq n$, $G \circ F$ is also isomorphic as follows.
        \begin{align}
            \left\{
                \begin{aligned}
                    \sigma_{\phi}^{\left(|I| \right)} \mathcal{X} &\overset{G \circ F}{\mapsto} 2|I|\left(n-|I|+1\right) \sigma_{\phi}^{\left(|I| \right)}\mathcal{X} \neq 0 \\
                    \prod_{j=1}^{L} \left(D_{i_{2j-1}} -D_{i_{2j}}\right) \sigma_{K}^{\left(|I|-L\right)} \mathcal{X} &\overset{G \circ F}{\mapsto} 2\left(|I|-L\right) \left(n-|I|+L+1\right) \prod_{j=1}^{L} \left(D_{i_{2j-1}} -D_{i_{2j}}\right) \sigma_{K}^{\left(|I|-L\right)} \mathcal{X} \neq 0
                \end{aligned}
            \right.~,
        \end{align}
        and
        \begin{align}
            \left\{
                \begin{aligned}
                    \sigma_{\phi}^{\left(|I| \right)} \mathcal{Y} &\overset{G \circ F}{\mapsto} -2 |I|\left(n-|I|+1\right) \sigma_{\phi}^{\left(|I| \right)}\mathcal{Y} \neq 0\\
                    \prod_{j=1}^{L} \left(D_{i_{2j-1}} -D_{i_{2j}}\right) \sigma_{K}^{\left(|I|-L\right)} \mathcal{Y} &\overset{G \circ F}{\mapsto} -2\left(|I|-L\right) \left(n-|I|+L+1\right) \prod_{j=1}^{L} \left(D_{i_{2j-1}} -D_{i_{2j}}\right) \sigma_{K}^{\left(|I|-L\right)} \mathcal{Y} \neq 0
                \end{aligned}
            \right.~,
        \end{align}
        Therefore, when (\ref{basis operators XI}) or (\ref{basis operators YI}) are basis operators for $\left[ \frac{n}{2} \right] < |I| \leq n$,
        (\ref{basis operators XI}) or (\ref{basis operators YI}) for $|I|-1$ are also basis. In particular, since $\sigma_{\phi}^{(n)} \mathcal{X}$ or $\sigma_{\phi}^{(n)} \mathcal{Y}$ are always basis,
        (\ref{basis operators XI}) or (\ref{basis operators YI}) for $\left[\frac{n}{2} \right] \leq |I| \leq n-1$ are also basis.
        Moreover, no primary operator consisting of $\mathcal{X}_{I}$ and $\mathcal{Y}_{I}$ for $\left[ \frac{n}{2} \right] < |I| \leq n$ exist.

        Thus, (\ref{basis operators XI}) or (\ref{basis operators YI}) are basis operators with respect to $|I|$ for $0 \leq |I| \leq n$.
        In particular, all primary operator consisting of $\mathcal{X}_{I}$ or $\mathcal{Y}_{I}$ can be given by (\ref{primary from XI}) or (\ref{primary from YI}) satisfying condition.\ref{no blank} and \ref{no intersection},
        and all non-primary operator composed of $\mathcal{X}_{I}$ or $\mathcal{Y}_{I}$ can be expanded by comparing the linear map $G$ with (\ref{composite XI 1}) or (\ref{composite YI 1}). 
    \end{proof}
    Therefore, we conclude the discussion of generators of this bosonic VOA from Sec.\ref{To derive non-trivial operators} to this section.
    As mentioned in \cite{Nishinaka:2025nbe}, T. Nishinaka and I explicitly solved the BRST cohomology problem with Mathematica for $n = 2,3,\dots,7$, and checked generators of this bosonic VOA at all the dimensions $1,\frac{3}{2},2,\dots,\frac{13}{2}$.
    By the series of discussions, we can solidify the validity of the conjecture regarding with generators;
    \begin{itemize}
        \item Generators of this bosonic VOA for $n=2$ are $U$, $\mathcal{X}$, $\mathcal{Y}$, $T_1 -T_2$, $\left(D_1 -D_2\right) \mathcal{X}$ and $\left(D_1 -D_2\right)\mathcal{Y}$.
        \item Generators of this bosonic VOA for $n=3$ are $U$, $\mathcal{X}$, $\mathcal{Y}$, $T$, $T_1 -T_2$, $T_2 -T_3$, $W_1 -W_2$, $W_2 -W_3$, $\left(D_1 -D_2\right) \mathcal{X}$, $\left(D_2 -D_3\right) \mathcal{X}$, $\left(D_1 -D_2 \right) \mathcal{Y}$ and $(D_2 -D_3)\mathcal{Y}$.
        \item Generators of this bosonic VOA for $n \geq 4$ are $U$, $\mathcal{X}$, $\mathcal{Y}$, $T$, (\ref{primary from T_i}), $W_{i}$ for $i=1,\dots,n$ and (\ref{primary from XI}) and (\ref{primary from YI}) satisfying condition.\ref{no blank} and \ref{no intersection}.
    \end{itemize}
\section{The OPEs with \texorpdfstring{$\mathcal{X}_{I}$}{} and \texorpdfstring{$\mathcal{Y}_{I}$}{}}
    \label{The OPEs with mathcal X I and mathcal Y I}
    Until the previous section, we have identified linearly independent generators of this bosonic VOA. Then, it is necessary for us to calculate the OPEs with these generators. Since the OPEs of generators are too complex, these OPEs regarding with each constituent element of generators are given in Appendix.\ref{OPE 1}, from (\ref{TT}) to (\ref{XY}).
    While the OPEs from (\ref{TT}) to (\ref{W-5}) are easy or simple to calculate, the remaining OPEs with $\mathcal{X}_{I}$ and $\mathcal{Y}_{I}$ are too complex and difficult.
    In this section, we derive the OPEs between $\mathcal{X}_{I}$ and $\mathcal{X}_{J}$, or $\mathcal{Y}_{I}$ and $\mathcal{Y}_{J}$.
    \footnote{At least, I do not understand the complete OPEs between $\mathcal{X}_{I}$ and $\mathcal{Y}_{J}$. Therefore, I only give some OPEs in Appendix.\ref{OPE 1}.}
    In order for us to derive the results, let us introduce the sets $I_i$, $\left(I_i\right)_j$ and $\mathcal{J}_{LM,p}$ as follows. 

    First, $\mathcal{I}_{n}$ and $I_i$ are defined by
    \begin{align}
        \mathcal{I}_{n} &\coloneqq \{1,\dots,n\}~, \\
        I_i &\coloneqq \left\{\text{All sets of}\ i\ \text{element removed from}\ I \subset \mathcal{I}_{n}\right\}~.
    \end{align}
    For $I = \{1,2,3\}$, we construct the sets $I_0$, $I_1$, $I_2$ and $I_3$ as follows.
    \begin{align*}
        I_0 &= \{\{1,2,3\}\}~,\\
        I_1 &= \{\{2,3\},\{1,3\},\{1,2\}\}~,\\
        I_2 &= \{\{1\},\{2\},\{3\}\}~,\\
        I_3 &= \{ \phi\}~.
    \end{align*}
    
    Second, since the set $I_i$ for $i=0,\dots,|I|$ is clearly finite, we can introduce $\left(I_i\right)_j$ as follows.
    \begin{align}
        \left(I_i\right)_j \coloneqq \left(\text{The } j \text{-th element of } I_i\right)~.
    \end{align}
    For the previous example $I_1 =\left\{\{2,3\},\ \{1,3\},\ \{1,2\}\right\}$, we set $\left(I_1\right)_1 = \{2,3\}$, $\left(I_1\right)_2 = \{1,3\}$ and $\left(I_1\right)_3 = \{1,2\}$. 
    Thus, we can generally reconstruct $I_i$ as follows;
    \begin{align}
        I_i = \left\{\left(I_i\right)_j \left|~ j=1,\dots,\binom{|I|}{i}\right.\right\}~,
    \end{align}
    where we set $\phi_0=\{\phi\}$ and $\left(\phi_0\right)_1 = \phi$ for convenience.

    Finally, for sets $L,M \in \mathcal{I}_{n}$ and a positive integer $p$, let us define the following set;
    \begin{align}
        \label{J LM p}
        \mathcal{J}_{LM,p} \coloneqq \left\{(a,b) \in \mathbb{Z} \times \mathbb{Z} \mid 0 \leq a \leq |L|,\ 0 \leq b \leq |M|,\ p \leq a+b\right\}~.
    \end{align}

    Using the above sets $I_i$, $\left(I_i\right)_j$ and $\mathcal{J}_{LM,p}$, the OPEs between $\mathcal{X}_{I}$ and $\mathcal{X}_{J}$ can be derived as follows.
    \begin{itemize}
        \item When $I, J \in \mathcal{I}_{n}$ and $I \neq J$,
    \begin{align}
    \label{XX1}
        \mathcal{X}_{I}(z) \mathcal{X}_{J}(0) &\sim 
        \sum_{p=1}^{|L|+|M|}\sum_{(a,b)\in \mathcal{J}_{LM,p}}\sum_{i=1}^{|L_a|}\sum_{j=1}^{|M_b|}\frac{1}{z^p} \notag \\
        &\quad \times \frac{(-1)^{a-p}}{\left(a+b-p\right)!}\binom{a+b-1}{p-1}\left(\partial^{a+b-p}\mathcal{X}_{K\cup (L_a)_i}\right)\mathcal{X}_{K\cup (M_b)_j}~,
    \end{align}
        where we set $K = I\cap J$, $L = I \setminus K$ and $M = J \setminus K$.
        \item When $I,J \in \mathcal{I}_{n}$ and $I = J$,
    \begin{align}
    \label{XX2}
        \mathcal{X}_{I}(z) \mathcal{X}_{I}(0) &\sim 0~.
    \end{align}
    \end{itemize}
    
    Similarly, the OPEs between $\mathcal{Y}_{I}$ and $\mathcal{Y}_{J}$ can be calculated as follows.
    \begin{itemize}
        \item When $I,J \in \mathcal{I}_{n}$ and $I \neq J$,
    \begin{align}
    \label{YY1}        
        \mathcal{Y}_{I}(z) \mathcal{Y}_{J}(0) 
        &\sim \sum_{p=1}^{|L|+|M|}\sum_{(a,b)\in \mathcal{J}_{LM,p}}\sum_{i=1}^{|L_a|}\sum_{j=1}^{|M_b|}\frac{1}{z^p} \notag \\
        &\quad \times \frac{(-1)^{b-p}}{\left(a+b-p\right)!} \binom{a+b-1}{p-1}\left(\partial^{a+b-p}\mathcal{Y}_{K\cup (L_a)_i}\right)\mathcal{Y}_{K\cup (M_b)_j}~,
    \end{align}
    where we set $K = I\cap J$, $L = I \setminus K$ and $M = J \setminus K$.
    \item When $I,J \in \mathcal{I}_{n}$ and $I =J$,
    \begin{align}
    \label{YY2}
        \mathcal{Y}_{I}(z) \mathcal{Y}_{I}(0) 
        &\sim 0~.
    \end{align}
    \end{itemize}

    In this section, let us derive the results (\ref{XX1}), (\ref{XX2}), (\ref{YY1}) and (\ref{YY2}). 
    Since (\ref{YY1}) and (\ref{YY2}) are similar to (\ref{XX1}) and (\ref{XX2}) respectively, we will proof the OPE between $\mathcal{X}_{I}$ and $\mathcal{X}_{J}$, namely (\ref{XX1}) and (\ref{XX2}) specifically. 
    The proof involves very complicated computations and elaborate reordering of sums. Therefore, let us prepare the identities and the relations first.
    \subsection{Identities for proofs}
        \label{Identities for proofs}
        In this subsection, let us introduce and prove the following identities used proof of (\ref{XX1}) and (\ref{XX2}).

        \begin{lem}
            \label{lem1}
            For two integers $a$ and $b$ satisfying $a \leq b$, the following identity holds.
            \begin{align}
                \sum_{i=a}^{b+1} \frac{(-1)^i}{(i-a)!(b+1-i)!} = 0~.
            \end{align}
        \end{lem}
        \begin{proof}
            We can prove Lemma \ref{lem1} by the following easy calculations.
            \begin{align*}
                \sum_{i=a}^{b+1} \frac{(-1)^i}{(i-a)!(b+1-i)!}
                &= \frac{(-1)^{a}}{(b+1-a)!} \sum_{i=a}^{b+1} (-1)^{i-a} \binom{b+1-a}{i-a} \\
                &= \frac{(-1)^{a}}{(b+1-a)!} \sum_{i=0}^{b+1-a} (-1)^{i} \binom{b+1-a}{i} \\
                &= \frac{(-1)^{a}}{(b+1-a)!} \left(1-1\right)^{b+1-a} \\
                &=0~.
            \end{align*}
        \end{proof}
        \begin{lem}
            \label{lem2}
            For two positive integers $a$ and $b$, and a non-negative integer $c$, the following identity holds.
            \begin{align}
                \sum_{i=\max(0,c-b)}^{\min(a,c)} \binom{a}{i} \binom{b}{c-i} = \binom{a+b}{c}~.
            \end{align}
        \end{lem}
        \begin{proof}
            Let us expand $(1+x)^{a+b}$ simply.
            \begin{align*}
                (1+x)^{a+b} = \sum_{c=0}^{a+b} \binom{a+b}{c} x^c.
            \end{align*}
            Since $(1+x)^{a+b} = (1+x)^a (1+x)^b$, $(1+x)^{a+b}$ can be also expanded as follows.
            \begin{align*}
                (1+x)^{a+b} = \sum_{i=0}^{a}\sum_{j=0}^{b} \binom{a}{i}\binom{b}{j} x^{i+j}.
            \end{align*}
            By comparing two polynomial expansions, we can prove the identity. 
            Consider the sum of the latter that satisfies $i+j=c$, where $c$ is the degree of $x$.
            \begin{itemize}
                \item When $0 \leq c \leq \min(a,b)$, we know that the following equality holds
                \begin{align*}
                    \binom{a}{0} \binom{b}{c} +\dots +\binom{a}{c} \binom{b}{0} = \sum_{i=0}^{c} \binom{a}{i}\binom{b}{c-i} =\binom{a+b}{c}~.
                \end{align*}
                \item For $a \leq c \leq b$, noting the upper bound of the range of the sum, the following equality holds
                \begin{align*}
                    \binom{a}{0} \binom{b}{c} +\dots +\binom{a}{a} \binom{b}{c-a} = \sum_{i=0}^{a} \binom{a}{i} \binom{b}{c-i} = \binom{a+b}{c}~.
                \end{align*}
                \item When $b \leq c \leq a$, noting the lower bound of the range of the sum similarly, the following equality holds
                \begin{align*}
                    \binom{a}{c-b} \binom{b}{b} +\dots + \binom{a}{c} \binom{b}{0} = \sum_{i=c-b}^{c} \binom{a}{i} \binom{b}{c-i} = \binom{a+b}{c}~.
                \end{align*}
                \item When $\max(a,b) \leq c \leq a+b$, the following equality holds Similarly
                \begin{align*}
                    \binom{a}{c-b} \binom{b}{b} +\dots +\binom{a}{a} \binom{b}{c-a} = \sum_{i=c-b}^{a} \binom{a}{i} \binom{b}{c-i} = \binom{a+b}{c}~.
                \end{align*}
            \end{itemize}
            Thus, we can prove that the following identity holds
            \begin{align*}
                \sum_{i=\max(0,c-b)}^{\min(a,c)} \binom{a}{i} \binom{b}{c-i} = \binom{a+b}{c}.
            \end{align*}
        \end{proof}
        \begin{lem}
            \label{lem3}
            For non-negative integers $a$ and $b$, the following identity holds
            \begin{align}
                \sum_{c=0}^{b} \frac{(-1)^{c}}{c+a+1} \binom{b}{c} = \frac{a! b!}{\left(a+b+1\right)!}~.
            \end{align}
        \end{lem}
        \begin{proof}
            Let us compute the following integral.
            \begin{align*}
                \int_{0}^{-1}t^a (1+t)^b dt = \int_{0}^{-1} \sum_{c=0}^{b} \binom{b}{c}t^{c+a} dt.
            \end{align*}
            The right-hand side of this integral can be computed as follows:
            \begin{align*}
                    \int_{0}^{-1} \sum_{c=0}^{b} \binom{b}{c}t^{c+a} dt 
                    &= \sum_{c=0}^{b} \frac{(-1)^{c+a+1}}{c+a+1} \binom{b}{c}~.
            \end{align*}
            On the other hand, the left-hand side is equal to the following result
            \begin{align*}
                    \int_{0}^{-1}t^a (1+t)^b dt &= (-1)^{a+1} \frac{a! b!}{\left(a+b+1\right)!}~.
            \end{align*}
            Thus, since the following equality holds
            \begin{align}
                \sum_{c=0}^{b} \frac{(-1)^{c+a+1}}{c+a+1} \binom{b}{c} = (-1)^{a+1} \frac{a! b!}{\left(a+b+1\right)!}~,
            \end{align}
            this lemma is proved.
        \end{proof}
        
        For the proof of the next Lemma, we need some technical calculations.
        \begin{lem}
            \label{lem4}
            For three positive integers $p~$, $q~$ and $r~$ satisfying $p \leq q < r$, the following identity holds.
            \begin{align}
                \sum_{s=p}^{q} (-1)^{p+q+s}\binom{r}{s} \binom{s-1}{p-1}\binom{r-s-1}{r-q-1}
                = (-1)^{p} \binom{r-1}{p-1} +(-1)^{q}\binom{r-1}{q}.
            \end{align}
        \end{lem}
        \begin{proof}
            Using the following identity,
            \begin{align*}
                \binom{r}{s} = \binom{r-1}{s} + \binom{r-1}{s-1}~,
            \end{align*}
            we can transform as follows.
            \begin{align*}
                \binom{r-1}{s} \binom{s-1}{p-1}\binom{r-s-1}{r-q-1} 
                &= \binom{r-1}{q} \binom{s-1}{p-1} \binom{q}{s} \\
                &= \frac{q}{s} \binom{r-1}{q} \binom{q-1}{p-1} \binom{q-p}{s-p}~.
            \end{align*}
            Similarly, 
            \begin{align*}
                \binom{r-1}{s-1} \binom{s-1}{p-1} \binom{r-s-1}{r-q-1} 
                &=\frac{r-p}{r-s}\binom{r-1}{p-1} \binom{r-p-1}{r-q-1} \binom{q-p}{q-s}~.
            \end{align*}
            From these results and Lemma \ref{lem3}, the identity can be derived by the following calculation.
            \begin{align*}
                \sum_{s=p}^{q} (-1)^{p+q+s}\binom{r}{s} \binom{s-1}{p-1}\binom{r-s-1}{r-q-1}
                &= (-1)^{q}\binom{r-1}{q}  q \binom{q-1}{p-1} \sum_{s=0}^{q-p} \frac{(-1)^{s}}{s+p} \binom{q-p}{s} \\
                &\quad +(-1)^{p} \binom{r-1}{p-1}  (r-p) \binom{r-p-1}{r-q-1} \sum_{s=0}^{q-p} \frac{(-1)^{s}}{s+r-q} \binom{q-p}{s} \\
                &= (-1)^{p} \binom{r-1}{p-1} +(-1)^{q}\binom{r-1}{q}~.
            \end{align*}
        \end{proof}

        The next Lemma is the most important in the proof of (\ref{XX1}).
        \begin{lem}
            \label{lem5}
            For three positive integer $p~$, $q~$ and $r~$ satisfying $p \leq q \leq r$, the following identity holds.
            \begin{align}
                \sum_{s=p}^{q} (-1)^{p+q+s}\binom{r}{s} \binom{s-1}{p-1}\binom{r-s-1}{r-q}
                = \delta_{q=r} (-1)^{p} \binom{r-1}{p-1} -(-1)^{p} \binom{r-1}{p-1} +(-1)^{q}\binom{r-1}{q-1}~.
            \end{align}
        \end{lem}
        \begin{proof}
            For $p=q$, we can prove it by the following easy calculation
            \begin{align*}
                \sum_{s=p}^{q} (-1)^{p+q+s}\binom{r}{s} \binom{s-1}{p-1}\binom{r-s-1}{r-q}
                &= (-1)^{p} \binom{r}{p} \binom{r-p-1}{r-p} \notag \\
                &= \delta_{p=r} (-1)^{p}~.
            \end{align*}
            When $p <q$, we can derive the following result by using Lemma \ref{lem4}
            \begin{align*}
                &\quad \sum_{s=p}^{q} (-1)^{p+q+s}\binom{r}{s} \binom{s-1}{p-1}\binom{r-s-1}{r-q} \notag \\
                &= -\sum_{s=p}^{q-1} (-1)^{p+(q-1)+s}\binom{r}{s} \binom{s-1}{p-1}\binom{r-s-1}{r-(q-1)-1}
                + (-1)^{p} \binom{r}{q} \binom{q-1}{p-1} \binom{r-q-1}{r-q} \notag \\
                &= -\left\{(-1)^p \binom{r-1}{p-1} +(-1)^{q-1} \binom{r-1}{q-1}\right\}
                + \delta_{q=r}(-1)^{p} \binom{r-1}{p-1} \notag \\
                &= \delta_{q=r}(-1)^p \binom{r-1}{p-1}
                -(-1)^p \binom{r-1}{p-1} +(-1)^q \binom{r-1}{q-1}~.
            \end{align*}
        \end{proof}
    \subsection{The construction of recursive relational formulas}
        \label{The construction of recursive relational formulas}
        In this subsection, let us derive the recursive relation for the proof of (\ref{XX1}) and (\ref{XX2}). 
        Since we use mathematical induciton to prove them, it is convenience for us to study the relation between the case for $n=k$ and $n= k+1$.

        Let $\mathcal{P}_{n}$ and $\mathcal{Q}_{n}$ set as follows;            
        \begin{align}
            \left\{
                \begin{aligned}
                    \mathcal{P}_{1} &\coloneqq P_1 &\quad \mathcal{Q}_{1} &\coloneqq Q_1\\
                    \mathcal{P}_{n+1} &\coloneqq \left(P_{n+1}\mathcal{P}_{n}\right)_0 &\quad \mathcal{Q}_{n+1} &\coloneqq \left(Q_{n+1} \mathcal{Q}_{n}\right)_0
                \end{aligned}
            \right.~,
        \end{align}  
        where $P_i$ and $Q_i$ are equal to $X_i$ or $D_i X_i$.
        The OPEs between $P_i$ and $Q_j$ are given by
        \begin{align}
            \left(P_i Q_j\right)_p =\left\{
                \begin{aligned}
                    &\frac{1}{(-p)!}\left(\partial^{(-p)}P_i Q_j\right)_0 &(p \leq 0)\\
                    &\delta_{ij}R_{P_i,Q_i} &(p=1)\\
                    &0 &(p \geq 2)
                \end{aligned}
            \right.~,
        \end{align}
        where $R_{P_i,Q_i}$ is the following expression;
        \begin{align}
            R_{P_i,Q_i} \coloneqq \left\{
                \begin{aligned}
                    &-\left(X_i X_i\right)_0 &(P_i = X_i,\ Q_i = D_i X_i)\\
                    &\left(X_i X_i\right)_0 &(P_i = D_i X_i,\ Q_i = X_i)\\
                    &0 &(P_i = Q_i)
                \end{aligned}
            \right.~.
        \end{align}

        Here, let derive the relation between $\left(\mathcal{P}_{n}\mathcal{Q}_{n}\right)_p$ and $\left(\mathcal{P}_{n+1}\mathcal{Q}_{n+1}\right)_q$.
        As the calculations are too complicated to derive this relationship, it is shown in the following two steps, Lemma \ref{lem6} and Theorem \ref{thm1}.
        \begin{lem}
            \label{lem6}
            For a positive integer $n$ and an integer $p$, the following recursive relational formula holds
            \begin{align}
                \left(\mathcal{P}_{n+1} \mathcal{Q}_{n+1}\right)_p 
                &= \sum_{l \geq p-1} \sum_{q \geq l} \frac{(-1)^l(-1)^q}{(q-l)!(l-p+1)!} \left(R_{P_{n+1},Q_{n+1}} \partial^{q -p +1}\left( \mathcal{P}_{n} \mathcal{Q}_{n}\right)_q \right)_{0} \notag \\
                &\quad +\sum_{l \geq p}\sum_{m=0}^{l-p} \sum_{q \geq l} \frac{(-1)^l(-1)^q}{(q-l)!(l-p-m)!} \left(P_{n+1} \left(Q_{n+1} \partial^{q -m -p}\left( \mathcal{P}_{n} \mathcal{Q}_{n}\right)_q \right)_0 \right)_{-m}.
            \end{align}
        \end{lem}
        \begin{proof}
            Let us transform $\left(\mathcal{P}_{n+1} \mathcal{Q}_{n+1} \right)$ into the following expression by using (\ref{OPEdefs2}).
        \begin{align}
            \label{1-lem6}
            \left(\mathcal{P}_{n+1}\mathcal{Q}_{n+1}\right)_p 
            &= \left(Q_{n+1}\left(\mathcal{P}_{n+1} \mathcal{Q}_{n}\right)_p\right)_0 
            + \sum_{l>0}\binom{p-1}{l-1}\left(\left(\mathcal{P}_{n+1}Q_{n+1}\right)_l \mathcal{Q}_{n}\right)_{p-l}~.
        \end{align}
        
        For the first term of the right-hand side of (\ref{1-lem6}), the following equality can be carefully derived by using (\ref{OPEdefs1}), (\ref{OPEdefs2}) and (\ref{OPEdefs3}).
        \begin{align}
            &\quad \left(Q_{n+1}\left(\mathcal{P}_{n+1} \mathcal{Q}_{n}\right)_p\right)_0 \notag \\
            &= \sum_{l \geq p}\sum_{m=0}^{l-p}\sum_{q \geq l} \frac{(-1)^l(-1)^q}{(l-p-m)!(q-l)!} \notag\\
            &\quad \times \left\{\left(P_{n+1} \left(Q_{n+1} \partial^{q-p-m}\left(\mathcal{P}_{n} \mathcal{Q}_{n}\right)_q\right)_{0} \right)_{-m}
            -\left(R_{P_{n+1},Q_{n+1}} \partial^{q-p-m}\left(\mathcal{P}_{n} \mathcal{Q}_{n}\right)_q \right)_{-m-1}\right\}~.
            \label{2-lem6}
        \end{align}

        Similarly, for the second term of the right-hand side of (\ref{1-lem6}),
        \begin{align}
            &\quad \sum_{l>0}\binom{p-1}{l-1}\left(\left(\mathcal{P}_{n+1}Q_{n+1}\right)_l \mathcal{Q}_{n}\right)_{p-l} \notag\\
            &= \sum_{l \geq p-1}\sum_{m=0}^{l-p+1} \sum_{q \geq l}
            \frac{(-1)^l(-1)^q}{(q-l)!(l-p+1-m)!} \left(R_{P_{n+1},Q_{n+1}} \partial^{q -m -p +1}\left( \mathcal{P}_{n} \mathcal{Q}_{n}\right)_q \right)_{-m}~.
            \label{3-lem6}
        \end{align}

        To summarize these equalities, it is necessary for us to transform the sum ``$\sum_{l \geq p-1} \sum_{m=0}^{l-p+1} \sum_{q \geq l} \cdots$'' in (\ref{3-lem6}) as follows.
        \begin{align}
            &\quad \sum_{l \geq p-1}\sum_{m=0}^{l-p+1} \sum_{q \geq l} \frac{(-1)^l(-1)^q}{(q-l)!(l-p+1-m)!} \left(R_{P_{n+1},Q_{n+1}} \partial^{q -m -p +1}\left( \mathcal{P}_{n} \mathcal{Q}_{n}\right)_q \right)_{-m} \notag\\
            &= \sum_{l \geq p-1} \sum_{q \geq l} \frac{(-1)^l(-1)^q}{(q-l)!(l-p+1)!} \left(R_{P_{n+1},Q_{n+1}} \partial^{q -p +1}\left( \mathcal{P}_{n} \mathcal{Q}_{n}\right)_q \right)_{0} \notag\\
            &\quad +\sum_{l \geq p}\sum_{m=0}^{l-p} \sum_{q \geq l} \frac{(-1)^l(-1)^q}{(q-l)!(l-p-m)!} \left(R_{P_{n+1},Q_{n+1}} \partial^{q -m -p}\left( \mathcal{P}_{n} \mathcal{Q}_{n}\right)_q \right)_{-m-1}~.
            \label{4-lem6}
        \end{align}

        Thus, we can derive the recursive relational formula for $\left(\mathcal{P}_{n}\mathcal{Q}_{n}\right)_p$ from (\ref{2-lem6}) and (\ref{4-lem6}).
        \begin{align*}
            \left(\mathcal{P}_{n+1}\mathcal{Q}_{n+1}\right)_{p} 
            &= \sum_{l \geq p}\sum_{m=0}^{l-p}\sum_{q \geq l} \frac{(-1)^l(-1)^q}{(l-p-m)!(q-l)!} \left(P_{n+1} \left(Q_{n+1} \partial^{q-p-m}\left(\mathcal{P}_{n} \mathcal{Q}_{n}\right)_q\right)_{0} \right)_{-m} \\
            &\quad +\sum_{l \geq p-1} \sum_{q \geq l} \frac{(-1)^l(-1)^q}{(l-p+1)!(q-l)!} \left(R_{P_{n+1},Q_{n+1}} \partial^{q -p +1}\left( \mathcal{P}_{n} \mathcal{Q}_{n}\right)_q \right)_{0}~.
        \end{align*}
        \end{proof}
        \begin{thm}
            \label{thm1}
            For a positive integer $n$ and a integer $p$, the following recursive relational formula holds
            \begin{align}
                \left(\mathcal{P}_{n+1} \mathcal{Q}_{n+1}\right)_p = \left\{
                    \begin{aligned}
                        &\left(R_{P_{n+1},Q_{n+1}} \left(\mathcal{P}_{n} \mathcal{Q}_{n}\right)_{p-1}\right)_0  \\
                        &\quad +\sum_{q=p}^{n} \left(P_{n+1} \left(Q_{n+1} \left(\mathcal{P}_{n}\mathcal{Q}_{n}\right)_{q} \right)_0 \right)_{p-q} &(p \leq n)\\
                        &\left(R_{P_{n+1},Q_{n+1}} \left(\mathcal{P}_{n} \mathcal{Q}_{n}\right)_{p-1}\right)_0 &(p = n+1)\\
                        &0 &(p \geq n+2)
                    \end{aligned}
                \right.~.
            \end{align}
        \end{thm}
        \begin{proof}
            First, we prove the statement that $\left(\mathcal{P}_{n} \mathcal{Q}_{n}\right)_{p} = 0$ always holds for $p \geq n+1$, using mathematical induction.
            \begin{enumerate}
                \item When $n=1$, it is trivial.
                \item When $n=k$, where $k$ is a positive integer, let us assume that this statement holds.
                For $n=k+1$, the following equality is derived by applying Lemma \ref{lem6}
                \begin{align}
                        &\quad \left(\mathcal{P}_{k+1}\mathcal{Q}_{k+1}\right)_{p} \notag \\
                        &= \sum_{l \geq p}\sum_{m=0}^{l-p}\sum_{q \geq l} \frac{(-1)^l(-1)^q}{(l-p-m)!(q-l)!} \left(P_{k+1} \left(Q_{k+1} \partial^{q-p-m}\left(\mathcal{P}_{k} \mathcal{Q}_{k}\right)_q\right)_{0} \right)_{-m} \notag\\
                        &\quad +\sum_{l \geq p-1} \sum_{q \geq l} \frac{(-1)^l(-1)^q}{(l-p+1)!(q-l)!} \left(R_{P_{k+1},Q_{k+1}} \partial^{q -p +1}\left( \mathcal{P}_{k} \mathcal{Q}_{k}\right)_q \right)_{0}~.
                        \label{1-thm1}
                \end{align}
                Since $\left(\mathcal{P}_{k} \mathcal{Q}_{k}\right)_q =0$ for $q \geq k+1$ holds from the assumption here, $\left(\mathcal{P}_{k+1} \mathcal{Q}_{k+1}\right)_p =0$ can be satisfied for $p \geq k+2$.
                Therefore, this statement for $n=k+1$ is also true.
            \end{enumerate}
            Thus, this statement is proved, namely $\left(\mathcal{P}_{n+1}\mathcal{Q}_{n+1}\right)_{p} = 0$ for $p \geq n+2$.

            Second, for $p \leq n+1$, we can derive that the following equality holds by using Lemma \ref{lem1}
            \begin{align}
                \label{2-thm1}
                \sum_{l \geq p-1} \sum_{q \geq l} \frac{(-1)^l(-1)^q}{(l-p+1)!(q-l)!} \left(R_{P_{n+1},Q_{n+1}} \partial^{q -p +1}\left( \mathcal{P}_{n} \mathcal{Q}_{n}\right)_q \right)_{0} =
                \left(R_{P_{n+1},Q_{n+1}} \left(\mathcal{P}_{n} \mathcal{Q}_{n}\right)_{n}\right)_0~.
            \end{align}

            Finally, when $p \leq n+1$, we can similarly derive that the following equality holds by using Lemma \ref{lem1}
            \begin{align}
                \label{3-thm1}
                &\quad \sum_{l \geq p}\sum_{m=0}^{l-p}\sum_{q \geq l} \frac{(-1)^l(-1)^q}{(l-p-m)!(q-l)!} \left(P_{n+1} \left(Q_{n+1} \partial^{q-p-m}\left(\mathcal{P}_{n} \mathcal{Q}_{n}\right)_q\right)_{0} \right)_{-m} \notag\\
                &= \left\{
                    \begin{aligned}
                        &\sum_{q=p}^{n}\left(P_{n+1} \left(Q_{n+1} \left(\mathcal{P}_{n} \mathcal{Q}_{n}\right)_q \right)_0 \right)_{p-q} &\quad (p \leq n) \\
                        &0 &\quad (p=n+1)
                    \end{aligned}
                \right.~.
            \end{align}

            Thus, since (\ref{1-thm1}), (\ref{2-thm1}) and (\ref{3-thm1}), the following equality can be derived
            \begin{align*}
                \left(\mathcal{P}_{n+1}\mathcal{Q}_{n+1}\right)_p =\left\{
                \begin{aligned}
                    &\left(R_{P_{n+1},Q_{n+1}} \left(\mathcal{P}_{n} \mathcal{Q}_{n}\right)_{p-1}\right)_0  \\
                    &\quad +\sum_{q=p}^{n} \left(P_{n+1} \left(Q_{n+1} \left(\mathcal{P}_{n}\mathcal{Q}_{n}\right)_{q} \right)_0 \right)_{p-q} &(p \leq n)\\
                    &\left(R_{P_{n+1},Q_{n+1}} \left(\mathcal{P}_{n} \mathcal{Q}_{n}\right)_{p-1}\right)_0 &(p = n+1)\\
                    &0 &(p \geq n+2)
                \end{aligned}
                \right.~.
            \end{align*}
        \end{proof}
        \begin{cor}
            \label{cor1}
            For a positive integer $n$ and a integer $p$, when $P_{n+1}=Q_{n+1}$, the recursive relational equality holds
            \begin{align}
                \left(\mathcal{P}_{n+1} \mathcal{Q}_{n+1}\right)_p = \left\{
                \begin{aligned}
                    &\sum_{q=p}^{n} \left(P_{n+1} \left(Q_{n+1} \left(\mathcal{P}_{n}\mathcal{Q}_{n}\right)_{q} \right)_0 \right)_{p-q} &(p \leq n)\\
                    &0 &(p \geq n+1)
                \end{aligned}                        
                \right.~.
            \end{align}
        \end{cor}
        \begin{proof}
            When $P_{n+1} = Q_{n+1}$, it can be proven from $R_{P_{n+1},Q_{n+1}}=0$ and Theorem \ref{thm1}.
        \end{proof}

        We have prepared the identities and the recursive relational equalities. In the next subsection, we will explain the proof of (\ref{XX1}) and (\ref{XX2}).
    \subsection{Proof of the OPEs between \texorpdfstring{$\mathcal{X}_I$}{} and \texorpdfstring{$\mathcal{X}_J$}{}, \texorpdfstring{$\mathcal{Y}_I$}{} and \texorpdfstring{$\mathcal{Y}_J$}{}}
        \label{Proof of the OPEs between mathcal X I and mathcal X J, mathcal Y I and mathcal Y J}
        Let us prove (\ref{XX1}) and (\ref{XX2}), that is the following theorem, by Lemma \ref{lem2}, Lemma \ref{lem5}, Theorem \ref{thm1} and Corollary \ref{cor1}.
        \begin{thm}
            \label{thm2}
            For a positive integer $n$ satisfying $n \geq 2$,
            the product $\left(\mathcal{X}_{I} \mathcal{X}_{J}\right)_p$ for $p \geq 1$ is equal to the following result:
            \begin{itemize}
                \item When $I \neq J$,
                \begin{align*}
                    \left(\mathcal{X}_{I} \mathcal{X}_{J}\right)_p
                    &= \left\{
                    \begin{aligned}
                            &\sum_{(a,b) \in \mathcal{J}_{LM,p}} \sum_{i=1}^{|L_a|} \sum_{j=1}^{|M_b|} (-1)^{a-p} \binom{a+b-1}{p-1}\\
                            &\quad \times \left(\mathcal{X}_{K \cup \left(L_a\right)_i} \mathcal{X}_{K \cup \left(M_b\right)_j}\right)_{p-a-b} &(1 \leq p \leq |L| +|M|)\\
                            &0 &(p \geq |L| + |M| +1)
                    \end{aligned}
                    \right.~,
                \end{align*}
                where $K = I \cap J$, $L = I \setminus K$ and $M = J \setminus K$.
                \item When $I=J$,
                \begin{align*}
                    \left(\mathcal{X}_{I} \mathcal{X}_{I}\right)_p = 0~.
                \end{align*}
            \end{itemize}
        \end{thm}

        Here, let us prove Theorem \ref{thm2} by mathematical induction, but we think that this proof is too long. Then, we divide the proof as follows.
        \begin{enumerate}
            \item Prove Theorem \ref{thm2} for $n=2$, which we set as Proposition \ref{prop1}.
            \item Assuming that Theorem \ref{thm2} for $n =k \geq 2$ is true, prove Theorem \ref{thm2} for $n=k+1$ and $P_{k+1} =Q_{k+1}$. Let this statement be Proposition \ref{prop2}.
            \item Assuming that Theorem \ref{thm2} for $n =k \geq 2$ holds, prove Theorem \ref{thm2} for $n=k+1$ and $(P_{k+1},Q_{k+1})=(D_{k+1} X_{k+1},X_{k+1})$ similarly. Let this statement be Proposition \ref{prop3}.
            \item Assuming that Theorem \ref{thm2} for $n =k \geq 2$ is true, prove Theorem \ref{thm2} for $n=k+1$ and $(P_{k+1},Q_{k+1})=(X_{k+1}, D_{k+1} X_{k+1})$ in the same way. Let this statement be Proposition \ref{prop4}.
        \end{enumerate}
        After we prove Proposition \ref{prop1} to Proposition \ref{prop4}, we would like to regard the series of proofs of these Propositions as the proof of Theorem \ref{thm2}.

        \begin{prop}
            \label{prop1}
            Theorem \ref{thm2} for $n=2$ is true.
        \end{prop}
        \begin{proof}
            In this proof, let us write down all cases.
            \begin{itemize}
                \item For $P_1=Q_1$ and $P_2 =Q_2$. By Corollary \ref{cor1} and $\left(P_1 Q_1\right)_1 = 0$, the following result for $p \geq 1$ holds.
                \begin{align*}
                    \left(\mathcal{P}_{2} \mathcal{Q}_{2}\right)_p = 0.
                \end{align*}
                \item For $P_1 \neq Q_1$ and $P_2 = Q_2$. By Corollary \ref{cor1}, the following results for $p \geq 1$ can be derived.
                \begin{align*}
                    \left(\mathcal{P}_{2} \mathcal{Q}_{2}\right)_p = \left\{
                    \begin{aligned}
                        &\left\{
                        \begin{aligned}
                            &\left(P_2 \left(Q_2 \left(X_1 X_1\right)_0 \right)_0\right)_0 &((P_1,Q_1) = (D_1 X_1,X_1)) \\
                            &-\left(P_2 \left(Q_2 \left(X_1 X_1\right)_0 \right)_0\right)_0 &((P_1,Q_1) = (X_1,D_1X_1)) \\
                        \end{aligned}
                        \right. &(p=1)\\
                        &0 &(p \geq 2)
                    \end{aligned}
                    \right.~.
                \end{align*}
                On the other hand, by Corollary \ref{cor1} we can also derive the following equality.
                \begin{align*}
                    \left(\left(P_2 X_1 \right)_0 \left(Q_2 X_1\right)_0\right)_0 = \left(P_2 \left(Q_2 \left(X_1 X_1\right)_0 \right)_0\right)_0~.
                \end{align*}
                Thus, we can find out that the following equality holds.
                \begin{align*}
                    \left(\mathcal{P}_{2} \mathcal{Q}_{2}\right)_p = \left\{
                    \begin{aligned}
                        &\left\{
                        \begin{aligned}                                
                            &\left(\left(P_2 X_1\right)_0 \left(Q_2 X_1\right)_0 \right)_0 &((P_1,Q_1) = (D_1 X_1,X_1)) \\
                            &-\left(\left(P_2 X_1\right)_0 \left(Q_2 X_1\right)_0\right)_0 &((P_1,Q_1) = (X_1,D_1X_1)) \\
                        \end{aligned}
                        \right. &(p=1)\\
                        &0 &(p \geq 2)
                    \end{aligned}
                    \right.~.
                \end{align*}
                \item For $P_1=Q_1$ and $P_2 \neq Q_2$. Using $\left(P_2 P_1\right)_p = \left(P_1 P_2\right)_p$ and $\left(Q_2 Q_1\right)_p = \left(Q_1 Q_2\right)_p$ for $p \geq 0$,
                we can similarly discover the following relation.
                \begin{align*}
                    \left(\mathcal{P}_{2} \mathcal{Q}_{2}\right)_p = \left\{
                    \begin{aligned}
                        &\left\{
                        \begin{aligned}                                
                            &\left(\left(X_2 P_1\right)_0 \left(X_2 Q_1\right)_0 \right)_0 &((P_2,Q_2) = (D_2 X_2,X_2)) \\
                            &-\left(\left(X_2 P_1\right)_0 \left(X_2 Q_1\right)_0\right)_0 &((P_2,Q_2) = (X_2,D_2X_2)) \\
                        \end{aligned}
                        \right. &(p=1)\\
                        &0 &(p \geq 2)
                    \end{aligned}
                    \right.~.
                \end{align*}
                \item For $P_1 \neq Q_1$ and $P_2 \neq Q_2$. Unlike the previous cases, we need four kind of terms in this case, as follows.
                \begin{align*}
                    &\quad \left(\left(X_2 P_1\right)_0 \left(X_2 Q_1\right)_0 \right)_{0} \\
                    \\
                    &= \left\{
                    \begin{aligned}
                        &\left(R_{P_2,Q_2}\left(P_1 Q_1\right)_0\right)_0 +\left(\left(X_2 X_2\right)_{-1} R_{P_1,Q_1}\right)_0 &\left(\left(P_2,Q_2\right) = \left(D_2X_2,X_2 \right) \right) \\
                        \\
                        &-\left(R_{P_2,Q_2}\left(P_1 Q_1\right)_0\right)_0 +\left(\left(X_2 X_2\right)_{-1} R_{P_1,Q_1}\right)_0 &\left(\left(P_2,Q_2\right) = \left(X_2,D_2X_2 \right) \right)
                    \end{aligned}
                \right.~,
                \end{align*}
                \begin{align*}
                    &\quad \left(\left(P_2 X_1\right)_0 \left(Q_2 X_1\right)_0 \right)_{0} \\
                    \\
                    &= \left\{
                        \begin{aligned}
                            &\left(R_{P_1,Q_1}\left(P_2 Q_2\right)_0\right)_0 +\left(\left(X_1 X_1\right)_{-1} R_{P_2,Q_2}\right)_0 &\left(\left(P_1,Q_1\right) = \left(D_1X_1,X_1 \right) \right) \\
                            \\
                            &-\left(R_{P_1,Q_1}\left(P_2 Q_2\right)_0\right)_0 +\left(\left(X_1 X_1\right)_{-1} R_{P_2,Q_2}\right)_0 &\left(\left(P_1,Q_1\right) = \left(X_1,D_1X_1 \right) \right)
                        \end{aligned}
                    \right.~,
                \end{align*}
                \begin{align*}
                    &\quad \left(\left(X_2 X_1\right)_0 \left(X_2 X_1\right)_0 \right)_{-1} \\
                    \\
                    &= \left\{
                        \begin{aligned}
                        &\left(R_{P_2,Q_2}\left(X_1 X_1\right)_{-1}\right)_0 + \left(\left(X_2 X_2\right)_{-1} R_{P_1,Q_1}\right)_0 
                        &\left\{
                            \begin{aligned}
                                \left(P_1,Q_1\right) &= \left(D_1X_1,X_1 \right) \\
                                \left(P_2,Q_2\right) &= \left(D_2X_2,X_2 \right)
                            \end{aligned}
                        \right. \\
                        \\
                        &\left(R_{P_2,Q_2}\left(X_1 X_1\right)_{-1}\right)_0 - \left(\left(X_2 X_2\right)_{-1} R_{P_1,Q_1}\right)_0
                        &\left\{
                            \begin{aligned}
                                \left(P_1,Q_1\right) &= \left(D_1X_1,X_1 \right) \\
                                \left(P_2,Q_2\right) &= \left(X_2,D_2X_2 \right)
                            \end{aligned}
                        \right. \\
                        \\
                        &-\left(R_{P_2,Q_2}\left(X_1 X_1\right)_{-1}\right)_0 + \left(\left(X_2 X_2\right)_{-1} R_{P_1,Q_1}\right)_0
                        &\left\{
                            \begin{aligned}    
                                \left(P_1,Q_1\right) &= \left(X_1,D_1X_1 \right) \\
                                \left(P_2,Q_2\right) &= \left(D_2X_2,X_2 \right)
                            \end{aligned}
                        \right. \\
                        \\
                        &-\left(R_{P_2,Q_2}\left(X_1 X_1\right)_{-1}\right)_0 - \left(\left(X_2 X_2\right)_{-1} R_{P_1,Q_1}\right)_0
                        &\left\{
                            \begin{aligned}    
                                \left(P_1,Q_1\right) &= \left(X_1,D_1X_1 \right) \\
                                \left(P_2,Q_2\right) &= \left(X_2,D_2X_2 \right)
                            \end{aligned}
                        \right.
                        \end{aligned}
                    \right.~,                  
                \end{align*}
                and
                \begin{align*}
                    \left(\left(X_2 X_1\right)_0 \left(X_2 X_1\right)_0 \right)_{0}
                    &= \left\{
                        \begin{aligned}
                        &\left(R_{P_2,Q_2} R_{P_1,Q_1}\right)_0
                        &\left\{
                            \begin{aligned}    
                                \left(P_1,Q_1\right) &= \left(D_1X_1,X_1 \right) \\
                                \left(P_2,Q_2\right) &= \left(D_2X_2,X_2 \right)
                            \end{aligned}
                        \right.\\
                        \\
                        &-\left(R_{P_2,Q_2} R_{P_1,Q_1}\right)_0
                        &\left\{
                            \begin{aligned}    
                                \left(P_1,Q_1\right) &= \left(D_1X_1,X_1 \right) \\
                                \left(P_2,Q_2\right) &= \left(X_2,D_2X_2 \right)
                            \end{aligned}
                        \right.\\
                        \\
                        &-\left(R_{P_2,Q_2} R_{P_1,Q_1}\right)_0
                        &\left\{
                            \begin{aligned}    
                                \left(P_1,Q_1\right) &= \left(X_1,D_1X_1 \right) \\
                                \left(P_2,Q_2\right) &= \left(D_2X_2,X_2 \right)
                            \end{aligned}
                        \right.\\
                        \\
                        &\left(R_{P_2,Q_2} R_{P_1,Q_1}\right)_0
                        &\left\{
                            \begin{aligned}    
                                \left(P_1,Q_1\right) &= \left(X_1,D_1X_1 \right) \\
                                \left(P_2,Q_2\right) &= \left(X_2,D_2X_2 \right)
                            \end{aligned}
                        \right.
                        \end{aligned}
                    \right.~.
                \end{align*}

                On the other hand, we know that $\left(\mathcal{P}_{2}\mathcal{Q}_{2}\right)_p$ for $p \geq 1$ is satisfying the following equality by Theorem \ref{thm1}.
                \begin{align*}
                    \left(\mathcal{P}_{2} \mathcal{Q}_2\right)_p =\left\{
                        \begin{aligned}
                            &\left(R_{P_2,Q_2} \left(P_1 Q_1\right)_0 \right)_0 +\left(P_2 \left(Q_2 R_{P_1,Q_1}\right)_0\right)_0&(p=1)\\
                            &\left(R_{P_2,Q_2} R_{P_1,Q_1}\right)_0 &(p =2)\\
                            &0 &(p \geq 3)
                        \end{aligned}
                    \right.~.
                \end{align*}
                Thus, when $p=1$, we can check that the following equality holds.
                \begin{align*}
                    &\quad \left(\mathcal{P}_{2}\mathcal{Q}_{2}\right)_1 \\
                    \\
                    &= \left\{
                        \begin{aligned}
                            &\left(\left(X_2 P_1\right)_0 \left(X_2 Q_1\right)_0 \right)_0 \\
                            &\quad +\left(\left(P_2 X_1\right)_0 \left(P_2 X_1\right)_0 \right)_0 
                            -\left(\left(X_2X_1\right)_0 \left(X_2 X_1\right)_0 \right)_{-1}
                            &\left\{
                                \begin{aligned}    
                                    \left(P_1,Q_1\right) &= \left(D_1X_1,X_1 \right) \\
                                        \left(P_2,Q_2\right) &= \left(D_2X_2,X_2 \right)
                                \end{aligned}
                            \right.\\
                            \\
                            &\left(\left(X_2 P_1\right)_0 \left(X_2 Q_1\right)_0 \right)_0 \\
                            &\quad -\left(\left(P_2 X_1\right)_0 \left(P_2 X_1\right)_0 \right)_0 
                            +\left(\left(X_2X_1\right)_0 \left(X_2 X_1\right)_0 \right)_{-1}
                            &\left\{
                                \begin{aligned}    
                                    \left(P_1,Q_1\right) &= \left(D_1X_1,X_1 \right) \\
                                    \left(P_2,Q_2\right) &= \left(X_2,D_2X_2 \right)
                                \end{aligned}
                            \right.\\
                            \\
                            &-\left(\left(X_2 P_1\right)_0 \left(X_2 Q_1\right)_0 \right)_0 \\
                            &\quad +\left(\left(P_2 X_1\right)_0 \left(P_2 X_1\right)_0 \right)_0 
                            +\left(\left(X_2X_1\right)_0 \left(X_2 X_1\right)_0 \right)_{-1}
                            &\left\{
                                \begin{aligned}
                                    \left(P_1,Q_1\right) &= \left(X_1,D_1X_1 \right) \\
                                    \left(P_2,Q_2\right) &= \left(D_2X_2,X_2 \right)
                                \end{aligned}
                            \right.\\
                            \\
                            &-\left(\left(X_2 P_1\right)_0 \left(X_2 Q_1\right)_0 \right)_0 \\
                            &\quad -\left(\left(P_2 X_1\right)_0 \left(P_2 X_1\right)_0 \right)_0 
                            -\left(\left(X_2X_1\right)_0 \left(X_2 X_1\right)_0 \right)_{-1}
                            &\left\{
                                \begin{aligned}    
                                    \left(P_1,Q_1\right) &= \left(X_1,D_1X_1 \right) \\
                                    \left(P_2,Q_2\right) &= \left(X_2,D_2X_2 \right)
                                \end{aligned}
                            \right.
                        \end{aligned}
                    \right.~.
                \end{align*}
                Similarly, when $p=2$, we can discover the following equality.
                \begin{align*}
                    \left(\mathcal{P}_{2}\mathcal{Q}_{2}\right)_{2}
                    &= \left\{
                        \begin{aligned}
                            &\left(\left(X_2X_1\right)_0 \left(X_2 X_1\right)_0\right)_0
                            &\left\{
                                \begin{aligned}
                                    \left(P_1,Q_1\right) &= \left(D_1X_1,X_1 \right) \\
                                    \left(P_2,Q_2\right) &= \left(D_2X_2,X_2 \right)
                                \end{aligned}
                            \right.\\
                            \\
                            &-\left(\left(X_2X_1\right)_0 \left(X_2 X_1\right)_0\right)_0
                            &\left\{
                                \begin{aligned}    
                                    \left(P_1,Q_1\right) &= \left(D_1X_1,X_1 \right) \\
                                    \left(P_2,Q_2\right) &= \left(X_2,D_2X_2 \right)
                                \end{aligned}
                            \right.\\
                            \\
                            &-\left(\left(X_2X_1\right)_0 \left(X_2 X_1\right)_0\right)_0
                            &\left\{
                                \begin{aligned}    
                                    \left(P_1,Q_1\right) &= \left(X_1,D_1X_1 \right) \\
                                    \left(P_2,Q_2\right) &= \left(D_2X_2,X_2 \right)
                                \end{aligned}
                            \right.\\
                            \\
                            &\left(\left(X_2X_1\right)_0 \left(X_2 X_1\right)_0\right)_0
                            &\left\{
                                \begin{aligned}    
                                    \left(P_1,Q_1\right) &= \left(X_1,D_1X_1 \right) \\
                                    \left(P_2,Q_2\right) &= \left(X_2,D_2X_2 \right)
                                \end{aligned}
                            \right.
                        \end{aligned}
                    \right.~.            
                \end{align*}
            \end{itemize}
            The above equalities are the very results of Theorem \ref{thm2} for $n=2$.
        \end{proof}
        \begin{prop}
            \label{prop2}
            Let $k$ be an integer satisfying $k \geq 2$. Assuming that Theorem \ref{thm2} for $n=k$ holds, 
            Theorem \ref{thm2} for $n=k+1$ and $P_{k+1} = Q_{k+1}$ also holds.
        \end{prop}
        \begin{proof}
            When $P_{k+1} = Q_{k+1}$, the following equality for $p \geq 1$ can be derived from Corollary \ref{cor1}.
            \begin{align*}
                \left(\mathcal{P}_{k+1}\mathcal{Q}_{k+1}\right)_p = \left\{
                    \begin{aligned}
                        &\sum_{q=p}^{k} \left(P_{k+1} \left(Q_{k+1}\left(\mathcal{P}_{k} \mathcal{Q}_{k}\right)_q \right)_0 \right)_{p-q} &(1 \leq p \leq k) \\
                        &0 &(p \geq k+1)
                    \end{aligned}
                \right.~,
            \end{align*}
            where we set $\mathcal{X} = \prod_{i=1}^{k} X_i$, $\mathcal{P}_{k} = \mathcal{X}_{I}$ and $\mathcal{Q}_{k} = \mathcal{X}_{J}$ here.
            The relationship between $I$ and $J$ is devided into the following cases.
                
            When $I=J$, from the assumption for $\left(\mathcal{X}_{I} \mathcal{X}_{I}\right)_q = 0$ for $q \geq 1$,
            we can easily find out that the following equality for $p \geq 1$ holds
            \begin{align*}
                \left(\mathcal{P}_{k+1} \mathcal{Q}_{k+1}\right)_p = 0~.
            \end{align*}
            
            For $I \neq J$, we can use the following equality from the assumption.
            \begin{align*}
                &\quad \left(\mathcal{P}_{k} \mathcal{Q}_{k}\right)_q \\
                &= \left\{
                    \begin{aligned}
                        &\sum_{(a,b) \in \mathcal{J}_{LM,q}} \sum_{i=1}^{|L_a|}\sum_{j=1}^{|M_b|}(-1)^{a-q}\binom{a+b-1}{q-1} \left(\mathcal{X}_{K \cup \left(L_a\right)_i} \mathcal{X}_{K \cup \left(M_b\right)_j}\right)_{-(a+b-q)} &(1 \leq q \leq |L| +|M|)\\
                        \\
                        &0 &(q >|L| +|M|)
                    \end{aligned}
                \right.~.
            \end{align*}
            Since $|L| +|M| \leq k$, the above relation for $1 \leq p \leq |L| +|M|$ can be transformed as follows.
            \begin{align}
                \label{1-prop2}
                    &\quad \left(\mathcal{P}_{k+1} \mathcal{Q}_{k+1}\right)_{p} \notag \\
                    &= \sum_{(a,b) \in \mathcal{J}_{LM,p}} \sum_{q=p}^{a+b}\sum_{i=1}^{|L_a|} \sum_{j=1}^{|M_b|} (-1)^{a-q} \binom{a+b-1}{q-1} \left(P_{k+1} \left(Q_{k+1} \left(\mathcal{X}_{K \cup \left(L_a\right)_i} \mathcal{X}_{K \cup \left(M_b\right)_j}\right)_{-(a+b-q)}\right)_0 \right)_{p-q}
            \end{align}

            Let us prove that the following equality is equal to (\ref{1-prop2}) for $p \geq 1$.
            \begin{align}
                \label{2-prop2}
                \left(\mathcal{P}_{k+1} \mathcal{Q}_{k+1}\right)_p
                &= \left\{
                    \begin{aligned}
                        &\sum_{(a,b) \in \mathcal{J}_{LM,p}} \sum_{i=1}^{|L_a|}\sum_{j=1}^{|M_b|}(-1)^{a-p}  \binom{a+b-1}{p-1}\\ 
                        &\quad \times\left(\left(P_{k+1}\mathcal{X}_{K \cup \left(L_a\right)_i}\right)_0 \left(Q_{k+1}\mathcal{X}_{K \cup \left(M_b\right)_j}\right)_0\right)_{-(a+b-p)} &(1 \leq p \leq |L| +|M|)\\
                        \\
                        &0 &(p >|L| +|M|)
                    \end{aligned}
                \right.~.
            \end{align}
            Since (\ref{2-prop2}) for $p > |L|+|M|$ is satisfying this Proposition, we will prove that (\ref{2-prop2}) for $1 \leq p \leq |L| + |M|$ is equal to (\ref{1-prop2}).

            For $1 \leq p \leq |L| +|M|$, (\ref{2-prop2}) can be transformed by using Corollary \ref{cor1} as follows.
            \begin{align}
                \label{3-prop2}
                    &\sum_{(c,d) \in \mathcal{J}_{LM,p}} \sum_{i=1}^{|L_c|}\sum_{j=1}^{|M_d|}(-1)^{c-p} \binom{c+d-1}{p-1}
                    \left(\left(P_{k+1}\mathcal{X}_{K \cup \left(L_c\right)_i}\right)_0 \left(Q_{k+1}\mathcal{X}_{K \cup \left(M_d\right)_j}\right)_0\right)_{-(c+d-p)} \notag \\
                    &= \sum_{(a,b) \in \mathcal{J}_{LM,p}} \sum_{q=p}^{a+b}\sum_{i=1}^{|L_a|} \sum_{j=1}^{|M_b|} (-1)^{a-p}  \binom{a+b-1}{p-1}
                    \left(P_{k+1}\left(Q_{k+1} \left(\mathcal{X}_{K \cup \left(L_a\right)_i} \mathcal{X}_{K \cup \left(M_b\right)_j}\right)_{-(a+b-q)}\right)_0 \right)_{p-q} \notag \\
                    &\quad +\delta_{p \leq |L|+ |M|-1}\sum_{\substack{(c,d) \in \mathcal{J}_{LM,p} \\ (c,d) \neq \left(|L|,|M|\right)}} \sum_{q=1}^{|L| +|M| -c-d} \sum_{i=1}^{|L_c|}\sum_{j=1}^{|M_d|}(-1)^{c-p} \binom{c+d-1}{p-1}\notag \\ 
                    &\quad \times \left(P_{k+1} \left(Q_{k+1} \left(\mathcal{X}_{K \cup \left(L_c\right)_i} \mathcal{X}_{K \cup \left(M_d\right)_j}\right)_{q}\right)_0\right)_{-(c+d-p)-q}~.
            \end{align}
            
            In the second term on the right-hand side of (\ref{3-prop2}), $\left(\mathcal{X}_{K \cup \left(L_c\right)_i} \mathcal{X}_{K \cup \left(M_d\right)_j}\right)_{q}$ for $q \geq 1$ can be expanded by using the assumption for $n=k$.
            Here, it is important to note the redundancy for $(L_a)_i$ and $(M_b)_j$ respectively. We can expand $\left(\mathcal{X}_{K \cup \left(L_c\right)_i} \mathcal{X}_{K \cup \left(M_d\right)_j}\right)_{q}$ for $1 \leq q \leq |L| +|M| -c -d$ by using the assumption,
            but if expanding it with brute force, the awful sums are emerged.
            Therefore, it is necessary for us to execute the following wiser expansion;
            \begin{align}
                \label{5-prop2}
                &\delta_{p \leq |L|+ |M|-1} \sum_{\substack{(c,d) \in \mathcal{J}_{LM,p} \\ (c,d) \neq \left(|L|,|M|\right)}} \sum_{q=1}^{|L| +|M| -c-d} \sum_{i=1}^{|L_c|}\sum_{j=1}^{|M_d|}(-1)^{c-p} \binom{c+d-1}{p-1} \notag \\ 
                &\quad \times \left(P_{k+1} \left(Q_{k+1} \left(\mathcal{X}_{K \cup \left(L_c\right)_i} \mathcal{X}_{K \cup \left(M_d\right)_j}\right)_{q}\right)_0\right)_{-(c+d-p)-q} \notag \\
                &= \delta_{p \leq |L|+ |M|-1} \sum_{\substack{(c,d) \in \mathcal{J}_{LM,p} \\ (c,d) \neq \left(|L|,|M|\right)}} \sum_{q=c+d+1}^{|L| +|M|} \sum_{\substack{0 \leq e \leq |L| -c \\0 \leq f \leq |M| -d \\q -c-d\leq e+f \leq |L| +|M| -c-d}} \sum_{i=1}^{|L_{c+e}|} \sum_{j=1}^{|M_{d+f}|} \notag \\
                &\quad \times (-1)^{d+e -(p+q)} \binom{c+e}{c} \binom{d+f}{d} \binom{c+d-1}{p-1}\binom{e+f-1}{q-c-d-1} \notag \\
                &\quad \times \left(P_{k+1} \left(Q_{k+1} \left(\mathcal{X}_{K \cup \left(L_{c+e}\right)_i} \mathcal{X}_{K \cup \left(M_{d+f}\right)_j}\right)_{-(c+d+e+f-q)}\right)_0\right)_{p-q}~,
            \end{align}
            where $\binom{c+e}{c}$ (\text{resp.}~$\binom{d+f}{d}$) is the number of times $\left(L_{c+e}\right)_i$ (resp.~$\left(M_{d+f}\right)_j$) appears, that is, $\left(L_{c+e}\right)_i$ (resp.~$\left(M_{d+f}\right)_j$) represents the degree of redundancy.
            In fact, in order to construct $\left(L_{c+e}\right)_i$, we need to remove $c$ elements from $L \setminus \left(L_{c+e}\right)_i$ and then $e$ elements.
            With $c$ and $e$ as fixed integers, the total number of the way that we get $\left(L_{c+e}\right)_i$ can be derived as follows.
            \begin{align*}
                \binom{|L|-\left(|L|-c-e\right)}{c} \binom{|L|-c -\left(|L|-c-e\right)}{e} = \binom{c+e}{c}~.
            \end{align*}
            Therefore, the set $\left(L_{c+e}\right)_i$ appears $\binom{c+e}{c}$ times. Similarly, the set $\left(M_{d+f}\right)_j$ emerges $\binom{d+f}{d}$ times.

            For (\ref{5-prop2}), we set $a=c+e$ and $b=d+f$, and interchange the order of the sums. Then we can derive the following result.
            \begin{align}
                &\sum_{\substack{(c,d) \in \mathcal{J}_{LM,p}}} \sum_{a=c}^{|L|} \sum_{b=d}^{|M|} \sum_{q= c+d}^{a+b} \sum_{i=1}^{|L_a|} \sum_{j=1}^{|M_b|}
                (-1)^{(a-c-d) -(p+q)} \binom{a}{c} \binom{b}{d} \binom{c+d-1}{p-1}\binom{a+b-c-d-1}{a+b-q} \notag\\
                &\quad \times \left(P_{k+1} \left(Q_{k+1} \left(\mathcal{X}_{K \cup \left(L_{a}\right)_i} \mathcal{X}_{K \cup \left(M_{b}\right)_j}\right)_{-(a+b-q)}\right)_0\right)_{p-q} \notag\\
                &\quad -\sum_{(a,b) \in \mathcal{J}_{LM,p}} \sum_{i=1}^{|L|_{a}} \sum_{j=1}^{|M|_{b}} (-1)^{a-p} \binom{a+b-1}{p-1} 
                \left(P_{k+1} \left(Q_{k+1} \left(\mathcal{X}_{K \cup \left(L_{a}\right)_i} \mathcal{X}_{K \cup \left(M_{b}\right)_j}\right)_{0}\right)_0\right)_{p-a-b}~.
                \label{6-prop2}
            \end{align}

            Let us swap the order of the sums with respect to $(a,b)$ and with respect to $(c,d)$ in the first term on the right-hand side of (\ref{6-prop2}). To do so, fix the pair $(a,b)$, use $(a,b)$ to change the range of $(c,d)$,
            and execute the summation with respect to $(c,d)$.
            Namely, let the range of $(c,d)$ set ``$0 \leq c \leq a,\ 0 \leq d \leq b,\ p \leq c+d \leq a+b$'', and $c$ and $d$ be eliminated by Lemma \ref{lem2} and \ref{lem5}.
            \begin{align}
                &\sum_{\substack{0 \leq c \leq a \\0 \leq d \leq b \\ p \leq c+d \leq a+b}} \sum_{q= c+d}^{a+b} \sum_{i=1}^{|L_a|} \sum_{j=1}^{|M_b|} 
                (-1)^{(a-c-d) -(p+q)} 
                \binom{a}{c} \binom{b}{d} \binom{c+d-1}{p-1}\binom{a+b-c-d-1}{a+b-q} \notag \\
                &\quad \times \left(P_{k+1} \left(Q_{k+1} \left(\mathcal{X}_{K \cup \left(L_{a}\right)_i} \mathcal{X}_{K \cup \left(M_{b}\right)_j}\right)_{-(a+b-q)}\right)_0\right)_{p-q} \notag \\
                &= \sum_{i=1}^{|L_a|} \sum_{j=1}^{|M_b|}(-1)^{a-p} \binom{a+b-1}{p-1} \left(P_{k+1} \left(Q_{k+1} \left(\mathcal{X}_{K \cup \left(L_{a}\right)_i} \mathcal{X}_{K \cup \left(M_{b}\right)_j}\right)_{0}\right)_0\right)_{-(a+b-p)} \notag \\
                &\quad - \sum_{q=p}^{a+b} \sum_{i=1}^{|L_a|} \sum_{j=1}^{|M_b|}(-1)^{a-p} \binom{a+b-1}{p-1} \left(P_{k+1} \left(Q_{k+1} \left(\mathcal{X}_{K \cup \left(L_{a}\right)_i} \mathcal{X}_{K \cup \left(M_{b}\right)_j}\right)_{0}\right)_{-(a+b-q)}\right)_{p-q} \notag \\
                &\quad +\sum_{q=p}^{a+b} \sum_{i=1}^{|L_a|} \sum_{j=1}^{|M_b|}(-1)^{a-q} \binom{a+b-1}{q-1} \left(P_{k+1} \left(Q_{k+1} \left(\mathcal{X}_{K \cup \left(L_{a}\right)_i} \mathcal{X}_{K \cup \left(M_{b}\right)_j}\right)_{0}\right)_{-(a+b-q)}\right)_{p-q}~.
                \label{7-prop2}
            \end{align}
            Therefore, we can execute the summation with respect to $(a,b)$ for (\ref{7-prop2}) as follows.
            \begin{align}
                \label{8-prop2}
                    &\sum_{\substack{(c,d) \in \mathcal{J}_{LM,p}}} \sum_{a=c}^{|L|} \sum_{b=d}^{|M|} \sum_{q= c+d}^{a+b} \sum_{i=1}^{|L_a|} \sum_{j=1}^{|M_b|} 
                    (-1)^{(a-c-d) -(p+q)} \binom{a}{c} \binom{b}{d} \binom{c+d-1}{p-1}\binom{a+b-c-d-1}{a+b-q} \notag \\
                    &\quad \times \left(P_{k+1} \left(Q_{k+1} \left(\mathcal{X}_{K \cup \left(L_{a}\right)_i} \mathcal{X}_{K \cup \left(M_{b}\right)_j}\right)_{-(a+b-q)}\right)_0\right)_{p-q} \notag \\
                    &= \sum_{(a,b) \in \mathcal{J}_{LM,p}} \sum_{q=p}^{a+b} \sum_{i=1}^{|L_a|} \sum_{j=1}^{|M_b|} (-1)^{a-q} \binom{a+b-1}{q-1} 
                    \left(P_{k+1} \left(Q_{k+1} \left(\mathcal{X}_{K \cup \left(L_a\right)_i} \mathcal{X}_{K \cup \left(M_b\right)_j} \right)_{-(a+b-q)} \right)_0 \right)_{p-q} \notag\\
                    &\quad -\sum_{(a,b) \in \mathcal{J}_{LM,p}} \sum_{q=p}^{a+b} \sum_{i=1}^{|L_a|} \sum_{j=1}^{|M_b|} (-1)^{a-p} \binom{a+b-1}{p-1} 
                    \left(P_{k+1} \left(Q_{k+1} \left(\mathcal{X}_{K \cup \left(L_a\right)_i} \mathcal{X}_{K \cup \left(M_b\right)_j} \right)_{-(a+b-q)} \right)_0 \right)_{p-q} \notag \\
                    &\quad + \sum_{(a,b) \in \mathcal{J}_{LM,p}} \sum_{i=1}^{|L_a|} \sum_{j=1}^{|M_b|} (-1)^{a-p} \binom{a+b-1}{p-1}
                    \left(P_{k+1} \left(Q_{k+1} \left(\mathcal{X}_{K \cup \left(L_a\right)_i} \mathcal{X}_{K \cup \left(M_b\right)_j} \right)_0 \right)_0 \right)_{p-a-b}~.
            \end{align}
            
            Thus, the right-hand side of (\ref{6-prop2}) is equal to the following terms by using (\ref{8-prop2}).
            \begin{align}
                \label{10-prop2}
                    &\sum_{(a,b) \in \mathcal{J}_{LM,p}} \sum_{q=p}^{a+b} \sum_{i=1}^{|L_a|} \sum_{j=1}^{|M_b|} (-1)^{a-q}\binom{a+b-1}{q-1} 
                    \left(P_{k+1} \left(Q_{k+1} \left(\mathcal{X}_{K \cup \left(L_a\right)_i} \mathcal{X}_{K \cup \left(M_b\right)_j} \right)_{-(a+b-q)} \right)_0 \right)_{p-q} \notag\\
                    &\quad -\sum_{(a,b) \in \mathcal{J}_{LM,p}} \sum_{q=p}^{a+b} \sum_{i=1}^{|L_a|} \sum_{j=1}^{|M_b|} (-1)^{a-p} \binom{a+b-1}{p-1} 
                    \left(P_{k+1} \left(Q_{k+1} \left(\mathcal{X}_{K \cup \left(L_a\right)_i} \mathcal{X}_{K \cup \left(M_b\right)_j} \right)_{-(a+b-q)} \right)_0 \right)_{p-q}~.
            \end{align}

            Finally, (\ref{3-prop2}) can be transformed by (\ref{10-prop2}) as follows.
            \begin{align}
                \label{11-prop2}
                    &\sum_{(c,d) \in \mathcal{J}_{LM,p}} \sum_{i=1}^{|L_c|}\sum_{j=1}^{|M_d|}(-1)^{c-p} \binom{c+d-1}{p-1} 
                    \left(\left(P_{k+1}\mathcal{X}_{K \cup \left(L_c\right)_i}\right)_0 \left(Q_{k+1}\mathcal{X}_{K \cup \left(M_d\right)_j}\right)_0\right)_{-(c+d-p)} \notag \\
                    &= \sum_{(a,b) \in \mathcal{J}_{LM,p}} \sum_{q=p}^{a+b} \sum_{i=1}^{|L_a|} \sum_{j=1}^{|M_b|} (-1)^{a-q}  \binom{a+b-1}{q-1} 
                    \left(P_{k+1} \left(Q_{k+1} \left(\mathcal{X}_{K \cup \left(L_a\right)_i} \mathcal{X}_{K \cup \left(M_b\right)_j} \right)_{-(a+b-q)} \right)_0 \right)_{p-q}.
            \end{align}
            This result (\ref{11-prop2}) is equal to (\ref{1-prop2}) for $1 \leq p \leq |L| +|M|$.
        \end{proof}
        \begin{prop}
            \label{prop3}
            Let $k$ be an integer satisfying $k \geq 2$. Assuming that Theorem \ref{thm2} for $n=k$ holds,
            Theorem \ref{thm2} for $n=k+1$ and $\left(P_{k+1},Q_{k+1}\right) = (D_{k+1}X_{k+1},X_{k+1})$ also holds.
        \end{prop}
        \begin{proof}
            Let us prove Proposition \ref{prop3} in the same way as we proved Proposition \ref{prop2}. So, we set $\mathcal{X} = \prod_{i=1}^{k} X_i$, $\mathcal{P}_{k} = \mathcal{X}_{I}$ and $\mathcal{Q}_{k} = \mathcal{X}_{J}$, and we will omit some of the similar calculations that appeared in Proposition \ref{prop2}.

            Since $P_{k+1} \neq Q_{k+1}$ in Proposition \ref{prop3}, we can expand $\left(\mathcal{P}_{k+1} \mathcal{Q}_{k+1}\right)_p$ using Thorem \ref{thm1} for $p \geq 1$ and substitute $R_{P_{k+1},Q_{k+1}} = \left(X_{k+1} X_{k+1}\right)_0$. 
            To $\left(\mathcal{X}_{K \cup L} \mathcal{X}_{K \cup M}\right)_{p-1}$ for $p \geq 2$ and $\left(\mathcal{X}_{K \cup L} \mathcal{X}_{K \cup M}\right)_q$ for $q \geq p$ derived from the above expansion, we can adapt Theorem \ref{thm2} in $n=k$.

            For $I = J$, 
            \begin{align}
                \label{1-prop3}
                \left(\mathcal{P}_{k+1} \mathcal{Q}_{k+1}\right)_p = \delta_{p=1} \left(X_{k+1} \left(X_{k+1} \left(\mathcal{X}_{I} \mathcal{X}_{I}\right)_0 \right)_0 \right)_0~.
            \end{align}
            
            When $I \neq J$,
            \begin{align}
                \label{2-prop3}
                    \left(\mathcal{P}_{k+1} \mathcal{Q}_{k+1}\right)_p 
                    &= \delta_{p=1} \left(X_{k+1} \left(X_{k+1} \left(\mathcal{X}_{K \cup L} \mathcal{X}_{K \cup M} \right)_0 \right)_0 \right)_0 \notag \\
                    &\quad +\delta_{2 \leq p \leq |L|+|M| +1} \sum_{(a,b) \in \mathcal{J}_{LM,p-1}} \sum_{i=1}^{|L_a|} \sum_{j=1}^{|M_b|} (-1)^{a-p+1} \binom{a+b-1}{p-2} \notag\\
                    &\quad \times \left(X_{k+1} \left(X_{k+1} \left(\mathcal{X}_{K \cup \left(L_a\right)_i} \mathcal{X}_{K \cup \left(M_b\right)_j} \right)_{-(a+b-p+1)} \right)_0 \right)_0 \notag\\
                    &\quad + \delta_{1 \leq p \leq |L|+ |M|} \sum_{(a,b) \in \mathcal{J}_{LM,p}} \sum_{q=p}^{a+b} \sum_{i=1}^{|L_a|} \sum_{j=1}^{|M_b|} (-1)^{a-q}\binom{a+b-1}{q-1} \notag\\
                    &\quad \times \left(P_{k+1} \left(Q_{k+1} \left(\mathcal{X}_{K \cup \left(L_a\right)_i} \mathcal{X}_{K \cup \left(M_b\right)_j} \right)_{-(a+b-q)} \right)_0 \right)_{p-q}~.
            \end{align}

            We prove the following equality    
            \begin{align}
                \label{3-prop3}
                \left(\mathcal{P}_{k+1} \mathcal{Q}_{k+1}\right)_{p} = \left\{
                    \begin{aligned}
                        &\sum_{(a,b) \in \mathcal{J}_{L'M,p}} \sum_{i=1}^{|L'_a|} \sum_{j=1}^{|M_b|} (-1)^{a-p} \binom{a+b-1}{p-1} \\
                        &\quad \times \left(\tilde{\mathcal{X}}_{K \cup \left(L'_a\right)_i} \tilde{\mathcal{X}}_{K \cup (M_b)_j}\right)_{-(a+b-p)} &(1 \leq p \leq |L|+|M|+1) \\
                        \\
                        &0 &(p \geq |L| + |M| +2)
                    \end{aligned}
                \right.~,
            \end{align}
            where the set $L'$ is defined by $L' \coloneqq L \cup \{k+1\}$ and we set
            \begin{align*}
                \tilde{\mathcal{X}}_{K \cup \left(L'_a\right)_i} &\coloneqq \left\{
                    \begin{aligned}
                        &\left(D_{k+1} X_{k+1} \mathcal{X}_{K \cup \left(\left(L'_a\right)_i \setminus \{k+1\} \right)}\right)_0 &\quad \left(k+1 \in \left(L'_a\right)_i \right)\\
                        \\
                        &\left(X_{k+1} \mathcal{X}_{K \cup \left(L'_a\right)_i}\right)_0 &\quad \left(k+1 \notin \left(L'_a\right)_i \right)
                    \end{aligned}
                \right.~,
            \end{align*}
            and
            \begin{align*}
                \tilde{\mathcal{X}}_{K \cup \left(M_b\right)_j} &\coloneqq \left(X_{k+1} \mathcal{X}_{K \cup \left(M_b\right)_j}\right)_0~.
            \end{align*}
            Then, for $p \geq |L|+ |M| +2$, (\ref{1-prop3}) and (\ref{2-prop3}) are equal to (\ref{3-prop3}). Therefore, we prove the equality for $1 \leq p \leq |L| + |M|+1$.

            Let us expand and decomposite the right-hand side of (\ref{3-prop3}) for $1 \leq p \leq |L| +|M| +1$, as follows. 
            \begin{align}
                \label{4-prop3}
                    &\sum_{(c,d) \in \mathcal{J}_{L'M,p}} \sum_{i=1}^{|L'_{c}|} \sum_{j=1}^{|M_d|} (-1)^{c-p} \binom{c+d-1}{p-1} 
                    \left(\tilde{\mathcal{X}}_{K \cup \left(L'_c\right)_i} \tilde{\mathcal{X}}_{K \cup (M_d)_j}\right)_{-(c+d-p)} \notag \\
                    &= \delta_{1 \leq p \leq |L| + |M|} \sum_{\substack{(c,d) \in \mathcal{J}_{L'M,p} \\ c \leq |L|}} \sum_{i=1}^{|L_c|} \sum_{j=1}^{|M_d|} (-1)^{c-p} \binom{c+d-1}{p-1} \notag\\
                    &\quad \times \left( \left(P_{k+1} \mathcal{X}_{K \cup \left(L_c\right)_i} \right)_0 \left(Q_{k+1} \mathcal{X}_{K \cup \left(M_d\right)_j} \right)_0 \right)_{-(c+d-p)} \notag \\
                    &\quad + \delta_{1 \leq p \leq |L| + |M|+1} \sum_{\substack{(c,d) \in \mathcal{J}_{L'M},p \\ c \geq 1}} \sum_{i=1}^{|L_{c-1}|} \sum_{j=1}^{|M_d|} (-1)^{c-p} \binom{c+d-1}{p-1} \notag \\
                    &\quad \times \left( \left(X_{k+1} \mathcal{X}_{K \cup \left(L_{c-1}\right)_i} \right)_0 \left(X_{k+1} \mathcal{X}_{K \cup \left(M_d\right)_j}\right)_0 \right)_{-(c+d-p)}~,
            \end{align}
            where the first term on the right-hand side of (\ref{4-prop3}) only emerges in $I \neq J$.
            
            From now on, we calculate each term on the right-hand side of (\ref{4-prop3}).
            \begin{itemize}
                \item The first term only emerges for $I \neq J$ and $1 \leq p \leq |L| +|M|$. We can derive the following result in the similar way in Proposition \ref{prop2}.
                \begin{align}
                    &\sum_{(a,b) \in \mathcal{J}_{LM,p}} \sum_{i=1}^{|L_a|} \sum_{j=1}^{|M_b|} (-1)^{a-p} \binom{a+b-1}{p-1} \left(X_{k+1} \left(X_{k+1} \left(\mathcal{X}_{K \cup \left(L_a\right)_i} \mathcal{X}_{K \cup \left(M_b\right)_j}\right)_{-(a+b-p+1)} \right)_0 \right)_0 \notag \\
                    &\quad + \sum_{(a,b) \in \mathcal{J}_{LM,p}} \sum_{q=p}^{a+b} \sum_{i=1}^{|L_a|} \sum_{j=1}^{|M_b|} (-1)^{a-q} \binom{a+b-1}{q-1} 
                    \left(P_{k+1} \left(Q_{k+1} \left(\mathcal{X}_{K \cup \left(L_a\right)_i} \mathcal{X}_{K \cup \left(M_b\right)_j} \right)_{-(a+b-q)} \right)_0 \right)_{p-q}~.
                    \label{5-prop3}
                \end{align}
                
                \item The second term emerges for $1 \leq p \leq |L| + |M|+1$. We can derive the following result.
                \begin{align}
                    &\sum_{(a,b) \in \mathcal{J}_{LM,p-1}} \sum_{q=p-1}^{a+b} \sum_{i=1}^{|L_a|} \sum_{j=1}^{|M_b|} (-1)^{a-p+1} \binom{a+b}{p-1} \notag \\
                    &\quad \times \left(X_{k+1} \left(X_{k+1} \left(\mathcal{X}_{K \cup \left(L_a\right)_i} \mathcal{X}_{K \cup \left(M_b\right)_j} \right)_{-(a+b-q)} \right)_0 \right)_{p-q-1} \notag \\
                    &\quad -\sum_{(a,b) \in \mathcal{J}_{LM,p-1}} \sum_{i=1}^{|L_a|} \sum_{j=1}^{|M_b|}(-1)^{a-p+1} \binom{a+b}{p-1} \left(X_{k+1} \left(X_{k+1} \left(\mathcal{X}_{K \cup \left(L_a\right)_i} \mathcal{X}_{K \cup \left(M_b\right)_j} \right)_{0} \right)_0 \right)_{-(a+b-p+1)} \notag \\
                    &\quad + \sum_{(c,d) \in \mathcal{J}_{LM,p-1}} \sum_{a=c}^{|L|} \sum_{b=d}^{|M|} \sum_{q=c+d}^{a+b} \sum_{i=1}^{|L_a|} \sum_{j=1}^{|M_b|}
                    (-1)^{a-(p+q+c+d)+1} \binom{a}{c} \binom{b}{d} \binom{c+d}{p-1} \binom{a+b-c-d-1}{a+b-q} \notag \\
                    &\quad \times \left(X_{k+1} \left(X_{k+1} \left(\mathcal{X}_{K \cup \left(L_a\right)_i} \mathcal{X}_{K \cup \left(M_b\right)_j} \right)_{-(a+b-q)} \right)_0 \right)_{p-q-1}~.
                    \label{6-prop3}
                \end{align}
                Here, the third term on the right-hand side of (\ref{6-prop3}) can be transformed in the similar way of (\ref{7-prop2}) when we fix the pair $(a,b)$.
                \begin{align}
                    &\quad \sum_{\substack{0 \leq c \leq a \\ 0 \leq d \leq b \\ p-1 \leq c+d \leq a+b}} \sum_{q=c+d}^{a+b} \sum_{i=1}^{|L_a|} \sum_{j=1}^{|M_b|}
                    (-1)^{a-(p+q+c+d)+1} \binom{a}{c} \binom{b}{d} \binom{c+d}{p-1} \binom{a+b-c-d-1}{a+b-q} \notag \\
                    &\quad \times \left(X_{k+1} \left(X_{k+1} \left(\mathcal{X}_{K \cup \left(L_a\right)_i} \mathcal{X}_{K \cup \left(M_b\right)_j} \right)_{-(a+b-q)} \right)_0 \right)_{p-q-1} \notag \\
                    &= \delta_{a+b=p-1} \sum_{i=1}^{|L_a|} \sum_{j=1}^{|M_b|}(-1)^{a-p+1} \left(X_{k+1} \left(X_{k+1} \left(\mathcal{X}_{K \cup \left(L_a\right)_i} \mathcal{X}_{K \cup \left(M_b\right)_j} \right)_0 \right)_0 \right)_0 \notag \\
                    &\quad + \delta_{a+b \geq p} \sum_{q=p}^{a+b} \sum_{r=p-1}^{q} \sum_{i=1}^{|L_a|} \sum_{j=1}^{|M_b|} (-1)^{a-(p+q+r)+1} \binom{a+b}{r} \binom{r}{p-1} \binom{a+b-r-1}{a+b-q} \notag \\
                    &\quad \times \left(X_{k+1} \left(X_{k+1} \left(\mathcal{X}_{K \cup \left(L_a\right)_i} \mathcal{X}_{K \cup \left(M_b\right)_j} \right)_{-(a+b-q)} \right)_0 \right)_{p-q-1}~.
                    \label{7-prop3}
                \end{align}
                For $a+b \geq p$, the second term on the right-hand side of (\ref{7-prop3}) can be transformed by using Lemma \ref{lem5}:
                \begin{align}
                    &\sum_{i=1}^{|L_a|} \sum_{j=1}^{|M_b|} (-1)^{a-p+1} \binom{a+b}{p-1} \left(X_{k+1} \left(X_{k+1} \left(\mathcal{X}_{K \cup \left(L_a\right)_i} \mathcal{X}_{K \cup \left(M_b\right)_j} \right)_{0} \right)_0 \right)_{-(a+b-p+1)} \notag \\
                    &\quad - \sum_{q=p}^{a+b}\sum_{i=1}^{|L_a|} \sum_{j=1}^{|M_b|} (-1)^{a-p+1} \binom{a+b}{p-1} \left(X_{k+1} \left(X_{k+1} \left(\mathcal{X}_{K \cup \left(L_a\right)_i} \mathcal{X}_{K \cup \left(M_b\right)_j} \right)_{-(a+b-q)} \right)_0 \right)_{p-q-1}~.
                    \label{8-prop3}
                \end{align}
                Therefore, (\ref{7-prop3}) can be equal to the following expression by (\ref{8-prop3}).
                \begin{align}
                    &\sum_{i=1}^{|L_a|} \sum_{j=1}^{|M_b|} (-1)^{a-p+1} \binom{a+b}{p-1} \left(X_{k+1} \left(X_{k+1} \left(\mathcal{X}_{K \cup \left(L_a\right)_i} \mathcal{X}_{K \cup \left(M_b\right)_j} \right)_{0} \right)_0 \right)_{-(a+b-p+1)} \notag \\
                    &\quad - \delta_{a+b \geq p}\sum_{q=p}^{a+b}\sum_{i=1}^{|L_a|} \sum_{j=1}^{|M_b|} (-1)^{a-p+1} \binom{a+b}{p-1} \left(X_{k+1} \left(X_{k+1} \left(\mathcal{X}_{K \cup \left(L_a\right)_i} \mathcal{X}_{K \cup \left(M_b\right)_j} \right)_{-(a+b-q)} \right)_0 \right)_{p-q-1}~.
                    \label{9-prop3}
                \end{align}
                So, we can execute the summation with respect to $(a,b)$ for (\ref{9-prop3}), and derive the following result.
                \begin{align}
                    &\quad \sum_{(a,b) \in \mathcal{J}_{LM,p-1}}\sum_{i=1}^{|L_a|} \sum_{j=1}^{|M_b|} (-1)^{a-p+1} \binom{a+b}{p-1} \left(X_{k+1} \left(X_{k+1} \left(\mathcal{X}_{K \cup \left(L_a\right)_i} \mathcal{X}_{K \cup \left(M_b\right)_j} \right)_{0} \right)_0 \right)_{-(a+b-p+1)} \notag \\
                    &\quad - \delta_{1 \leq p \leq |L| +|M|}\sum_{(a,b) \in \mathcal{J}_{LM,p}}\sum_{q=p}^{a+b}\sum_{i=1}^{|L_a|} \sum_{j=1}^{|M_b|} (-1)^{a-p+1} \binom{a+b}{p-1} \notag \\
                    &\quad \times \left(X_{k+1} \left(X_{k+1} \left(\mathcal{X}_{K \cup \left(L_a\right)_i} \mathcal{X}_{K \cup \left(M_b\right)_j} \right)_{-(a+b-q)} \right)_0 \right)_{p-q-1}~.
                    \label{10-prop3}
                \end{align}
                
                Thus, (\ref{6-prop3}) is equal to the following summation by the above calculations.
                \begin{align}
                    &\sum_{(a,b) \in \mathcal{J}_{LM,p-1}} \sum_{i=1}^{|L_a|} \sum_{j=1}^{|M_b|} (-1)^{a-p+1} \binom{a+b}{p-1}
                    \left(X_{k+1} \left(X_{k+1} \left(\mathcal{X}_{K \cup \left(L_a\right)_i} \mathcal{X}_{K \cup \left(M_b\right)_j} \right)_{-(a+b-p+1)} \right)_0 \right)_{0}~.
                    \label{11-prop3}
                \end{align}
            \end{itemize}
            
            Since summarizing the above results (\ref{5-prop3}) and (\ref{11-prop3}), we can derive the following result.
            \begin{align}
                &\sum_{(c,d) \in \mathcal{J}_{L'M,p}} \sum_{i=1}^{|L'_{c}|} \sum_{j=1}^{|M_d|} (-1)^{c-p} \binom{c+d-1}{p-1} 
                \left(\tilde{\mathcal{X}}_{K \cup \left(L'_c\right)_i} \tilde{\mathcal{X}}_{K \cup (M_d)_j}\right)_{-(c+d-p)} \notag \\
                &= \delta_{1 \leq p \leq |L|+|M|}\sum_{(a,b) \in \mathcal{J}_{LM,p}} \sum_{i=1}^{|L_a|} \sum_{j=1}^{|M_b|} (-1)^{a-p} \binom{a+b-1}{p-1} \left(X_{k+1} \left(X_{k+1} \left(\mathcal{X}_{K \cup \left(L_a\right)_i} \mathcal{X}_{K \cup \left(M_b\right)_j}\right)_{-(a+b-p+1)} \right)_0 \right)_0 \notag \\
                &\quad + \sum_{(a,b) \in \mathcal{J}_{LM,p-1}} \sum_{i=1}^{|L_a|} \sum_{j=1}^{|M_b|} (-1)^{a-p+1} \binom{a+b}{p-1}
                \left(X_{k+1} \left(X_{k+1} \left(\mathcal{X}_{K \cup \left(L_a\right)_i} \mathcal{X}_{K \cup \left(M_b\right)_j} \right)_{-(a+b-p+1)} \right)_0 \right)_{0} \notag \\
                &\quad +\delta_{1 \leq p \leq |L|+|M|} \sum_{(a,b) \in \mathcal{J}_{LM,p}} \sum_{q=p}^{a+b} \sum_{i=1}^{|L_a|} \sum_{j=1}^{|M_b|} (-1)^{a-q} \binom{a+b-1}{q-1} \notag
            \end{align}
            \begin{align}
                \times \left(P_{k+1} \left(Q_{k+1} \left(\mathcal{X}_{K \cup \left(L_a\right)_i} \mathcal{X}_{K \cup \left(M_b\right)_j} \right)_{-(a+b-q)} \right)_0 \right)_{p-q}~.
                \label{12-prop3}
            \end{align}

            Finally, let us transform the first and second terms on the right-hand side of (\ref{12-prop3}) as follows.
            \begin{align}
                &\delta_{p=1} \left(X_{k+1} \left(X_{k+1} \left(\mathcal{X}_{K \cup L} \mathcal{X}_{K \cup M} \right)_{0} \right)_0 \right)_{0} \notag \\
                &\quad +\delta_{2 \leq p \leq |L|+|M|+1}\sum_{(a,b) \in \mathcal{J}_{LM,p-1}} \sum_{i=1}^{|L_a|} \sum_{j=1}^{|M_b|} (-1)^{a-p+1} \binom{a+b-1}{p-2} \notag \\
                &\quad \times \left(X_{k+1} \left(X_{k+1} \left(\mathcal{X}_{K \cup \left(L_a\right)_i} \mathcal{X}_{K \cup \left(M_b\right)_j}\right)_{-(a+b-p+1)} \right)_0 \right)_0~.
                \label{13-prop3}
            \end{align}
            Since (\ref{13-prop3}), we can make the right-hand side of (\ref{12-prop3}) the following expression.
            \begin{align}
                \label{14-prop3}
                    &\delta_{p=1} \left(X_{k+1} \left(X_{k+1} \left(\mathcal{X}_{K \cup L} \mathcal{X}_{K \cup M}\right)_0 \right)_0 \right)_0 \notag \\
                    &\quad +\delta_{2 \leq p \leq |L|+ |M|+1} \sum_{(a,b) \in \mathcal{J}_{LM,p-1}} \sum_{i=1}^{|L_a|} \sum_{j=1}^{|M_b|} (-1)^{a-p+1} \binom{a+b-1}{p-2} \notag \\
                    &\quad \times \left(X_{k+1} \left(X_{k+1} \left(\mathcal{X}_{K \cup \left(L_a\right)_i} \mathcal{X}_{K \cup \left(M_b\right)_j} \right)_{-(a+b-p+1)} \right)_0 \right)_0 \notag\\
                    &\quad +\delta_{1 \leq p \leq |L|+ |M|}\sum_{(a,b) \in \mathcal{J}_{LM,p}} \sum_{q=p}^{a+b} \sum_{i=1}^{|L_a|} \sum_{j=1}^{|M_b|} (-1)^{a-q} \binom{a+b-1}{q-1} \notag \\
                    &\quad \times \left(P_{k+1} \left(Q_{k+1} \left(\mathcal{X}_{K \cup \left(L_a\right)_i} \mathcal{X}_{K \cup \left(M_b\right)j} \right)_{-(a+b-q)} \right)_0 \right)_{p-q}.
            \end{align}
            We can find out that (\ref{14-prop3}) is equal to (\ref{2-prop3}) in $1 \leq p \leq |L|+|M|+1$.
        \end{proof}
        \begin{prop}
            \label{prop4}
            Let $k$ be an integer satisfying $k \geq 2$. Assuming that Theorem \ref{thm2} for $n=k$ holds,
            Theorem \ref{thm2} for $n=k+1$ and $(P_{k+1},Q_{k+1})=(X_{k+1},D_{k+1} X_{k+1})$ also holds.
        \end{prop}
        \begin{proof}
            If $P_i =Q_i$ for some $i=1,\dots,k$, this case is equivalent to Proposition \ref{prop2} because we can execute the following commutations for $1 \leq i < j \leq k+1$.
            \begin{align*}
                &\left(P_{k+1} \left(\dots \left(P_j \left(\dots \left(P_i \left(\dots \left(P_2 P_1\right)_0 \dots \right)_0 \right)_0 \dots \right)_0 \right)_0 \dots \right)_0 \right)_0 \\
                &= \left(P_{k+1} \left(\dots \left(P_i \left(\dots \left(P_j \left(\dots \left(P_2 P_1\right)_0 \dots \right)_0 \right)_0 \dots \right)_0 \right)_0 \dots \right)_0 \right)_0~,
            \end{align*}
            and
            \begin{align*}
                &\left(Q_{k+1} \left(\dots \left(Q_j \left(\dots \left(Q_i \left(\dots \left(Q_2 Q_1\right)_0 \dots \right)_0 \right)_0 \dots \right)_0 \right)_0 \dots \right)_0 \right)_0 \\
                &= \left(Q_{k+1} \left(\dots \left(Q_i \left(\dots \left(Q_j \left(\dots \left(Q_2 Q_1\right)_0 \dots \right)_0 \right)_0 \dots \right)_0 \right)_0 \dots \right)_0 \right)_0~.
            \end{align*}
            Similarly, if $(P_i,Q_i) = (D_i X_i,X_i)$ for some $i=1,\dots,k$, we can also make this case Proposition \ref{prop3}.
            Thus, we must only prove the case when $(P_i,Q_i)=(X_i,D_iX_i)$ for all $i=1,\dots,k+1$, that is, we set $\mathcal{X} = \prod_{i=1}^{k} X_i$, $\mathcal{P}_{k} = \mathcal{X}$ and $\mathcal{Q}_{k} = \mathcal{X}_{M} = \prod_{i=1}^{k} D_i \mathcal{X}$.
            
            Substituing $R_{P_{k+1},Q_{k+1}} = -\left(X_{k+1} X_{k+1}\right)_0$, let us prove Proposition \ref{prop4} in the similar way how to prove Proposition \ref{prop2} and \ref{prop3}.
            \begin{align}
                \label{2-prop4}
                    \left(\mathcal{P}_{k+1} \mathcal{Q}_{k+1}\right)_p 
                    &= -\delta_{p=1} \left(X_{k+1} \left(X_{k+1} \left(\mathcal{X} \mathcal{X}_{M} \right)_0 \right)_0 \right)_0 \notag\\
                    &\quad +\delta_{2 \leq p \leq |M| +1} \sum_{b=p-1}^{|M|} \sum_{j=1}^{|M_b|} (-1)^{-p} \binom{b-1}{p-2}
                    \left(X_{k+1} \left(X_{k+1} \left(\mathcal{X} \mathcal{X}_{\left(M_b\right)_j} \right)_{-(b-p+1)} \right)_0 \right)_0 \notag\\
                    &\quad + \delta_{1 \leq p \leq |M|} \sum_{b=p}^{|M|} \sum_{q=p}^{b} \sum_{j=1}^{|M_b|} (-1)^{-q} \binom{b-1}{q-1} 
                    \left(P_{k+1} \left(Q_{k+1} \left(\mathcal{X} \mathcal{X}_{\left(M_b\right)_j} \right)_{-(b-q)} \right)_0 \right)_{p-q}~.
            \end{align}

            When the set $M'$ is defined by $M' \coloneqq M \cup \{k+1\}$ and we set 
            \begin{align*}
                \tilde{\mathcal{X}}_{\phi} &\coloneqq \left(X_{k+1} \mathcal{X}\right)_0~,\\
                \tilde{\mathcal{X}}_{\left(M'_b\right)_j} &\coloneqq \left\{
                    \begin{aligned}
                        &\left(D_{k+1} X_{k+1} \mathcal{X}_{\left(M'_b\right)_j \setminus \{k+1\}}\right)_0 &\quad (k+1 \in \left(M'_b\right)_j) \\
                        &\left(X_{k+1} \mathcal{X}_{\left(M'_b\right)_j}\right)_0 &\quad (k+1 \notin \left(M'_b\right)_j)
                    \end{aligned}
                \right.~,
            \end{align*}
            we will prove the following equality.
            \begin{align}
                \label{3-prop4}
                \left(\mathcal{P}_{k+1} \mathcal{Q}_{k+1}\right)_{p} = \left\{
                    \begin{aligned}
                        &\sum_{d=p}^{|M'|} \sum_{j=1}^{|M'_d|} (-1)^{-p} \binom{d-1}{p-1} 
                        \left(\tilde{\mathcal{X}}_{\phi}\tilde{\mathcal{X}}_{\left(M'_d\right)_j}\right)_{-(d-p)} &(1 \leq p \leq |M|+1) \\
                        \\
                        &0 &(p \geq |M| +2)
                    \end{aligned}
                \right.~.
            \end{align}
            For $p \geq |M| +2$, (\ref{2-prop4}) are equal to (\ref{3-prop4}). So we will prove the equality for $1 \leq p \leq |M|+1$.

            Expanding the right-hand side of (\ref{3-prop4}) for $1 \leq p \leq |M|+1$, we can get the following result.
            \begin{align}
                \label{4-prop4}
                    &\sum_{d=p}^{|M'|} \sum_{j=1}^{|M'_d|} (-1)^{-p} \binom{d-1}{p-1} \left(\tilde{\mathcal{X}}_{\phi}\tilde{\mathcal{X}}_{\left(M'_d\right)_j}\right)_{-(d-p)} \notag\\
                    &= \delta_{1 \leq p \leq |M|}\sum_{d=p}^{|M|} \sum_{j=1}^{|M_d|} (-1)^{-p} \binom{d-1}{p-1} 
                    \left(\left(P_{k+1} \mathcal{X}\right)_0 \left(Q_{k+1} \mathcal{X}_{\left(M_d\right)_j}\right)_0 \right)_{-(d-p)} \notag\\
                    &\quad +\delta_{1 \leq p \leq |M|+1}\sum_{d=p}^{|M|+1} \sum_{j=1}^{|M_{d-1}|} (-1)^{-p} \binom{d-1}{p-1} 
                    \left(\left(X_{k+1} \mathcal{X}\right)_0 \left(X_{k+1} \mathcal{X}_{\left(M_{d-1}\right)_j}\right)_0 \right)_{-(d-p)}~,
            \end{align}
            where the first term on the right-hand side of (\ref{4-prop4}) only emerges in $J \neq \phi$.

            Similar to Proposition \ref{prop3}, the first term on the right-hand side of (\ref{4-prop4}) can be transformed as follows.
            \begin{align}
                \label{5-prop4}
                    &\delta_{1 \leq p \leq |M|}\sum_{d=p}^{|M|} \sum_{j=1}^{|M_d|} (-1)^{-p} \binom{d-1}{p-1} 
                    \left(\left(P_{k+1} \mathcal{X}\right)_0 \left(Q_{k+1} \mathcal{X}_{\left(M_d\right)_j}\right)_0 \right)_{-(d-p)} \notag \\
                    &= -\delta_{1 \leq p \leq |M|}\sum_{b=p}^{|M|} \sum_{j=1}^{|M_b|} (-1)^{-p} \binom{b-1}{p-1} \left(X_{k+1} \left(X_{k+1} \left(\mathcal{X} \mathcal{X}_{\left(M_b\right)_j} \right)_{-(b-p+1)} \right)_0 \right)_0 \notag \\
                    &\quad + \delta_{1 \leq p \leq |M|}\sum_{b=p}^{|M|} \sum_{q=p}^b \sum_{j=1}^{|M_b|} (-1)^{-q} \binom{b-1}{q-1}
                    \left(P_{k+1} \left(Q_{k+1} \left(\mathcal{X} \mathcal{X}_{\left(M_b\right)_j} \right)_{-(b-q)} \right)_0 \right)_{p-q}~.
            \end{align}
            The second term can also be the following equality.
            \begin{align}
                \label{6-prop4}
                    &\delta_{1 \leq p \leq |M|+1}\sum_{b=p}^{|M|+1} \sum_{j=1}^{|M_{d-1}|} (-1)^{-p} \binom{d-1}{p-1} 
                    \left(\left(X_{k+1} \mathcal{X}\right)_0 \left(X_{k+1} \mathcal{X}_{\left(M_{d-1}\right)_j}\right)_0 \right)_{-(d-p)} \notag\\
                    &= \delta_{1 \leq p \leq |M|+1} \sum_{b=p-1}^{|M|} \sum_{j=1}^{|M_b|} (-1)^{-p} \binom{b}{p-1}
                    \left(X_{k+1} \left(X_{k+1} \left(\mathcal{X} \mathcal{X}_{\left(M_b\right)_j}\right)_{-(b-p+1)} \right)_0 \right)_0~.
            \end{align}
            
            From (\ref{5-prop4}) and (\ref{6-prop4}), we can derive the following result.
            \begin{align}
                \label{7-prop4}
                    &-\delta_{p=1} \left(X_{k+1} \left(X_{k+1} \left(\mathcal{X} \mathcal{X}_{M}\right)_0 \right)_0 \right)_0 \notag\\
                    &\quad + \delta_{2 \leq p \leq |M|+1} \sum_{b=p-1}^{|M|} \sum_{j=1}^{|M_b|}(-1)^{-p} \binom{b-1}{p-2}
                    \left(X_{k+1} \left(X_{k+1} \left(\mathcal{X} \mathcal{X}_{\left(M_b\right)_j} \right)_{-(b-p+1)} \right)_0 \right)_0 \notag\\
                    &\quad + \delta_{1 \leq p \leq |M|} \sum_{b=p}^{|M|} \sum_{q=p}^{b} \sum_{j=1}^{|M_b|} (-1)^{-q} \binom{b-1}{q-1}
                    \left(P_{k+1} \left(Q_{k+1} \left(\mathcal{X} \mathcal{X}_{\left(M_b\right)_j} \right)_{-(b-q)}\right)_0 \right)_{p-q}~.
            \end{align}
            (\ref{7-prop4}) is equal to (\ref{2-prop4}) for $1 \leq p \leq |M|+1$.
        \end{proof}

        Summarizing the above propositions, we can finally prove Theorem \ref{thm2}.
        \begin{proof}
            The proof of Theorem \ref{thm2} is executed by mathematical induction. 
            Proposition \ref{prop1}, \ref{prop2}, \ref{prop3} and \ref{prop4} are equivalent to each step of this inductive proof.
        \end{proof}

        Rewriting Theorem \ref{thm2} in the OPE, we get the following Corollary.
        \begin{cor}
            \label{cor2}
            When $n$ is an integer satisfying $n \geq 2$, the OPE between $\mathcal{X}_{I}$ and $\mathcal{X}_{J}$ is given as follows.
            \begin{itemize}
                \item When $I \neq J$,
                \begin{align*}
                        \mathcal{X}_{I}(z) \mathcal{X}_{J}(0)
                        &\sim \sum_{p=1}^{|L|+|M|}\sum_{(a,b)\in \mathcal{J}_{LM,p}}\sum_{i=1}^{|L_a|}\sum_{j=1}^{|M_b|}\frac{1}{(z-w)^p} \\
                        &\quad \times \frac{(-1)^{a-p}}{\left(a+b-p\right)!} \binom{a+b-1}{p-1}\left(\partial^{a+b-p}\mathcal{X}_{K\cup (L_a)_i}\right)\mathcal{X}_{K\cup (M_b)_j}~,
                \end{align*}
                where we set $K= I \cap J$, $L = I \setminus K$ and $M= J \setminus K$.
                \item When $I=J$,
                \begin{align*}
                    \mathcal{X}_{I}(z) \mathcal{X}_{I}(0) \sim 0~.
                \end{align*}
            \end{itemize}
        \end{cor}
            
            Similarly, the following theorem also holds.
            \begin{thm}
                \label{thm3}
                When $n$ is an integer satisfying $n \geq 2$, the product $\left(\mathcal{Y}_{I} \mathcal{Y}_{J}\right)_p$ for $p \geq 1$ is given as follows.
                \begin{itemize}
                    \item When $I \neq J$,
                    \begin{align*}
                            \left(\mathcal{Y}_{I} \mathcal{Y}_{J}\right)_p 
                            &= \left\{
                                \begin{aligned}
                                    &\sum_{(a,b) \in \mathcal{J}_{LM,p}} \sum_{i=1}^{|L_a|} \sum_{j=1}^{|M_b|} (-1)^{b-p} \binom{a+b-1}{p-1}\\
                                    &\quad \times \left(\mathcal{Y}_{K \cup \left(L_a\right)_i} \mathcal{Y}_{K \cup \left(M_b\right)_j}\right)_{p-a-b} &(1 \leq p \leq |L| +|M|)\\
                                    &0 &(p \geq |L| + |M| +1)
                                \end{aligned}
                            \right.~,
                    \end{align*}
                    where we set $K = I \cap J$, $L = I \setminus K$ and $M = J \setminus K$.
                    \item When $I=J$,
                    \begin{align*}
                        \left(\mathcal{Y}_{I} \mathcal{Y}_{I}\right)_p = 0~.
                    \end{align*}
                \end{itemize}
            \end{thm}
            \begin{proof}
                In this case, $P_i$ and $Q_i$ are equal to $Y_i$ or $D_i Y_i$. Then, we define $R_{P_i,Q_i}$ as follows. 
                \begin{align}
                    R_{P_i,Q_i} \coloneqq \left\{
                        \begin{aligned}
                            &\left(Y_i Y_i\right)_0 &\left(P_i =Y_i,\ Q_{i} = D_i Y_i\right) \\
                            &- \left(Y_i Y_i\right)_0 &\left(P_i =D_i Y_i,\ Q_{i} = Y_i\right) \\
                            &0 &\left(P_i = Q_{i}\right)
                        \end{aligned}
                    \right.~.
                \end{align}
                Using this $R_{P_i,Q_i}$, we can prove Theorem \ref{thm3}, similar to Theorem \ref{thm2}.
            \end{proof}

            Similarly, we can get the following Corollary from Theorem \ref{thm3}.
            \begin{cor}
                \label{cor3}
                When $n$ is an integer satisfying $n \geq 2$, the OPE between $\mathcal{Y}_{I}$ and $\mathcal{Y}_{J}$ for $p \geq 1$ is given as follows.
                \begin{itemize}
                    \item When $I \neq J$,
                    \begin{align*}
                            \mathcal{Y}_{I}(z) \mathcal{Y}_{J}(0)
                            &\sim \sum_{p=1}^{|L|+|M|}\sum_{(a,b)\in \mathcal{J}_{LM,p}}\sum_{i=1}^{|L_a|}\sum_{j=1}^{|M_b|}\frac{1}{(z-w)^p} \\
                            &\quad \times \frac{(-1)^{b-p}}{\left(a+b-p\right)!} \binom{a+b-1}{p-1}\left(\partial^{a+b-p}\mathcal{Y}_{K\cup (L_a)_i}\right)\mathcal{Y}_{K\cup (M_b)_j}~,
                    \end{align*}
                    where we set $K = I \cap J$, $L = I \setminus K$ and $M = J \setminus K$.
                    \item When $I=J$,
                    \begin{align*}
                        \mathcal{Y}_{I}(z) \mathcal{Y}_{I}(0) \sim 0~.
                    \end{align*}
                \end{itemize}
            \end{cor}

            Therefore, we have shown the OPEs (\ref{XX1}), (\ref{XX2}), (\ref{YY1}) and (\ref{YY2}). 
            On the other hand, we do not perfectly find out that the remaining OPE (\ref{XY}) is written by $U$, $T_i$ and $W_i$. Studying this remaining OPE is one of future works regarding with this bosonic VOA.
\section{The bosonic algebra associated with \texorpdfstring{$T^{[2,1^{n-1}]}_{[n-1,1^2]}(SU(n+1))$}{}}
    \label{The bosonic algebra associated with T^{[2,1^{n-1}]}_{[n-1,1^2]}(SU(n+1))}
    From this section to Sec.\ref{Generators of the bosonic VOA of T^{[2,1^{n-1}]}_{[n-1,1^2]}(SU(n+1))}, we study another theory, $T^{[2,1^{n-1}]}_{[n-1,1^2]}(SU(n+1))$. This theory is similar to $T_{[n-1,1]}^{[1^n]}(SU(n))$, but the difference is whether $SU(2)$ flavor symmetry is included or not.
    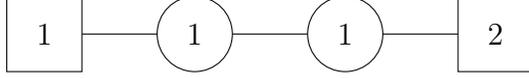
\begin{figure}[t]
        \center
        \begin{tikzpicture}
            \draw (0,-0.5) rectangle (1,0.5);
            \draw (0.5,0) node{1};
            \draw (1,0)--(2,0);
            \draw (2.5,0) circle[radius=0.5cm];
            \draw (2.5,0) node{1};
            \draw (3,0)--(4,0);
            \draw (4.5,0) circle[radius=0.5cm];
            \draw (4.5,0) node{1};
            \draw (5,0)--(6,0);
            \draw (6,-0.5) rectangle (7,0.5);
            \draw (6.5,0) node{2};
        \end{tikzpicture}
        \caption{The quiver diagram for $T^{[2,1^{2}]}_{[3,1^2]}(SU(4))$.}
    \end{figure}
    Then, the definition of the symplectic bosons and the $bc$-ghost are same for $T_{[n-1,1]}^{[1^n]}(SU(n))$ except for those regarding with $SU(2)$ flavor symmetry. The remined bosons $(X_{n}^{(1)},Y_{n,(1)})$ and $(X^{(2)}_{n},Y_{n,(2)})$ are satisfying the OPEs.
    \begin{align}
        X^{(i)}_{n}(z) Y_{n,(j)}(0) \sim \frac{\delta^i_j}{z}~.
    \end{align}
    Therefore, by using $(X_i,Y_i)$, $(X^{(i)}_{n},Y_{n,(i)})$, $(b_a,c^a)$ and the currents removing $U(1)$ gauge anomlies, let us construct the BRST cohomology and the bosonic VOA associated with $T^{[2,1^{n-1}]}_{[n-1,1^2]}(SU(n+1))$.
    \subsection{BRST reduction}
        To do the BRST reduction for $T^{[2,1^{n-1}]}_{[n-1,1^2]}(SU(n+1))$, we define the currents $J_{XY,i}$ as follows.
        \begin{align}
            \left\{
                \begin{aligned}
                    J_{XY,i} &\coloneqq X_i Y_i - X_{i+1}Y_{i+1} &\quad (i=1,\dots,n-2)\\
                    J_{XY,n-1} &\coloneqq X_{n-1} Y_{n-1} -\sum_{j=1}^{2}X_n^{(j)}Y_{n,(j)}
                \end{aligned}
            \right.~.
        \end{align}
        These currents are satisfying the following OPEs by the definition.
        \begin{align}
            J_{XY,i}(z) J_{XY,j}(0) \sim \frac{-\tilde{C}_{ij}}{z^2}~,
        \end{align}
        where $\tilde{C}_{ij} \coloneqq C_{ij} + \delta_{i,n-1} \delta_{j,n-1}$, and $C_{ij}$ is the Cartan matrix of $A_{n-1}$.
        In order to remove $U(1)$ gauge anomalies, we introduce the Heisenberg currents $h_i$ satisfying the following OPEs.
        \begin{align}
            h_i(z) h_j(0) \sim \frac{\tilde{C}_{ij}}{z^2}~.
        \end{align}

        Using $J_{XY,i}$ and $h_i$, we can define the BRST current $J_{BRST}$ as follows.
        \begin{align}
            J_{BRST} \coloneqq \sum_{a=1}^{n-1}c^a \left(J_{XY,a} + h_a\right)~.
        \end{align}
        By the above definition, the following result holds.
        \begin{align}
            J_{BRST}(z) J_{BRST}(0) &\sim 0~.
        \end{align}
        We also set the BRST charge $Q$ by $J_{BRST}$,
        \begin{align}
            Q \coloneqq \oint\frac{dz}{2\pi i}J_{BRST}(z)~,
        \end{align}
        and we can find out the nilpotency of $Q$, $Q^2 = 0$. Thus, we can construct the BRST cohomology for $T^{[2,1^{n-1}]}_{[n-1,1^2]}(SU(n+1))$.
    \subsection{Closed operators}
        In this subsection, let us derive closed operators under the BRST cohomology for this theory, $T^{[2,1^{n-1}]}_{[n-1,1^2]}(SU(n+1))$.
        For the bosonic VOA of $T_{[n-1,1]}^{[1^n]}(SU(n))$, we know that closed operators using $\mathcal{D}_i$ and $\hat{\mathcal{A}}_i$ exist and are needed.
        Similarly, we also need the following operations $\mathcal{D}_i,\hat{\mathcal{A}}_i$ and $\mathcal{D}_n^{(j)},\hat{\mathcal{A}}_{n}^{(j)}$ to construct closed operators;
        \begin{align}
            \left\{
                \begin{aligned}
                    \mathcal{D}_i \mathcal{O} &\coloneqq \mathcal{A}_i \mathcal{O} + X_i Y_i \mathcal{O} \\
                    \hat{\mathcal{A}}_i \mathcal{O} &\coloneqq \mathcal{A}_i \mathcal{O}\\
                    D_i \mathcal{O} &\coloneqq \mathcal{D}_i \mathcal{O} + \hat{\mathcal{A}}_i \mathcal{O}
                \end{aligned}
            \right.\quad (i=1,\dots,n-1)~,
        \end{align}
        and
        \begin{align}
            \left\{
                \begin{aligned}
                    \mathcal{D}_n^{(j)} \mathcal{O} &\coloneqq \mathcal{A}_n^{(j)} \mathcal{O} + X_n^{(j)} Y_{n,(j)} \mathcal{O} \\
                    \hat{\mathcal{A}}_n^{(j)} \mathcal{O} &\coloneqq \mathcal{A}_n^{(j)} \mathcal{O}\\
                    D_n^{(j)} \mathcal{O} &\coloneqq \mathcal{D}_n^{(j)} \mathcal{O} + \hat{\mathcal{A}}_n^{(j)} \mathcal{O}
                \end{aligned}
            \right.\quad (j=1,2)~,
        \end{align}
        where $\mathcal{A}_{n}^{(j)} \coloneqq -X^{(j)}_{n} Y_{n,(j)}$, $\mathcal{O}$ is an operator, and $(\mathcal{D}_i,\hat{\mathcal{A}}_i,D_i)$ and $(\mathcal{D}_n^{(j)},\hat{\mathcal{A}}_n^{(j)},D_n^{(j)})$ are studied in Sec.\ref{To derive non-trivial operators}.

        Thus, we can construct closed operators. The following operators come from the Higgs branch chiral operators for $T^{\left[2,1^{n-1}\right]}_{\left[n-1,1^2\right]} (SU(n+1))$.
        \begin{align}
            U \coloneqq \frac{1}{2n-1}\left(2 \sum_{i=1}^{n-1}\mathcal{A}_{i} +\sum_{j=1}^{2}\mathcal{A}^{(j)}_{n}\right)~,
        \end{align}
        \begin{align}
            H \coloneqq \frac{1}{2}\left(\mathcal{A}^{(1)}_{n} -\mathcal{A}^{(2)}_{n}\right)~,\quad E \coloneqq X_{n}^{(1)}Y_{n,(2)}~,\quad F \coloneqq X_{n}^{(2)}Y_{n,(1)}~,
        \end{align}
        \begin{align}
            \mathcal{X}^{+} \coloneqq \left(\prod_{i=1}^{n-1} X_{i}\right) X_{n}^{(1)}~,\quad \mathcal{X}^{-} \coloneqq \left(\prod_{i=1}^{n-1} X_{i}\right) X_{n}^{(2)}~,
        \end{align}
        \begin{align}
            \mathcal{Y}^{+} \coloneqq \left(\prod_{i=1}^{n-1} Y_i\right) Y_{n,(2)}~,\quad \mathcal{Y}^{-} \coloneqq \left(\prod_{i=1}^{n-1} Y_i\right) Y_{n,(1)}~.
        \end{align}
        Other operators using $\mathcal{D}_i$ and $\hat{\mathcal{A}}_i$ for $i =1,\dots,n-1$ can be given as follows.
        \begin{align}
            T_i \coloneqq \frac{1}{2}D_i \mathcal{A}_i~,
        \end{align}
        \begin{align}
            W_i \coloneqq \frac{1}{3}\sqrt{\frac{2}{3}} \left\{\mathcal{D}_i \left(\mathcal{D}_i \mathcal{A}_i\right) +\frac{1}{2}\mathcal{D}_{i} \left(\hat{\mathcal{A}}_{i} \mathcal{A}_{i}\right) +\frac{1}{2}\hat{\mathcal{A}}_i \left(\mathcal{D}_i\mathcal{A}_i\right) +\hat{\mathcal{A}}_{i} \left(\hat{\mathcal{A}}_{i} \mathcal{A}_{i}\right)\right\}~,
        \end{align}
        \begin{align}
            \mathcal{X}^{+}_{I} \coloneqq \prod_{i \in I}D_i \mathcal{X}^{+}~,\quad \mathcal{X}^{-}_{I} \coloneqq \prod_{i \in I}D_i \mathcal{X}^{-}~,\quad \mathcal{Y}^{+}_{I} \coloneqq \prod_{i \in I}D_i \mathcal{Y}^{+}~,\quad \mathcal{Y}^{-}_{I} \coloneqq \prod_{i \in I}D_i \mathcal{Y}^{-}~,
        \end{align}
        where we set $I \subset \{1,\dots, n-1\}$, and $\mathcal{X}^{\pm}_{\phi} =\mathcal{X}^{\pm}$ and $\mathcal{Y}^{\pm}_{\phi} = \mathcal{Y}^{\pm}$ for convenience.
        Moreover, the following operators using $\mathcal{D}_n^{(j)}$ and $\hat{\mathcal{A}}_n^{(j)}$ for $j =1,2$ can be also given as follows.
        \begin{align}
            T_n^{(j)} \coloneqq \frac{1}{2}D_{n}^{(j)} \mathcal{A}^{(j)}_{n}~,
        \end{align}
        \begin{align}
            W_n^{(j)} \coloneqq \frac{1}{3}\sqrt{\frac{2}{3}} \left\{\mathcal{D}_n^{(j)} \left(\mathcal{D}_n^{(j)} \mathcal{A}^{(j)}_{n}\right) +\frac{1}{2}\mathcal{D}_{n}^{(j)} \left(\hat{\mathcal{A}}_{n}^{(j)} \mathcal{A}^{(j)}_{n}\right) +\frac{1}{2}\hat{\mathcal{A}}_n^{(j)} \left(\mathcal{D}_n^{(j)} \mathcal{A}^{(j)}_{n}\right) +\hat{\mathcal{A}}_{n}^{(j)} \left(\hat{\mathcal{A}}_{n}^{(j)} \mathcal{A}^{(j)}_{n}\right)\right\}~,
        \end{align}
        \begin{align}
            E_{J} \coloneqq \prod_{j \in J} D_{n}^{(j)} E~,\quad F_{J} \coloneqq \prod_{j \in J} D_{n}^{(j)} F~,
        \end{align}
        \begin{align}
            D_n^{(1)} \mathcal{X}^{+}_{I}~,\quad D_n^{(2)} \mathcal{X}^{-}_{I}~,\quad D_n^{(2)}\mathcal{Y}^{+}_{I}~,\quad D_n^{(1)} \mathcal{Y}^{-}_{I}~,
        \end{align}
        where we set $J \subset \{1,2\}$, and $E_{\phi} =E$ and $F_{\phi} = F$ for convenience.

        Let us introduce the stress tensor $T$ for $T^{\left[2,1^{n-1}\right]}_{\left[n-1,1^2\right]} (SU(n+1))$.\footnote{This can be checked in \cite{Nishinaka:2025nbe}.} Since $\tilde{C}_{ij}$ is invertible,
        we can set the following formula by using $(X_i,Y_i)$, $(X_n^{(j)},Y_{n,(j)})$, $h_i$ and $(b_a,c^a)$.
        \begin{align}
            T \coloneqq T_{sb} + T_h + T_{bc}~,
        \end{align}
        \begin{align}
            T_{sb} \coloneqq \frac{1}{2}\sum_{i=1}^{n-1}\left(X_i \partial Y_i - \partial X_i Y_i\right)+\frac{1}{2}\sum_{j=1}^{2}\left(X_n^{(j)} \partial Y_{n,(j)} - \partial X_n^{(j)} Y_{n,(j)}\right)~,
        \end{align}
        \begin{align}
            T_h \coloneqq \frac{1}{2}\sum_{i,j=1}^{n-1} \tilde{C}^{ij}h_i h_j~,\quad T_{bc} \coloneqq -\sum_{a=1}^{n-1} b_a \partial c^a~.
        \end{align}
        Removing $T_h +T_{bc}$ from the above stress tensor under the BRST cohomology, we can also transform it to the following equivalent formula.
        \begin{align}
            T \equiv \sum_{i=1}^{n-1}T_i + \sum_{j=1}^{2}T_n^{(j)} -\frac{2n-1}{4}U^2 -H^2 ~.
        \end{align}
        Thus, the stress tensor $T$ can be also written by $T_i$, $T_{n}^{(j)}$, $U$ and $H$.
    \subsection{Composite operators}
        In the previous section, we derived closed operators. Before constructing generators from closed operators,
        we give composite relations to operators using $X_{n}^{(j)}$ or $Y_{n,(j)}$. The part of these relations can be inferred from the bosonic VOA associated with $T^{\left[1^n\right]}_{\left[n-1,1\right]} (SU(n))$ for $n=2$, as follows.
        \begin{align}
            T_{n}^{(1)} + T_{n}^{(2)} = 2H^2 +\frac{1}{2} \left(EF + FE\right)~,
        \end{align}
        \begin{align}
            W_n^{(j)} &= \frac{1}{3}\sqrt{\frac{2}{3}} \left\{
            2H \left(T_n^{(1)} -T_{n}^{(2)}\right) +\frac{1}{2}F \left(D_n^{(1)} +D_{n}^{(2)}\right)E +\frac{1}{4}\partial \left(T_n^{(1)}-T_{n}^{(2)}\right)\right\} \notag \\
            &\quad +(-1)^{j}\sqrt{\frac{2}{3}}\left\{H \partial H +\frac{1}{4}E \partial F - \frac{1}{4} F \partial E -\frac{4}{3} H^3 -H E F -\frac{1}{3} \partial^2 H\right\}~,
        \end{align}
        \begin{align}
            \label{composite E}
            \left\{
                \begin{aligned}
                    \begin{bmatrix}
                        1 & -1
                    \end{bmatrix}
                    \begin{bmatrix}
                        E_{\{1\}} \\
                        E_{\{2\}}
                    \end{bmatrix}
                    &=
                    \begin{bmatrix}
                        \partial E +2 HE
                    \end{bmatrix}\\
                    \begin{bmatrix}
                        -1 \\
                        1
                    \end{bmatrix}
                    \begin{bmatrix}
                        E_{\{1,2\}}
                    \end{bmatrix}
                    &=
                    \begin{bmatrix}
                        \partial E_{\{1\}} -2T_n^{(1)}E +2 HE_{\{1\}} \\
                        \partial E_{\{2\}} +2T_n^{(2)}E +2 HE_{\{2\}}
                    \end{bmatrix} \\
                \end{aligned}
            \right.~,
        \end{align}
        and
        \begin{align}
            \label{composite F}
            \left\{
                \begin{aligned}
                    \begin{bmatrix}
                        1 & -1
                    \end{bmatrix}
                    \begin{bmatrix}
                        F_{\{1\}} \\
                        F_{\{2\}}
                    \end{bmatrix}
                    &=
                    \begin{bmatrix}
                        -\partial F +2 HF
                    \end{bmatrix}\\
                    \begin{bmatrix}
                        1 \\
                        -1
                    \end{bmatrix}
                    \begin{bmatrix}
                        F_{\{1,2\}}
                    \end{bmatrix}
                    &=
                    \begin{bmatrix}
                        \partial F_{\{1\}} +2T_n^{(1)}F -2 HF_{\{1\}} \\
                        \partial F_{\{2\}} -2T_n^{(2)}F -2 HF_{\{2\}}
                    \end{bmatrix}
                \end{aligned}
            \right.~.
        \end{align}
        Furthermore, we can find out the following equalities.
        \begin{align}
            \left\{
                \begin{aligned}
                    D_{n}^{(1)}\mathcal{X}^{+}_{I} &= 2H \mathcal{X}^{+}_{I} -E \mathcal{X}^{-}_{I}\\
                    D_{n}^{(2)}\mathcal{X}^{-}_{I} &= -2H \mathcal{X}^{-}_{I} -F \mathcal{X}^{+}_{I}\\
                    D_{n}^{(2)}\mathcal{Y}^{+}_{I} &= -2H \mathcal{Y}^{+}_{I} -E \mathcal{Y}^{-}_{I}\\
                    D_{n}^{(1)}\mathcal{Y}^{-}_{I} &= 2H \mathcal{Y}^{-}_{I} -F \mathcal{Y}^{+}_{I}
                \end{aligned}
            \right.~.
        \end{align}

        Therefore, when we construct generators regarding with the bosonic VOA associated with $T^{\left[2,1^{n-1}\right]}_{\left[n-1,1^2\right]} (SU(n+1))$, we can remove the above operators,
        and we should set and use the following operators:
        \begin{align}
            \hat{H} \coloneqq T_{n}^{(1)} -T_{n}^{(2)}~,
        \end{align}
        \begin{align}
            \hat{E} \coloneqq E_{1} + E_{2}~,
        \end{align}
        \begin{align}
            \hat{F} \coloneqq F_{1} + F_{2}~.
        \end{align}

        Thus, the candidates of generators are $U$, $H$, $E$, $F$, $T_i$, $W_i$, $\hat{H}$, $\hat{E}$, $\hat{F}$, $\mathcal{X}^{\pm}_{I}$ and $\mathcal{Y}^{\pm}_{I}$.
        The OPEs with these candidates are given in Appendix.\ref{OPE 2}.
\section{Generators of the bosonic VOA of \texorpdfstring{$T^{[2,1^{n-1}]}_{[n-1,1^2]}(SU(n+1))$}{}}
    \label{Generators of the bosonic VOA of T^{[2,1^{n-1}]}_{[n-1,1^2]}(SU(n+1))}
    In this section, we construct generators in the same manner as in Sec.\ref{How to construct primary fields from T i, mathcal X I, mathcal Y I} to \ref{Linearly independent operators consisting of X I and Y I}.
    Since we know how to derive primary or composite operators, we apply a similar way in this case.
    \subsection{\texorpdfstring{$T_i$}{} and \texorpdfstring{$T_n^{(j)}$}{}}
        Since we can find out that $T_1-T_2,\dots,T_{n-2}-T_{n-1}$ and $\hat{H}$ are primary because of the OPEs with the stress tensor $T$ in Appendix.\ref{OPE 2}, we set the following relations for $T_i$ and $T_{n}^{(j)}$,
        \begin{align}
            \begin{bmatrix}
                1 & 1 & 1 & \cdots & 1 & 1 & 1 & 1 \\
                1 & -1 & 0 & \cdots & 0 & 0 & 0 & 0 \\
                0 & 1 & -1 & \cdots & 0 & 0 & 0 & 0 \\
                \vdots & \vdots & \vdots & & \vdots & \vdots & \vdots & \vdots\\
                0 & 0 & 0 & \cdots & 1 & -1 & 0 & 0 \\
                0 & 0 & 0 & \cdots & 0 & 0 & 1 & -1 \\
                0 & 0 & 0 & \cdots & 0 & 0 & 1 & 1 
            \end{bmatrix}
            \begin{bmatrix}
                T_1 \\
                T_2 \\
                \vdots \\
                T_{n-2}\\
                T_{n-1} \\
                T_n^{(1)} \\
                T_n^{(2)}
            \end{bmatrix}
            \equiv
            \begin{bmatrix}
                T +\frac{2n-1}{4}U^2 + H^2 \\
                T_1 -T_2\\
                \vdots \\
                T_{n-2} -T_{n-1} \\
                \hat{H} \\
                2H^2 +\frac{1}{2}(EF +FE)
            \end{bmatrix}~,
        \end{align}
        where the matrix of the left-hand side is invertible.
        Thus, generators composed of $T_i$ and $T_n^{(j)}$ are $n-1$ primary operators, 
        $T_1-T_2,\dots,T_{n-2}-T_{n-1}$, $\hat{H}$ and the stress tensor $T$.\footnote{When $n=2$, these operators become composite, in \cite{Nishinaka:2025nbe}.}
    \subsection{\texorpdfstring{$W_i$}{} and \texorpdfstring{$W_n^{(j)}$}{}}
        Since $W_i$ and $W_{n}^{(j)}$ are primary operators respectively, we would like to choose them as the generators.
        We, however, know that the following relations regarding $W_n^{(j)}$ hold.
        Thus, $W_1,\dots,W_{n-1}$ are only generators.\footnote{When $n=2$, $W_1$ is also composite. When $n=3$, we choose the generator as $W_1-W_2$ because $W_1+W_2$ is composite, in \cite{Nishinaka:2025nbe}.}
    \subsection{\texorpdfstring{$E_{J}$}{} and \texorpdfstring{$F_{J}$}{}}
        Let us study generators consisting $E_{J}$ and $F_{J}$ here. We define $T_U$ and $T_H$ as follows.
        \begin{align}
            T_U &\coloneqq -\frac{2n-1}{4}U^2~,\\
            T_H &\coloneqq - H^2~.
        \end{align}
        Then, we calculate the OPE with $T_U$ and $T_H$.
        \begin{align}            
            T_U(z) E_{J}(0) &\sim 0~,\\
            T_U(z) F_{J}(0) &\sim 0~,\\
            T_H(z) E_{J}(0) &\sim -\frac{E_{J}}{z^2} - \frac{2H E_{J}}{z}~,\\
            T_H(z) F_{J}(0) &\sim -\frac{F_{J}}{z^2} + \frac{2H F_{J}}{z}~.
        \end{align}
        Since $T \equiv \sum_{i=1}^{n-1} T_i + \sum_{j=1}^{2}T_n^{(j)} +T_U +T_H$, we can derive the following relations from the $(-1)$-th degree of the OPEs with $T$.
        \begin{align}
            \partial E_{J} &= \left\{
                \begin{aligned}
                    &\sum_{j=1}^{2}(-1)^{j+1} E_{\{j\}} -2H E_{J} &\quad (J=\phi) \\
                    &\sum_{j \in \{1,2\} \setminus J} (-1)^{j+1}E_{\{1,2\}} +2\sum_{j \in J} (-1)^{j+1}T_n^{(j)} E_{\phi}-2HE_{J} &\quad (J = \{1\},\{2\}) \\
                    &2\sum_{j \in J} (-1)^{j+1}T_n^{(j)}E_{J \setminus \{j\}} -2H E_{J} &\quad (J=\{1,2\})
                \end{aligned}
            \right.~,
        \end{align}
        and
        \begin{align}
            \partial F_{J} &= \left\{
                \begin{aligned}
                    &\sum_{j=1}^{2}(-1)^{j} F_{\{j\}} +2H F_{J} &\quad (J=\phi) \\
                    &\sum_{j \in \{1,2\} \setminus J} (-1)^{j}F_{\{1,2\}} +2\sum_{j \in J} (-1)^{j}T_n^{(j)} F_{\phi} +2HF_{J} &\quad (J = \{1\},\{2\}) \\
                    &2\sum_{j \in J} (-1)^{j}T_n^{(j)}F_{J \setminus \{j\}} +2H F_{J} &\quad (J=\{1,2\})
                \end{aligned}
            \right.~.
        \end{align}
        We can transform the above results to (\ref{composite E}) and (\ref{composite F}) like (\ref{composite XI 1}) and (\ref{composite YI 1}).
        
        On the other hand, $\hat{E}$ and $\hat{F}$ are primary.
        We can respectively summarize $E_{\{1\}} -E_{\{2\}}$ and $E_{\{1\}}+E_{\{2\}}$, and $F_{\{1\}} -F_{\{2\}}$ and $F_{\{1\}} +F_{\{2\}}$ as follows.
        \begin{align}
            \begin{bmatrix}
                1 & -1 \\
                1 & 1 
            \end{bmatrix}
            \begin{bmatrix}
                E_{\{1\}} \\
                E_{\{2\}}
            \end{bmatrix}
            =
            \begin{bmatrix}
                \partial E +2 HE \\
                \hat{E}
            \end{bmatrix}~,
            \quad
            \begin{bmatrix}
                1 & -1 \\
                1 & 1 
            \end{bmatrix}
            \begin{bmatrix}
                F_{\{1\}} \\
                F_{\{2\}}
            \end{bmatrix}
            =
            \begin{bmatrix}
                -\partial F +2 HF \\
                \hat{F}
            \end{bmatrix}
            ~,
        \end{align}
        where the square matrices of the left-hand side which emerge in the above formulas are clearly invertible.
        Thus, generators consisting of $E_{J}$ or $F_{J}$ are $E$, $F$, $\hat{E}$ and $\hat{F}$.\footnote{When $n=2$, $E_{J}$ and $F_{J}$ for $J \neq \phi$ become composite, in \cite{Nishinaka:2025nbe}.}
    \subsection{\texorpdfstring{$\mathcal{X}^{\pm}_{I}$}{} and \texorpdfstring{$\mathcal{Y}^{\pm}_{I}$}{}}
        Here, let us deal with $\mathcal{X}^{\pm}_{I}$ and $\mathcal{Y}^{\pm}_{I}$. Fortunately, since we presented important statements in Sec.\ref{How to construct primary fields from T i, mathcal X I, mathcal Y I} through \ref{Linearly independent operators consisting of X I and Y I},
        we can use these statements.

        In order to derive the $(-1)$-th degree of the OPE with $T$, let us calclate the OPEs with $T_U$ and $T_H$ as follows.(Re-check)
        \begin{align}
            &\left\{
                \begin{aligned}
                    T_U(z) \mathcal{X}^{\pm}_{I}(0) &\sim -\frac{\left(2n-1\right)\mathcal{X}^{\pm}_{I}}{4z^2} -\frac{\left(2n-1\right)U\mathcal{X}^{\pm}_{I}}{2z}\\
                    T_U(z) \mathcal{Y}^{\pm}_{I}(0) &\sim -\frac{\left(2n-1\right)\mathcal{Y}^{\pm}_{I}}{4z^2} +\frac{\left(2n-1\right)U\mathcal{Y}^{\pm}_{I}}{2z}\\
                \end{aligned}
            \right.~,
        \end{align}
        and
        \begin{align}
            &\left\{
                \begin{aligned}
                    T_H(z) \mathcal{X}^{+}_{I} (0) &\sim -\frac{\mathcal{X}^{+}_{I}}{4z^2} -\frac{H \mathcal{X}^{+}_{I}}{z}\\
                    T_H(z) \mathcal{X}^{-}_{I} (0) &\sim -\frac{\mathcal{X}^{-}_{I}}{4z^2} +\frac{H \mathcal{X}^{-}_{I}}{z}\\
                    T_H(z) \mathcal{Y}^{+}_{I} (0) &\sim -\frac{\mathcal{Y}^{+}_{I}}{4z^2} -\frac{H \mathcal{Y}^{+}_{I}}{z}\\
                    T_H(z) \mathcal{Y}^{-}_{I} (0) &\sim -\frac{\mathcal{Y}^{-}_{I}}{4z^2} +\frac{H \mathcal{Y}^{-}_{I}}{z}
                \end{aligned}
            \right.~.
        \end{align}
        Similarly, we also calculate the OPEs with $T_{n}^{(1)}$ and $T_{n}^{(2)}$ as follows.
        \begin{align}
            \left\{
                \begin{aligned}
                    T_{n}^{(1)}(z) \mathcal{X}^{+}_{I} (0) &\sim \frac{\mathcal{X}^{+}_{I}}{z^2} +\frac{2H \mathcal{X}^{+}_{I} -E \mathcal{X}^{-}_{I}}{z}\\
                    T_{n}^{(1)}(z) \mathcal{X}^{-}_{I} (0) &\sim 0\\
                    T_{n}^{(1)}(z) \mathcal{Y}^{+}_{I} (0) &\sim 0\\
                    T_{n}^{(1)}(z) \mathcal{Y}^{-}_{I} (0) &\sim \frac{\mathcal{Y}^{-}_{I}}{z^2} -\frac{2H \mathcal{Y}^{-} -F \mathcal{Y}^{+}_{I}}{z}
                \end{aligned}
            \right.~,
        \end{align}
        and
        \begin{align}
            \left\{
                \begin{aligned}
                    T_{n}^{(2)}(z) \mathcal{X}^{+}_{I} (0) &\sim 0\\
                    T_{n}^{(2)}(z) \mathcal{X}^{-}_{I} (0) &\sim \frac{\mathcal{X}^{-}_{I}}{z^2} -\frac{2H \mathcal{X}^{-}_{I} +F \mathcal{X}^{+}_{I}}{z}\\
                    T_{n}^{(2)}(z) \mathcal{Y}^{+}_{I} (0) &\sim \frac{\mathcal{Y}^{+}_{I}}{z^2} +\frac{2H \mathcal{Y}^{+}_{I} + E\mathcal{Y}^{-}_{I}}{z}\\ 
                    T_{n}^{(2)}(z) \mathcal{Y}^{-}_{I} (0) &\sim 0
                \end{aligned}
            \right.~,
        \end{align}

        From the above results and Appendix.\ref{OPE 2}, we can check that the following equalities regarding with $\mathcal{X}^{\pm}_{I}$ and $\mathcal{Y}^{\pm}_{I}$ hold
        \begin{align}
            \partial \mathcal{X}^{+}_{I} &\equiv \left\{
                \begin{aligned}
                    &\sum_{i \in \mathcal{I}_{n-1}} \mathcal{X}^{+}_{\{i\}} +2H \mathcal{X}^{+} -E \mathcal{X}^{-} -\frac{2n-1}{2}U \mathcal{X}^{+} -H\mathcal{X}^{+} &\quad (I = \phi)\\
                    &\sum_{i \in \mathcal{I}_{n-1} \setminus I} \mathcal{X}^{+}_{I \cup \{i\}} +2\sum_{i \in I} T_i \mathcal{X}^{+}_{I \setminus \{i\}} +2H \mathcal{X}^{+}_{I} -E \mathcal{X}^{-}_{I} \\
                    &\quad -\frac{2n-1}{2}U \mathcal{X}^{+}_{I} -H\mathcal{X}^{+}_{I} &\quad (I \neq \phi,\mathcal{I}_{n-1}) \\
                    &2\sum_{i \in I} T_i \mathcal{X}^{+}_{I \setminus \{i\}} +2H \mathcal{X}^{+}_{I} -E \mathcal{X}^{-}_{I} -\frac{2n-1}{2}U \mathcal{X}^{+}_{I} -H\mathcal{X}^{+}_{I} &\quad (I = \mathcal{I}_{n-1})
                \end{aligned}
            \right.~,
        \end{align}
        \begin{align}
            \partial \mathcal{X}^{-}_{I} &\equiv \left\{
                \begin{aligned}
                    &\sum_{i \in \mathcal{I}_{n-1}} \mathcal{X}^{-}_{\{i\}} -2H \mathcal{X}^{-} -F \mathcal{X}^{+} -\frac{2n-1}{2}U \mathcal{X}^{-} +H\mathcal{X}^{-} &\quad (I= \phi)\\
                    &\sum_{i \in \mathcal{I}_{n-1} \setminus I} \mathcal{X}^{-}_{I \cup \{i\}} +2\sum_{i \in I} T_i \mathcal{X}^{-}_{I \setminus \{i\}} -2H \mathcal{X}^{-}_{I} -F \mathcal{X}^{+}_{I} \\
                    &\quad -\frac{2n-1}{2}U \mathcal{X}^{-}_{I} +H\mathcal{X}^{-}_{I} &\quad (I \neq \phi,\mathcal{I}_{n-1}) \\
                    &2\sum_{i \in I} T_i \mathcal{X}^{-}_{I \setminus \{i\}} -2H \mathcal{X}^{-}_{I} -F \mathcal{X}^{+}_{I} -\frac{2n-1}{2}U \mathcal{X}^{-}_{I} +H\mathcal{X}^{-}_{I} &\quad (I = \mathcal{I}_{n-1})
                \end{aligned}
            \right.~,
        \end{align}
        \begin{align}
            \partial \mathcal{Y}^{+}_{I} &\equiv \left\{
                \begin{aligned}
                    &-\sum_{i \in \mathcal{I}_{n-1}} \mathcal{Y}^{+}_{\{i\}} +2H \mathcal{Y}^{+} +E \mathcal{Y}^{-} +\frac{2n-1}{2}U \mathcal{Y}^{+} -H\mathcal{Y}^{+} &\quad (I = \phi)\\
                    &-\sum_{i \in \mathcal{I}_{n-1} \setminus I} \mathcal{Y}^{+}_{I \cup \{i\}} -2\sum_{i \in I} T_i \mathcal{Y}^{+}_{I \setminus \{i\}} +2H \mathcal{Y}^{+}_{I} +E \mathcal{Y}^{-}_{I} \\
                    &\quad + \frac{2n-1}{2}U \mathcal{Y}^{+}_{I} -H\mathcal{Y}^{+}_{I} &\quad (I \neq \phi,\mathcal{I}_{n-1}) \\
                    &-2\sum_{i \in I} T_i \mathcal{Y}^{+}_{I \setminus \{i\}} +2H \mathcal{Y}^{+}_{I} +E \mathcal{Y}^{-}_{I} + \frac{2n-1}{2}U \mathcal{Y}^{+}_{I} -H\mathcal{Y}^{+}_{I} &\quad (I = \mathcal{I}_{n-1})
                \end{aligned}
            \right.~,
        \end{align}
        and
        \begin{align}
            \partial \mathcal{Y}^{-}_{I} &\equiv \left\{
                \begin{aligned}
                    &-\sum_{i \in \mathcal{I}_{n-1}} \mathcal{Y}^{-}_{\{i\}} -2H \mathcal{Y}^{-} +F \mathcal{Y}^{+} +\frac{2n-1}{2}U \mathcal{Y}^{-} +H\mathcal{Y}^{-} &\quad (I = \phi)\\
                    &-\sum_{i \in \mathcal{I}_{n-1} \setminus I} \mathcal{Y}^{-}_{I \cup \{i\}} -2\sum_{i \in I} T_i \mathcal{Y}^{-}_{I \setminus \{i\}} -2H \mathcal{Y}^{-}_{I} +F \mathcal{Y}^{+}_{I} \\
                    &\quad +\frac{2n-1}{2}U \mathcal{Y}^{-}_{I} +H\mathcal{Y}^{-}_{I} &\quad (I \neq \phi, \mathcal{I}_{n-1}) \\
                    &-2\sum_{i \in I} T_i \mathcal{Y}^{-}_{I \setminus \{i\}} -2H \mathcal{Y}^{-}_{I} +F \mathcal{Y}^{+}_{I} +\frac{2n-1}{2}U \mathcal{Y}^{-}_{I} +H\mathcal{Y}^{-}_{I} &\quad (I \neq \phi)
                \end{aligned}
            \right.~.
        \end{align}
        Then, we can transform the above equalities to the following relations respectively.
        \begin{align}
            \sum_{i \in \mathcal{I}_{n-1} \setminus I} \mathcal{X}^{+}_{I \cup \{i\}} &\equiv \left\{
                \begin{aligned}
                    &\partial \mathcal{X}^{+} -2H \mathcal{X}^{+} +E \mathcal{X}^{-} +\frac{2n-1}{2}U \mathcal{X}^{+} +H\mathcal{X}^{+} &\quad (I= \phi) \\
                    &\partial \mathcal{X}^{+}_{I} -2\sum_{i \in I} T_i \mathcal{X}^{+}_{I \setminus \{i\}} -2H \mathcal{X}^{+}_{I} +E \mathcal{X}^{-}_{I} \\
                    &\quad +\frac{2n-1}{2}U \mathcal{X}^{+}_{I} +H\mathcal{X}^{+}_{I} &\quad (I \neq \phi, \mathcal{I}_{n-1})
                \end{aligned}
            \right.~, \\
            0 &\equiv \partial \mathcal{X}^{+}_{I} -2\sum_{i \in I} T_i \mathcal{X}^{+}_{I \setminus \{i\}} -2H \mathcal{X}^{+}_{I} +E \mathcal{X}^{-}_{I} 
            +\frac{2n-1}{2}U \mathcal{X}^{+}_{I} +H\mathcal{X}^{+}_{I} \quad (I = \mathcal{I}_{n-1})~,
        \end{align}
        \begin{align}
            \sum_{i \in \mathcal{I}_{n-1} \setminus I} \mathcal{X}^{-}_{I \cup \{i\}} &\equiv \left\{
                \begin{aligned}
                    &\partial \mathcal{X}^{-} +2H \mathcal{X}^{-} +F \mathcal{X}^{+} +\frac{2n-1}{2}U \mathcal{X}^{-} -H\mathcal{X}^{-} &\quad (I= \phi) \\
                    &\partial \mathcal{X}^{-}_{I} -2\sum_{i \in I} T_i \mathcal{X}^{-}_{I \setminus \{i\}} +2H \mathcal{X}^{-}_{I} +F \mathcal{X}^{+}_{I} \\
                    &\quad +\frac{2n-1}{2}U \mathcal{X}^{-}_{I} -H\mathcal{X}^{-}_{I} &\quad (I \neq \phi, \mathcal{I}_{n-1})
                \end{aligned}
            \right.~, \\
            0 &\equiv \partial \mathcal{X}^{-}_{I} -2\sum_{i \in I} T_i \mathcal{X}^{-}_{I \setminus \{i\}} +2H \mathcal{X}^{-}_{I} +F \mathcal{X}^{+}_{I}
            +\frac{2n-1}{2}U \mathcal{X}^{-}_{I} -H\mathcal{X}^{-}_{I} \quad (I = \mathcal{I}_{n-1})~,
        \end{align}
        \begin{align}
            \sum_{i \in \mathcal{I}_{n-1} \setminus I} \mathcal{Y}^{+}_{I \cup \{i\}} &\equiv \left\{
                \begin{aligned}
                    &-\partial \mathcal{Y}^{+} +2H \mathcal{Y}^{+} +E \mathcal{Y}^{-} +\frac{2n-1}{2}U \mathcal{Y}^{+} -H\mathcal{Y}^{+} &\quad (I = \phi) \\
                    &-\partial \mathcal{Y}^{+}_{I} -2\sum_{i \in I} T_i \mathcal{Y}^{+}_{I \setminus \{i\}} +2H \mathcal{Y}^{+}_{I} +E \mathcal{Y}^{-}_{I} \\
                    &\quad + \frac{2n-1}{2}U \mathcal{Y}^{+}_{I} -H\mathcal{Y}^{+}_{I} &\quad (I \neq \phi, \mathcal{I}_{n-1})
                \end{aligned}
            \right.~,\\
            0 &\equiv -\partial \mathcal{Y}^{+}_{I} -2\sum_{i \in I} T_i \mathcal{Y}^{+}_{I \setminus \{i\}} +2H \mathcal{Y}^{+}_{I} +E \mathcal{Y}^{-}_{I}
            + \frac{2n-1}{2}U \mathcal{Y}^{+}_{I} -H\mathcal{Y}^{+}_{I} \quad (I = \mathcal{I}_{n-1}) ~,
        \end{align}
        and
        \begin{align}
            \sum_{i \in \mathcal{I}_{n-1} \setminus I} \mathcal{Y}^{-}_{I \cup \{i\}} &\equiv \left\{
                \begin{aligned}
                    &-\partial \mathcal{Y}^{-} -2H \mathcal{Y}^{-} +F \mathcal{Y}^{+} +\frac{2n-1}{2}U \mathcal{Y}^{-} +H\mathcal{Y}^{-} &\quad (I = \phi) \\
                    &-\partial \mathcal{Y}^{-}_{I} -2\sum_{i \in I} T_i \mathcal{Y}^{-}_{I \setminus \{i\}} -2H \mathcal{Y}^{-}_{I} +F \mathcal{Y}^{+}_{I} \\
                    &\quad +\frac{2n-1}{2}U \mathcal{Y}^{-}_{I} +H\mathcal{Y}^{-}_{I} &\quad (I \neq \phi, \mathcal{I}_{n-1})
                \end{aligned}
            \right.~, \\
            0 &\equiv -\partial \mathcal{Y}^{-}_{I} -2\sum_{i \in I} T_i \mathcal{Y}^{-}_{I \setminus \{i\}} -2H \mathcal{Y}^{-}_{I} +F \mathcal{Y}^{+}_{I} 
            +\frac{2n-1}{2}U \mathcal{Y}^{-}_{I} +H\mathcal{Y}^{-}_{I} \quad (I = \mathcal{I}_{n-1})~.
        \end{align}

        Therefore, it is sufficient for us to construct generators consisting of $\mathcal{X}^{\pm}_{I}$ or $\mathcal{Y}^{\pm}_{I}$ for $I \subset \mathcal{I}_{n-1}$.
        As derived In Sec.\ref{How to count independent operators} and \ref{Linearly independent operators consisting of X I and Y I},
        the following operators are generators:
        \begin{align}
            \label{primary from XPI}
            \prod_{j=1}^{L}\left(D_{i_{2j-1}} -D_{i_{2j}}\right) \mathcal{X}^{+}~,\\
            \label{primary from XMI}
            \prod_{j=1}^{L}\left(D_{i_{2j-1}} -D_{i_{2j}}\right) \mathcal{X}^{-}~,\\
            \label{primary from YPI}
            \prod_{i=1}^{L}\left(D_{i_{2j-1}} -D_{i_{2j}}\right) \mathcal{Y}^{+}~,\\
            \label{primary from YMI}
            \prod_{i=1}^{L}\left(D_{i_{2j-1}} -D_{i_{2j}}\right) \mathcal{Y}^{-}~,
        \end{align}
        where $\{i_1,\dots,i_{2L}\} \subset \mathcal{I}_{n-1}$ satisfying $1 \leq L \leq \left[\frac{n-1}{2} \right]$, and the above operators are satisfying condition.\ref{no blank} and \ref{no intersection}.
        Then, the number of generators composed of $\mathcal{X}^{\pm}_{I}$ or $\mathcal{Y}^{\pm}_{I}$ is equal to 
        \begin{align}
            \frac{n-2L}{n-L} \binom{n-1}{L}
        \end{align}
        with respect to $L$ for $0 \leq L \leq \left[\frac{n-1}{2}\right]$.
        In particular, when $n=2$, $\mathcal{X}^{\pm}$ and $\mathcal{Y}^{\pm}$ can only exist, and when $n \geq 3$, these operators (\ref{primary from XPI}), (\ref{primary from XMI}), (\ref{primary from YPI}) and (\ref{primary from YMI}) for $1 \leq L \leq \left[\frac{n-1}{2}\right]$
        can be constructed.

        The OPEs between $\mathcal{X}^{\pm}_{I}$ and $\mathcal{X}^{\pm}_{J}$, or $\mathcal{Y}^{\pm}_{I}$ and $\mathcal{Y}^{\pm}_{J}$ can be similarly derived in Sec.\ref{The OPEs with mathcal X I and mathcal Y I}.
        Since we, however, do not wholly calculate the OPEs between $\mathcal{X}^{\pm}_{I}$ and $\mathcal{Y}^{\pm}_{J}$, this is also a future work.
\section{Fermionic extension}
    \label{Fermionic extension}
    Until the previous section, we used the Heisenberg currents $h_i$ to remove $U(1)$ gauge anomalies \cite{10.1093/imrn/rnaa031}.
    On the other hand, fermi multiplets which are $\mathcal{N} = (0,2)$ multiplets on the boundary are also used to remove $U(1)$ gauge anomalies \cite{Costello:2018fnz}.
    Then, we would like to study the algebraic relationship between the former and the latter for VOAs associated with $T^{1^n}_{n-1,1} (SU(n))$.
    In \cite{Yoshida:2023wyt}, the latter VOA is called for a boundary VOA or a fermionic extension of the former VOA.
    In this section, While respecting Y. Yoshida's idea, we provide deeper insights into the algebraic relationship.

    Let us set the OPEs between fermi multiplets $\left(\psi_i,\chi_i\right)$ satisfied as follows.
    \begin{align}
        \psi_i(z) \chi_{j}(0) \sim \frac{\delta_{ij}}{z}~.
    \end{align}
    When we replace Heisenberg currents $h_i$ to $J_{f,i} \coloneqq \psi_i \chi_i -\psi_{i+1} \chi_{i+1}$, we can define the BRST current and the BRST charge after fermionic extension.
    \begin{align}
        J_{BRST}(z) &\coloneqq \sum_{a=1}^{n}c^a \left(X_a Y_a -X_{a+1} Y_{a+1} + \psi_i \chi_i -\psi_{i+1} \chi_{i+1} \right)~, \\
        Q_{BRST} &\coloneqq \oint \frac{dz}{2 \pi i} J_{BRST}(z)~.
    \end{align}
    Since $J_{BRST}(z) J_{BRST}(0) \sim 0$, $Q_{BRST}^2 = 0$.
    Then, we can also construct the BRST cohomology using the BRST charge $Q_{BRST}$.

    We also set generators by $(X_i,Y_i)$ and $(\psi_i,\chi_i)$ based on \cite{Yoshida:2023wyt,Beem:2023dub} as follows;
    \begin{align}
        \label{bosonic only}
        U \coloneqq -\frac{1}{n} \sum_{i=1}^{n} X_i Y_i~,\quad \hat{\mathcal{X}} \coloneqq \prod_{i=1}^{n} X_i~,\quad \hat{\mathcal{Y}} \coloneqq \prod_{i=1}^{n} Y_i~,\\
        \label{U_F}
        U_F \coloneqq -\frac{1}{N} \sum_{i=1}^{n} \psi_i \chi_i~,\\
        \label{M pm}
        M^{+}_{i} \coloneqq X_i \psi_i~,\quad M^{-}_i \coloneqq Y_i \psi_i~,\\
        \label{X I and Y I}
        \hat{\mathcal{X}}_I \coloneqq \prod_{i \in \{1,\dots,n\} \setminus I} X_i \prod_{i \in I} \chi_i~,\quad \hat{\mathcal{Y}}_I \coloneqq \prod_{i \in \{1,\dots,n\} \setminus I} Y_i \prod_{i \in I} \psi_i~,
    \end{align}
    where $\hat{\mathcal{X}} = \mathcal{X}$, $\hat{\mathcal{Y}} = \mathcal{Y}$, and we can regard $\hat{\mathcal{X}}_{\phi}$ as $\hat{\mathcal{X}}$ and $\hat{\mathcal{Y}}_{\phi}$ as $\hat{\mathcal{Y}}$, for convenience.
    In particular, $\hat{\mathcal{X}}_{I}$ and $\hat{\mathcal{Y}}_{I}$ are constructed as follows:
    \begin{align*}
        \hat{\mathcal{X}}_I = A_1 \cdots A_n~,\quad \hat{\mathcal{Y}}_{I} = B_1 \cdots B_n~,
    \end{align*} 
    where
    \begin{align*}
        A_i = \left\{
            \begin{aligned}
                &X_i &(i \notin I)\\
                &\chi_i &(i \in I)
            \end{aligned}
        \right.~,\quad 
        B_i = \left\{
            \begin{aligned}
                &Y_i &(i \notin I)\\
                &\psi_i &(i \in I)
            \end{aligned}
        \right.~.
    \end{align*}

    In Sec.\ref{To derive non-trivial operators} to \ref{Linearly independent operators consisting of X I and Y I}, generators of the bosonic VOA are $U$, $\mathcal{X}$, $\mathcal{Y}$, and operators composed of $T_i$, $W_i$, $\mathcal{X}_{I}$ or $\mathcal{Y}_{I}$.
    Compared with (\ref{bosonic only}), (\ref{U_F}), (\ref{M pm}) and (\ref{X I and Y I}), we can easily find out that $T_i$, $W_i$, $\mathcal{X}_{I}$ and $\mathcal{Y}_{I}$ are not generators of VOA after fermionic extension.
    In other words, $T_i$, $W_i$, $\mathcal{X}_{I}$ and $\mathcal{Y}_{I}$ can be written by using (\ref{bosonic only}), (\ref{U_F}), (\ref{M pm}), (\ref{X I and Y I}) and these derivatives. By using these equivalences,
    we can properly understand the algebraic relation between these VOAs.
    \subsection{\texorpdfstring{$T_i$}{}}
        Let us transform $T_{i}=\frac{1}{2}D_i \mathcal{A}_i$ to the following formula.
        \begin{align}
            \label{extension of T_i}
            T_i \equiv \frac{1}{2}\left(U +U_F\right)^2 -\frac{1}{2}[M^{+}_i, M^{-}_i]~.
        \end{align}
        \begin{proof}
            We can derive the following equivalences from the definition of $J_{BRST}$.
            \footnote{When using the following relation,
            \begin{align}
                \label{equivalency between fermi and symplectic boson}
                &\psi_i \chi_i - \psi_{i+1} \chi_{i+1} \equiv -X_i Y_i + X_{i+1} Y_{i+1} &\quad (i=1,\dots,n-1)~,
            \end{align}
            we find out that this situation is similar to the one using Heisenberg currents $h_i$}
            \begin{align}
                &X_i Y_i +\psi_i \chi_i \equiv X_{i+1} Y_{i+1} + \psi_{i+1} \chi_{i+1} &\quad (i=1,\dots,n-1)~.
            \end{align}
            From these relations, we can construct $U +U_F$ as follows,
            \begin{align}
                \label{equivalency of U+U_F}
                X_1 Y_1 +\psi_1 \chi_1 \equiv \dots \equiv X_n Y_n +\psi_n \chi_n \equiv \frac{1}{n} \sum_{i=1}^{n}\left(X_i Y_i +\psi_i \chi_i\right) = -\left(U+U_F\right)~.
            \end{align}
            Therefore, when we set $\mathcal{A}_{f,i} \coloneqq -\psi_i \chi_i$, the following relations hold.
            \begin{align*}
                \mathcal{A}_i = \left(\mathcal{A}_i +\mathcal{A}_{f,i}\right) - \mathcal{A}_{f,i} \equiv (U+U_F) - \mathcal{A}_{f,i}~.
            \end{align*}

            From now, we transform $D_i \mathcal{A}_i$ using $\mathcal{A}_{f,i}$ as follows.
            \footnote{
                Similar to $\hat{\mathcal{A}}_i$, $\hat{\mathcal{A}}_{f,i}$ is defined by
                \begin{align*}
                    \hat{\mathcal{A}}_{f,i} \mathcal{O} \coloneqq \left(\mathcal{A}_{f,i} \mathcal{O}\right)_0 ~,
                \end{align*}
                where $\mathcal{O}$ is an operator.
            }
            \begin{align}
                    D_i \mathcal{A}_i &= \hat{\mathcal{A}}_i \mathcal{A}_{i} +\mathcal{D}_i \mathcal{A}_i \notag\\
                    &= \left(\hat{\mathcal{A}}_i +\hat{\mathcal{A}}_{f,i}\right)\left(\mathcal{A}_i +\mathcal{A}_{f,i}\right) +\mathcal{D}_i \mathcal{A}_i -\hat{\mathcal{A}}_i \mathcal{A}_{f,i} -\hat{\mathcal{A}}_{f,i} \mathcal{A}_i -\hat{\mathcal{A}}_{f,i} \mathcal{A}_{f,i} \notag\\
                    &\equiv \left(U+U_f\right)^2 +\mathcal{D}_i \mathcal{A}_i -2\mathcal{A}_{i} \mathcal{A}_{f,i} -\mathcal{A}_{f,i}^2~.
            \end{align}
            
            By simple calculations, we can derive the following equalities.
            \begin{align}
                \label{equalities of U_i, U_fi, M pm}
                \left\{
                    \begin{aligned}
                        2 \mathcal{A}_i \mathcal{A}_{f,i} &= 2X_i Y_i \psi_i \chi_i\\
                        \mathcal{A}^2_{f,i} &= \partial \psi_i \chi_i -\psi_i \partial \chi_i\\
                        M^{+}_i M^{-}_i &= X_i Y_i \psi_i \chi_i + \partial \psi_i \chi_i + \partial X_i Y_i\\
                        M^{-}_i M^{+}_i &= -X_i Y_i \psi_i \chi_i + \psi_i \partial \chi_i +X_i \partial Y_i
                    \end{aligned}
                \right.~.
            \end{align}
            By (\ref{equalities of U_i, U_fi, M pm}), the following commutator can hold.
            \begin{align}
                \label{commutator of M}
                [M^{+}_i, M^{-}_i] = -\left(\mathcal{D}_i \mathcal{A}_i -2 \mathcal{A}_i \mathcal{A}_{f,i} -\mathcal{A}_{f,i}^2\right)~.
            \end{align}
            Thus, (\ref{extension of T_i}) holds.
        \end{proof}
    \subsection{\texorpdfstring{$W_i$}{}}
        Let us derive the following formula regarding with $W_{i}$.
        \begin{align}
            \label{extension of W_i}
            W_{i} &\equiv \frac{1}{3}\sqrt{\frac{2}{3}} \left\{\frac{3}{2}\left(M^{+}_{i} \partial M^{-}_{i} +M^{-}_{i} \partial M^{+}_{i}\right)
            -\frac{3}{2}\left(U+U_F\right)[M^{+}_{i},M^{-}_{i}] \right. \notag \\
            &\quad \left. +\left(U+U_F\right)^3 +\partial^2 \left(U+U_F\right) \right\}~.
        \end{align}
        \begin{proof}
            Let us use (\ref{definition of W_i ver.2}) and transform $W_i$ to (\ref{extension of W_i}).
            First, each term in (\ref{definition of W_i ver.2}) can be expressed using $\mathcal{A}_i$, $\mathcal{A}_{f,i}$ and $M^{\pm}_i$ as follows.
            \begin{align}
                    \frac{1}{2}\mathcal{D}_{i} \left(\mathcal{D}_{i} \mathcal{A}_i\right)
                    &= \frac{1}{2}\mathcal{D}_i \left\{-[M^{+}_i,M^{-}_i]
                    +\mathcal{A}_{f,i}^2
                    +2 \mathcal{A}_i\mathcal{A}_{f,i}\right\} \notag\\
                    &= -\frac{1}{2}\mathcal{D}_i[M^{+}_i,M^{-}_i]
                    -\mathcal{A}_{f,i} [M^{+}_i,M^{-}_i]
                    +2 \mathcal{A}_i\mathcal{A}_{f,i}^2 +\mathcal{A}_{f,i}^3~,
            \end{align}
            \begin{align}
                    \frac{3}{2} \hat{\mathcal{A}}_{i} \left(\mathcal{D}_i \mathcal{A}_i\right)
                    = 3 \mathcal{A}_{f,i} \mathcal{A}_i^2+\frac{3}{2}\mathcal{A}_i \mathcal{A}_{f,i}^2 -\frac{3}{2}\mathcal{A}_{i}[M_i^{+},M_i^{-}]~,
            \end{align}
            and
            \begin{align}
                \hat{\mathcal{A}}_i \left(\hat{\mathcal{A}}_i \mathcal{A}_i\right) = \mathcal{A}_i^3~,
            \end{align}
            where we can apply (\ref{commutator of M}) for $\mathcal{D}_i \mathcal{A}_i$.
            Summarizing these results, we can derive the following relation by (\ref{equivalency of U+U_F}).
            \begin{align}
                \label{relationship of W_i}
                \frac{1}{2}\mathcal{D}_i \left(\mathcal{D}_i \mathcal{A}_i\right) +\frac{3}{2} \hat{\mathcal{A}}_{i} \left(\mathcal{D}_i \mathcal{A}_i\right) +\hat{\mathcal{A}}_{i}\left(\hat{\mathcal{A}}_{i} \mathcal{A}_{i}\right)
                &= -\frac{1}{2}\mathcal{D}_i[M^{+}_i,M^{-}_i] +\frac{1}{2} \mathcal{A}_i\mathcal{A}_{f,i}^2 \notag\\
                &\quad -\left(\frac{3}{2}\mathcal{A}_i +\mathcal{A}_{f,i}\right) [M^{+}_i,M^{-}_i] + \left(\mathcal{A}_{i} +\mathcal{A}_{f,i}\right)^3 \notag \\
                &\equiv -\frac{1}{2}\mathcal{D}_i[M^{+}_i,M^{-}_i] +\frac{1}{2} \mathcal{A}_{f,i} [M^{+}_i,M^{-}_i]+\frac{1}{2} \mathcal{A}_i\mathcal{A}_{f,i}^2 \notag \\
                &\quad -\frac{3}{2}\left(U + U_F\right) [M^{+}_i,M^{-}_i] + \left(U + U_F\right)^3~.
            \end{align}
        
            Second, let us replace the terms using $\mathcal{D}_i$ and composed of $\mathcal{A}_i$ and $\mathcal{A}_{f,i}$ in (\ref{relationship of W_i}) with the terms consisting of $U$, $U_F$ and $M^{\pm}_i$.
            To do so, we calculate (\ref{relationship of W_i}) specifically as follows.
            \begin{align}
                &\quad -\frac{1}{2}\mathcal{D}_i[M^{+}_i,M^{-}_i] +\frac{1}{2} \mathcal{A}_{f,i} [M^{+}_i,M^{-}_i]+\frac{1}{2} \mathcal{A}_i\mathcal{A}_{f,i}^2 \notag \\
                &= -\frac{3}{2} \partial X_i Y_i \psi_i \chi_i +\frac{3}{2}X_i \partial Y_i \psi_i \chi_i -\frac{3}{2} X_i Y_i \partial \psi_i \chi_i +\frac{3}{2} X_i Y_i \psi_i \partial \chi_i \notag \\
                &\quad -\frac{1}{4} \partial^2 X_i Y_i + \partial X_i \partial Y_i -\frac{1}{4} X_i \partial^2 Y_i
                -\frac{1}{4} \partial^2 \psi_i \chi_i + \partial \psi_i \partial \chi_i -\frac{1}{4} \psi_i \partial^2 \chi_i~.
            \end{align}
            We also compute the following terms composed of $M^{\pm}_i$ or $U+U_F$ which emerge in (\ref{extension of W_i}), that is $M^{+}_i \partial M^{-}_i + M^{-}_i \partial M^{+}_i$ and $\partial^2 \left(U+U_F\right)$, as follows.
            \begin{align}
                M^{+}_i \partial M^{-}_i + M^{-}_i \partial M^{+}_i 
                &= -\partial X_i Y_i \psi_i \chi_i +X_i \partial Y_i \psi_i \chi_i - X_i Y_i \partial \psi_i \chi_i + X_i Y_i \psi_i \partial \chi_i \notag \\
                &\quad +\frac{1}{2} \partial^2 X_i Y_i +2 \partial X_i \partial Y_i +\frac{1}{2} X_i \partial^2 Y_i \notag \\
                &\quad +\frac{1}{2} \partial^2 \psi_i \chi_i +2 \partial \psi_i \partial \chi_i +\frac{1}{2} \psi_i \partial^2 \chi_i~,
            \end{align}
            and
            \begin{align}
                \partial^2 \left(U+U_F\right) &\equiv \partial^2 \left(\mathcal{A}_i + \mathcal{A}_{f,i}\right) \notag \\
                &= -\partial^2 X_i Y_i -2 \partial X_i \partial Y_i - X_i \partial^2 Y_i \notag \\
                &\quad - \partial^2 \psi_i \chi_i -2 \partial \psi_i \partial \chi_i - \psi_i \partial^2 \chi_i~.
            \end{align}
            Therefore, from these above results, we can find out that the following relation holds.
            \begin{align}
                -\frac{1}{2}\mathcal{D}_i[M^{+}_i,M^{-}_i] +\frac{1}{2} \mathcal{A}_{f,i} [M^{+}_i,M^{-}_i]+\frac{1}{2} \mathcal{A}_i\mathcal{A}_{f,i}^2
                \equiv \frac{3}{2}\left(M^{+}_i \partial M^{-}_i + M^{-}_i \partial M^{+}_i\right) + \partial^2 \left(U+ U_F\right)~.
            \end{align} 

            Finally, summarizing the above results, we can derive the following result.
            \begin{align}
                \frac{1}{2}\mathcal{D}_i \left(\mathcal{D}_i \mathcal{A}_i\right) +\frac{3}{2} \hat{\mathcal{A}}_{i} \left(\mathcal{D}_i \mathcal{A}_i\right) +\hat{\mathcal{A}}_{i} \left(\hat{\mathcal{A}}_{i} \mathcal{A}_{i}\right)
                &\equiv \frac{3}{2}\left(M^{+}_i \partial M^{-}_i + M^{-}_i \partial M^{+}_i \right) -\frac{3}{2}\left(U + U_F\right)[M^{+}_i,M^{-}_i] \notag \\
                &\quad + \left(U+U_F\right)^3 + \partial^2 \left(U+U_F\right)~.
            \end{align}
            Thus, (\ref{extension of W_i}) holds.
        \end{proof}
        Since (\ref{extension of T_i}) and (\ref{extension of W_i}) are satisfied, $W_3$-algebra $(T_i, W_i)$ can be constructed by $U+U_F$ and $M^{\pm}_{i}$ with respect to $i=1,\dots,n$. 
    \subsection{\texorpdfstring{$\mathcal{X}_{I}$}{} and \texorpdfstring{$\mathcal{Y}_{I}$}{}}
        To study $\mathcal{X}_{I}$ and $\mathcal{Y}_{I}$, let us define $\hat{\mathcal{X}}_{I,J}$ and $\hat{\mathcal{Y}}_{I,J}$ for $I \subset \mathcal{I}_{n}$ and $J \subset \mathcal{I} \setminus I$, as follows.
        \begin{align}
            \hat{\mathcal{X}}_{I,J} &\coloneqq \prod_{j \in J}D_j \hat{\mathcal{X}}_I~,\\
            \hat{\mathcal{Y}}_{I,J} &\coloneqq \prod_{j \in J}D_j \hat{\mathcal{Y}}_I~.
        \end{align}
        For example, $\hat{\mathcal{X}}_{I,\phi} = \mathcal{X}_{I}$, $\hat{\mathcal{X}}_{\phi,J} = \prod_{j \in J} D_j \hat{\mathcal{X}}_{\phi} = \mathcal{X}_{J}$ and so on.

        Transforming $\hat{\mathcal{X}}_{I,J}$ and $\hat{\mathcal{Y}}_{I,J}$ for $J \neq \phi$ by (\ref{bosonic only}),  (\ref{U_F}), (\ref{M pm}) and (\ref{X I and Y I}), the following relations can be derived.
        \begin{align}
            \label{transformation of XI, YI}
            \left\{
                \begin{aligned}
                    \hat{\mathcal{X}}_{I,J} = D_j \hat{\mathcal{X}}_{I,J \setminus \{j\}}\equiv \left(U+ U_F\right)\hat{\mathcal{X}}_{I,J \setminus \{j\}} + (-1)^{f_j} M_{j}^{+} \hat{\mathcal{X}}_{I\cup \{j\},J \setminus \{j\}} \\
                    \hat{\mathcal{Y}}_{I,J} = D_j \hat{\mathcal{Y}}_{I,J \setminus \{j\}}\equiv \left(U+ U_F\right)\hat{\mathcal{Y}}_{I,J \setminus \{j\}} - (-1)^{f_j} M_{j}^{-} \hat{\mathcal{Y}}_{I\cup \{j\},J \setminus \{j\}}
                \end{aligned}
            \right.~,
        \end{align}
        where $f_j$ is defined by
        \begin{align*}
            f_{j} =
            \left\{
                \begin{aligned}    
                    &0 &\quad (j=1) \\
                    &\left(
                        \begin{aligned}
                            &\text{The number of $\psi_k$ or $\chi_k$ ($k=1,\dots,i-1$)} \\
                            &\quad \text{contained in $\hat{\mathcal{X}}_{I,J \setminus \{j\}}$ or $\hat{\mathcal{Y}}_{I,J \setminus \{j\}}$}
                        \end{aligned}
                    \right) &\quad (2 \leq j \leq n)
                \end{aligned}
            \right.~.
        \end{align*}
        
        \begin{proof}
            First, we can transform $\mathcal{A}_i \hat{\mathcal{X}}_{I,J \setminus \{i\}}$ and $\mathcal{A}_i \hat{\mathcal{Y}}_{I,J \setminus \{i\}}$ as follows.
            \begin{align}
                \label{correction terms ver.fermi X} 
                \hat{\mathcal{A}}_i \hat{\mathcal{X}}_{I,J \setminus \{i\}} &= \left(\hat{\mathcal{A}}_i +\hat{\mathcal{A}}_{f,i}\right)\hat{\mathcal{X}}_{I,J \setminus \{i\}} -\hat{\mathcal{A}}_{f,i} \hat{\mathcal{X}}_{I, J \setminus \{i\}}\notag\\
                &\equiv \left(U+U_F\right) \hat{\mathcal{X}}_{I,J\setminus \{i\}} -\mathcal{A}_{f,i} \hat{\mathcal{X}}_{I, J \setminus \{i\}}~,
            \end{align}
            and
            \begin{align}
                \label{correction terms ver.fermi Y} 
                \hat{\mathcal{A}}_i \hat{\mathcal{Y}}_{I,J \setminus \{i\}} 
                &\equiv \left(U+U_F\right) \hat{\mathcal{Y}}_{I,J\setminus \{i\}} -\mathcal{A}_{f,i} \hat{\mathcal{Y}}_{I, J \setminus \{i\}}~.
            \end{align}
            Second, by using $M^{\pm}_i$ respectively, $\mathcal{A}_{f,i} \hat{\mathcal{X}}_{I, J \setminus \{i\}}$ and $\mathcal{A}_{f,i} \hat{\mathcal{Y}}_{I, J \setminus \{i\}}$ can be transformed as follows. 
            \begin{align}
                \left\{
                    \begin{aligned}
                        \mathcal{A}_{f,i} \hat{\mathcal{X}}_{I,J \setminus \{i\}} &= \{M^{+}_i,\chi_i\} \tilde{\mathcal{X}}_{I,J \setminus \{i\}} -(-1)^{f_i}M^{+}_i \hat{\mathcal{X}}_{I \cup \{i\},J \setminus \{i\}}\\
                        \mathcal{A}_{f,i} \hat{\mathcal{Y}}_{I,J \setminus \{i\}} &= -\{M^{-}_i,\psi_i\} \tilde{\mathcal{Y}}_{I,J \setminus \{i\}} +(-1)^{f_i}M^{-}_i \hat{\mathcal{Y}}_{I \cup \{i\},J \setminus \{i\}}
                    \end{aligned}
                \right.~,
            \end{align}
            where $\tilde{\mathcal{X}}_{I,J \setminus \{i\}}$ and $\tilde{\mathcal{Y}}_{I,J \setminus \{i\}}$ are defined by the following equalities
            \begin{align}
                \left\{
                    \begin{aligned}
                        \hat{\mathcal{X}}_{I,J \setminus \{i\}} &\eqqcolon X_i \tilde{\mathcal{X}}_{I, J \setminus \{i\}},\\
                        \hat{\mathcal{Y}}_{I,J \setminus \{i\}} &\eqqcolon Y_i \tilde{\mathcal{Y}}_{I, J \setminus \{i\}}
                    \end{aligned}
                \right.~,
            \end{align}
            respectively. Then, by easy calculations, $\{M^{+},\chi_i\}$ and $\{M^{-},\psi_i\}$ are given by 
            \begin{align}
                \left\{
                    \begin{aligned}
                        \{M^{+}_i, \chi_i\} &= \partial X_i\\
                        \{M^{-}_i, \psi_i\} &= \partial Y_i
                    \end{aligned}
                \right.~.
            \end{align}
            Therefore, $\{M^{+}_i, \chi_i\}\tilde{\mathcal{X}}_{I,J \setminus \{i\}}$ and $\{M^{-}_i, \psi_i\}\tilde{\mathcal{Y}}_{I,J \setminus \{i\}}$ are equal to 
            \begin{align}
                \left\{
                    \begin{aligned}
                        \{M^{+}_i, \chi_i\}\tilde{\mathcal{X}}_{I,J \setminus \{i\}} 
                        &= \mathcal{D}_i \hat{X}_{I,J \setminus \{i\}}\\
                        \{M^{-}_i, \psi_i\}\tilde{\mathcal{Y}}_{I,J \setminus \{i\}} 
                        &= -\mathcal{D}_i \hat{\mathcal{Y}}_{I,J \setminus \{i\}}
                    \end{aligned}
                \right.~,
            \end{align}
            respectively. 

            Thus, we can derive the following equalities.
            \begin{align}
                \label{covariant derivative ver.fermi X}
                \mathcal{D}_i \hat{\mathcal{X}}_{I,J \setminus \{i\}} - \hat{\mathcal{A}}_{f,i} \hat{\mathcal{X}}_{I,J \setminus \{i\}} = (-1)^{f_i} M^{+}_{i} \hat{\mathcal{X}}_{I \cup \{i\},J \setminus \{i\}}~,
            \end{align}
            and
            \begin{align}
                \label{covariant derivative ver.fermi Y}
                \mathcal{D}_i \hat{\mathcal{Y}}_{I,J \setminus \{i\}} - \hat{\mathcal{A}}_{f,i} \hat{\mathcal{Y}}_{I,J \setminus \{i\}} = - (-1)^{f_i} M^{-}_{i} \hat{\mathcal{Y}}_{I \cup \{i\},J \setminus \{i\}}~.
            \end{align}
            So, by (\ref{correction terms ver.fermi X}), (\ref{correction terms ver.fermi Y}), (\ref{covariant derivative ver.fermi X}) and (\ref{covariant derivative ver.fermi Y}), we prove that (\ref{transformation of XI, YI}) holds.
        \end{proof}
        
        By repeatedly using this relation (\ref{transformation of XI, YI}), we can easily find out that $\mathcal{X}_{I}$ (resp. $\mathcal{Y}_{I}$) can be written by (\ref{bosonic only}), (\ref{U_F}), (\ref{M pm}) and (\ref{X I and Y I}).
        In particular, applying (\ref{transformation of XI, YI}) for $\mathcal{X}_{I}$ (resp. $\mathcal{Y}_{I}$), we can finally find out that $\hat{\mathcal{X}}_{I}$ (reps. $\hat{\mathcal{Y}}_{I}$) emerges.
    \subsection{Stress tensor \texorpdfstring{$T$}{}}
        When removing $U(1)$ gauge anomalies by Heisenberg currents $h_i$, we can construct the stress tensor $T$ (\ref{definition of T});
        \begin{align*}
                \sum_{i=1}^{n} T_i -\frac{n}{2} U^2~.
        \end{align*}
        This stress tensor can be also derived by substituting $h_i$ to $J_{f,i} = \psi_i \chi_i - \psi_{i+1} \chi_{i+1}$ and using (\ref{equivalency between fermi and symplectic boson}).
        Therefore, after fermionic extension, we naively consider that (\ref{definition of T}) is equivalent of
        \begin{align}
            \label{definition of T ver.2}
            \frac{n}{2}\left(U+U_F\right)^2 -\frac{n}{2}U^2 -\frac{1}{2} \sum_{i=1}^{n}[M_i^{+},M^{-}_i]~,
        \end{align}
        by (\ref{extension of T_i}). 
        
        In \cite{Yoshida:2023wyt}, the stress tensor corresponding to (\ref{definition of T ver.2}) was used, 
        and from its stress tensor and the generators (\ref{bosonic only}), (\ref{U_F}), (\ref{M pm}) and (\ref{X I and Y I}), several examples of the algebraic structures after fermionic extension were derived.
        The stress tensor (\ref{definition of T ver.2}), however, does not work well with generators using $\psi_i$ and $\chi_i$.
        Therefore, let us reconstruct a stress tensor that is more appropriate than (\ref{definition of T ver.2}).

        In order to construct more appropriate stress tensor, it is necessary for us to use the following equality.\footnote{In fact, this equality (\ref{canonical stress tensor}) was given in \cite{Yoshida:2023wyt}.
        The method for deriving (\ref{canonical stress tensor}) is similar to the method for deriving $T_h+T_{bc} \equiv \frac{1}{2}\sum_{i=1}^{n}\hat{\mathcal{A}}_i \mathcal{A}_i-\frac{n}{2}U^2$. This specific calculation was given in \cite{Nishinaka:2025nbe}.
        In Appendix.\ref{To generalize fermionic extension for abelian quiver gauge theories}, we will generalize these relationship.}
        \begin{align}
            \label{canonical stress tensor}
            T_f +\frac{n}{2} U_F^2 = \frac{1}{2} \sum_{i=1}^{n} \left(\partial \psi_i \chi_i -\psi_i \partial \chi_i\right)~,
        \end{align}
        where $T_f \coloneqq \frac{1}{2}\sum_{i,j=1}^{n-1}C^{ij}J_{f,i} J_{f,j}$ and the right-hand side is the canoninal stress tensor for $\psi_i$ and $\chi_i$.
        By the equality (\ref{canonical stress tensor}), we should replace $T_h$ with $T_f + \frac{n}{2}U_F^2$ to do fermionic extension properly.
        In other words, after fermionic extension, we define the stress tensor $T$ as follows.
        \begin{align}
            \label{definition of T ver.3}
            T &\coloneqq T_{sb} +T_f + T_{bc} + \frac{n}{2} U_F^2 \\
            \label{definition of T ver.4}
            &\equiv \sum_{i=1}^{n}T_i -\frac{n}{2}U^2 +\frac{n}{2}U_F^{2} \\
            \label{definition of T ver.5}
            &\equiv \frac{n}{2}\left(U+U_F\right)^2 -\frac{n}{2}U^2 + \frac{n}{2}U_F^2 -\frac{1}{2}\sum_{i=1}^{n}[M^{+}_i,M^{-}_i]~,
        \end{align}
        where (\ref{definition of T ver.3}) is the usual definition, (\ref{definition of T ver.4}) is the result by replacing $T_f + T_{bc}$ with $\frac{1}{2}\sum_{i=1}^{n} T_i -\frac{n}{2}U^2$ under the BRST cohomology,
        and (\ref{definition of T ver.5}) is derived from (\ref{definition of T ver.4}) by using (\ref{extension of T_i}). In particular, since (\ref{definition of T ver.5}) can be written by $U$, $U_F$ and $M^{\pm}_i$,
        the stress tensor $T$ is no longer one of generators, but there is no doubt that $T$ is an important operator.

        Therefore, two VOAs associated with $T^{[1^n]}_{[n-1,1]}(SU(n))$ are related as follows.
        \begin{itemize}
            \item In order to remove $U(1)$ gauge anomalies, we use Heisenberg currents $h_i$ or $J_{f,i} = \psi_i \chi_i - \psi_{i+1} \chi_{i+1}$, where $(\psi_i, \chi_i)$ are fermimultiplets.
            \item When using $J_{f,i}$, the VOA after fermionic extension is generated by $U$, $U_F$, $M^{\pm}_{i}$, $\hat{\mathcal{X}}_{I}$ and $\hat{\mathcal{Y}}_{I}$ for $I \subset \{1,\dots,n \}$.
            These operators are represented in a form that is easy to understand intuitively. 
            \item When using $h_i$, the bosonic VOA is generated by using $U$, $T_i$, $W_i$, $\mathcal{X}_{I}$ and $\mathcal{Y}_{I}$ for $I \subset \{1,\dots,n\}$.
            While $U$, $\mathcal{X}$ and $\mathcal{Y}$ which come from Higgs branch chiral operators can be easily constructed, the others $T_i$, $W_i$, $\mathcal{X}_{I}$ and $\mathcal{Y}_{I}$ for $I \subset \{1,\dots,n\}$ and $I \neq \phi$ are not.
            This difficulty is shown in Sec.\ref{To derive non-trivial operators} to \ref{Linearly independent operators consisting of X I and Y I}. 
            \item When replacing $J_{f,i}$ to $h_i$, that is removing $(\psi_i,\chi_i)$ and constructed by symplectic bosons $(X_i,Y_i)$ only,
            the bosonic VOA cannot be simply derived from the VOA after fermionic extension in generic case. We must discover the relations (\ref{extension of T_i}), (\ref{extension of W_i}) and (\ref{transformation of XI, YI}).
        \end{itemize}

\section{Summary and conclusions}
    In this paper, we discussed issues that were not addressed in detail in \cite{Nishinaka:2025nbe}; how to find $\mathcal{D}_i$ and $\hat{\mathcal{A}}_i$ in Sec.\ref{Preliminary of the bosonic VOA associated with T_{[n-1,1]}^{[1^n]}(SU(n))}, how to construct operators by introducing $\mathcal{D}_{i}$ and $\hat{A}_{i}$ in Sec.\ref{To derive non-trivial operators},
    and how to the number of generators consisting of $\mathcal{X}_{I}$ and $\mathcal{Y}_{I}$ for $T^{[1^n]}_{[n-1,1]}(SU(n))$ from Sec.\ref{How to construct primary fields from T i, mathcal X I, mathcal Y I} to \ref{Linearly independent operators consisting of X I and Y I},
    or $\mathcal{X}^{\pm}_{I}$ and $\mathcal{Y}^{\pm}_{I}$ for $T^{[2,1^{n-1}]}_{[n-1,1^2]}(SU(n+1))$ in Sec.\ref{Generators of the bosonic VOA of T^{[2,1^{n-1}]}_{[n-1,1^2]}(SU(n+1))}. In particular, the number of generators composed of $\mathcal{X}_{I}$ and $\mathcal{Y}_{I}$ is related to the Catalan number.
    Moreover, by calculating the OPEs in Appendix.\ref{OPE 1} for $T^{[1^n]}_{[n-1,1]}(SU(n))$ or Appendix.\ref{OPE 2} for $T^{[2,1^{n-1}]}_{[n-1,1^2]}(SU(n+1))$,
    we solidified the conjecture that the bosonic VOA can be constructed by generators given in the previous paper and this paper. In other words, even if other generators exist, we find out that holomorphic dimension of these operators must be $\frac{11}{2}$ or more. 
    However, it is difficult to verify it through more specific calculations, so establishing other methods is desirable.

    The other result is that we partially discovered how to derive a bosonic VOA, which is consisting of symplectic bosons only, as a sub-VOA of a boundary VOA, which is consisting of symplectic bosons and fermimultiplets, in Section.\ref{Fermionic extension}, through the example for the VOAs associated with $T^{[1^n]}_{[n-1,1]}(SU(n))$.
    Then, we conjecuture that this procedure given in Section. \ref{Fermionic extension} can be applied for other abelian quiver gauge theories. In short, the method is formally described in Appendix.\ref{To generalize fermionic extension for abelian quiver gauge theories}.
    At least I am not dealing with other abelian quiver gauge theories, but it may be possible to construct a sub-VOA using similar procedures compared with \cite{Beem:2023dub,Ballin:2023rmt}.
    
    Furthermore, I consider that there are the following future works;
    \begin{enumerate}
        \item Can the sequence of discussion in this paper be generalized and applied for non-abelian quiver gauge theories \cite{Coman:2023xcq,Kim:2025rog}?
        In other words, can a sub-VOA composed of symplectic bosons be always constructed from a boundary VOA associated with a $3D$ $\mathcal{N}=4$ gauge theory?
        \item If the sequence of discussion in this paper can apply for non-abelian quiver gauge theories, is it possible to extend this technique to other VOAs ?
        For example, would the similar situations occur in the case of $3D$ $\mathcal{N}=2$ gauge theories with a boundary or boundaries \cite{Costello:2020ndc,Alekseev:2022gnr} ?
        \item In \cite{Nishinaka:2025nbe}, T. Nishinaka and I considered the theories $T^{[1^n]}_{[n-1,1]}(SU(n))$ and $T^{[2,1^{n-1}]}_{[n-1,1^2]}(SU(n+1))$,
        where these theories are $3D$ $\mathcal{N}=4$ abelian quiver gauge theories obtained by compactifying the $4D$ $\mathcal{N}=2$ Argyres-Douglas theories called $(A_1,A_{2n-1})$ and $(A_1,D_{2n})$ respectively \cite{Cecotti:2010fi,Xie:2012hs}.
        Then, we compared bosonic VOAs with VOAs given in \cite{Creutzig:2017qyf}. In this paper, the construction of generators related of these theories $T^{[1^n]}_{[n-1,1]}(SU(n))$ and $T^{[2,1^{n-1}]}_{[n-1,1^2]}(SU(n+1))$ are also shown.
        On the other hand, the twisted compactification discussed in \cite{Dedushenko:2023cvd,ArabiArdehali:2024ysy} are not included.
        Then, does this compactification affect the operators by using $\mathcal{D}_i$ and $\hat{\mathcal{A}}_i$, e.g. $T_i$, $W_i$, $\mathcal{X}_{I}$ and $\mathcal{Y}_{I}$ for $T^{[1^n]}_{[n-1,1]}(SU(n))$?
        \item It has been known that the duality called mirror symmetry in $3D$ $\mathcal{N}=4$ gauge theories exists \cite{deBoer:1996mp,Webster:2023jtq}. Then, how do the results in \cite{Nishinaka:2025nbe} and this paper,
        affect the VOA by performing C-twist? In particular, do operators corresponding to the operators constructed by using $\mathcal{D}_i$ and $\hat{\mathcal{A}}_i$ appear?
        \item From the generators of the bosonic VOAs of $T^{[1^n]}_{[n-1,1]}(SU(n))$ and $T^{[2,1^{n-1}]}_{[n-1,1^2]}(SU(n+1))$ derived in this paper, how about constructing the associated varieties \cite{arakawa2010remarkc2cofinitenessconditionvertex} or Zhu-algebras \cite{6afcd487-972b-3da9-824d-e922f84b29f5}?
    \end{enumerate}
\bigskip
\bigskip

\begin{center}
\noindent{\bf Acknowledgments}
\end{center}
    I am grateful to Takahiro Nishinaka, Yutaka Yoshida and Ryo Hamachika for illuminating discussions.
    Most of computations have been carried out with the Mathematica package OPEdefs \cite{doi:10.1142/S0129183191001001}
    provided by K. Thielemans to whom I am grateful. 
    This work was partially supported by JSPS KAKENHI Grant Number JP21H04993.
\clearpage
\appendix
\section{The OPEs regarding with the bosonic VOA associated with \texorpdfstring{$T_{[n-1,1]}^{[1^n]}(SU(n))$}{}}
    \label{OPE 1}
    Let us write down the candidates of generators of VOA associated with $T_{[n-1,1]}^{[1^n]}(SU(n))$.

    The candidates of generators are defined by 
    \begin{align}
        \label{U}
        U \coloneqq -\frac{1}{n} \sum_{i=1}^{n} X_i Y_i~,\\
        \label{X Y}
        \mathcal{X} \coloneqq \prod_{i=1}^{n}X_i~,\quad \mathcal{Y} \coloneqq \prod_{i=1}^{n} Y_i~,\\
        \label{T_i}
        T_i \coloneqq \frac{1}{2}D_i \mathcal{A}_i~,\\
        \label{W_i}
        W_i \coloneqq \frac{1}{3}\sqrt{\frac{2}{3}} \left(\mathcal{D}_i \left(\mathcal{D}_i \mathcal{A}_i\right) + \frac{1}{2}\hat{\mathcal{A}}_i \left(\mathcal{D}_i \mathcal{A}_i \right)+ \frac{1}{2}\mathcal{D}_i \left(\hat{\mathcal{A}}_i \mathcal{A}_i\right) + \hat{\mathcal{A}}_i \left(\hat{\mathcal{A}}_i \mathcal{A}_i \right)\right)~,\\
        \label{XI YI}
        \mathcal{X}_{I} = \prod_{i \in I} D_i \mathcal{X}~,\quad \mathcal{Y}_{I} = \prod_{i \in I} D_i \mathcal{Y}~,
    \end{align}
    and the stress tensor $T$ is given by
    \begin{align}
        T \equiv \sum_{i=1}^{n} T_i -\frac{n}{2}U^2~.
    \end{align}

    Then, the OPEs with the above operators, except for the OPEs between $\mathcal{X}_{I}$ and $\mathcal{Y}_{J}$, are derived as follows.\footnote{We set $\Lambda_i \coloneqq T_iT_i -\frac{3}{10}\partial^2 T_i$.}
    \begin{align}
        \label{TT}
        T(z)T(0) &\sim \frac{-(2n-1)}{2 z^4} + \frac{2 T}{z^2} +\frac{\partial T}{z}~,\\
        \label{TU}
        T(z)U(0) &\sim \frac{U}{z^2} +\frac{\partial U}{z}~,\\
        \label{TT_i}
        T(z)T_{i}(0) &\sim \frac{-2}{2 z^4} + \frac{2 T_{i}}{z^2} +\frac{\partial T_{i}}{z}~,\\
        \label{TW_i}
        T(z)W_{i}(0) &\sim \frac{3 W_{i}}{z^2} +\frac{\partial W_{i}}{z}~,\\
        \label{TX}
        T(z) \mathcal{X}(0) &\sim \frac{\frac{n}{2} \mathcal{X}}{z^2} +\frac{\partial \mathcal{X}}{z} ~,\\
        \label{T XI}
        T(z) \mathcal{X}_{I}(0) &\sim \frac{2\sum_{i \in I}\mathcal{X}_{I \setminus \{i\}}}{z^3}+\frac{\left(\frac{n}{2} +|I|\right)\mathcal{X}_{I}}{z^2} +\frac{\partial \mathcal{X}_{I}}{z}
        &\quad (I \neq \phi)~,
    \end{align}
    \begin{align}
        \label{TY}
        T(z) \mathcal{Y}(0) &\sim \frac{\frac{n}{2} \mathcal{Y}}{z^2} +\frac{\partial \mathcal{Y}}{z} ~,\\
        \label{T YI}
        T(z) \mathcal{Y}_{I}(0) &\sim -\frac{2\sum_{i \in I}\mathcal{Y}_{I \setminus \{i\}}}{z^3} +\frac{\left(\frac{n}{2} +|I|\right)\mathcal{Y}_{I}}{z^2} +\frac{\partial \mathcal{Y}_{I}}{z}
        &\quad (I \neq \phi)~,\\
        \label{UU}
        U(z)U(0) &\sim -\frac{1}{n z^2}~,\\
        \label{UT_i}
        U(z)T_{i}(0) &\sim 0~,\\
        \label{UW_i}
        U(z)W_{i}(0) &\sim 0~,\\
        \label{U XI}
        U(z)\mathcal{X}_{I}(0) &\sim \frac{\mathcal{X}_{I}}{z}~,\\
        \label{U YI}
        U(z)\mathcal{Y}_{I}(0) &\sim -\frac{\mathcal{Y}_{I}}{z}~, \\
        \label{T-1}
        T_{i}(z) T_{j}(0) &\sim \delta_{ij}\left\{\frac{-2}{2 z^4} +\frac{2 T_{i}}{z^2} +\frac{\partial T_{i}}{z}\right\}~,\\
        \label{T-2}
        T_{i}(z) W_{j}(0) &\sim \delta_{ij}\left\{\frac{3 W_{i}}{z^2} +\frac{\partial W_{i}}{z} \right\}~,\\
        \label{T-3}
        T_{i}(z) \mathcal{X}_{I}(0) &\sim \frac{2\mathcal{X}_{I\setminus \{i\}}}{z^3} +\frac{2 \mathcal{X}_{I}}{z^2} +\frac{2T_{i} \mathcal{X}_{I \setminus \{i\}}}{z}
        &\quad (i \in I)~,\\
        \label{T-4}
        T_{i}(z) \mathcal{X}_{I}(0) &\sim \frac{\mathcal{X}_{I}}{z^2} +\frac{\mathcal{X}_{I \cup \{i\}}}{z}
        &\quad (i \notin I)~,\\
        \label{T-5}
        T_{i}(z) \mathcal{Y}_{I}(0) &\sim -\frac{2\mathcal{Y}_{I\setminus \{i\}}}{z^3} +\frac{2 \mathcal{Y}_{I}}{z^2} -\frac{2T_{i} \mathcal{Y}_{I \setminus \{i\}}}{z}
        &\quad (i \in I)~,\\
        \label{T-6}
        T_{i}(z) \mathcal{Y}_{I}(0) &\sim \frac{\mathcal{Y}_{I}}{z^2} -\frac{\mathcal{Y}_{I \cup \{i\}}}{z}
        &\quad (i \notin I)~,\\
        \label{W-1}
        W_{i}(z) W_{j}(0) &\sim \delta_{ij} \left\{\frac{-2}{3 z^6} +\frac{2T_{i}}{z^4} +\frac{\partial T_{i}}{z^3}
        +\frac{1}{z^2}\left\{\frac{32}{22+5\times (-2)}\Lambda_{i} +\frac{3}{10}\partial^2 T_{i} \right\} \right. \notag\\
        &\quad \left.+\frac{1}{z}\left\{\frac{16}{22+5\times(-2)}\partial \Lambda_{i} +\frac{1}{15}\partial^3 T_{i} \right\}\right\}~\\
        \label{W-2}
        W_{i}(z) \mathcal{X}_{I}(0) &\sim \sqrt{\frac{2}{3}}\left\{\frac{3\mathcal{X}_{I\setminus \{i\}}}{z^4} +\frac{4\mathcal{X}_{I}}{z^3} +\frac{5T_{i} \mathcal{X}_{I \setminus \{i\}}}{z^2} \right.\notag \\
        &\quad \left.+\frac{1}{z}\left(3\sqrt{\frac{3}{2}}W_{i} -\frac{3}{2}\partial T_{i}\right)\mathcal{X}_{I \setminus \{i\}}\right\}
        &\quad (i \in I)~, \\
        \label{W-3}
        W_{i}(z) \mathcal{X}_{I}(0) &\sim \sqrt{\frac{2}{3}} \left\{\frac{\mathcal{X}_{I}}{z^3} +\frac{3\mathcal{X}_{I \cup \{i\}}}{2 z^2} +\frac{2 T_{i}\mathcal{X}_{I}}{z} \right\}
        &\quad (i \notin I)~,
    \end{align}
    \begin{align}
        \label{W-4}
        W_{i}(z) \mathcal{Y}_{I}(0) &\sim \sqrt{\frac{2}{3}}\left\{\frac{3\mathcal{Y}_{I\setminus \{i\}}}{z^4} -\frac{4\mathcal{Y}_{I}}{z^3} +\frac{5T_{i} \mathcal{Y}_{I \setminus \{i\}}}{z^2} \right. \notag \\
        &\quad \left.+\frac{1}{z}\left(-3\sqrt{\frac{3}{2}}W_{i} -\frac{3}{2}\partial T_{i}\right)\mathcal{Y}_{I \setminus \{i\}}\right\}
        &\quad (i \in I)~, \\
        \label{W-5}
        W_{i}(z) \mathcal{Y}_{I}(0) &\sim \sqrt{\frac{2}{3}} \left\{-\frac{\mathcal{Y}_{I}}{z^3} +\frac{3\mathcal{Y}_{I \cup \{i\}}}{2 z^2} -\frac{2 T_{i}\mathcal{Y}_{I}(w)}{z-w} \right\}
        &\quad (i \notin I)~,\\
        \mathcal{X}_{I}(z) \mathcal{X}_{I}(0) &\sim 0~,\\
        \mathcal{X}_{I}(z) \mathcal{X}_{J}(0) &\sim 
        \sum_{p=1}^{|L|+|M|}\sum_{(a,b)\in \mathcal{J}_{LM,p}}\sum_{i=1}^{|L_a|}\sum_{j=1}^{|M_b|}\frac{1}{z^p} \frac{(-1)^{a-p}}{\left(a+b-p\right)!}\binom{a+b-1}{p-1} \notag\\
        &\quad \times \left(\partial^{a+b-p}\mathcal{X}_{K\cup (L_a)_i}\right)\mathcal{X}_{K\cup (M_b)_j} &\quad (I \neq J)~, \\      
        \mathcal{Y}_{I}(z) \mathcal{Y}_{I}(0) &\sim 0~,\\
        \mathcal{Y}_{I}(z) \mathcal{Y}_{J}(0) 
        &\sim \sum_{p=1}^{|L|+|M|}\sum_{(a,b)\in \mathcal{J}_{LM,p}}\sum_{i=1}^{|L_a|}\sum_{j=1}^{|M_b|}\frac{1}{z^p} \frac{(-1)^{b-p}}{\left(a+b-p\right)!} \binom{a+b-1}{p-1} \notag \\
        &\quad \times \left(\partial^{a+b-p}\mathcal{Y}_{K\cup (L_a)_i}\right)\mathcal{Y}_{K\cup (M_b)_j} &\quad (I \neq J)~,
    \end{align}
    where $K = I \cap J$, $L = I \setminus K$, $M = J \setminus K$, and $\mathcal{J}_{LM,p}$ is defined by (\ref{J LM p}).

    For the OPEs between $\mathcal{X}_{I}$ and $\mathcal{Y}_{J}$, it is too difficult and complex to calculate them, but we conjecture that these OPEs can be written by using $U$, $T_i$ and $W_i$.
    In order to justify this conjecture, we compute some OPEs. For convenience, let us define the following symmetric product,
    \begin{align}
        \label{symmetric product}
        \left\{
            \begin{aligned}
                \sigma \left[A_1\right] &=A_1 \\
                \sigma \left[A_1^{a_1} \dots A_k^{a_k}\right] &= \sum_{m=1}^{k} \left(A_m \sigma \left[A_1^{a_1} \dots A_m^{a_m-1} \dots A_k^{a_k}\right] \right)_0
            \end{aligned}
        \right.~,
    \end{align}
    where $A_1,\dots,A_{m}$ are $U,\partial U,\dots,T_i, \partial T_i,\dots$, and $W_i, \partial W_i,\dots$ for $i=1,\dots,n$, satisfying $A_i \neq A_j$ for $i \neq j$.
    For convenience, we introduce the following constant
    \begin{align}
        \label{constant zeta}
        \zeta_{I,J} \coloneqq (-2)^{|I \cap J^c| + |I^c \cap J|} \times 6^{|I \cap J|}\quad \left(I,J \in \mathcal{I}_{n}\right)~,
    \end{align}
    where $I^c$ (resp. $J^c$) is the compliment set of $I$ (resp. $J$).
    $\tilde{U}$ and $\tilde{W}_i$ are also defined by
    \begin{align}
        \tilde{U} &\coloneqq n U~, \\
        \tilde{W}_i &\coloneqq 3\sqrt{\frac{3}{2}}W_i~.
    \end{align}
    Then, we can partially derive the singular part of the OPE between $\mathcal{X}_{I}$ and $\mathcal{Y}_{J}$ as follows.
    \begin{align}
        \mathcal{X}_{I} (z) \mathcal{Y}_{J}(0) &\sim \frac{\zeta_{I,J}}{z^{n +|I|+|J|}} +\frac{\zeta_{I,J}}{z^{n +|I| +|J|-1}}\left(-\tilde{U}\right)
        +\frac{\zeta_{I,J}}{z^{n+|I|+|J|-2}} \left(
            \frac{1}{2} \tilde{U}^2 -\sum_{\substack{i \in \mathcal{I}_{n} \\ i \notin I \cup J}} T_{i} -\frac{1}{2} \partial \tilde{U}
        \right) \notag \\
        &\quad +\frac{\zeta_{I,J}}{z^{n +|I| +|J| -3}} \left(
            -\frac{1}{6} \tilde{U}^3 +\frac{1}{2} \sum_{\substack{i \in \mathcal{I}_{n} \\ i \notin I \cup J}} \sigma \left[ \tilde{U} T_{i} \right] +\frac{1}{4} \sigma \left[ \tilde{U} \partial \tilde{U}\right]
            +\frac{1}{6} \sum_{\substack{i \in I \cup J \\ i \notin I \cap J}} \tilde{W}_{i} \right. \notag \\
            &\quad \quad \left. -\frac{1}{3} \sum_{\substack{i \in \mathcal{I}_{n} \\ i \notin I \cup J}} \tilde{W}_{i} +\frac{1}{4} \sum_{\substack{i \in I \\ i \notin I \cap J}} \partial T_{i} -\frac{1}{4} \sum_{\substack{i \in J \\ i \notin I \cap J}} \partial T_{i} -\frac{1}{2} \sum_{\substack{i \in \mathcal{I}_{n} \\ i \notin I \cup J}} \partial T_{i} -\frac{1}{6} \partial^2 \tilde{U}
        \right) \notag \\
        &\quad + \frac{\zeta_{I,J}}{z^{n+|I|+|J|-4}} \left(
            \frac{1}{24} \tilde{U}^4 -\frac{1}{6}\sum_{\substack{i \in \mathcal{I}_{n} \\ i \notin I \cup J}} \sigma\left[\tilde{U}^2 T_{i}\right] -\frac{1}{12} \sigma\left[\tilde{U}^2 \partial \tilde{U} \right]
            -\frac{1}{12} \sum_{\substack{i \in I \cup J \\ i \notin I \cap J}} \sigma \left[ \tilde{U} \tilde{W}_{i} \right] \right. \notag \\
            &\quad \quad +\frac{1}{6} \sum_{\substack{i \in \mathcal{I}_{n} \\ i \notin I \cup J}} \sigma \left[\tilde{U} \tilde{W}_{i} \right] -\frac{1}{8} \sum_{\substack{i \in I \\ i \notin I \cap J}} \sigma \left[\tilde{U} \partial T_{i} \right] 
            +\frac{1}{8} \sum_{\substack{i \in J \\ i \notin I \cap J}} \sigma \left[\tilde{U} \partial T_{i} \right] +\frac{1}{4} \sum_{\substack{i \in \mathcal{I}_{n} \\ i \notin I \cup J}} \sigma \left[\tilde{U} \partial T_{i} \right] \notag \\
            &\quad \quad +\frac{1}{12} \sigma \left[\tilde{U} \partial^2 \tilde{U} \right] -\frac{1}{6} \sum_{i \in I \cap J} \sigma \left[T_{i}^2 \right] +\frac{1}{2} \sum_{\substack{i \in I \cup J \\ i \notin I \cap J}} \sigma \left[ T_{i}^2\right]
            -\frac{1}{2} \sum_{\substack{i \in \mathcal{I}_{n} \\ i \notin I \cup J}} \sigma \left[T_{i}^2 \right] \notag \\
            &\quad \quad +\frac{1}{2} \sum_{\substack{i,j \in \mathcal{I}_{n} \\ i,j \notin I \cup J \\ i < j}} \sigma \left[T_{i} T_{j} \right] +\frac{1}{4} \sum_{\substack{i \in \mathcal{I}_{n} \\ i \notin I \cup J}} \sigma \left[T_{i} \partial \tilde{U}\right]
            +\frac{1}{8} \left(\partial \tilde{U}\right)^2 
            +\frac{1}{6} \sum_{\substack{i \in I \\ i \notin I \cap J}} \partial \tilde{W}_{i} -\frac{1}{6} \sum_{\substack{i \in \mathcal{I}_{n} \\ i \notin I \cup J}} \partial \tilde{W}_{i} \notag \\
            &\quad \quad \left. +\frac{1}{12} \sum_{i \in I \cap J} \partial^2 T_{i} -\frac{1}{4} \sum_{\substack{i \in J \\ i \notin I \cap J}} \partial^2 T_{i} -\frac{1}{24} \partial^3 \tilde{U}
        \right) \notag \\
        &\quad + \frac{\zeta_{I,J}}{z^{n +|I|+|J|-5}} \left(
            -\frac{1}{120} \tilde{U}^5 +\frac{1}{24} \sum_{\substack{i \in \mathcal{I}_{n} \\ i \notin I \cup J}} \sigma \left[\tilde{U}^3 T_{i}\right] +\frac{1}{48} \sigma \left[\tilde{U}^3 \partial \tilde{U} \right] +\frac{1}{36} \sum_{\substack{i \in I \cup J \\ i \notin I \cap J}} \sigma \left[\tilde{U}^2 \tilde{W}_{i} \right] \right. \notag \\
            &\quad \quad -\frac{1}{18} \sum_{\substack{i \in \mathcal{I}_{n} \\ i \notin I \cup J}} \sigma \left[\tilde{U}^2 \tilde{W}_{i} \right]
            +\frac{1}{24} \sum_{\substack{i \in I \\ i \notin I \cap J}} \sigma \left[ \tilde{U}^2 \partial T_i \right] -\frac{1}{24} \sum_{\substack{i \in J \\ i \notin I \cap J}} \sigma \left[\tilde{U}^2 \partial T_{i} \right] -\frac{1}{12} \sum_{\substack{i \in \mathcal{I}_{n} \\ i \notin I \cup J}} \sigma \left[\tilde{U}^2 \partial T_{i} \right] \notag
    \end{align}
    \begin{align}
            &\quad \quad -\frac{1}{36} \sigma \left[\tilde{U}^2 \partial^2 \tilde{U} \right]
            +\frac{1}{18} \sum_{\substack{i \in I \cap J}} \sigma \left[\tilde{U} T_{i}^2 \right] -\frac{1}{6} \sum_{\substack{i \in I \cup J \\ i \notin I \cap J}} \sigma \left[\tilde{U} T_{i}^2 \right] +\frac{1}{6} \sum_{\substack{i \in \mathcal{I}_{n} \\ i \notin I \cup J}} \sigma \left[\tilde{U} T_{i}^2 \right] \notag \\
            &\quad \quad -\frac{1}{6} \sum_{\substack{i,j \in \mathcal{I}_{n} \\ i,j \notin I \cup J \\ i < j}} \sigma \left[\tilde{U} T_{i} T_{j} \right] 
            -\frac{1}{12} \sum_{\substack{i \in \mathcal{I}_{n} \\ i \notin I \cup J}} \sigma \left[\tilde{U} T_{i} \partial \tilde{U} \right]
            -\frac{1}{24} \sigma \left[\tilde{U} \left(\partial \tilde{U}\right)^2 \right]
            -\frac{1}{12} \sum_{\substack{i \in I \\ i \notin I \cap J}} \sigma \left[ \tilde{U} \partial \tilde{W}_{i} \right] \notag \\
            &\quad \quad +\frac{1}{12} \sum_{\substack{i \in \mathcal{I}_{n} \\ i \notin I \cup J}} \sigma \left[\tilde{U} \partial \tilde{W}_{i} \right] 
            -\frac{1}{24} \sum_{i \in I \cap J}\sigma \left[\tilde{U} \partial^2 T_{i} \right] +\frac{1}{8} \sum_{\substack{i \in J \\ i \notin I \cap J}} \sigma \left[\tilde{U} \partial^2 T_{i} \right]
            +\frac{1}{48} \sigma \left[\tilde{U} \partial^3 \tilde{U} \right] \notag \\
            &\quad \quad -\frac{1}{30} \sum_{i \in I \cap J} \sigma \left[T_{i} \tilde{W}_{i} \right] +\frac{1}{20} \sum_{\substack{i \in I \cup J \\ i \notin I \cap J}} \sigma \left[T_{i} \tilde{W}_{i} \right]  -\frac{1}{30} \sum_{\substack{i \in \mathcal{I}_{n}\\ i \notin I \cup J}} \sigma \left[T_{i} \tilde{W}_{i} \right] -\frac{1}{12} \sum_{\substack{i \in \mathcal{I}_{n} \\ i \notin I \cup J}} \sum_{\substack{j \in I \\ j \notin I \cap J}}\sigma \left[T_{i} \tilde{W}_{j} \right] \notag \\
            &\quad \quad -\frac{1}{12} \sum_{\substack{i \in \mathcal{I}_{n} \\ i \notin I \cup J}} \sum_{\substack{j \in J \\ j \notin I \cap J}} \sigma \left[T_{i} \tilde{W}_{j} \right] +\frac{1}{6} \sum_{\substack{i,j \in \mathcal{I}_{n} \\ i,j \notin I \cup J \\ i \neq j}} \sigma \left[T_{i} \tilde{W}_{j} \right]
            -\frac{1}{12} \sum_{i \in I \cap J} \sigma \left[T_{i} \partial T_{i} \right] +\frac{3}{8} \sum_{\substack{i \in I \\ i \notin I \cap J}} \sigma \left[T_{i} \partial T_{i} \right] \notag \\
            &\quad \quad +\frac{1}{8} \sum_{\substack{i \in J \\ i \notin I \cap J}} \sigma \left[T_{i} \partial T_{i} \right] -\frac{1}{4} \sum_{\substack{i \in \mathcal{I}_{n} \\ i \notin I \cup J}} \sigma \left[T_{i} \partial T_{i} \right] 
            -\frac{1}{8} \sum_{\substack{i \in \mathcal{I}_{n} \\ i \notin I \cup J}} \sum_{\substack{j \in I \\ j \notin I \cap J}} \sigma \left[T_{i} \partial T_{j} \right] +\frac{1}{8} \sum_{\substack{i \in \mathcal{I}_{n} \\ i \notin I \cup J}} \sum_{\substack{j \in J \\ j \notin I \cap J}} \sigma \left[T_{i} \partial T_{j} \right] \notag \\
            &\quad \quad +\frac{1}{4} \sum_{\substack{i,j \in \mathcal{I}_{n} \\ i,j \notin I \cup J \\ i \neq j}} \sigma \left[T_{i} \partial T_{j} \right]
            +\frac{1}{12} \sum_{\substack{i \in \mathcal{I}_{n} \\ i \notin I \cup J}} \sigma \left[T_{i} \partial^2 \tilde{U} \right]
            -\frac{1}{24} \sum_{\substack{i \in I \cup J \\ i \notin I \cap J}} \sigma \left[ \tilde{W}_{i} \partial \tilde{U} \right] +\frac{1}{12} \sum_{\substack{i \in \mathcal{I}_{n} \\ i \notin I \cup J}} \sigma \left[\tilde{W}_{i} \partial \tilde{U}\right] \notag \\
            &\quad \quad -\frac{1}{16} \sum_{\substack{i \in I \\ i \notin I \cap J}} \sigma \left[ \partial \tilde{U} \partial T_{i} \right] +\frac{1}{16} \sum_{\substack{i \in J \\ i \notin I \cap J}} \sigma \left[\partial \tilde{U} \partial T_{i} \right] +\frac{1}{8} \sum_{\substack{i \in \mathcal{I}_{n} \\ i \notin I \cup J}} \sigma \left[\partial \tilde{U} \partial T_{i}\right]
            +\frac{1}{24} \sigma \left[\partial \tilde{U} \partial^2 \tilde{U} \right] \notag \\
            &\quad \quad +\frac{7}{180} \sum_{i \in I \cap J} \partial^2 \tilde{W}_{i} +\frac{1}{40} \sum_{\substack{i \in I \\ i \notin I \cap J}} \partial^2 \tilde{W}_i -\frac{7}{120} \sum_{\substack{i \in J \\ i \notin I \cap J}} \partial^2 \tilde{W}_i -\frac{1}{60} \sum_{\substack{i \in \mathcal{I}_{n} \\ i \notin I \cup J}} \partial^2 \tilde{W}_{i} +\frac{1}{24} \sum_{i \in I \cap J} \partial^3 T_{i} \notag \\
            &\quad \quad \left. -\frac{1}{16} \sum_{\substack{i \in I \cup J \\ i \notin I \cap J}} \partial^3 T_{i} +\frac{1}{24} \sum_{\substack{i \in \mathcal{I}_{n} \\ i \notin I \cup J}} \partial^3 T_{i} -\frac{1}{120} \partial^4 \tilde{U}
        \right) \notag \\
        &\quad + \cdots~,
        \label{XY}
    \end{align}
    where the above OPEs are checked for $n=2,3,4,5,6$.
    \subsection{Example: \texorpdfstring{$n=2$}{}}
        \label{Example: T^{[1^n]}_{[n-1,1]}(SU(n)) for n=2}
        In \cite{Nishinaka:2025nbe}, T.Nishinaka and I gave the OPEs for $n=2$.
        Then, since I discover better expressions, I show the OPEs in this subsection.
        
        Let us set the generators consisting of $U$, $\mathcal{X}$, $\mathcal{Y}$, $T_1 -T_2$, $\left(D_1 -D_2\right) \mathcal{X}$ and $\left(D_1 -D_2\right)\mathcal{Y}$ as follows:
        \begin{align}
            J_1 :=\frac{\mathcal{X} -\mathcal{Y}}{2}~,\quad J_2 \coloneqq -\frac{i(\mathcal{X}+\mathcal{Y})}{2}~,\quad J_3 \coloneqq U~,
        \end{align}
        \begin{align}
            \hat{J}_1 \coloneqq \frac{\left(D_1 -D_2\right)\left(\mathcal{X} -\mathcal{Y}\right)}{4}~,\quad \hat{J}_2 \coloneqq \frac{-i \left(D_1 -D_2\right)\left(\mathcal{X} + \mathcal{Y}\right)}{4}~,\quad \hat{J}_3 \coloneqq \frac{T_1 -T_2}{2}~.
        \end{align}
        The stress tensor $T$ is derived from the Sugawara stress tensor for $\mathfrak{su}(2)$:
        \begin{align}
            T \coloneqq \sum_{i=1}^{3}J_i J_i~.
        \end{align}
        
        Then, we can derive the following OPEs with the generators $J_j$ and $\hat{J}_{j}$ for $j=1,2,3$.
        \begin{align}
            T(z)T(0) &\sim \frac{-3}{2z^4} +\frac{2T}{z^2}+\frac{\partial T}{z}~,\\
            T(z)J_j(0) &\sim \frac{J_j}{z^2} +\frac{\partial J_j}{z}~,\\
            T(z) \hat{J}_{j}(0) &\sim \frac{2 \hat{J}_{j}}{z^2} +\frac{\partial \hat{J}_{j}}{z}~,\\
            J_j(z) J_k(0) &\sim \frac{-\delta_{jk}}{2z^2} +\frac{i \sum_{l=1}^{3} f_{jkl} J_l}{z}~,\\
            J_j(z) \hat{J}_{k}(0) &\sim \frac{i \sum_{l=1}^{3} f_{jkl} \hat{J}_l}{z}~,\\
            \hat{J}_1 (z) \hat{J}_1 (0) &\sim \frac{-1}{2z^4} +\frac{J_1 J_1 +T}{2z^2} +\frac{2 i J_3 J_1 J_2 -i J_1 J_2 J_3 -i J_2 J_3 J_1 +2\partial T}{6z}~,\\
            \hat{J}_2 (z) \hat{J}_2 (0) &\sim \frac{-1}{2z^4} +\frac{J_2 J_2 +T}{2z^2} +\frac{-i J_3 J_1 J_2 + 2i J_1 J_2 J_3 -i J_2 J_3 J_1 +2\partial T}{6z}~,\\
            \hat{J}_3 (z) \hat{J}_3 (0) &\sim \frac{-1}{2z^4} +\frac{J_3 J_3 +T}{2z^2} +\frac{-i J_3 J_1 J_2 -i J_1 J_2 J_3 +2 i J_2 J_3 J_1 +2\partial T}{6z}~,\\
            \hat{J}_j (z) \hat{J}_k (0) &\sim \frac{i \sum_{l=1}^{3}f_{jkl}J_l}{z^3} +\frac{2J_j J_k +i \sum_{l=1}^{3} f_{jkl} \partial J_l}{4z^2} \notag \\
            &\quad +\frac{-4 J_j \partial J_k +6 \partial J_j J_k +i \sum_{l=1}^{3} f_{jkl} \partial^2 J_l +4i \sum_{l=1}^{3} f_{jkl} J_l T}{4z} \quad (j \neq k)~,
        \end{align}
        where $f_{jkl}$ are structure constants of $\mathfrak{su}(2)$ satisfying $\sum_{k,l=1}^{3} f_{jkl}f_{mkl} =2 \delta_{jm}$.
\clearpage
\section{The OPEs regarding the bosonic VOA associated with \texorpdfstring{$T_{[n-1,1^2]}^{[2,1^{n-1}]}(SU(n+1))$}{}}
    \label{OPE 2}
    Let us write down the candidates of generators of bosonic VOA associated with $T^{[2,1^{n-1}]}_{[n-1,1^2]}(SU(n+1))$.

    The candidates of generators are defined by
    \begin{align}
        U \coloneqq \frac{1}{2n-1}\left(2 \sum_{i=1}^{n-1}\mathcal{A}_{i} +\sum_{j=1}^{2}\mathcal{A}^{(j)}_{n}\right)~,\\
        H \coloneqq \frac{1}{2}\left(\mathcal{A}^{(1)}_{n} -\mathcal{A}^{(2)}_{n}\right)~,\quad E \coloneqq X_{n}^{(1)}Y_{n,(2)}~,\quad F \coloneqq X_{n}^{(2)}Y_{n,(1)}~,\\
        \mathcal{X}^{+} \coloneqq \left(\prod_{i=1}^{n-1} X_{i}\right) X_{n}^{(1)}~,\quad \mathcal{X}^{-} \coloneqq \left(\prod_{i=1}^{n-1} X_{i}\right) X_{n}^{(2)}~,\\
        \mathcal{Y}^{+} \coloneqq \left(\prod_{i=1}^{n-1} Y_i\right) Y_{n,(2)}~,\quad \mathcal{Y}^{-} \coloneqq \left(\prod_{i=1}^{n-1} Y_i\right) Y_{n,(1)}~,\\
        T_i \coloneqq \frac{1}{2}D_i \mathcal{A}_i~,\\
        \hat{H} \coloneqq T_{n}^{(1)} -T_{n}^{(2)}~,\quad \hat{E} \coloneqq E_{1} + E_{2}~,\quad \hat{F} \coloneqq F_{1} + F_{2}~,\\
        W_i \coloneqq \frac{1}{3}\sqrt{\frac{2}{3}} \left\{\mathcal{D}_i \left(\mathcal{D}_i \mathcal{A}_i\right) +\frac{1}{2}\mathcal{D}_{i} \left(\hat{\mathcal{A}}_{i} \mathcal{A}_{i}\right) 
        +\frac{1}{2}\hat{\mathcal{A}}_i \left(\mathcal{D}_i\mathcal{A}_i\right) +\hat{\mathcal{A}}_{i} \left(\hat{\mathcal{A}}_{i} \mathcal{A}_{i}\right)\right\}~,\\
        \mathcal{X}^{+}_{I} \coloneqq \prod_{i \in I}D_i \mathcal{X}^{+}~,\quad \mathcal{X}^{-}_{I} \coloneqq \prod_{i \in I}D_i \mathcal{X}^{-}~,\quad \mathcal{Y}^{+}_{I} \coloneqq \prod_{i \in I}D_i \mathcal{Y}^{+}~,\quad \mathcal{Y}^{-}_{I} \coloneqq \prod_{i \in I}D_i \mathcal{Y}^{-}~.
    \end{align}
    The stress tensor $T$ is given by
    \begin{align}
        T \equiv \sum_{i=1}^{n-1}T_i + \sum_{j=1}^{2}T_n^{(j)} -\frac{2n-1}{4}U^2 -H^2 ~.
    \end{align}
    
    Then, the OPEs with the above operators, except for the OPEs between $\mathcal{X}^{\pm}_{I}$ and $\mathcal{Y}_{J}^{\pm}$, are derived as follows.\footnote{We set $\Lambda_i \coloneqq T_iT_i -\frac{3}{10}\partial^2 T_i$.}
    \begin{align}
        T(z)T(0) &\sim -\frac{2n}{2z^4} +\frac{2T}{z^2} +\frac{\partial T}{z}~,\\
        T(z)U(0) &\sim \frac{U}{z^2} +\frac{\partial U}{z}~,\\
        T(z)H(0) &\sim \frac{H}{z^2} +\frac{\partial H}{z}~,
    \end{align}
    \begin{align}
        T(z)E(0) &\sim \frac{E}{z^2} +\frac{\partial E}{z}~,\\
        T(z)F(0) &\sim \frac{F}{z^2} +\frac{\partial F}{z}~,\\
        T(z)T_i(0) &\sim -\frac{2}{2z^4} +\frac{2T_i}{z^2} + \frac{\partial T_i}{z}~,\\
        T(z)\hat{H}(0) &\sim \frac{2 \hat{H}}{z^2} + \frac{\partial \hat{H}}{z}~,\\
        T(z)\hat{E}(0) &\sim \frac{2 \hat{E}}{z^2} +\frac{\partial \hat{E}}{z}~,\\
        T(z)\hat{F}(0) &\sim \frac{2 \hat{F}}{z^2} +\frac{\partial \hat{F}}{z}~,\\
        T(z)W_i(0) &\sim \frac{3W_i}{z^2} + \frac{\partial W_i}{z}~,\\
        T(z) \mathcal{X}^{\pm}(0) &\sim \frac{\frac{n}{2} \mathcal{X}^{\pm}}{z^2} +\frac{\partial \mathcal{X}^{\pm}}{z}~,\\
        T(z) \mathcal{X}^{\pm}_{I}(0) &\sim \frac{2 \sum_{i \in I} \mathcal{X}^{\pm}_{I \setminus \{i\}}}{z^3} +\frac{\left(\frac{n}{2}+|I|\right)\mathcal{X}^{\pm}_{I}}{z^2} +\frac{\partial \mathcal{X}^{\pm}_{I}}{z}~,\\
        T(z) \mathcal{Y}^{\pm}(0) &\sim \frac{\frac{n}{2} \mathcal{Y}^{\pm}}{z^2} +\frac{\partial \mathcal{Y}^{\pm}}{z}~,\\
        T(z) \mathcal{Y}^{\pm}_{I}(0) &\sim -\frac{2 \sum_{i \in I} \mathcal{Y}^{\pm}_{I \setminus \{i\}}}{z^3} +\frac{\left(\frac{n}{2}+|I|\right)\mathcal{Y}^{\pm}_{I}}{z^2} +\frac{\partial \mathcal{Y}^{\pm}_{I}}{z}~, \\
        U(z) U(0) &\sim -\frac{2}{2n-1}\frac{1}{z^2}~,\\
        U(z) H(0) &\sim 0~,\\
        U(z) E(0) &\sim 0~,\\
        U(z) F(0) &\sim 0~,\\
        U(z) T_i(0) &\sim 0~,\\
        U(z) \hat{H} (0) &\sim 0~,\\
        U(z) \hat{E}(0) &\sim 0~,\\
        U(z) \hat{F}(0) &\sim 0~,\\
        U(z) W_i(0) &\sim 0~,\\
        U(z) \mathcal{X}^{\pm}_{I} (0) &\sim \frac{\mathcal{X}^{\pm}_{I}}{z}~,\\
        U(z) \mathcal{Y}^{\pm}_{I} (0) &\sim -\frac{\mathcal{Y}^{\pm}_{I}}{z}~,\\
        H(z) H(0) &\sim -\frac{1}{2z^2}~,\\
        H(z) E(0) &\sim \frac{E}{z}~,
    \end{align}
    \begin{align}
        H(z) F(0) &\sim -\frac{F}{z}~,\\
        H(z) T_i(0) &\sim 0~,\\
        H(z) \hat{H}(0) &\sim 0~,\\
        H(z) \hat{E}(0) &\sim \frac{\hat{E}}{z}~,\\
        H(z) \hat{F}(0) &\sim -\frac{\hat{F}}{z}~,\\
        H(z) W_i(0) &\sim 0~,\\
        H(z) \mathcal{X}^{+}_{I} (0) &\sim \frac{\mathcal{X}^{+}_{I}}{2z}~,\\
        H(z) \mathcal{X}^{-}_{I} (0) &\sim -\frac{\mathcal{X}^{-}_{I}}{2z}~,\\
        H(z) \mathcal{Y}^{+}_{I} (0) &\sim \frac{\mathcal{Y}^{+}_{I}}{2z}~,\\
        H(z) \mathcal{Y}^{-}_{I} (0) &\sim -\frac{\mathcal{Y}^{-}_{I}}{2z}~,\\
        E(z) E(0) &\sim 0~,\\
        E(z) F(0) &\sim -\frac{1}{z^2} +\frac{2H}{z}~,\\
        E(z) T_i(0) &\sim 0~,\\
        E(z) \hat{H}(0) &\sim -\frac{\hat{E}}{z}~,\\
        E(z) \hat{E}(0) &\sim 0~,\\
        E(z) \hat{F}(0) &\sim \frac{2 \hat{H}}{z}~,\\
        E(z) W_i(0) &\sim 0~,\\
        E(z) \mathcal{X}^{+}_{I} (0) &\sim 0~,\\
        E(z) \mathcal{X}^{-}_{I} (0) &\sim -\frac{\mathcal{X}^{+}_{I}}{z}~,\\
        E(z) \mathcal{Y}^{+}_{I} (0) &\sim 0~,\\
        E(z) \mathcal{Y}^{-}_{I} (0) &\sim \frac{\mathcal{Y}^{+}_{I}}{z}~,\\
        F(z) F(0) &\sim 0~,\\
        F(z) T_i(0) &\sim 0~,\\
        F(z) \hat{H}(0) &\sim \frac{\hat{F}}{z}~,\\
        F(z) \hat{E} (0) &\sim -\frac{2 \hat{H}}{z}~,
    \end{align}
    \begin{align}
        F(z) \hat{F}(0) &\sim 0~,\\
        F(z) W_i(0) &\sim 0~,\\
        F(z) \mathcal{X}^{+}_{I} (0) &\sim -\frac{\mathcal{X}^{-}_{I}}{z}~,\\
        F(z) \mathcal{X}^{-}_{I} (0) &\sim 0~,\\
        F(z) \mathcal{Y}^{+}_{I} (0) &\sim \frac{\mathcal{Y}^{-}_{I}}{z}~,\\
        F(z) \mathcal{Y}^{-}_{I} (0) &\sim 0~,\\
        T_i(z) T_j(0) &\sim \delta_{ij}\left\{\frac{-2}{2z^4} +\frac{2 T_i}{z^2} +\frac{\partial T_i}{z}\right\}~,\\
        T_i(z) \hat{H}(0) &\sim 0~,\\
        T_i(z) \hat{E}(0) &\sim 0~,\\
        T_i(z) \hat{F}(0) &\sim 0~,\\
        T_i(z) W_j(0) &\sim \delta_{ij} \left\{\frac{3 W_i}{z^2} +\frac{\partial W_i}{z}\right\}~,\\
        T_i(z) \mathcal{X}^{\pm}_{I} (0) &\sim \frac{\mathcal{X}^{\pm}_{I}}{z^2} +\frac{\mathcal{X}^{\pm}_{I \cup \{i\}}}{z} &\quad (i \notin I)~,\\
        T_i(z) \mathcal{X}^{\pm}_{I} (0) &\sim \frac{2 \mathcal{X}^{\pm}_{I \setminus \{i\}}}{z^3}+\frac{2\mathcal{X}^{\pm}_{I}}{z^2} +\frac{2T_i \mathcal{X}^{\pm}_{I \setminus \{i\}}}{z} &\quad (i \in I)~,\\
        T_i(z) \mathcal{Y}^{\pm}_{I} (0) &\sim \frac{\mathcal{Y}^{\pm}_{I}}{z^2} -\frac{\mathcal{Y}^{\pm}_{I \cup \{i\}}}{z} &\quad (i \notin I)~,\\
        T_i(z) \mathcal{Y}^{\pm}_{I} (0) &\sim -\frac{2 \mathcal{Y}^{\pm}_{I \setminus \{i\}}}{z^3}+\frac{2\mathcal{Y}^{\pm}_{I}}{z^2} -\frac{2T_i \mathcal{Y}^{\pm}_{I \setminus \{i\}}}{z} &\quad (i \in I)~, \\
        \hat{H}(z) \hat{H}(0) &\sim \frac{-2}{z^4} + \frac{4H^2 +EF +FE}{z^2} +\frac{4 H \partial H + E \partial F + F \partial E}{z}~,\\
        \hat{H}(z) \hat{E}(0) &\sim \frac{4 E}{z^3} +\frac{2 HE +\partial E}{z^2} \notag \\
        &\quad \quad +\frac{4HHE +4 EEF -12 H\partial E +2 E \partial H +6 \partial^2 E}{z}~, \\
        \hat{H}(z) \hat{F}(0) &\sim -\frac{4 F}{z^3} +\frac{2 HF -\partial F}{z^2} \notag \\
        &\quad \quad -\frac{4 HHF +4 FFE +12 H\partial F -2 F \partial H +6 \partial^2 F}{z} ~,\\
        \hat{H}(z) W_i(0) &\sim 0~,\\
        \hat{H}(z) \mathcal{X}^{+}_{I}(0) &\sim \frac{\mathcal{X}^{+}_{I}}{z^2} +\frac{2H \mathcal{X}^{+}_{I}-E \mathcal{X}^{-}_{I}}{z}~,\\
        \hat{H}(z) \mathcal{X}^{-}_{I}(0) &\sim -\frac{\mathcal{X}^{-}_{I}}{z^2} +\frac{2H \mathcal{X}^{-}_{I} +F \mathcal{X}^{+}_{I}}{z}~,
    \end{align}
    \begin{align}
        \hat{H}(z) \mathcal{Y}^{+}_{I}(0) &\sim -\frac{\mathcal{Y}^{+}_{I}}{z^2} -\frac{2H \mathcal{Y}^{+}_{I} + E \mathcal{Y}^{-}_{I}}{z}~,\\
        \hat{H}(z) \mathcal{Y}^{-}_{I}(0) &\sim \frac{\mathcal{Y}^{-}_{I}}{z^2} -\frac{2H \mathcal{Y}^{-}_{I} - F \mathcal{Y}^{+}_{I}}{z}~,\\
        \hat{E}(z) \hat{E}(0) &\sim \frac{2E^2}{z^2} +\frac{2 E\partial E}{z}~,\\
        \hat{E}(z) \hat{F}(0) &\sim \frac{-4}{z^4} +\frac{8 H}{z^3} +\frac{4 H^2 +6 EF -2 \partial H}{z^2} \notag \\
        &\quad +\frac{8 H^3 -4 H \partial H +8 H E F -2 E \partial F +8 F \partial E + 8 \partial^2 H}{z}~,\\
        \hat{E}(z) W_i(0) &\sim 0~,\\
        \hat{E}(z) \mathcal{X}^{+}_{I} (0) &\sim \frac{E\mathcal{X}^{+}_{I}}{z}~,\\
        \hat{E}(z) \mathcal{X}^{-}_{I} (0) &\sim -\frac{2 \mathcal{X}^{+}_{I}}{z^2} +\frac{-2 H \mathcal{X}^{+}_{I} +3 E \mathcal{X}^{-}_{I}}{z}~,\\
        \hat{E}(z) \mathcal{Y}^{+}_{I} (0) &\sim -\frac{E \mathcal{Y}^{+}_{I}}{z}~,\\
        \hat{E}(z) \mathcal{Y}^{-}_{I} (0) &\sim -\frac{2 \mathcal{Y}^{+}_{I}}{z^2} + \frac{-2 H \mathcal{Y}^{+}_{I} -3 E \mathcal{Y}^{-}_{I}}{z}~,\\
        \hat{F}(z) \hat{F}(0) &\sim \frac{2F^2}{z^2} + \frac{2 F \partial F}{z}~,\\
        \hat{F}(z) W_i(0) &\sim 0~,\\
        \hat{F}(z) \mathcal{X}^{+}_{I} (0) &\sim -\frac{2 \mathcal{X}^{-}_{I}}{z^2} +\frac{2 H \mathcal{X}^{-}_{I} +3 F \mathcal{X}^{+}_{I}}{z}~,\\
        \hat{F}(z) \mathcal{X}^{-}_{I} (0) &\sim \frac{F \mathcal{X}^{-}_{I}}{z}~,\\
        \hat{F}(z) \mathcal{Y}^{+}_{I} (0) &\sim -\frac{2 \mathcal{Y}^{-}_{I}}{z^2} + \frac{2 H \mathcal{Y}^{-}_{I} -3 F \mathcal{Y}^{+}_{I}}{z}~,\\
        \hat{F}(z) \mathcal{Y}^{-}_{I} (0) &\sim -\frac{F \mathcal{Y}^{-}_{I}}{z}~,\\
        W_i(z) W_j(0) &\sim \delta_{ij} \left\{\frac{-2}{3 z^6} +\frac{2T_{i}}{z^4} +\frac{\partial T_{i}}{z^3} \right. \notag\\
        &\quad +\frac{1}{z^2}\left\{\frac{32}{22+5\times (-2)}\Lambda_{i} +\frac{3}{10}\partial^2 T_{i} \right\} \notag\\
        &\quad \left.+\frac{1}{z}\left\{\frac{16}{22+5\times(-2)}\partial \Lambda_{i} +\frac{1}{15}\partial^3 T_{i} \right\}\right\} ~,\\
        W_i(z) \mathcal{X}^{\pm}_{I} (0) &\sim \sqrt{\frac{2}{3}} \left\{\frac{\mathcal{X}^{\pm}_{I}}{z^3} +\frac{3\mathcal{X}^{\pm}_{I \cup \{i\}}}{2 z^2} +\frac{2 T_{i}\mathcal{X}^{\pm}_{I}}{z} \right\} &\quad (i \notin I)~,
    \end{align}
    \begin{align}
        W_i(z) \mathcal{X}^{\pm}_{I} (0) &\sim \sqrt{\frac{2}{3}}\left\{\frac{3\mathcal{X}^{\pm}_{I\setminus \{i\}}}{z^4} +\frac{4\mathcal{X}^{\pm}_{I}}{z^3} +\frac{5T_{i} \mathcal{X}^{\pm}_{I \setminus \{i\}}}{z^2} \right. \notag\\
        &\quad \left.+\frac{1}{z}\left(3\sqrt{\frac{3}{2}}W_{i} -\frac{3}{2}\partial T_{i}\right)\mathcal{X}^{\pm}_{I \setminus \{i\}}\right\} &\quad (i \in I)~,\\
        W_i(z) \mathcal{Y}^{\pm}_{I} (0) &\sim \sqrt{\frac{2}{3}} \left\{-\frac{\mathcal{Y}^{\pm}_{I}}{z^3} +\frac{3\mathcal{Y}^{\pm}_{I \cup \{i\}}}{2 z^2} -\frac{2 T_{i}\mathcal{Y}^{\pm}_{I}}{z} \right\} &\quad (i \notin I)~,\\
        W_i(z) \mathcal{Y}^{\pm}_{I} (0) &\sim \sqrt{\frac{2}{3}}\left\{\frac{3\mathcal{Y}^{\pm}_{I\setminus \{i\}}}{z^4} -\frac{4\mathcal{Y}^{\pm}_{I}}{z^3} +\frac{5T_{i} \mathcal{Y}^{\pm}_{I \setminus \{i\}}}{z^2} \right. \notag\\
        &\quad \left.+\frac{1}{z}\left(-3\sqrt{\frac{3}{2}}W_{i} -\frac{3}{2}\partial T_{i}\right)\mathcal{Y}^{\pm}_{I \setminus \{i\}}\right\} &\quad (i \in I)~,\\
        \mathcal{X}^{\pm}_{I}(z) \mathcal{X}^{\pm}_{I}(0) &\sim 0~, \\
        \mathcal{X}^{\pm}_{I}(z) \mathcal{X}^{\pm}_{J}(0)
        &\sim \sum_{p=1}^{|L| + |M|} \sum_{(a,b) \in J_{LM,p}} \sum_{i=1}^{|L_a|} \sum_{j=1}^{|M_b|} \frac{1}{z^p} \frac{(-1)^{a-p}}{(a+b-p)!} \binom{a+b-1}{p-1} \notag\\
        &\quad \times \left(\partial^{a+b-p} \mathcal{X}^{\pm}_{K \cup \left(L_a \right)} \right)\mathcal{X}^{\pm}_{K \cup \left(M_b\right)_j} &\quad (I \neq J)~, \\
        \mathcal{X}^{\pm}_{I}(z) \mathcal{X}^{\mp}_{I}(0) &\sim 0~, \\
        \mathcal{X}^{\pm}_{I}(z) \mathcal{X}^{\mp}_{J}(0)
        &\sim \sum_{p=1}^{|L| + |M|} \sum_{(a,b) \in J_{LM,p}} \sum_{i=1}^{|L_a|} \sum_{j=1}^{|M_b|} \frac{1}{z^p} \frac{(-1)^{a-p}}{(a+b-p)!} \binom{a+b-1}{p-1} \notag\\
        &\quad \times \left(\partial^{a+b-p} \mathcal{X}^{\pm}_{K \cup \left(L_a \right)} \right)\mathcal{X}^{\mp}_{K \cup \left(M_b\right)_j} &\quad (I \neq J)~, \\
        \mathcal{Y}^{\pm}_{I}(z) \mathcal{Y}^{\pm}_{I}(0) &\sim 0~,\\
        \mathcal{Y}^{\pm}_{I}(z) \mathcal{Y}^{\pm}_{J}(0)
        &\sim \sum_{p=1}^{|L| + |M|} \sum_{(a,b) \in J_{LM,p}} \sum_{i=1}^{|L_a|} \sum_{j=1}^{|M_b|} \frac{1}{z^p} \frac{(-1)^{b-p}}{(a+b-p)!} \binom{a+b-1}{p-1} \notag\\
        &\quad \times \left(\partial^{a+b-p} \mathcal{Y}^{\pm}_{K \cup \left(L_a \right)_i} \right) \mathcal{Y}^{\pm}_{K \cup \left(M_b\right)_j} &\quad (I \neq J)~,\\
         \mathcal{Y}^{\pm}_{I}(z) \mathcal{Y}^{\mp}_{I}(0) &\sim 0~,\\
        \mathcal{Y}^{\pm}_{I}(z) \mathcal{Y}^{\mp}_{J}(0)
        &\sim \sum_{p=1}^{|L| + |M|} \sum_{(a,b) \in J_{LM,p}} \sum_{i=1}^{|L_a|} \sum_{j=1}^{|M_b|} \frac{1}{z^p} \frac{(-1)^{b-p}}{(a+b-p)!} \binom{a+b-1}{p-1} \notag\\
        &\quad \times \left(\partial^{a+b-p} \mathcal{Y}^{\pm}_{K \cup \left(L_a \right)_i} \right) \mathcal{Y}^{\mp}_{K \cup \left(M_b\right)_j} &\quad (I \neq J)~,
    \end{align}
    where $K = I \cap J$, $L = I \setminus J$, $M = J \setminus K$, and $\mathcal{J}_{LM,p}$ is defined by (\ref{J LM p}).

    For the OPEs between $\mathcal{X}_{I}^{\pm}$ and $\mathcal{Y}_{J}^{\pm}$, it is too difficult and complex to calculate them.
    However, we also conjecture that these OPEs can be written by using $U$, $H$, $E$, $F$,
    $T_i$, $\hat{H}$, $\hat{E}$, $\hat{F}$ and $W_i$.
    In order to justify this conjecture, we compute some OPEs. For convenience, let us define the following symmetric product (\ref{symmetric product}) and constant $\zeta_{I,J}$ by (\ref{constant zeta}).
    $\tilde{U}$ and $\tilde{W}_i$ are also defined by
    \begin{align}
        \tilde{U} &\coloneqq \frac{2n-1}{2} U~, \\
        \tilde{W}_i &\coloneqq 3\sqrt{\frac{3}{2}}W_i~.
    \end{align}
    Then, we can partially derive the singular part of the OPE between $\mathcal{X}_{I}$ and $\mathcal{Y}_{J}$ as follows.
    \begin{align}
        \mathcal{X}^{+}_{I}(z) \mathcal{Y}^{+}_{J}(0) &\sim \frac{\zeta_{I,J} E}{z^{n+|I|+|J|-1}} +\frac{\zeta_{I,J}}{z^{n+|I|+|J|-2}} \left(
        -\frac{1}{2} \sigma \left[\tilde{U} E \right] +\frac{1}{2} \hat{E} +\frac{1}{2} \partial E
        \right)\notag \\
        &\quad + \frac{\zeta_{I,J}}{z^{n+|I|+|J|-3}} \left(
            \frac{1}{6} \sigma \left[\tilde{U}^2 E \right] -\frac{1}{4} \sigma \left[\tilde{U} \hat{E} \right] -\frac{1}{4} \sigma \left[\tilde{U} \partial E \right]
            +\frac{1}{6} \sigma \left[H^2 E \right] -\frac{1}{4} \sigma \left[H \hat{E} \right] \right. \notag \\
            &\quad \quad -\frac{1}{2} \sigma \left[H \partial E \right] +\frac{1}{6} \sigma \left[E^2 F \right] -\frac{1}{2} \sum_{\substack{i \in \mathcal{I}_{n-1} \\ i \notin I \cup J}} \sigma \left[E T_{i} \right] +\frac{1}{4} \sigma \left[E \hat{H} \right] 
            -\frac{1}{4} \sigma \left[E \partial \tilde{U} \right] \notag \\
            &\quad \quad \left. +\frac{1}{2} \sigma \left[E \partial H \right] +\frac{1}{3} \partial^2 E
        \right) \notag \\
        &\quad + \frac{\zeta_{I,J}}{z^{n+|I|+|J|-4}} \left(
            -\frac{1}{24} \sigma \left[\tilde{U}^3 E \right] +\frac{1}{12} \sigma \left[\tilde{U}^2 \hat{E} \right] +\frac{1}{12} \sigma \left[\tilde{U}^2 \partial E \right] -\frac{1}{24} \sigma \left[\tilde{U} H^2 E \right] \right. \notag \\
            &\quad \quad +\frac{1}{12} \sigma \left[\tilde{U} H \hat{E} \right] +\frac{1}{6} \sigma \left[\tilde{U} H \partial E \right] -\frac{1}{24} \sigma \left[\tilde{U} E^2 F \right] +\frac{1}{6} \sum_{\substack{i \in \mathcal{I}_{n-1} \\ i \notin I \cup J}} \sigma \left[\tilde{U} E T_{i} \right] \notag \\
            &\quad \quad -\frac{1}{12} \sigma \left[\tilde{U} E \hat{H} \right] +\frac{1}{12} \sigma \left[\tilde{U} E \partial \tilde{U} \right] -\frac{1}{6} \sigma \left[\tilde{U} E \partial H \right] -\frac{1}{6} \sigma \left[\tilde{U} \partial^2 E \right]
            -\frac{3}{2} \sigma \left[H^3 E \right] \notag \\
            &\quad \quad -\frac{1}{108} \sigma \left[H^2 \hat{E} \right] +\frac{37}{12} \sigma \left[H^2 \partial E \right] -\frac{1}{2} \sigma \left[H E^2 F \right] +\frac{1}{54} \sigma \left[H E \hat{H} \right]
            -\frac{17}{12} \sigma \left[H E \partial H \right] \notag \\
            &\quad \quad -\frac{2}{9} \sigma \left[H \partial \hat{E} \right] -\frac{1}{2} \sigma \left[H \partial^2 E \right]
            -\frac{1}{18} \sigma \left[E^2 \hat{F} \right] +\frac{19}{12} \sigma \left[E^2 \partial F \right] +\frac{1}{24} \sigma \left[E F \hat{E} \right] \notag \\
            &\quad \quad -\frac{2}{3} \sigma \left[E F \partial E \right] +\frac{1}{12} \sum_{\substack{i \in I \cup J \\ i \notin I \cap J}} \sigma \left[E \tilde{W}_{i} \right] -\frac{1}{6} \sum_{\substack{i \in \mathcal{I}_{n-1} \\ i \notin I \cup J}} \sigma \left[E \tilde{W}_{i} \right]
            +\frac{1}{8} \sum_{\substack{i \in I \\ i \notin I \cap J}} \sigma \left[E \partial T_{i} \right] \notag \\
            &\quad \quad -\frac{1}{8} \sum_{\substack{i \in J \\ i \notin I \cap J}} \sigma \left[E \partial T_{i} \right] -\frac{1}{4} \sum_{\substack{i \in \mathcal{I}_{n-1} \\ i \notin I \cup J}} \sigma \left[E \partial T_{i} \right] -\frac{1}{12} \sigma \left[E \partial^2 \tilde{U} \right] -\frac{1}{4} \sum_{\substack{i \in \mathcal{I}_{n-1} \\ i \notin I \cup J}} \sigma \left[T_{i} \hat{E} \right] \notag \\
            &\quad \quad -\frac{1}{4} \sum_{\substack{i \in \mathcal{I}_{n-1} \\ i \notin I \cup J}} \sigma \left[T_{i} \partial E \right] +\frac{3}{4} \sigma \left[\hat{H} \hat{E} \right] +\frac{1}{9} \sigma \left[\hat{H} \partial E \right]
            -\frac{1}{8} \sigma \left[\hat{E} \partial \tilde{U} \right] -\frac{1}{8} \sigma \left[\partial \tilde{U} \partial E \right] \notag \\
            &\quad \quad \left. -\frac{15}{4} \sigma \left[\partial H \partial E \right]
            +\frac{1}{8} \partial^3 E
        \right) \notag
    \end{align}
    \begin{align}
        &\quad + \frac{\zeta_{I,J}}{z^{n+|I|+|J|-5}} \left(
            \frac{1}{120} \sigma \left[\tilde{U}^4 E \right] -\frac{1}{48} \sigma \left[\tilde{U}^3 \hat{E} \right] -\frac{1}{48} \sigma \left[\tilde{U}^3 \partial E \right]  +\frac{1}{120} \sigma \left[\tilde{U}^2 H^2 E \right] \right. \notag \\
            &\quad \quad -\frac{1}{48} \sigma \left[\tilde{U}^2 H \hat{E} \right] -\frac{1}{24} \sigma \left[\tilde{U}^2 H \partial E \right] 
            +\frac{1}{120} \sigma \left[\tilde{U}^2 E^2 F \right] -\frac{1}{24} \sum_{\substack{i \in \mathcal{I}_{n-1} \\ i \notin I \cup J}} \sigma \left[\tilde{U}^2 E T_{i} \right] \notag \\
            &\quad \quad +\frac{1}{48} \sigma \left[\tilde{U}^2 E \hat{H} \right] -\frac{1}{48} \sigma \left[\tilde{U}^2 E \partial \tilde{U} \right] +\frac{1}{24} \sigma \left[\tilde{U}^2 E \partial H \right] +\frac{1}{18} \sigma \left[\tilde{U}^2 \partial^2 E \right] \notag \\
            &\quad \quad +\frac{3}{10} \sigma \left[\tilde{U} H^3 E \right]  +\frac{1}{432} \sigma \left[\tilde{U} H^2 \hat{E} \right] -\frac{37}{48} \sigma \left[\tilde{U} H^2 \partial E \right] +\frac{1}{10} \sigma \left[\tilde{U} H E^2 F \right] \notag \\
            &\quad \quad -\frac{1}{216} \sigma \left[\tilde{U} H E \hat{H} \right] 
            +\frac{17}{48} \sigma \left[\tilde{U} H E \partial H \right] +\frac{2}{27} \sigma \left[\tilde{U} H \partial \hat{E} \right] +\frac{1}{6} \sigma \left[\tilde{U} H \partial^2 E \right] +\frac{1}{72} \sigma \left[\tilde{U} E^2 \hat{F} \right] \notag \\
            &\quad \quad -\frac{19}{48} \sigma \left[\tilde{U} E^2 \partial F \right]
            -\frac{1}{96} \sigma \left[\tilde{U} E F \hat{E} \right] +\frac{1}{6} \sigma \left[\tilde{U} E F \partial E \right] -\frac{1}{36} \sum_{\substack{i \in I \cup J \\ i \notin I \cap J}} \sigma\left[\tilde{U} E \tilde{W}_{i} \right] \notag \\
            &\quad \quad +\frac{1}{18} \sum_{\substack{i \in \mathcal{I}_{n-1} \\ i \notin I \cup J}} \sigma \left[\tilde{U} E \tilde{W}_{i} \right] -\frac{1}{24} \sum_{\substack{i \in I \\ i \notin I \cap J}} \sigma \left[\tilde{U} E \partial T_{i} \right] +\frac{1}{24} \sum_{\substack{i \in J \\ i \notin I \cap J}} \sigma \left[\tilde{U} E \partial T_{i} \right] +\frac{1}{12} \sum_{\substack{i \in \mathcal{I}_{n-1} \\ i \notin I \cup J}} \sigma \left[\tilde{U} E \partial T_{i} \right] \notag \\
            &\quad \quad +\frac{1}{36} \sigma \left[\tilde{U} E \partial^2 \tilde{U} \right] +\frac{1}{12} \sum_{\substack{i \in \mathcal{I}_{n-1} \\ i \notin I \cup J}} \sigma \left[\tilde{U} T_{i} \hat{E} \right] +\frac{1}{12} \sum_{\substack{i \in \mathcal{I}_{n-1}\\ i \notin I \cup J}} \sigma \left[\tilde{U} T_{i} \partial E \right]
            -\frac{1}{4} \sigma \left[\tilde{U} \hat{H} \hat{E} \right] -\frac{1}{27} \sigma \left[\tilde{U} \hat{H} \partial E \right] \notag \\
            &\quad \quad +\frac{1}{24} \sigma \left[\tilde{U} \hat{E} \partial \tilde{U} \right] +\frac{1}{24} \sigma \left[\tilde{U} \partial \tilde{U} \partial E \right]
            +\frac{5}{4} \sigma \left[\tilde{U} \partial H \partial E \right] -\frac{1}{16} \sigma \left[\tilde{U} \partial^3 E \right] +\frac{352933}{56760} \sigma \left[H^4 E \right] \notag \\
            &\quad \quad +\frac{11}{432} \sigma \left[H^3 \hat{E} \right] -\frac{715675}{68112} \sigma \left[H^3 \partial E \right] +\frac{127403}{170280} \sigma \left[H^2 E^2 F \right] -\frac{1}{24} \sum_{\substack{i \in \mathcal{I}_{n-1} \\ i \notin I \cup J}} \sigma \left[H^2 E T_{i} \right] \notag \\
            &\quad \quad -\frac{11}{1296} \sigma \left[H^2 E \hat{H} \right]
            -\frac{1}{48} \sigma \left[H^2 E \partial \tilde{U} \right] +\frac{715675}{204336} \sigma \left[H^2 E \partial H \right] +\frac{19}{324} \sigma \left[H^2 \partial \hat{E} \right] \notag \\
            &\quad \quad -\frac{35201}{51084} \sigma \left[H^2 \partial^2 E \right] -\frac{5}{324} \sigma \left[H E^2 \hat{F} \right] -\frac{112765}{25542} \sigma \left[H E^2 \partial F \right] +\frac{17}{864} \sigma \left[H E F \hat{E} \right] \notag \\
            &\quad \quad +\frac{362855}{136224} \sigma \left[H E F \partial E \right]
            -\frac{5}{54} \sigma \left[H E \partial \hat{H} \right] -\frac{112765}{102168} \sigma \left[H E \partial^2 H \right]
            +\frac{1}{12} \sum_{\substack{i \in \mathcal{I}_{n-1} \\ i \notin I \cup J}} \sigma \left[H T_{i} \hat{E} \right] \notag \\
            &\quad \quad +\frac{1}{6}\sum_{\substack{i \in \mathcal{I}_{n-1} \\ i \notin I \cup J}} \sigma \left[H T_{i} \partial E \right]
            -\frac{20843}{25542} \sigma \left[H \hat{H} \hat{E} \right] -\frac{1}{72} \sigma \left[H \hat{H} \partial E \right] +\frac{1}{24} \sigma \left[H \hat{E} \partial \tilde{U} \right] +\frac{191}{1944} \sigma \left[H \hat{E} \partial H \right] \notag \\
            &\quad \quad +\frac{1}{12} \sigma \left[H \partial \tilde{U} \partial E \right] +\frac{1876463}{204336} \sigma \left[H \partial H \partial E \right] -\frac{32709}{37840} \sigma \left[E^3 F^2 \right] -\frac{1}{24} \sum_{\substack{i \in \mathcal{I}_{n-1} \\ i \notin I \cup J}} \sigma \left[E^2 F T_{i} \right] \notag \\
            &\quad \quad -\frac{31}{1296} \sigma \left[E^2 F \hat{H} \right] -\frac{1}{48} \sigma \left[E^2 F \partial \tilde{U} \right] -\frac{186445}{204336} \sigma \left[E^2 F \partial H \right]
            +\frac{49}{648} \sigma \left[E F \partial \hat{E} \right] +\frac{1569}{7568} \sigma \left[E F \partial^2 E \right] \notag
    \end{align}
    \begin{align}
            &\quad \quad -\frac{1}{18} \sum_{i \in I \cap J} \sigma \left[E T_{i}^2 \right] +\frac{1}{6} \sum_{\substack{i \in I \cup J \\ i \notin I \cap J}} \sigma \left[E T_{i}^2 \right] -\frac{1}{6} \sum_{\substack{i \in \mathcal{I}_{n-1} \\ i \notin I \cup J}} \sigma \left[E T_{i}^2 \right]
            +\frac{1}{6} \sum_{\substack{i,j \in \mathcal{I}_{n-1} \\ i,j \notin I \cup J \\ i < j}} \sigma \left[E T_{i} T_{j}\right] \notag \\
            &\quad \quad -\frac{1}{12} \sum_{\substack{i \in \mathcal{I}_{n-1} \\ i \notin I \cup J}} \sigma \left[E T_{i} \hat{H} \right] +\frac{1}{12} \sum_{\substack{i \in \mathcal{I}_{n-1} \\ i \notin I \cup J}} ETDU[i, x, y] \sigma \left[E T_{i} \partial \tilde{U} \right] 
            -\frac{1}{6} \sum_{\substack{i \in \mathcal{I}_{n-1} \\ i \notin I \cup J}} \sigma \left[E T_{i} \partial H \right] \notag \\
            &\quad \quad -\frac{97601}{102168} \sigma \left[E \hat{H}^2 \right] -\frac{1}{24} \sigma \left[E \hat{H} \partial \tilde{U} \right] -\frac{71}{1944} \sigma \left[E \hat{H} \partial H \right]
            +\frac{44557}{204336} \sigma \left[E \hat{E} \hat{F} \right] -\frac{109}{1944} \sigma \left[E \hat{F} \partial E \right] \notag \\
            &\quad \quad +\frac{1}{24} \sigma \left[E \left(\partial \tilde{U}\right)^2 \right] -\frac{1}{12} \sigma \left[E \partial \tilde{U} \partial H \right]
            +\frac{1114447}{102168} \sigma \left[E \left(\partial H \right)^2 \right] -\frac{856303}{408672} \sigma \left[E \partial E \partial F \right] \notag \\
            &\quad \quad +\frac{1}{12} \sum_{\substack{i \in I \\ i \notin I \cap J}} \sigma \left[E \partial \tilde{W}_{i} \right] -\frac{1}{12} \sum_{\substack{i \in \mathcal{I}_{n-1} \\ i \notin I \cup J}} \sigma \left[E \partial \tilde{W}_{i} \right] +\frac{1}{24} \sum_{i \in I \cap J} \sigma \left[E \partial^2 T_{i} \right] 
            -\frac{1}{8} \sum_{\substack{i \in J \\ i \notin I \cap J}} \sigma \left[E \partial^2 T_{i} \right] \notag \\
            &\quad \quad -\frac{1}{48} \sigma \left[E \partial^3 \tilde{U} \right] +\frac{29393}{102168} \sigma \left[F \hat{E}^2 \right]
            +\frac{11}{972} \sigma \left[F \hat{E} \partial E \right] -\frac{604367}{204336} \sigma \left[F \partial E \partial E \right] -\frac{1}{6} \sum_{\substack{i \in \mathcal{I}_{n-1} \\ i \notin I \cup J}} \sigma \left[T_{i} \partial^2 E \right] \notag \\
            &\quad \quad +\frac{5801}{45408} \sigma \left[\hat{H} \partial \hat{E} \right]
            +\frac{1}{24} \sum_{\substack{i \in I \cup J \\ i \notin I \cap J}} \sigma \left[\hat{E} \tilde{W}_{i} \right] -\frac{1}{12} \sum_{\substack{i \in \mathcal{I}_{n-1} \\ i \notin I \cup J}} \sigma \left[\hat{E} \tilde{W}_{i} \right] +\frac{1}{16} \sum_{\substack{i \in I \\ i \notin I \cap J}} \sigma \left[\hat{E} \partial T_{i} \right] \notag \\
            &\quad \quad -\frac{1}{16} \sum_{\substack{i \in J \\ i \notin I \cap J}} \sigma \left[\hat{E} \partial T_{i} \right] -\frac{1}{8} \sum_{\substack{i \in \mathcal{I}_{n-1} \\ i \notin I \cup J}} \sigma \left[\hat{E} \partial T_{i} \right]
            -\frac{5801}{45408} \sigma \left[\hat{E} \partial \hat{H} \right] -\frac{1}{24} \sigma \left[\hat{E} \partial^2 \tilde{U} \right] +\frac{1}{24} \sum_{\substack{i \in I \cup J \\ i \notin I \cap J}} \sigma \left[ \tilde{W}_{i} \partial E \right] \notag \\
            &\quad \quad -\frac{1}{12} \sum_{\substack{i \in \mathcal{I}_{n-1}\\ i \notin I \cup J}} \sigma \left[\tilde{W}_{i} \partial E \right] -\frac{1}{12} \sigma \left[\partial \tilde{U} \partial^2 E \right]
            +\frac{1}{16} \sum_{\substack{i \in I \\ i \notin I \cap J}} \sigma \left[\partial E \partial T_{i} \right] -\frac{1}{16} \sum_{\substack{i \in J \\ i \notin I \cap J}} \sigma \left[ \partial E \partial T_{i} \right] \notag \\
            &\quad \quad \left. -\frac{1}{8} \sum_{\substack{i \in \mathcal{I}_{n-1}\\ i \notin I \cup J}} \sigma \left[\partial E \partial T_{i} \right] -\frac{1}{24} \sigma \left[\partial E \partial^2 \tilde{U} \right]
        \right) \notag \\
        &\quad + \cdots ~,
    \end{align}

    \begin{align}
        \mathcal{X}^{+}_{I}(z) \mathcal{Y}^{-}_{J}(0) &\sim \frac{\zeta_{I,J}}{z^{n+|I|+|J|}} + \frac{\zeta_{I,J}}{z^{n+|I|+|J|-1}}\left(-\tilde{U} -H\right) \notag \\
        &\quad +\frac{\zeta_{I,J}}{z^{n+|I|+|J|-2}} \left(
            \frac{1}{2} \tilde{U}^2 +\frac{1}{2} \sigma \left[\tilde{U} H \right] -\frac{1}{2} H^2 -\frac{1}{4} \sigma \left[E F \right] -\sum_{\substack{i \in \mathcal{I}_{n-1} \\ i \notin I \cup J}} T_{i}
            -\frac{1}{2} \hat{H} \right. \notag \\
            &\quad \quad \left. -\frac{1}{2} \partial \tilde{U} -\frac{1}{2} \partial H
        \right) \notag \\
        &\quad + \frac{\zeta_{I,J}}{z^{n+|I|+|J|-3}} \left(
            -\frac{1}{6} \tilde{U}^3 -\frac{1}{6} \sigma \left[\tilde{U}^2 H \right] +\frac{1}{6} \sigma \left[\tilde{U} H^2 \right] +\frac{1}{12} \sigma \left[\tilde{U} E F \right] +\frac{1}{2} \sum_{\substack{i \in \mathcal{I}_{n-1} \\ i \notin I \cup J}} \sigma \left[\tilde{U} T_{i} \right] \right. \notag \\
            &\quad \quad +\frac{1}{4} \sigma \left[\tilde{U} \hat{H} \right] +\frac{1}{4} \sigma \left[\tilde{U} \partial \tilde{U} \right] +\frac{1}{4} \sigma \left[\tilde{U} \partial H \right] -\frac{1}{2} H^3 -\frac{1}{12} \sigma \left[H E F \right]
            +\frac{1}{2} \sum_{\substack{i \in \mathcal{I}_{n-1} \\ i \notin I \cup J}} \sigma \left[H T_{i} \right] \notag \\
            &\quad \quad -\frac{1}{12} \sigma \left[H \hat{H} \right] +\frac{1}{4} \sigma \left[H \partial \tilde{U} \right] -\frac{1}{4} \sigma \left[H \partial H \right]
            +\frac{1}{12} \sigma \left[E \hat{F} \right] +\frac{1}{8} \sigma \left[E \partial F \right] -\frac{1}{6} \sigma \left[F \hat{E} \right] \notag \\
            &\quad \quad -\frac{3}{8} \sigma \left[F \partial E \right] +\frac{1}{6} \sum_{\substack{i \in I \cup J \\ i \notin I \cap J}} \tilde{W}_{i} -\frac{1}{3} \sum_{\substack{i \in \mathcal{I}_{n-1} \\ i \notin I \cup J}} \tilde{W}_{i}
            +\frac{1}{4} \sum_{\substack{i \in I \\ i \notin I \cap J}} \partial T_{i} -\frac{1}{4} \sum_{\substack{i \in J \\ i \notin I \cap J}} \partial T_{i} -\frac{1}{2} \sum_{\substack{i \in \mathcal{I}_{n-1} \\ i \notin I \cup J}} \partial T_{i} \notag \\
            &\quad \quad \left. -\frac{1}{6} \partial^2 \tilde{U} -\frac{1}{3} \partial^2 H
        \right) \notag \\
        &\quad + \frac{\zeta_{I,J}}{z^{n+|I|+|J|-4}} \left(
            \frac{1}{24} \tilde{U}^4 +\frac{1}{24} \sigma \left[\tilde{U}^3 H \right] -\frac{1}{24} \sigma \left[\tilde{U}^2 H^2 \right] -\frac{1}{48} \sigma \left[\tilde{U}^2 E F \right] \right. \notag \\
            &\quad \quad -\frac{1}{6} \sum_{\substack{i \in \mathcal{I}_{n-1} \\ i \notin I \cup J}} \sigma \left[\tilde{U}^2 T_{i} \right] -\frac{1}{12} \sigma \left[\tilde{U}^2 \hat{H} \right] -\frac{1}{12} \sigma \left[\tilde{U}^2 \partial \tilde{U} \right] -\frac{1}{12} \sigma \left[\tilde{U}^2 \partial H \right] +\frac{1}{8} \sigma \left[\tilde{U} H^3 \right] \notag \\
            &\quad \quad +\frac{1}{48} \sigma \left[\tilde{U} H E F \right]-\frac{1}{6} \sum_{\substack{i \in \mathcal{I}_{n-1} \\ i \notin I \cup J}} \sigma \left[\tilde{U} H T_{i} \right] +\frac{1}{36} \sigma \left[\tilde{U} H \hat{H} \right] -\frac{1}{12} \sigma \left[\tilde{U} H \partial \tilde{U} \right] \notag \\
            &\quad \quad +\frac{1}{12} \sigma \left[\tilde{U} H \partial H \right] -\frac{1}{36} \sigma \left[\tilde{U} E \hat{F} \right] 
            -\frac{1}{24} \sigma \left[\tilde{U} E \partial F \right] +\frac{1}{18} \sigma \left[\tilde{U} F \hat{E} \right] +\frac{1}{8} \sigma \left[\tilde{U} F \partial E \right] \notag \\
            &\quad \quad -\frac{1}{12} \sum_{\substack{i \in I \cup J \\ i \notin I \cap J}} \sigma \left[\tilde{U} \tilde{W}_{i} \right]
            +\frac{1}{6} \sum_{\substack{i \in \mathcal{I}_{n-1} \\ i \notin I \cup J}} \sigma \left[\tilde{U} \tilde{W}_{i} \right] -\frac{1}{8} \sum_{\substack{i \in I \\ i \notin I \cap J}} \sigma \left[\tilde{U} \partial T_{i} \right] + \frac{1}{8} \sum_{\substack{i \in J \\ i \notin I \cap J}} \sigma \left[ \tilde{U} \partial T_{i} \right] \notag \\
            &\quad \quad +\frac{1}{4} \sum_{\substack{i \in \mathcal{I}_{n-1} \\ i \notin I \cup J}} \sigma \left[\tilde{U} \partial T_{i} \right]
            +\frac{1}{12} \sigma \left[\tilde{U} \partial^2 \tilde{U} \right] +\frac{1}{6} \sigma \left[\tilde{U} \partial^2 H \right] +\frac{3}{16} H^4 -\frac{1}{64} \sigma \left[H^2 E F \right] \notag \\
            &\quad \quad +\frac{1}{6} \sum_{\substack{i \in \mathcal{I}_{n-1} \\ i \notin I \cup J}} \sigma \left[H^2 T_{i} \right] -\frac{1}{36} \sigma \left[H^2 \hat{H} \right]
            +\frac{1}{12} \sigma \left[H^2 \partial \tilde{U} \right]
            -\frac{1}{4} \sigma \left[H^2 \partial H \right] -\frac{1}{72} \sigma \left[H E \hat{F} \right] \notag
    \end{align}
    \begin{align}
            &\quad \quad -\frac{35}{192} \sigma \left[H E \partial F \right] +\frac{23}{216} \sigma \left[H F \hat{E} \right]
            +\frac{19}{192} \sigma \left[H F \partial E \right]
            -\frac{1}{12} \sum_{\substack{i \in I \cup J \\ i \notin I \cap J}} \sigma \left[H \tilde{W}_{i} \right] +\frac{1}{6} \sum_{\substack{i \in \mathcal{I}_{n-1} \\ i \notin I \cup J}} \sigma \left[H \tilde{W}_{i} \right] \notag \\
            &\quad \quad -\frac{1}{8} \sum_{\substack{i \in I \\ i \notin I \cap J}} \sigma \left[H \partial T_{i} \right] +\frac{1}{8} \sum_{\substack{i \in J \\ i \notin I \cap J}} \sigma \left[H \partial T_{i} \right] 
            +\frac{1}{4} \sum_{\substack{i \in \mathcal{I}_{n-1} \\ i \notin I \cup J}} \sigma \left[H \partial T_{i} \right] -\frac{13}{36} \sigma \left[H \partial \hat{H} \right] +\frac{1}{12} \sigma \left[H \partial^2 \tilde{U} \right] \notag \\
            &\quad \quad -\frac{13}{64} \sigma \left[H \partial^2 H \right] -\frac{1}{16} \sigma \left[E^2 F^2 \right]
            +\frac{1}{12} \sum_{\substack{i \in \mathcal{I}_{n-1} \\ i \notin I \cup J}} \sigma \left[E F T_{i} \right] -\frac{23}{216} \sigma \left[E F \hat{H} \right] +\frac{1}{24} \sigma \left[E F \partial \tilde{U} \right] \notag \\
            &\quad \quad -\frac{1}{24} \sigma \left[E F \partial H \right]
            -\frac{1}{4} \sigma \left[F \partial^2 E \right] 
            -\frac{1}{6} \sum_{i \in I \cap J} T_{i}^2 +\frac{1}{2} \sum_{\substack{i \in I \cup J \\ i \notin I \cap J}} T_{i}^2 -\frac{1}{2} \sum_{\substack{i \in \mathcal{I}_{n-1} \\ i \notin I \cup J}} T_{i}^2 +\frac{1}{2} \sum_{\substack{i,j \in \mathcal{I}_{n-1} \\ i,j \notin I \cup J\\i < j}} \sigma \left[T_{i} T_{j} \right] \notag \\
            &\quad \quad +\frac{1}{4} \sum_{\substack{i \in \mathcal{I}_{n-1} \\ i \notin I \cup J}} \sigma \left[T_{i} \hat{H} \right]
            +\frac{1}{4} \sum_{\substack{i \in \mathcal{I}_{n-1} \\ i \notin I \cup J }} \sigma \left[T_{i} \partial \tilde{U} \right]
            +\frac{1}{4} \sum_{\substack{i \in \mathcal{I}_{n-1} \\ i \notin I \cup J}} \sigma \left[T_{i} \partial H \right] -\frac{5}{64} \hat{H}^2 
            +\frac{1}{8} \sigma \left[\hat{H} \partial \tilde{U} \right] \notag \\
            &\quad \quad +\frac{17}{72} \sigma \left[\hat{H} \partial H \right] +\frac{1}{32} \sigma \left[\hat{E} \hat{F} \right]
            +\frac{1}{18} \sigma \left[\hat{F} \partial E \right]
            +\frac{1}{8} \left(\partial \tilde{U} \right)^2 +\frac{1}{8} \sigma \left[\partial \tilde{U} \partial H \right] +\frac{45}{64} \left(\partial H \right)^2 \notag \\
            &\quad \quad \left. +\frac{1}{6} \sum_{\substack{i \in I \\ i \notin I \cap J}} \partial \tilde{W}_{i} -\frac{1}{6} \sum_{\substack{i \in \mathcal{I}_{n-1} \\ i \notin I \cup J}} \partial \tilde{W}_{i}
            +\frac{1}{12} \sum_{i \in I \cap J} \partial^2 T_{i} -\frac{1}{4} \sum_{\substack{i \in J \\ i \notin I \cap J}} \partial^2 T_{i} -\frac{1}{24} \partial^3 \tilde{U} -\frac{1}{8} \partial^3 H
        \right) \notag \\
        &\quad +\frac{\zeta_{I,J}}{z^{n+|I|+|J|-5}} \left(
            -\frac{1}{120} \tilde{U}^5 -\frac{1}{120} \sigma \left[\tilde{U}^4 H\right] +\frac{1}{120} \sigma \left[\tilde{U}^3 H^2 \right] +\frac{1}{240} \sigma \left[\tilde{U}^3 E F \right] +\frac{1}{24} \sum_{\substack{i \in \mathcal{I}_{n-1} \\ i \notin I \cup J}} \sigma \left[\tilde{U}^3  T_{i}\right] \right. \notag \\
            &\quad \quad + \frac{1}{48} \sigma \left[\tilde{U}^3 \hat{H}\right] + \frac{1}{48} \sigma \left[\tilde{U}^3 \partial \tilde{U} \right] + \frac{1}{48} \sigma \left[\tilde{U}^3 \partial H \right] -\frac{1}{40} \sigma \left[\tilde{U}^2 H^3 \right] 
            - \frac{1}{240} \sigma \left[\tilde{U}^2 H E F \right] \notag \\
            &\quad \quad +\frac{1}{24} \sum_{\substack{i \in \mathcal{I}_{n-1} \\ i \notin I \cup J}} \sigma \left[\tilde{U}^2 H T_{i} \right] - \frac{1}{144} \sigma \left[\tilde{U}^2 H \hat{H} \right] + \frac{1}{48} \sigma \left[\tilde{U}^2 H \partial \tilde{U} \right] -\frac{1}{48} \sigma \left[\tilde{U}^2 H \partial H \right] +\frac{1}{144} \sigma \left[\tilde{U}^2 E \hat{F} \right] \notag \\
            &\quad \quad +\frac{1}{96} \sigma \left[\tilde{U}^2 E \partial F \right] -\frac{1}{72} \sigma \left[\tilde{U}^2 F \hat{E} \right] -\frac{1}{32} \sigma \left[\tilde{U}^2 F \partial E \right] 
            +\frac{1}{36} \sum_{\substack{i \in I \cup J \\ i \notin I \cap J}} \sigma \left[\tilde{U}^2 \tilde{W}_{i} \right] -\frac{1}{18} \sum_{\substack{i \in \mathcal{I}_{n-1} \\ i \notin I \cup J}} \sigma \left[\tilde{U}^2 \tilde{W}_{i} \right] \notag \\
            &\quad \quad +\frac{1}{24} \sum_{\substack{i \in I \\ i \notin I \cap J}} \sigma \left[\tilde{U}^2 \partial T_{i} \right] -\frac{1}{24} \sum_{\substack{i \in J \\ i \notin I \cap J}} \sigma \left[\tilde{U} \partial T_{i} \right] 
            -\frac{1}{12} \sum_{\substack{i \in \mathcal{I}_{n-1} \\ i \notin I \cup J}} \sigma \left[\tilde{U}^2 \partial T_{i} \right] -\frac{1}{36} \sigma \left[\tilde{U}^2 \partial^2 \tilde{U} \right] \notag \\
            &\quad \quad -\frac{1}{18} \sigma \left[\tilde{U}^2 \partial^2 H \right] -\frac{3}{80} \sigma \left[\tilde{U} H^4 \right] +\frac{1}{320} \sigma \left[H^2 \tilde{U} E F \right] -\frac{1}{24} \sum_{\substack{i \in \mathcal{I}_{n-1} \\ i \notin I \cup J}} \sigma \left[\tilde{U} H^2 T_{i} \right] + \frac{1}{144} \sigma \left[\tilde{U} H^2 \hat{H} \right] \notag \\
            &\quad \quad -\frac{1}{48} \sigma \left[\tilde{U} H^2 \partial \tilde{U} \right] 
            + \frac{1}{16} \sigma \left[\tilde{U} H^2 \partial H \right] +\frac{1}{288} \sigma \left[\tilde{U} H E \hat{F} \right]
            +\frac{35}{768} \sigma \left[\tilde{U} H E \partial F \right] - \frac{23}{864} \sigma \left[\tilde{U} H F \hat{E} \right] \notag
    \end{align}
    \begin{align}
            &\quad \quad -\frac{19}{768} \sigma \left[\tilde{U} H F \partial E \right] +\frac{1}{36} \sum_{\substack{i \in I \cup J \\ i \notin I \cap J}} \sigma \left[\tilde{U} H \tilde{W}_{i}\right] -\frac{1}{18} \sum_{\substack{i \in \mathcal{I}_{n-1} \\ i \notin I \cup J}} \sigma \left[\tilde{U} H \tilde{W}_{i} \right] 
            +\frac{1}{24} \sum_{\substack{i \in I \\ i \notin I \cap J}} \sigma \left[\tilde{U} H \partial T_{i} \right] \notag \\
            &\quad \quad  -\frac{1}{24} \sum_{\substack{i \in J \\ i \notin I \cap J}} \sigma \left[\tilde{U} H \partial T_{i} \right] -\frac{1}{12} \sum_{\substack{i \in \mathcal{I}_{n-1} \\ i \notin I \cup J}} \sigma \left[\tilde{U} H \partial T_{i} \right] +\frac{13}{108} \sigma \left[\tilde{U} H \partial \hat{H} \right] - \frac{1}{36} \sigma \left[\tilde{U} H \partial^2 \tilde{U} \right] \notag \\
            &\quad \quad + \frac{13}{192} \sigma \left[\tilde{U} H \partial^2 H \right] +\frac{1}{80} \sigma \left[\tilde{U} E^2 F^2 \right] -\frac{1}{48} \sum_{\substack{i \in \mathcal{I}_{n-1} \\ i \notin I \cup J}} \sigma \left[\tilde{U} E F T_{i} \right] +\frac{23}{864} \sigma \left[\tilde{U} E F \hat{H} \right] - \frac{1}{96} \sigma \left[\tilde{U} E F \partial \tilde{U} \right] \notag \\
            &\quad \quad +\frac{1}{96} \sigma \left[\tilde{U} E F \partial H \right] + \frac{1}{12} \sigma \left[\tilde{U} F \partial^2 E \right] +\frac{1}{18} \sum_{\substack{i \in I \cap J}} \sigma \left[\tilde{U} T_{i}^2 \right] -\frac{1}{6} \sum_{\substack{i \in I \cup J \\ i \notin I \cap J}} \sigma \left[\tilde{U} T_{i}^2 \right] +\frac{1}{6} \sum_{\substack{i \in \mathcal{I}_{n-1} \\ i \notin I \cup J}} \sigma \left[\tilde{U} T_{i}^2 \right] \notag \\
            &\quad \quad -\frac{1}{6} \sum_{\substack{i,j \in \mathcal{I}_{n-1} \\ i,j \notin I \cup J \\ i < j}} \sigma \left[\tilde{U} T_{i} T_{j} \right] -\frac{1}{12} \sum_{\substack{i \in \mathcal{I}_{n-1} \\ i \notin I \cup J}} \sigma \left[\tilde{U} T_{i} \hat{H} \right] -\frac{1}{12} \sum_{\substack{i \in \mathcal{I}_{n-1} \\ i \notin I \cup J}} \sigma \left[\tilde{U} T_{i} \partial \tilde{U} \right] -\frac{1}{12} \sum_{\substack{i \in \mathcal{I}_{n-1} \\ i \notin I \cup J}} \sigma \left[\tilde{U} T_{i} \partial H \right] \notag \\
            &\quad \quad +\frac{5}{192} \sigma \left[\tilde{U} \hat{H}^2 \right]
            -\frac{1}{24} \sigma \left[\tilde{U} \hat{H} \partial \tilde{U} \right] -\frac{17}{216} \sigma \left[\tilde{U} \hat{H} \partial H \right]
            - \frac{1}{96} \sigma \left[\tilde{U} \hat{E} \hat{F} \right] - \frac{1}{54} \sigma \left[\tilde{U} \hat{F} \partial E \right] \notag \\
            &\quad \quad -\frac{1}{24} \sigma \left[\tilde{U} \left(\partial \tilde{U}\right)^2 \right]- \frac{1}{24} \sigma \left[\tilde{U} \partial \tilde{U} \partial H \right]
            - \frac{15}{64} \sigma \left[\tilde{U}, \left(\partial H\right)^2 \right] -\frac{1}{12} \sum_{\substack{i \in I \\ i \notin I \cap J}} \sigma \left[ \tilde{U} \partial \tilde{W}_{i} \right] \notag \\
            &\quad \quad +\frac{1}{12} \sum_{\substack{i \in \mathcal{I}_{n-1} \\ i \notin I \cup J}} \sigma \left[\tilde{U} \partial \tilde{W}_{i} \right]
            -\frac{1}{24} \sum_{i \in I \cap J} \sigma \left[\tilde{U} \partial^2 T_{i} \right] +\frac{1}{8} \sum_{\substack{i \in J \\ i \notin I \cap J}} \sigma \left[\tilde{U} \partial^2 T_{i} \right] 
            + \frac{1}{48} \sigma \left[\tilde{U} \partial^3 \tilde{U} \right]
            + \frac{1}{16} \sigma \left[\tilde{U} \partial^3 H \right] \notag \\
            &\quad \quad +\frac{236063}{124872} H^5 + \frac{2317}{58080} \sigma \left[H^3 E F \right] +\frac{1}{8} \sum_{\substack{i \in \mathcal{I}_{n-1} \\ i \notin I \cup J}} \sigma \left[H^3 T_{i} \right]
            - \frac{1}{240} \sigma \left[H^3 \hat{H} \right] + \frac{1}{16} \sigma \left[H^3 \partial \tilde{U} \right] \notag \\
            &\quad \quad +\frac{5}{168} \sigma \left[H^3 \partial H \right] - \frac{1}{27} \sigma \left[H^2 E \hat{F} \right]
            - \frac{406843}{998976} \sigma \left[H^2 E \partial F \right] +\frac{257}{2160} \sigma \left[H^2 F \hat{E} \right] +\frac{139083}{332992} \sigma \left[H^2 F \partial E \right] \notag \\
            &\quad \quad -\frac{1}{36} \sum_{\substack{i \in I \cup J \\ i \notin I \cap J}} \sigma \left[ H^2 \tilde{W}_{i} \right] +\frac{1}{18} \sum_{\substack{i \in \mathcal{I}_{n-1} \\ i \notin I \cup J}} \sigma \left[H^2 \tilde{W}_{i} \right] -\frac{1}{24} \sum_{\substack{i \in I \\ i \notin I \cap J}} \sigma \left[H^2 \partial T_{i} \right] 
            +\frac{1}{24} \sum_{\substack{i \in J \\ i \notin I \cap J}} \sigma \left[H^2 \partial T_{i}\right] \notag \\
            &\quad \quad +\frac{1}{12} \sum_{\substack{i \in \mathcal{I}_{n-1} \\ i \notin I \cup J}} \sigma \left[H^2 \partial T_{i} \right] -\frac{49}{90} \sigma \left[H^2 \partial \hat{H} \right]
            + \frac{1}{36} \sigma \left[H^2 \partial^2 \tilde{U} \right] + \frac{111005}{749232} \sigma \left[H^2 \partial^2 H\right]
            - \frac{8527}{234135} \sigma \left[H E^2 F^2 \right] \notag \\
            &\quad \quad +\frac{1}{48} \sum_{\substack{i \in \mathcal{I}_{n-1} \\ i \notin I \cup J}} \sigma \left[H E F T_{i} \right]
            - \frac{61}{1440} \sigma \left[H E F \hat{H}\right] + \frac{1}{96} \sigma \left[H E F \partial \tilde{U} \right] - \frac{11}{4032} \sigma \left[H E F \partial H \right] \notag \\
            &\quad \quad -\frac{13}{720} \sigma \left[H E \partial \hat{F} \right]
            + \frac{522065}{4495392} \sigma \left[H E \partial^2 F \right] +\frac{107}{720} \sigma \left[H F \partial \hat{E} \right] + \frac{173441}{1123848} \sigma \left[H F \partial^2 E \right] \notag \\
            &\quad \quad +\frac{1}{18} \sum_{i \in I \cap J} \sigma \left[H T_{i}^2 \right] 
            -\frac{1}{6} \sum_{\substack{i \in I \cup J \\ i \notin I \cap J}} \sigma \left[H T_{i}^2 \right] +\frac{1}{6} \sum_{\substack{i \in \mathcal{I}_{n-1} \\ i \notin I \cup J}} \sigma \left[H T_{i}^2 \right]
            -\frac{1}{6} \sum_{\substack{i,j \in \mathcal{I}_{n-1} \\ i,j \notin I \cup J \\ i < j}} \sigma \left[H T_{i} T_{j} \right] \notag
    \end{align}
    \begin{align}
            &\quad \quad +\frac{1}{36} \sum_{\substack{i \in \mathcal{I}_{n-1} \\ i \notin I \cup J}} \sigma \left[H T_{i} \hat{H} \right]
            -\frac{1}{12} \sum_{\substack{i \in \mathcal{I}_{n-1} \\ i \notin I \cup J}} \sigma \left[H T_{i} \partial \tilde{U} \right] +\frac{1}{12} \sum_{\substack{i \in \mathcal{I}_{n-1} \\ i \notin I \cup J}} \sigma \left[H T_{i} \partial H \right] - \frac{40211}{249744} \sigma \left[H \hat{H}^2 \right] \notag \\
            &\quad \quad + \frac{1}{72} \sigma \left[H \hat{H} \partial \tilde{U} \right] + \frac{221}{1080} \sigma \left[H \hat{H} \partial H \right] + \frac{20319}{332992} \sigma \left[H \hat{E} \hat{F}\right] + \frac{13}{540} \sigma \left[H \hat{E} \partial F \right] + \frac{1}{180} \sigma \left[H \hat{F} \partial E \right] \notag \\
            &\quad \quad - \frac{1}{24} \sigma \left[H \left(\partial \tilde{U}\right)^2 \right] 
            + \frac{1}{24} \sigma \left[H \partial \tilde{U} \partial  H \right] + \frac{7247}{7568} \sigma \left[H \left(\partial H\right)^2 \right]
            - \frac{237319}{272448} \sigma \left[H \partial E \partial F \right] \notag \\
            &\quad \quad -\frac{1}{12} \sum_{\substack{i \in I \\ i \notin I \cap J}} \sigma \left[H \partial \tilde{W}_{i} \right] +\frac{1}{12} \sum_{\substack{i \in \mathcal{I}_{n-1} \\ i \notin I \cup J}} \sigma \left[H \partial \tilde{W}_{i} \right] -\frac{1}{24} \sum_{i \in I \cap J} \sigma \left[H \partial^2 T_{i} \right] +\frac{1}{8} \sum_{\substack{i \in J \\ i \notin I \cap J}} \sigma \left[H \partial^2 T_{i} \right] \notag \\
            &\quad \quad - \frac{1}{16} \sigma \left[H \partial^2 \hat{H} \right]
            + \frac{1}{48} \sigma \left[H \partial^3 \tilde{U} \right] + \frac{1}{64} \sigma \left[H \partial^3 H \right] - \frac{101}{4320} \sigma \left[E^2 F \hat{F} \right] + \frac{26951}{327789} \sigma \left[E^2 F \partial F \right] \notag \\
            &\quad \quad + \frac{19}{864} \sigma \left[E F^2 \hat{E} \right] - \frac{457231}{5244624} \sigma \left[E F^2 \partial E \right] -\frac{1}{72} \sum_{\substack{i \in I \cup J \\ i \notin I \cap J}} \sigma \left[E F \tilde{W}_{i} \right] +\frac{1}{36} \sum_{\substack{i \in \mathcal{I}_{n-1} \\ i \notin I \cup J}} \sigma \left[E F \tilde{W}_{i} \right] \notag \\
            &\quad \quad -\frac{1}{48} \sum_{\substack{i \in I \\ i \notin I \cap J}} \sigma \left[E F \partial T_{i} \right] + \frac{1}{48} \sum_{\substack{i \in J \\ i \notin I \cap J}} \sigma \left[E F \partial T_{i} \right] +\frac{1}{24} \sum_{\substack{i \in \mathcal{I}_{n-1} \\ i \notin I \cup J}} \sigma \left[E F \partial T_{i} \right] - \frac{17}{240} \sigma \left[E F \partial \hat{H} \right] \notag \\
            &\quad \quad + \frac{1}{72} \sigma \left[E F \partial^2 \tilde{U} \right] 
            - \frac{648143}{8990784} \sigma \left[E F \partial^2 H \right]
            -\frac{1}{36} \sum_{\substack{i \in \mathcal{I}_{n-1} \\ i \notin I \cup J}} \sigma \left[E T_{i} \hat{F} \right] -\frac{1}{24} \sum_{\substack{i \in \mathcal{I}_{n-1} \\ i \notin I \cup J}} \sigma \left[E T_{i} \partial F \right] \notag \\
            &\quad \quad - \frac{14251}{20978496} \sigma \left[E \hat{H} \hat{F} \right] - \frac{1}{72} \sigma \left[E \partial \tilde{U} \hat{F} \right] + \frac{1}{54} \sigma \left[E \partial H \hat{F} \right] -\frac{1}{48} \sigma \left[E \partial \tilde{U} \partial F \right] +\frac{1}{18} \sum_{\substack{i \in \mathcal{I}_{n-1} \\ i \notin I \cup J}} \sigma \left[F T_{i} \hat{E} \right] \notag \\
            &\quad \quad +\frac{1}{8} \sum_{\substack{i \in \mathcal{I}_{n-1} \\ i \notin I \cup J}} \sigma \left[F T_{i} \partial E\right] -\frac{347243}{20978496} \sigma \left[F \hat{H} \hat{E}\right]
            -\frac{7}{108} \sigma \left[F \hat{H} \partial E\right] +\frac{1}{36} \sigma \left[F \hat{E} \partial \tilde{U} \right] +\frac{1}{16} \sigma \left[F \partial \tilde{U} \partial E \right] \notag \\
            &\quad \quad -\frac{5}{288} \sigma \left[F \partial H \partial E \right] -\frac{1}{30} \sum_{i \in I \cap J} \sigma \left[T_{i} \tilde{W}_{i} \right] +\frac{1}{20} \sum_{\substack{i \in I \cup J \\ i \notin I \cap J}} \sigma \left[T_{i} \tilde{W}_{i} \right]
            -\frac{1}{30} \sum_{\substack{i \in \mathcal{I}_{n-1} \\ i \notin I \cup J}} \sigma \left[T_{i} \tilde{W}_{i} \right] \notag \\
            &\quad \quad -\frac{1}{12} \sum_{\substack{i \in \mathcal{I}_{n-1} \\ i \notin I \cup J}} \sum_{\substack{j \in I \cup J \\ j \notin I \cap J}} \sigma \left[T_{i} \tilde{W}_{j} \right] +\frac{1}{6} \sum_{\substack{i,j \in \mathcal{I}_{n-1} \\ i,j \notin I \cup J}} \sigma \left[T_{i} \tilde{W}_{j} \right]
            -\frac{1}{12} \sum_{i \in I \cap J} \sigma \left[T_{i} \partial T_{i} \right] 
            + \frac{3}{8} \sum_{\substack{i \in I \\ i \notin I \cap J}} \sigma \left[T_{i} \partial T_{i} \right] \notag \\
            &\quad \quad +\frac{1}{8} \sum_{\substack{i \in J \\ i \notin I \cap J}} \sigma \left[T_{i} \partial T_{i} \right] -\frac{1}{4} \sum_{\substack{i \in \mathcal{I}_{n-1} \\ i \notin I \cap J}} \sigma \left[T_{i} \partial T_{i} \right] -\frac{1}{8} \sum_{\substack{i \in \mathcal{I}_{n-1} \\ i \notin I \cup J}} \sum_{\substack{j \in I \\ j \notin I \cap J}} \sigma \left[T_{i} \partial T_{j} \right] 
            +\frac{1}{8} \sum_{\substack{i \in \mathcal{I}_{n-1} \\ i \notin I \cup J}} \sum_{\substack{j \in J \\ j \notin I \cap J}} \sigma \left[T_{i} \partial T_{j} \right] \notag \\
            &\quad \quad +\frac{1}{4} \sum_{\substack{i, j \in \mathcal{I}_{n-1} \\ i,j \notin I \cup J}} \sigma \left[T_{i} \partial T_{j} \right] +\frac{1}{12} \sum_{\substack{i \in \mathcal{I}_{n-1} \\ i \notin I \cup J}} \sigma \left[T_{i} \partial^2 \tilde{U} \right]
            +\frac{1}{6} \sum_{\substack{i \in \mathcal{I}_{n-1} \\ i \notin I \cup J}} \sigma \left[T_{i} \partial^2 H \right]
            -\frac{1}{24} \sum_{\substack{i \in I \cup J \\ i \notin I \cap J}} \sigma \left[\hat{H} \tilde{W}_{i} \right] \notag \\
            &\quad \quad +\frac{1}{12} \sum_{\substack{i \in \mathcal{I}_{n-1} \\ i \notin I \cup J}} \sigma \left[\hat{H} \tilde{W}_{i} \right] -\frac{1}{16} \sum_{\substack{i \in I \\ i \notin I \cap J}} \sigma \left[\hat{H} \partial T_{i} \right] +\frac{1}{16} \sum_{\substack{i \in J \\ i \notin I \cap J}} \sigma \left[\hat{H} \partial T_{i} \right] +\frac{1}{8} \sum_{\substack{i \in \mathcal{I}_{n-1} \\ i \notin I \cup J}} \sigma \left[\hat{H} \partial T_{i}\right] \notag
    \end{align}
    \begin{align}
            &\quad \quad + \frac{1}{1344} \sigma \left[\hat{H} \partial \hat{H} \right] +\frac{1}{24} \sigma \left[\hat{H} \partial^2 \tilde{U} \right] -\frac{48175}{635712} \sigma \left[\hat{E} \partial \hat{F}\right] +\frac{11013}{211904} \sigma \left[\hat{F} \partial \hat{E} \right] -\frac{1}{24} \sum_{\substack{i \in I \cup J\\ i \notin I \cap J}} \sigma \left[ \tilde{W}_{i} \partial \tilde{U} \right] \notag \\
            &\quad \quad +\frac{1}{12} \sum_{\substack{i \in \mathcal{I}_{n-1} \\ i \notin I \cup J}} \sigma \left[\tilde{W}_{i} \partial \tilde{U} \right]
            -\frac{1}{24} \sum_{\substack{i \in I \cup J \\ i \notin I \cap J}} \sigma \left[\tilde{W}_{i} \partial H \right] +\frac{1}{12} \sum_{\substack{i \in \mathcal{I}_{n-1} \\ i \notin I \cup J}} \sigma \left[\tilde{W}_{i} \partial H\right] -\frac{1}{16} \sum_{\substack{i \in I \\ i \notin I \cap J}} \sigma \left[\partial \tilde{U} \partial T_{i} \right] \notag \\
            &\quad \quad +\frac{1}{16} \sum_{\substack{i \in J \\ i \notin I \cap J}} \sigma \left[\partial \tilde{U} \partial T_{i} \right] 
            +\frac{1}{8} \sum_{\substack{i \in \mathcal{I}_{n-1} \\ i \notin I \cup J}} \sigma \left[\partial \tilde{U} \partial T_{i}\right] +\frac{1}{24} \sigma \left[\partial \tilde{U} \partial^2 \tilde{U} \right]
            + \frac{1}{12} \sigma \left[\partial \tilde{U} \partial^2 H\right] -\frac{1}{16} \sum_{\substack{i \in I \\ i \notin I \cap J}} \sigma \left[\partial H \partial T_{i} \right] \notag \\
            &\quad \quad +\frac{1}{16} \sum_{\substack{i \in J \\ i \notin I \cap J}} \sigma \left[\partial H \partial T_{i}\right] +\frac{1}{8} \sum_{\substack{i \in \mathcal{I}_{n-1} \\ i \notin I \cup J}} \sigma \left[\partial H \partial T_{i} \right] + \frac{1}{24} \sigma \left[\partial H \partial^2 \tilde{U}\right] +\frac{7}{180} \sum_{i \in I \cap J} \partial^2 \tilde{W}_{i} + \frac{1}{40} \sum_{\substack{i \in I \\ i \notin I \cap J}} \partial^2 \tilde{W}_{i} \notag \\
            &\quad \quad \left. -\frac{7}{120} \sum_{\substack{i \in J \\ i \notin I \cap J}} \partial^2 \tilde{W}_{i} -\frac{1}{60} \sum_{\substack{i \in \mathcal{I}_{n-1} \\ i \notin I \cup J}} \partial^2 \tilde{W}_{i} +\frac{1}{24} \sum_{i \in I \cap J} \partial^3 T_{i} -\frac{1}{16} \sum_{\substack{i \in I \cup J \\ i \notin I \cap J}} \partial^3 T_{i} +\frac{1}{24} \sum_{\substack{i \in \mathcal{I}_{n-1} \\ i \notin I \cup J}} \partial^3 T_{i} - \frac{1}{120} \partial^4 \tilde{U}
        \right) \notag \\
        &\quad +\cdots ~,
    \end{align}

    \begin{align}
        \mathcal{X}^{-}_{I}(z) \mathcal{Y}^{+}_{J}(0) &\sim \frac{\zeta_{I,J}}{z^{n+|I|+|J|}} +\frac{\zeta_{I,J}}{z^{n+|I|+|J|-1}} \left(-\tilde{U} +H\right) \notag \\
        &\quad +\frac{\zeta_{I,J}}{z^{n+|I|+|J|-2}} \left(
            \frac{1}{2} \tilde{U}^2 -\frac{1}{2} \sigma \left[\tilde{U} H \right] -\frac{1}{2} H^2 -\frac{1}{4} \sigma \left[E F \right] -\sum_{\substack{i \in \mathcal{I}_{n-1} \\ i \notin I \cup J}} T_{i} +\frac{1}{2} \hat{H} \right. \notag \\
            &\quad \quad \left. -\frac{1}{2} \partial \tilde{U} +\frac{1}{2} \partial H
        \right) \notag \\
        &\quad + \frac{\zeta_{I,J}}{z^{n+|I|+|J|-3}} \left(
            -\frac{1}{6} \tilde{U}^3 +\frac{1}{6} \sigma \left[\tilde{U}^2 H \right] +\frac{1}{6} \sigma \left[\tilde{U} H^2 \right] +\frac{1}{12} \sigma \left[\tilde{U} E F \right] +\frac{1}{2} \sum_{\substack{i \in \mathcal{I}_{n-1} \\ i \notin I \cup J}} \sigma \left[\tilde{U} T_{i} \right] \right. \notag \\
            &\quad \quad -\frac{1}{4} \sigma \left[\tilde{U} \hat{H} \right] 
            +\frac{1}{4} \sigma \left[\tilde{U} \partial \tilde{U} \right] -\frac{1}{4} \sigma \left[\tilde{U} \partial H \right] +\frac{1}{2} H^3  +\frac{1}{12} \sigma \left[H E F \right] -\frac{1}{2} \sum_{\substack{i \in \mathcal{I}_{n-1} \\ i \notin I \cup J}} \sigma \left[H T_{i} \right] \notag \\
            &\quad \quad -\frac{1}{12} \sigma \left[H \hat{H} \right]
            -\frac{1}{4} \sigma \left[H \partial \tilde{U} \right] -\frac{1}{4} \sigma \left[H \partial H \right] -\frac{1}{6} \sigma \left[E \hat{F} \right] -\frac{3}{8} \sigma \left[E \partial F \right]
            +\frac{1}{12} \sigma \left[F \hat{E} \right] \notag \\
            &\quad \quad +\frac{1}{8} \sigma \left[F \partial E \right] +\frac{1}{6} \sum_{\substack{i \in I \cup J \\ i \notin I \cap J}} \tilde{W}_{i}
            -\frac{1}{3} \sum_{\substack{i \in \mathcal{I}_{n-1} \\ i \notin I \cup J}} \tilde{W}_{i} +\frac{1}{4} \sum_{\substack{i \in I \\ i \notin I \cap J}} \partial T_{i} -\frac{1}{4} \sum_{\substack{i \in J \\ i \notin I \cap J}} \partial T_{i} 
            -\frac{1}{2} \sum_{\substack{i \in \mathcal{I}_{n-1} \\ i \notin I \cup J}} \partial T_{i} \notag \\
            &\quad \quad \left. -\frac{1}{6} \partial^2 \tilde{U} +\frac{1}{3} \partial^2 H
        \right) \notag
    \end{align}
    \begin{align}
        &\quad + \frac{\zeta_{I,J}}{z^{n+|I|+|J|-4}} \left(
            \frac{1}{24} \tilde{U}^4 -\frac{1}{24} \sigma \left[\tilde{U}^3 H \right] -\frac{1}{24} \sigma \left[\tilde{U}^2 H^2 \right] -\frac{1}{48} \sigma \left[\tilde{U}^2 E F \right] \right. \notag \\
            &\quad \quad -\frac{1}{6} \sum_{\substack{i \in \mathcal{I}_{n-1} \\ i \notin I \cup J}} \sigma \left[\tilde{U}^2 T_{i} \right] +\frac{1}{12} \sigma \left[\tilde{U}^2 \hat{H} \right] -\frac{1}{12} \sigma \left[\tilde{U}^2 \partial \tilde{U} \right] +\frac{1}{12} \sigma \left[\tilde{U}^2 \partial H \right] -\frac{1}{8} \sigma \left[\tilde{U} H^3 \right] \notag \\
            &\quad \quad -\frac{1}{48} \sigma \left[\tilde{U} H E F \right] +\frac{1}{6} \sum_{\substack{i \in \mathcal{I}_{n-1} \\ i \notin I \cup J}} \sigma \left[\tilde{U} H T_{i} \right] +\frac{1}{36} \sigma \left[\tilde{U} H \hat{H} \right] +\frac{1}{12} \sigma \left[\tilde{U} H \partial \tilde{U} \right] \notag \\
            &\quad \quad +\frac{1}{12} \sigma \left[\tilde{U} H \partial H \right] +\frac{1}{18} \sigma \left[\tilde{U} E \hat{F} \right] +\frac{1}{8} \sigma \left[\tilde{U} E \partial F \right] -\frac{1}{36} \sigma \left[\tilde{U} F \hat{E} \right] -\frac{1}{24} \sigma \left[\tilde{U} F \partial E \right] \notag \\
            &\quad \quad -\frac{1}{12} \sum_{\substack{i \in I \cup J \\ i \notin I \cap J}} \sigma \left[\tilde{U} \tilde{W}_{i} \right] +\frac{1}{6} \sum_{\substack{i \in \mathcal{I}_{n-1} \\ i \notin I \cup J}} \sigma \left[\tilde{U} \tilde{W}_{i} \right] 
            -\frac{1}{8} \sum_{\substack{i \in I \\ i \notin I \cap J}} \sigma \left[\tilde{U} \partial T_{i} \right] +\frac{1}{8} \sum_{\substack{i \in J \\ i \notin I \cap J}} \sigma \left[\tilde{U} \partial T_{i} \right] \notag \\
            &\quad \quad +\frac{1}{4} \sum_{\substack{i \in \mathcal{I}_{n-1} \\ i \notin I \cup J}} \sigma \left[\tilde{U} \partial T_{i} \right] +\frac{1}{12} \sigma \left[\tilde{U} \partial^2 \tilde{U} \right] -\frac{1}{6} \sigma \left[\tilde{U} \partial^2 H \right] 
            -\frac{9}{16} H^4 +\frac{1}{192} \sigma \left[H^2 E F \right] \notag \\
            &\quad \quad +\frac{1}{6} \sum_{\substack{i \in \mathcal{I}_{n-1} \\ i \notin I \cup J}} \sigma \left[H^2 T_{i} \right] +\frac{1}{36} \sigma \left[H^2 \hat{H} \right] +\frac{1}{12} \sigma \left[H^2 \partial \tilde{U} \right] +\frac{1}{4} \sigma \left[H^2 \partial H \right]
            -\frac{5}{72} \sigma \left[H E \hat{F} \right] \notag \\
            &\quad \quad +\frac{65}{192} \sigma \left[H E \partial F \right] -\frac{5}{216} \sigma \left[H F \hat{E} \right] -\frac{49}{192} \sigma \left[H F \partial E \right] +\frac{1}{12} \sum_{\substack{i \in I \cup J \\ i \notin I \cap J}} \sigma \left[H \tilde{W}_{i} \right] -\frac{1}{6} \sum_{\substack{i \in \mathcal{I}_{n-1} \\ i \notin I \cup J}} \sigma \left[H \tilde{W}_{i} \right] \notag \\
            &\quad \quad +\frac{1}{8} \sum_{\substack{i \in I \\ i \notin I \cap J}} \sigma \left[H \partial T_{i} \right] -\frac{1}{8} \sum_{\substack{i \in J \\ i \notin I \cap J}} \sigma \left[H \partial T_{i} \right] -\frac{1}{4} \sum_{\substack{i \in \mathcal{I}_{n-1} \\ i \notin I \cup J}} \sigma \left[H \partial T_{i} \right]
            -\frac{5}{36} \sigma \left[H \partial \hat{H} \right] -\frac{1}{12} \sigma \left[H \partial^2 \tilde{U} \right] \notag \\
            &\quad \quad +\frac{29}{192} \sigma \left[H \partial^2 H \right]
            +\frac{5}{48} \sigma \left[E^2 F^2 \right] +\frac{1}{12} \sum_{\substack{i \in \mathcal{I}_{n-1} \\ i \notin I \cup J}} \sigma \left[E F T_{i} \right] +\frac{23}{216} \sigma \left[E F \hat{H} \right] +\frac{1}{24} \sigma \left[E F \partial \tilde{U} \right] \notag \\
            &\quad \quad +\frac{1}{24} \sigma \left[E F \partial H \right] +\frac{1}{4} \sigma \left[F \partial^2 E \right] -\frac{1}{6} \sum_{i \in I \cap J} T_{i}^2 +\frac{1}{2} \sum_{\substack{i \in I \cup J \\ i \notin I \cap J}} T_{i}^2 -\frac{1}{2} \sum_{\substack{i \in \mathcal{I}_{n-1} \\ i \notin I \cup J}} T_{i}^2 +\frac{1}{2} \sum_{\substack{i,j \in \mathcal{I}_{n-1} \\ i,j \notin I \cup J \\ i < j}} \sigma \left[T_{i} T_{j} \right] \notag \\
            &\quad \quad -\frac{1}{4} \sum_{\substack{i \in \mathcal{I}_{n-1} \\ i \notin I \cup J}} \sigma \left[T_{i} \hat{H} \right] +\frac{1}{4} \sum_{\substack{i \in \mathcal{I}_{n-1} \\ i \notin I \cup J}} \sigma \left[T_{i} \partial \tilde{U} \right] -\frac{1}{4} \sum_{\substack{i \in \mathcal{I}_{n-1} \\ i \notin I \cup J}} \sigma \left[T_{i} \partial H \right] 
            +\frac{7}{64} \hat{H}^2 -\frac{1}{8} \sigma \left[\hat{H} \partial \tilde{U} \right] \notag \\
            &\quad \quad +\frac{1}{72} \sigma \left[\hat{H} \partial H \right] -\frac{3}{32} \sigma \left[\hat{E} \hat{F} \right] -\frac{1}{18} \sigma \left[\hat{F} \partial E \right]
            +\frac{1}{8} \left(\partial \tilde{U} \right)^2 -\frac{1}{8} \sigma \left[\partial \tilde{U} \partial H \right]
            -\frac{95}{64} \left(\partial H \right)^2 \notag \\
            &\quad \quad \left. +\frac{1}{6} \sum_{\substack{i \in I \\ i \notin I \cap J}} \partial \tilde{W}_{i} -\frac{1}{6} \sum_{\substack{i \in \mathcal{I}_{n-1} \\ i \notin I \cup J}} \partial \tilde{W}_{i}
            +\frac{1}{12} \sum_{i \in I \cap J} \partial^2 T_{i} -\frac{1}{4} \sum_{\substack{i \in J \\ i \notin I \cap J}} \partial^2 T_{i}  -\frac{1}{24} \partial^3 \tilde{U} +\frac{1}{8} \partial^3 H
        \right) \notag
    \end{align}
    \begin{align}
        &\quad + \frac{\zeta_{I,J}}{z^{n+|I|+|J|-5}} \left(
            -\frac{1}{120} \tilde{U}^5 +\frac{1}{120} \sigma \left[\tilde{U}^4 H \right] + \frac{1}{120} \sigma \left[\tilde{U}^3 H^2 \right] + \frac{1}{240} \sigma \left[\tilde{U}^3 E F \right]
            +\frac{1}{24} \sum_{\substack{i \in \mathcal{I}_{n-1}\\ i \notin I \cup J}} \sigma \left[\tilde{U}^3 T_{i} \right] \right. \notag \\
            &\quad \quad - \frac{1}{48} \sigma \left[\tilde{U}^3 \hat{H} \right] +\frac{1}{48} \sigma \left[\tilde{U}^3 \partial \tilde{U} \right]
            - \frac{1}{48} \sigma \left[\tilde{U}^3 \partial H \right] + \frac{1}{40} \sigma \left[\tilde{U}^2 H^3 \right] + \frac{1}{240} \sigma \left[\tilde{U}^2 H E F \right] \notag \\
            &\quad \quad -\frac{1}{24} \sum_{\substack{i \in \mathcal{I}_{n-1}\\ i \notin I \cup J}} \sigma \left[\tilde{U}^2 H T_{i} \right] - \frac{1}{144} \sigma \left[\tilde{U}^2 H \hat{H} \right] - \frac{1}{48} \sigma \left[\tilde{U}^2 H \partial \tilde{U} \right]
            - \frac{1}{48} \sigma \left[\tilde{U}^2 H \partial H \right] - \frac{1}{72} \sigma \left[\tilde{U}^2 E \hat{F} \right] \notag \\
            &\quad \quad - \frac{1}{32} \sigma \left[\tilde{U}^2 E \partial F \right] 
            + \frac{1}{144} \sigma \left[\tilde{U}^2 F \hat{E} \right] + \frac{1}{96} \sigma \left[\tilde{U}^2 F \partial E \right] +\frac{1}{36} \sum_{\substack{i \in I \cup J \\ i \notin I \cap J}} \sigma \left[\tilde{U}^2 \tilde{W}_{i} \right] -\frac{1}{18} \sum_{\substack{i \in \mathcal{I}_{n-1} \\ i \notin I \cup J}} \sigma \left[\tilde{U}^2 \tilde{W}_{i} \right] \notag \\
            &\quad \quad +\frac{1}{24} \sum_{\substack{i \in I \\ i \notin I \cap J}} \sigma \left[\tilde{U}^2 \partial T_{i} \right] -\frac{1}{24} \sum_{\substack{i \in J \\ i \notin I \cap J}} \sigma \left[\tilde{U}^2 \partial T_{i} \right]
            -\frac{1}{12} \sum_{\substack{i \in \mathcal{I}_{n-1} \\ i \notin I \cup J}} \sigma \left[\tilde{U}^2 \partial T_{i} \right] -\frac{1}{36} \sigma \left[\tilde{U}^2 \partial^2 \tilde{U} \right] \notag \\
            &\quad \quad + \frac{1}{18} \sigma \left[\tilde{U}^2 \partial^2 H \right] + \frac{9}{80} \sigma \left[\tilde{U} H^4 \right] - \frac{1}{960} \sigma \left[\tilde{U} H^2 E F \right] -\frac{1}{24} \sum_{\substack{i \in \mathcal{I}_{n-1}\\i \notin I \cup J}} \sigma \left[\tilde{U} H^2 T_{i} \right] - \frac{1}{144} \sigma \left[\tilde{U} H^2 \hat{H} \right] \notag \\
            &\quad \quad - \frac{1}{48} \sigma \left[\tilde{U} H^2 \partial \tilde{U} \right]
            - \frac{1}{16} \sigma \left[\tilde{U} H^2 \partial H \right] + \frac{5}{288} \sigma \left[\tilde{U} H E \hat{F} \right]
            - \frac{65}{768} \sigma \left[\tilde{U} H E \partial F \right]
            + \frac{5}{864} \sigma \left[\tilde{U} H F \hat{E} \right] \notag \\
            &\quad \quad + \frac{49}{768} \sigma \left[\tilde{U} H F \partial E \right]
            -\frac{1}{36} \sum_{\substack{i \in I \cup J \\ i \notin I \cap J}} \sigma \left[\tilde{U} H \tilde{W}_{i} \right] +\frac{1}{18} \sum_{\substack{i \in \mathcal{I}_{n-1} \\ i \notin I \cup J}} \sigma \left[\tilde{U} H \tilde{W}_{i} \right]
            -\frac{1}{24} \sum_{\substack{i \in I \\ i \notin I \cap J}} \sigma \left[\tilde{U} H \partial T_{i} \right] \notag \\
            &\quad \quad +\frac{1}{24} \sum_{\substack{i \in J \\ i \notin I \cap J}} \sigma \left[\tilde{U} H \partial T_{i} \right] +\frac{1}{12} \sum_{\substack{i \in \mathcal{I}_{n-1} \\ i \notin I \cup J}} \sigma \left[\tilde{U} H \partial T_{i} \right]
            + \frac{5}{108} \sigma \left[\tilde{U} H \partial \hat{H} \right] + \frac{1}{36} \sigma \left[\tilde{U} H \partial^2 \tilde{U} \right] \notag \\
            &\quad \quad - \frac{29}{576} \sigma \left[\tilde{U} H \partial^2 H \right]
            - \frac{1}{48} \sigma \left[\tilde{U} E^2 F^2 \right] -\frac{1}{48} \sum_{\substack{i \in \mathcal{I}_{n-1} \\ i \notin I \cup J}} \sigma \left[\tilde{U} E F T_{i} \right] - \frac{23}{864} \sigma \left[\tilde{U} E F \hat{H} \right]
            - \frac{1}{96} \sigma \left[\tilde{U} E F \partial \tilde{U} \right] \notag \\
            &\quad \quad - \frac{1}{96} \sigma \left[\tilde{U} E F \partial H \right]
            - \frac{1}{12} \sigma \left[\tilde{U} F \partial^2 E \right] +\frac{1}{18} \sum_{i \in I \cap J} \sigma \left[\tilde{U} T_{i}^2 \right] -\frac{1}{6} \sum_{\substack{i \in I \cup J \\ i \notin I \cap J}} \sigma \left[\tilde{U} T_{i}^2 \right] +\frac{1}{6} \sum_{\substack{i \in \mathcal{I}_{n-1} \\ i \notin I \cup J}} \sigma \left[\tilde{U} T_{i}^2 \right] \notag \\
            &\quad \quad -\frac{1}{6} \sum_{\substack{i,j \in \mathcal{I}_{n-1} \\ i,j \notin I \cup J \\ i < j}} \sigma \left[\tilde{U} T_{i} T_{j} \right]
            +\frac{1}{12} \sum_{\substack{i \in \mathcal{I}_{n-1} \\ i \notin I \cup J}} \sigma \left[\tilde{U} T_{i} \hat{H} \right] -\frac{1}{12} \sum_{\substack{i \in \mathcal{I}_{n-1} \\ i \notin I \cup J}} \sigma \left[\tilde{U} T_{i} \partial \tilde{U} \right] +\frac{1}{12} \sum_{\substack{i \in \mathcal{I}_{n-1} \\ i \notin I \cup J}} \sigma \left[\tilde{U} T_{i} \partial H \right] \notag \\
            &\quad \quad - \frac{7}{192} \sigma \left[\tilde{U} \hat{H}^2 \right]
            + \frac{1}{24} \sigma \left[\tilde{U} \hat{H} \partial \tilde{U} \right] - \frac{1}{216} \sigma \left[\tilde{U} \hat{H} \partial H \right]
            + \frac{1}{32} \sigma \left[\tilde{U} \hat{E} \hat{F} \right] + \frac{1}{54} \sigma \left[\tilde{U} \hat{F} \partial E \right] \notag \\
            &\quad \quad - \frac{1}{24} \sigma \left[\tilde{U} \left(\partial \tilde{U}\right)^2 \right]
            + \frac{1}{24} \sigma \left[\tilde{U} \partial \tilde{U} \partial H \right]
            + \frac{95}{192} \sigma \left[\tilde{U} \left(\partial H\right)^2 \right] -\frac{1}{12} \sum_{\substack{i \in I \\ i \notin I \cap J}} \sigma \left[\tilde{U} \partial \tilde{W}_{i} \right] \notag \\
            &\quad \quad +\frac{1}{12} \sum_{\substack{i \in \mathcal{I}_{n-1} \\ i \notin I \cup J}} \sigma \left[\tilde{U} \partial \tilde{W}_{i} \right]
            -\frac{1}{24} \sum_{i \in I \cap J} \sigma \left[\tilde{U} \partial^2 T_{i} \right] +\frac{1}{8} \sum_{\substack{i \in J \\ i \notin I \cap J}} \sigma \left[\tilde{U} \partial^2 T_{i} \right]
            + \frac{1}{48} \sigma \left[\tilde{U} \partial^3 \tilde{U} \right] - \frac{1}{16} \sigma \left[\tilde{U} \partial^3 H \right] \notag
    \end{align}
    \begin{align}
            &\quad \quad - \frac{226603}{124872} H^5 - \frac{809}{19360} \sigma \left[H^3 E F \right] -\frac{1}{8} \sum_{\substack{i \in \mathcal{I}_{n-1} \\ i \notin I \cup J}} \sigma \left[H^3 T_{i} \right] - \frac{1}{240} \sigma \left[H^3 \hat{H} \right] - \frac{1}{16} \sigma \left[H^3 \partial \tilde{U} \right] \notag \\
            &\quad \quad + \frac{5}{168} \sigma \left[H^3 \partial H \right] + \frac{931}{2160} \sigma \left[H^2 E \hat{F} \right] + \frac{1176067}{2996928} \sigma \left[H^2 E \partial F \right]
            - \frac{377}{1080} \sigma \left[H^2 F \hat{E} \right] \notag \\
            &\quad \quad - \frac{1144849}{2996928} \sigma \left[H^2 F \partial E \right] -\frac{1}{36} \sum_{\substack{i \in I \cup J \\ i \notin I \cap J}} \sigma \left[H^2 \tilde{W}_{i} \right] +\frac{1}{18} \sum_{\substack{i \in \mathcal{I}_{n-1} \\ i \notin I \cup J}} \sigma \left[H^2 \tilde{W}_{i} \right]
            -\frac{1}{24} \sum_{\substack{i \in I \\ i \notin I \cap J}} \sigma \left[H^2 \partial T_{i} \right] \notag \\
            &\quad \quad + \frac{1}{24} \sum_{\substack{i \in J \\ i \notin I \cap J}} \sigma \left[H^2 \partial T_{i} \right] +\frac{1}{12} \sum_{\substack{i \in \mathcal{I}_{n-1} \\ i \notin I \cup J}} \sigma \left[H^2 \partial T_{i} \right] + \frac{8}{3} \sigma \left[H^2 \partial \hat{H} \right] + \frac{1}{36} \sigma \left[H^2 \partial^2 \tilde{U} \right]
            - \frac{115735}{749232} \sigma \left[H^2 \partial^2 H \right] \notag \\
            &\quad \quad + \frac{61121}{1873080} \sigma \left[H E^2 F^2 \right]
            -\frac{1}{48} \sum_{\substack{i \in \mathcal{I}_{n-1} \\ i \notin I \cup J}} \sigma \left[H E F T_{i} \right] - \frac{61}{1440} \sigma \left[H E F \hat{H} \right] - \frac{1}{96} \sigma \left[H E F \partial \tilde{U} \right] \notag \\
            &\quad \quad - \frac{11}{4032} \sigma \left[H E F \partial H \right]
            - \frac{17}{48} \sigma \left[H E \partial \hat{F} \right] - \frac{686669}{4495392} \sigma \left[H E \partial^2 F \right] - \frac{3}{16} \sigma \left[H F \partial \hat{E} \right] \notag \\
            &\quad \quad - \frac{257485}{2247696} \sigma \left[H F \partial^2 E \right] -\frac{1}{18} \sum_{i \in I \cap J} \sigma \left[H T_{i}^2 \right] +\frac{1}{6} \sum_{\substack{i \in I \cup J \\ i \notin I \cap J}} \sigma \left[H T_{i}^2 \right] -\frac{1}{6} \sum_{\substack{i \in \mathcal{I}_{n-1} \\ i \notin I \cup J}} \sigma \left[H T_{i}^2 \right] \notag \\
            &\quad \quad +\frac{1}{6} \sum_{\substack{i,j \in \mathcal{I}_{n-1} \\ i,j \notin I \cup J \\ i < j}} \sigma \left[H T_{i} T_{j} \right]
            +\frac{1}{36} \sum_{\substack{i \in \mathcal{I}_{n-1} \\ i \notin I \cup J}} \sigma \left[H T_{i} \hat{H} \right] +\frac{1}{12} \sum_{\substack{i \in \mathcal{I}_{n-1}\\ i \notin I \cup J}} \sigma \left[H T_{i} \partial \tilde{U} \right]
            +\frac{1}{12} \sum_{\substack{i \in \mathcal{I}_{n-1} \\ i \notin I \cup J}} \sigma \left[H T_{i} \partial H \right] \notag \\
            &\quad \quad + \frac{115903}{749232} \sigma \left[H \hat{H}^2 \right] + \frac{1}{72} \sigma \left[H \hat{H} \partial \tilde{U} \right]
            - \frac{233}{216} \sigma \left[H \hat{H} \partial H \right] - \frac{178141}{2996928} \sigma \left[H \hat{E} \hat{F} \right]
            + \frac{5}{108} \sigma \left[H \hat{E} \partial F \right] \notag \\
            &\quad \quad + \frac{1}{36} \sigma \left[H \hat{F} \partial E \right]
            + \frac{1}{24} \sigma \left[H \left(\partial \tilde{U} \right)^2 \right] + \frac{1}{24} \sigma \left[H \partial \tilde{U} \partial H \right] - \frac{60493}{68112} \sigma \left[H \left(\partial H \right)^2 \right] \notag \\
            &\quad \quad + \frac{232589}{272448} \sigma \left[H \partial E \partial F \right] +\frac{1}{12} \sum_{\substack{i \in I \\ i \notin I \cap J}} \sigma \left[H \partial \tilde{W}_{i} \right] -\frac{1}{12} \sum_{\substack{i \in \mathcal{I}_{n-1} \\ i \notin I \cup J}} \sigma \left[H \partial \tilde{W}_{i} \right]
            +\frac{1}{24} \sum_{i \in I \cap J} \sigma \left[H \partial^2 T_{i} \right] \notag \\
            &\quad \quad -\frac{1}{8} \sum_{\substack{i \in J \\ i \notin I \cap J}} \sigma \left[H \partial^2 T_{i} \right]
            - \frac{1}{16} \sigma \left[H \partial^2 \hat{H} \right] - \frac{1}{48} \sigma \left[H \partial^3 \tilde{U} \right] + \frac{1}{64} \sigma \left[H  \partial^3 H \right] + \frac{487}{4320} \sigma \left[E^2 F \hat{F} \right] \notag \\
            &\quad \quad - \frac{424121}{5244624} \sigma \left[E^2 F \partial F \right] - \frac{493}{4320} \sigma \left[E F^2 \hat{E} \right]
            + \frac{22117}{291368} \sigma \left[E F^2 \partial E \right] -\frac{1}{72} \sum_{\substack{i \in I \cup J \\ i \notin I \cap J}} \sigma \left[E F \tilde{W}_{i} \right] \notag \\
            &\quad \quad +\frac{1}{36} \sum_{\substack{i \in \mathcal{I}_{n-1} \\ i \notin I \cup J}} \sigma \left[E F \tilde{W}_{i}\right]
            -\frac{1}{48} \sum_{\substack{i \in I \\ i \notin I \cap J}} \sigma \left[E F \partial T_{i} \right]
            +\frac{1}{48} \sum_{\substack{i \in J \\ i \notin I \cap J}} \sigma \left[E F \partial T_{i} \right] +\frac{1}{24} \sum_{\substack{i \in \mathcal{I}_{n-1} \\ i \notin I \cup J}} \sigma \left[E F \partial T_{i} \right] \notag \\
            &\quad \quad + \frac{31}{144} \sigma \left[E F \partial \hat{H} \right] + \frac{1}{72} \sigma \left[E F \partial^2 \tilde{U} \right]
            + \frac{662333}{8990784} \sigma \left[E F \partial^2 H \right]
            +\frac{1}{18} \sum_{\substack{i \in \mathcal{I}_{n-1} \\ i \notin I \cup J}} \sigma \left[E T_{i} \hat{F} \right] \notag
    \end{align}
    \begin{align}
            &\quad \quad +\frac{1}{8} \sum_{\substack{i \in \mathcal{I}_{n-1} \\ i \notin I \cup J}} \sigma \left[E T_{i} \partial F \right]
            + \frac{380353}{20978496} \sigma \left[E \hat{H} \hat{F} \right] + \frac{1}{36} \sigma \left[E \partial \tilde{U} \hat{F} \right] + \frac{1}{54} \sigma \left[E \partial H \hat{F} \right]
            + \frac{1}{16} \sigma \left[E \partial \tilde{U} \partial F \right] \notag \\
            &\quad \quad -\frac{1}{36} \sum_{\substack{i \in \mathcal{I}_{n-1} \\ i \notin I \cup J}} \sigma \left[F T_{i} \hat{E} \right] -\frac{1}{24} \sum_{\substack{i \in \mathcal{I}_{n-1} \\ i \notin I \cup J}} \sigma \left[F T_{i} \partial E \right]
            + \frac{15787}{6992832} \sigma \left[F \hat{H} \hat{E} \right] - \frac{7}{108} \sigma \left[F \hat{H} \partial E \right] \notag \\
            &\quad \quad - \frac{1}{72} \sigma \left[F \hat{E} \partial \tilde{U} \right]
            - \frac{1}{48} \sigma \left[F \partial \tilde{U} \partial E \right] - \frac{5}{288} \sigma \left[F \partial H \partial E \right] 
            -\frac{1}{30} \sum_{i \in I \cap J} \sigma \left[T_{i} \tilde{W}_{i} \right] +\frac{1}{20} \sum_{\substack{i \in I \cup J \\ i \notin I \cap J}} \sigma \left[T_{i} \tilde{W}_{i} \right] \notag \\
            &\quad \quad -\frac{1}{30} \sum_{\substack{i \in \mathcal{I}_{n-1} \\ i \notin I \cup J}} \sigma \left[T_{i} \tilde{W}_{i} \right] -\frac{1}{12} \sum_{\substack{i \in \mathcal{I}_{n-1} \\ i \notin I \cup J}} \sum_{\substack{j \in I \cup J \\ i \notin I \cap J}} \sigma \left[T_{i} \tilde{W}_{j} \right]
            +\frac{1}{6} \sum_{\substack{i,j \in \mathcal{I}_{n-1} \\ i,j \notin I \cup J \\ i \neq j}} \sigma \left[T_{i} \tilde{W}_{j} \right] -\frac{1}{12} \sum_{i \in I \cap J} \sigma \left[T_{i} \partial T_{i} \right] \notag \\
            &\quad \quad +\frac{3}{8} \sum_{\substack{i \in I \\ i \notin I \cap J}} \sigma \left[T_{i} \partial T_{i} \right] +\frac{1}{8} \sum_{\substack{i \in J \\ i \notin I \cap J}} \sigma \left[T_{i} \partial T_{i} \right]
            -\frac{1}{4} \sum_{\substack{i \in \mathcal{I}_{n-1} \\ i \notin I \cup J}} \sigma \left[T_{i} \partial T_{i} \right]
            -\frac{1}{8} \sum_{\substack{i \in \mathcal{I}_{n-1} \\ i \notin I \cup J}} \sum_{\substack{j \in I \\ j \notin I \cap J}} \sigma \left[T_{i} \partial T_{j} \right] \notag \\
            &\quad \quad +\frac{1}{8} \sum_{\substack{i \in \mathcal{I}_{n-1} \\ i \notin I \cup J}} \sum_{\substack{j \in J \\ j \notin I \cap J}} \sigma \left[T_{i} \partial T_{j} \right] +\frac{1}{4} \sum_{\substack{i,j \in \mathcal{I}_{n-1} \\ i,j \notin I \cup J \\ i \neq j}} \sigma \left[T_{i} \partial T_{j} \right] +\frac{1}{12} \sum_{\substack{i \in \mathcal{I}_{n-1}\\ i \notin I \cup J}} \sigma \left[T_{i} \partial^2 \tilde{U} \right]
            -\frac{1}{6} \sum_{\substack{i \in \mathcal{I}_{n-1} \\ i \notin I \cup J}} \sigma \left[T_{i} \partial^2 H \right] \notag \\
            &\quad \quad +\frac{1}{24} \sum_{\substack{i \in I \cup J \\ i \notin I \cap J}} \sigma \left[\hat{H} \tilde{W}_{i} \right] -\frac{1}{12} \sum_{\substack{i \in \mathcal{I}_{n-1} \\ i \notin I \cup J}} \sigma \left[\hat{H} \tilde{W}_{i} \right] +\frac{1}{16} \sum_{\substack{i \in I \\ i \notin I \cap J}} \sigma \left[\hat{H} \partial T_{i} \right] -\frac{1}{16} \sum_{\substack{i \in J \\ i \notin I \cap J}} \sigma \left[\hat{H} \partial T_{i} \right] \notag \\
            &\quad \quad -\frac{1}{8} \sum_{\substack{i \in \mathcal{I}_{n-1} \\ i \notin I \cup J}} \sigma \left[\hat{H} \partial T_{i} \right] + \frac{1}{1344} \sigma \left[\hat{H} \partial \hat{H} \right]
            - \frac{1}{24} \sigma \left[\hat{H} \partial^2 \tilde{U} \right] + \frac{11013}{211904} \sigma \left[\hat{E} \partial \hat{F} \right] - \frac{48175}{635712} \sigma \left[\hat{F} \partial \hat{E} \right] \notag \\
            &\quad \quad -\frac{1}{24} \sum_{\substack{i \in I \cup J \\ i \notin I \cap J}} \sigma \left[\tilde{W}_{i} \partial \tilde{U} \right] 
            +\frac{1}{12} \sum_{\substack{i \in \mathcal{I}_{n-1} \\ i \notin I \cup J}} \sigma \left[\tilde{W}_{i} \partial \tilde{U} \right] +\frac{1}{24} \sum_{\substack{i \in I \cup J \\ i \notin I \cap J}} \sigma \left[\tilde{W}_{i} \partial H\right] -\frac{1}{12} \sum_{\substack{i \in \mathcal{I}_{n-1} \\ i \notin I \cup J}} \sigma \left[\tilde{W}_{i} \partial H \right] \notag \\
            &\quad \quad -\frac{1}{16} \sum_{\substack{i \in I \\ i \notin I \cap J}} \sigma \left[\partial \tilde{U} \partial T_{i} \right]
            +\frac{1}{16} \sum_{\substack{i \in J \\ i \notin I \cap J}} \sigma \left[\partial \tilde{U} \partial T_{i} \right] +\frac{1}{8} \sum_{\substack{i \in \mathcal{I}_{n-1} \\ i \notin I \cup J}} \sigma \left[\partial \tilde{U} \partial T_{i} \right] + \frac{1}{24} \sigma \left[\partial \tilde{U} \partial^2 \tilde{U} \right] \notag \\
            &\quad \quad - \frac{1}{12} \sigma \left[\partial \tilde{U} \partial^2 H \right] +\frac{1}{16}\sum_{\substack{i \in I \\ i \notin I \cap J}} \sigma \left[\partial H \partial T_{i} \right] 
            -\frac{1}{16} \sum_{\substack{i \in J \\ i \notin I \cap J}} \sigma \left[\partial H \partial T_{i} \right] -\frac{1}{8} \sum_{\substack{i \in \mathcal{I}_{n-1} \\ i \notin I \cup J}} \sigma \left[\partial H \partial T_{i} \right] \notag \\
            &\quad \quad - \frac{1}{24} \sigma \left[\partial H \partial^2 \tilde{U} \right]
            +\frac{7}{180} \sum_{i \in I \cap J} \partial^2 \tilde{W}_{i} +\frac{1}{40} \sum_{\substack{i \in I \\ i \notin I \cap J}} \partial^2 \tilde{W}_{i} -\frac{7}{120} \sum_{\substack{i \in J \\ i \notin I \cap J}} \partial^2 \tilde{W}_{i} -\frac{1}{60} \sum_{\substack{i \in \mathcal{I}_{n-1} \\ i \notin I \cup J}} \partial^2 \tilde{W}_{i} \notag \\
            &\quad \quad \left. +\frac{1}{24} \sum_{i \in I \cap J} \partial^3 T_{i} -\frac{1}{16} \sum_{\substack{i \in I \cup J \\ i \notin I \cap J}} \partial^3 T_{i} +\frac{1}{24} \sum_{\substack{i \in \mathcal{I}_{n-1} \\ i \notin I \cup J}} \partial^3 T_{i} - \frac{1}{120} \partial^4 \tilde{U}
        \right) \notag \\
        &\quad + \cdots ~,
    \end{align}

    \begin{align}
        \mathcal{X}^{-}_{I}(z) \mathcal{Y}^{-}_{J}(0) &\sim \frac{\zeta_{I,J} F}{z^{n+|I|+|J|-1}} +\frac{\zeta_{I,J}}{z^{n+|I|+|J|-2}} \left(
            -\frac{1}{2} \sigma \left[\tilde{U} F \right] +\frac{1}{2} \hat{F} +\frac{1}{2} \partial F
        \right) \notag \\
        &\quad + \frac{\zeta_{I,J}}{z^{n+|I|+|J|-3}} \left(
            \frac{1}{6} \sigma \left[\tilde{U}^2 F \right] -\frac{1}{4} \sigma \left[\tilde{U} \hat{F} \right] -\frac{1}{4} \sigma \left[\tilde{U} \partial F \right] +\frac{1}{6} \sigma \left[H^2 F \right] +\frac{1}{4} \sigma \left[H \hat{F} \right] \right. \notag \\
            &\quad \quad +\frac{1}{2} \sigma \left[H \partial F \right] +\frac{1}{6} \sigma \left[E F^2 \right] -\frac{1}{2} \sum_{\substack{i \in \mathcal{I}_{n-1} \\ i \notin I \cup J}} \sigma \left[F T_{i} \right] -\frac{1}{4} \sigma \left[F \hat{H} \right] -\frac{1}{4} \sigma \left[F \partial \tilde{U} \right] \notag \\
            &\quad \quad \left. -\frac{1}{2} \sigma \left[F \partial H \right] +\frac{1}{3} \partial^2 F
        \right) \notag \\
        &\quad + \frac{1}{z^{n-4}} \left(
            -\frac{1}{24} \sigma \left[\tilde{U}^3 F \right] +\frac{1}{12} \sigma \left[\tilde{U}^2 \hat{F} \right] +\frac{1}{12} \sigma \left[\tilde{U}^2 \partial F \right] -\frac{1}{24} \sigma \left[\tilde{U} H^2 F \right] -\frac{1}{12} \sigma \left[\tilde{U} H \hat{F} \right] \right. \notag \\
            &\quad \quad -\frac{1}{6} \sigma \left[\tilde{U} H \partial F \right]
            -\frac{1}{24} \sigma \left[\tilde{U} E F^2 \right] +\frac{1}{6} \sum_{\substack{i \in \mathcal{I}_{n-1} \\ i \notin I \cup J}} \sigma \left[\tilde{U} F T_{i} \right] +\frac{1}{12} \sigma \left[\tilde{U} F \hat{H} \right] +\frac{1}{12} \sigma \left[\tilde{U} F \partial \tilde{U} \right] \notag \\
            &\quad \quad +\frac{1}{6} \sigma \left[\tilde{U} F \partial H \right]
            -\frac{1}{6} \sigma \left[\tilde{U} \partial^2 F \right] -\frac{1}{108} \sigma \left[H^2 \hat{F} \right] +\frac{1}{12} \sigma \left[H^2 \partial F \right] +\frac{1}{54} \sigma \left[H F \hat{H} \right] \notag \\
            &\quad \quad +\frac{1}{12} \sigma \left[H F \partial H \right] +\frac{2}{9} \sigma \left[H \partial \hat{F} \right] +\frac{1}{4} \sigma \left[H \partial^2 F\right]
            +\frac{1}{24} \sigma \left[E F \hat{F} \right] +\frac{1}{12} \sigma \left[E F \partial F \right] \notag \\
            &\quad \quad -\frac{1}{18} \sigma \left[F^2 \hat{E} \right] +\frac{1}{12} \sigma \left[F^2 \partial E \right] +\frac{1}{12} \sum_{\substack{i \in I \cup J \\ i \notin I \cap J}} \sigma \left[F \tilde{W}_{i} \right] -\frac{1}{6} \sum_{\substack{i \in \mathcal{I}_{n-1} \\ i \notin I \cup J}} \sigma \left[F \tilde{W}_{i} \right] \notag \\
            &\quad \quad +\frac{1}{8} \sum_{\substack{i \in I \\ i \notin I \cap J}} \sigma \left[F \partial T_{i} \right] -\frac{1}{8} \sum_{\substack{i \in J \\ i \notin I \cap J}} \sigma \left[F \partial T_{i} \right]
            -\frac{1}{4} \sum_{\substack{i \in \mathcal{I}_{n-1} \\ i \notin I \cup J}} \sigma \left[F \partial T_{i} \right] -\frac{1}{12} \sigma \left[F \partial^2 \tilde{U} \right] \notag \\
            &\quad \quad -\frac{1}{4} \sigma \left[F \partial^2 H \right] -\frac{1}{4} \sum_{\substack{i \in \mathcal{I}_{n-1} \\ i \notin I \cup J}} \sigma \left[T_{i} \hat{F} \right] -\frac{1}{4} \sum_{\substack{i \in \mathcal{I}_{n-1} \\ i \notin I \cup J}} \sigma \left[T_{i} \partial F \right]
            -\frac{1}{9} \sigma \left[\hat{H} \partial F \right] -\frac{1}{8} \sigma \left[\hat{F} \partial \tilde{U} \right] \notag \\
            &\quad \quad \left. -\frac{1}{8} \sigma \left[\partial \tilde{U} \partial F \right]
            +\frac{1}{8} \partial^3 F
        \right) \notag \\
        &\quad +\frac{\zeta_{I,J}}{z^{n+|I|+|J|-5}} \left(
            \frac{1}{120} \sigma \left[\tilde{U}^4 F \right] - \frac{1}{48} \sigma \left[\tilde{U}^3 \hat{F} \right] - \frac{1}{48} \sigma \left[\tilde{U}^3 \partial F \right] + \frac{1}{120} \sigma \left[\tilde{U}^2 H^2 F \right] \right. \notag \\
            &\quad \quad + \frac{1}{48} \sigma \left[\tilde{U}^2 H \hat{F} \right]
            + \frac{1}{24} \sigma \left[\tilde{U}^2 H \partial F \right] + \frac{1}{120} \sigma \left[\tilde{U}^2 E F^2 \right]
            -\frac{1}{24} \sum_{\substack{i \in \mathcal{I}_{n-1} \\ i \notin I \cup J}} \sigma \left[\tilde{U}^2 F T_{i} \right] \notag \\
            &\quad \quad - \frac{1}{48} \sigma \left[\tilde{U}^2 F \hat{H} \right] - \frac{1}{48} \sigma \left[\tilde{U}^2 F \partial \tilde{U} \right] - \frac{1}{24} \sigma \left[\tilde{U}^2 F \partial H \right] + \frac{1}{18} \sigma \left[\tilde{U}^2 \partial^2 F \right] + \frac{1}{432} \sigma \left[\tilde{U} H^2 \hat{F} \right] \notag \\
            &\quad \quad - \frac{1}{48} \sigma \left[\tilde{U} H^2 \partial F \right] - \frac{1}{216} \sigma \left[\tilde{U} H F \hat{H} \right] 
            - \frac{1}{48} \sigma \left[\tilde{U} H F \partial H \right]
            - \frac{2}{27} \sigma \left[\tilde{U} H \partial \hat{F} \right] \notag \\
            &\quad \quad - \frac{1}{12} \sigma \left[\tilde{U} H \partial^2 F \right] - \frac{1}{96} \sigma \left[\tilde{U} E F \hat{F} \right]
            - \frac{1}{48} \sigma \left[\tilde{U} E F \partial F \right]
            + \frac{1}{72} \sigma \left[\tilde{U} F^2 \hat{E} \right] - \frac{1}{48} \sigma \left[\tilde{U} F^2 \partial E \right] \notag
    \end{align}
    \begin{align}
            &\quad \quad -\frac{1}{36} \sum_{\substack{i \in I \cup J \\ i \notin I \cap J}} \sigma \left[\tilde{U} F \tilde{W}_{i} \right] +\frac{1}{18} \sum_{\substack{i \in \mathcal{I}_{n-1} \\ i \notin I \cup J}} \sigma \left[\tilde{U} F \tilde{W}_{i} \right]
            -\frac{1}{24} \sum_{\substack{i \in I \\ i \notin I \cap J}} \sigma \left[\tilde{U} F \partial T_{i} \right] + \frac{1}{24} \sum_{\substack{i \in J \\ i \notin I \cap J}} \sigma \left[\tilde{U} F \partial T_{i} \right] \notag \\
            &\quad \quad +\frac{1}{12} \sum_{\substack{i \in \mathcal{I}_{n-1} \\ i \notin I \cup J}} \sigma \left[\tilde{U} F \partial T_{i} \right]
            + \frac{1}{36} \sigma \left[\tilde{U} F \partial^2 \tilde{U} \right] 
            + \frac{1}{12} \sigma \left[\tilde{U} F \partial^2 H \right] +\frac{1}{12} \sum_{\substack{i \in \mathcal{I}_{n-1} \\ i \notin I \cup J}} \sigma \left[\tilde{U} T_{i} \hat{F} \right] \notag \\
            &\quad \quad +\frac{1}{12} \sum_{\substack{i \in \mathcal{I}_{n-1} \\ i \notin I \cup J}} \sigma \left[\tilde{U} T_{i} \partial F \right] + \frac{1}{27} \sigma \left[\tilde{U} \hat{H} \partial F \right] + \frac{1}{24} \sigma \left[\tilde{U} \hat{F} \partial \tilde{U} \right] 
            + \frac{1}{24} \sigma \left[\tilde{U} \partial \tilde{U} \partial F \right] - \frac{1}{16} \sigma \left[\tilde{U} \partial^3 F \right] \notag \\
            &\quad \quad - \frac{952673}{624360} \sigma \left[H^4 F \right]
            - \frac{403}{336} \sigma \left[H^3 \hat{F} \right] - \frac{599149}{249744} \sigma \left[H^3 \partial F \right]
            - \frac{102917}{468270} \sigma \left[H^2 E F^2 \right] \notag \\
            &\quad \quad -\frac{1}{24} \sum_{\substack{i \in \mathcal{I}_{n-1} \\ i \notin I \cup J}} \sigma \left[H^2 F T_{i} \right] + \frac{403}{1008} \sigma \left[H^2 F \hat{H} \right] - \frac{1}{48} \sigma \left[H^2 F \partial \tilde{U} \right]
            + \frac{599149}{749232} \sigma \left[H^2 F \partial H \right] \notag \\
            &\quad \quad + \frac{35}{18} \sigma \left[H^2 \partial \hat{F} \right] + \frac{54289}{140481} \sigma \left[H^2 \partial^2 F \right]
            + \frac{53}{672} \sigma \left[H E F \hat{F} \right]
            + \frac{6539}{11616} \sigma \left[H E F \partial F \right] - \frac{281}{1008} \sigma \left[H F^2 \hat{E} \right] \notag \\
            &\quad \quad - \frac{180335}{187308} \sigma \left[H F^2 \partial E \right] + \frac{25}{18} \sigma \left[H F \partial \hat{H} \right]
            - \frac{31681}{1123848} \sigma \left[H F \partial^2 H \right]
            -\frac{1}{12} \sum_{\substack{i \in \mathcal{I}_{n-1} \\ i \notin I \cup J}} \sigma \left[H T_{i} \hat{F} \right] \notag \\
            &\quad \quad -\frac{1}{6} \sum_{\substack{i \in \mathcal{I}_{n-1} \\ i \notin I \cup J}} \sigma \left[H T_{i} \partial F \right] + \frac{48607}{187308} \sigma \left[H \hat{H} \hat{F}\right]
            - \frac{95}{168} \sigma \left[H \hat{H} \partial F \right] - \frac{1}{24} \sigma \left[H \hat{F} \partial \tilde{U}\right] - \frac{277}{504} \sigma \left[H \hat{F} \partial H \right] \notag \\
            &\quad \quad - \frac{1}{12} \sigma \left[H \partial \tilde{U} \partial F \right]
            - \frac{180335}{68112} \sigma \left[H \partial H \partial F \right] 
            + \frac{129337}{1248720} \sigma \left[E^2 F^3 \right] -\frac{1}{24} \sum_{\substack{i \in \mathcal{I}_{n-1} \\ i \notin I \cup J}} \sigma \left[E F^2 T_{i} \right] \notag \\
            &\quad \quad + \frac{61}{504} \sigma \left[E F^2 \hat{H} \right]
            - \frac{1}{48} \sigma \left[E F^2 \partial \tilde{U}\right] - \frac{122191}{749232} \sigma \left[E F^2 \partial H \right] - \frac{1}{168} \sigma \left[E F \partial \hat{F} \right]
            - \frac{1609}{26136} \sigma \left[E F \partial^2 F \right] \notag \\
            &\quad \quad + \frac{14093}{749232} \sigma \left[E \hat{F}^2 \right]
            - \frac{143}{504} \sigma \left[F^2 \partial \hat{E} \right] - \frac{604367}{2247696} \sigma \left[F^2 \partial^2 E \right]
            -\frac{1}{18} \sum_{i \in I \cap J} \sigma \left[F T_{i}^2 \right] +\frac{1}{6} \sum_{\substack{i \in I \cup J \\ i \notin I \cap J}} \sigma \left[F T_{i}^2 \right] \notag \\
            &\quad \quad -\frac{1}{6} \sum_{\substack{i \in \mathcal{I}_{n-1} \\ i \notin I \cup J}} \sigma \left[F T_{i}^2 \right]
            +\frac{1}{6} \sum_{\substack{i,j \in\mathcal{I}_{n-1} \\ i,j \notin I \cup J \\ i < j}} \sigma \left[F T_{i} T_{j} \right]
            +\frac{1}{12} \sum_{\substack{i \in \mathcal{I}_{n-1} \\ i \notin I \cup J}} \sigma \left[F T_{i} \hat{H} \right] +\frac{1}{12} \sum_{\substack{i \in \mathcal{I}_{n-1} \\ i \notin I \cup J}} \sigma \left[F T_{i} \partial \tilde{U} \right] \notag \\
            &\quad \quad +\frac{1}{6} \sum_{\substack{i \in \mathcal{I}_{n-1} \\ i \notin I \cup J}} \sigma \left[F T_{i} \partial H \right]
            + \frac{45041}{374616} \sigma \left[F \hat{H}^2 \right] + \frac{1}{24} \sigma \left[F \hat{H} \partial \tilde{U} \right]
            - \frac{277}{504} \sigma \left[F \hat{H} \partial H \right] - \frac{2380}{46827} \sigma \left[F \hat{E} \hat{F} \right] \notag \\
            &\quad \quad + \frac{1}{126} \sigma \left[F \hat{F} \partial E \right]
            + \frac{1}{24} \sigma \left[F \left(\partial \tilde{U} \right)^2 \right] + \frac{1}{12} \sigma \left[F \partial \tilde{U} \partial H \right]
            - \frac{31429}{34056} \sigma \left[F \left(\partial H \right)^2 \right]
            + \frac{13053}{15136} \sigma \left[F \partial E \partial F \right] \notag \\
            &\quad \quad +\frac{1}{12} \sum_{\substack{i \in I \\ i \notin I \cap J}} \sigma \left[F \partial \tilde{W}_{i} \right] -\frac{1}{12} \sum_{\substack{i \in \mathcal{I}_{n-1} \\ i \notin I \cup J}} \sigma \left[F \partial \tilde{W}_{i} \right]
            +\frac{1}{24} \sum_{i \in I \cap J} \sigma \left[F \partial^2 T_{i} \right] -\frac{1}{8} \sum_{\substack{i \in J \\ i \notin I \cap J}} \sigma \left[F \partial^2 T_{i} \right] \notag
    \end{align}
    \begin{align}
            &\quad \quad - \frac{1}{48} \sigma \left[F \partial^3 \tilde{U} \right]
            -\frac{1}{6} \sum_{\substack{i \in \mathcal{I}_{n-1} \\ i \notin I \cup J}} \sigma \left[T_{i} \partial^2 F \right] - \frac{5801}{45408} \sigma \left[\hat{H} \partial \hat{F} \right] +\frac{1}{24} \sum_{\substack{i \in I \cup J \\ i \notin I \cap J}} \sigma \left[\hat{F} \tilde{W}_{i} \right] -\frac{1}{12} \sum_{\substack{i \in \mathcal{I}_{n-1} \\ i \notin I \cup J}} \sigma \left[\hat{F} \tilde{W}_{i} \right] \notag \\
            &\quad \quad +\frac{1}{16} \sum_{\substack{i \in I \\ i \notin I \cap J}} \sigma \left[\hat{F} \partial T_{i} \right] -\frac{1}{16} \sum_{\substack{i \in J \\ i \notin I \cap J}} \sigma \left[\hat{F} \partial T_{i} \right] -\frac{1}{8} \sum_{\substack{i \in \mathcal{I}_{n-1} \\ i \notin I \cup J}} \sigma \left[\hat{F} \partial T_{i} \right] + \frac{5801}{45408} \sigma \left[\hat{F} \partial \hat{H} \right]
            - \frac{1}{24} \sigma \left[\hat{F} \partial^2 \tilde{U} \right] \notag \\
            &\quad \quad +\frac{1}{24} \sum_{\substack{i \in I \cup J \\ i \notin I \cap J}} \sigma \left[\tilde{W}_{i} \partial F\right] -\frac{1}{12} \sum_{\substack{i \in \mathcal{I}_{n-1} \\ i \notin I \cup J}} \sigma \left[\tilde{W}_{i} \partial F \right] - \frac{1}{12} \sigma \left[\partial \tilde{U} \partial^2 F \right] +\frac{1}{16} \sum_{\substack{i \in I \\ i \notin I \cap J}} \sigma \left[\partial F \partial T_{i} \right] \notag \\
            &\quad \quad \left. -\frac{1}{16} \sum_{\substack{i \in J \\ i \notin I \cap J}} \sigma \left[\partial F \partial T_{i} \right] -\frac{1}{8} \sum_{\substack{i \in \mathcal{I}_{n-1} \\ i \notin I \cup J}} \sigma \left[\partial F \partial T_{i} \right]
            - \frac{1}{24} \sigma \left[\partial F \partial^2 \tilde{U} \right]
        \right) \notag \\
        &\quad +\cdots ~,
    \end{align}
    where the above OPEs are checked for $n=2,3,4,5$.
\section{To generalize fermionic extension for abelian quiver gauge theories}
    \label{To generalize fermionic extension for abelian quiver gauge theories}
    Through the specific example in Sec.\ref{Fermionic extension}, we checked that the operators $T_i$, $W_i$, $\mathcal{X}_{I}$, $\mathcal{Y}_{I}$, for $I \neq \phi$ which were composed of symplectic bosons $X_i$ and $Y_i$ could be constructed by the generators after fermionic extension, that is (\ref{bosonic only}), (\ref{U_F}), (\ref{M pm}) and (\ref{X I and Y I}).
    In this subsection, we try to generalize this statement and understand how to extend and restrict the algebraic structure.\footnote{Here, we deal with a quiver gauge theory with an abelian gauge group.} 
    
    We know that it is necessary to modify the stress tensor from (\ref{definition of T}) to (\ref{definition of T ver.4}) and apply (\ref{extension of T_i}) in our case, $T_{[n-1,1]}^{[1^n]}(SU(n))$. 
    Therefore, let us generalize the proper modification of the stress tensor as follows.

    First, we must check whether $C_{ij}$ is invertible. Let us define the currents $J_{XY,i}$ as follows.
    \begin{align}
        J_{XY,i} \coloneqq \sum_{k} A_{ik} \mathcal{A}_{k}~,
    \end{align} 
    where $J_{XY,i}$ are used by the definition of the BRST current $J_{BRST} \coloneqq \sum_{i}c^i \left(J_{XY,i} + h_i\right)~$. 
    Then, we can calculate the OPEs between $J_{XY,i}$ and $J_{XY,j}$.
    \begin{align}
        \label{definition of C_ij}
        J_{XY,i}(z) J_{XY,j}(0) \sim \frac{-C_{ij}}{z^2}~,
    \end{align}
    where $C_{ij}$ is defined by using $A_{ik}$ as follows.
    \begin{align}
        \label{AA}
        C_{ij} \coloneqq \sum_{k} A_{ik} A_{jk}~.
    \end{align}
    Therefore, when the following condition holds,
    \begin{align}
        \label{inversion}
        \text{Rank}\left(A\right) 
        &= \left(\text{the number of } U(1) \text{ gauge groups in the theory}\right)~,
    \end{align}
    $C_{ij}$ can be invertible.
    Then, we can use the stress tensor
    \begin{align*}
        T &= \frac{1}{2}\sum_{i} \left(X_i \partial Y_i -\partial X_i Y_i\right) + \frac{1}{2}\sum_{i,j}C^{ij} h_i h_j -\sum_{a} b_a \partial c^a \\
        &\equiv \frac{1}{2}\sum_{i} \left(X_i \partial Y_i -\partial X_i Y_i\right) +\frac{1}{2}\sum_{i,j}C^{ij} J_{XY,i} J_{XY,j}~.
    \end{align*}
    where $C^{ij}$ is the inverse of $C_{ij}$.

    Second, let us conjecture that we can generally construct $T_{BC}$ satisfying
    \begin{align}
        \label{T BC}
        \frac{1}{2}\sum_{i,j}C^{ij} J_{XY,i} J_{XY,j} = \frac{1}{2} \sum_{i} \mathcal{A}_i \mathcal{A}_i +T_{BC}~,
    \end{align}
    by the generators of dimension one using only symplectic bosons.
    Since $T_{BC}$ can be composed of some linear sum of $\mathcal{A}_i$, we set $m$ closed operators $U_1,\dots,U_m$ satisfying
    \begin{align}
        U_i(z) U_j(0) \sim \frac{k_i\delta_{ij}}{z^2}~,
    \end{align}
    where $k_i$ is a negative integer and the level of $U(1)$ current $U_i$. Since we can define $U_i$ as follows,
    \begin{align}
        U_i \coloneqq \sum_{j} B_{ij} \mathcal{A}_j~,
    \end{align}
    we can derive the following relations from the OPEs between $U_i$ and $U_j$,
    \begin{align}
        \label{BB}
        \sum_{k} B_{ik} B_{jk} = -k_i \delta_{ij}~.
    \end{align}
    Therefore, $T_{BC}$ can be defined by
    \begin{align}
        \label{the definition of T_BC}
        T_{BC} \coloneqq \frac{1}{2}\sum_{a} \sum_{i,j} \frac{B_{ai}B_{aj}}{k_a} \mathcal{A}_i \mathcal{A}_j~. 
    \end{align}
    In fact, we can derive the OPEs of $T_{BC}$ as follows.
    \begin{align}
        T_{BC}(z) T_{BC}(0) &\sim \frac{m}{2z^4} + \frac{2 T_{BC}}{z^2} + \frac{\partial T_{BC}}{z}~, \\
        T_{BC}(z) U_k(0) &\sim \frac{U_k}{z^2} + \frac{\partial U_k}{z}~.
    \end{align}
    By the definition of $T_{BC}$, $T_{BC}$ is clearly a closed operator, namely $J_{BRST}(z) T_{BC}(0) \sim 0~$.
    This closed condition is equivalent of the following equality.
    \begin{align}
        \label{AB}
        \sum_{k}A_{ik} B_{jk}= 0~.
    \end{align}

    Finally, we expand the predicted relation (\ref{T BC}) by $A$ and $B$. Then, we can derive the following equivalence relation.
    \begin{align}
        \label{T BC 2}
        \sum_{a} \left\{\sum_{b} C^{ab}A_{ai} A_{bj}\right\} +\sum_{a}\left(-\frac{B_{ai} B_{aj}}{k_a} \right) = \delta_{ij}~.
    \end{align}
    We can find out that $\left\{\sum_{b}C^{ab} A_{ai} A_{bj}\right\}_{a}$ and $\left\{-\frac{B_{ai} B_{aj}}{k_a}\right\}_a$ are satisfying the following relations by (\ref{AA}), (\ref{BB}) and (\ref{AB}).
    \begin{align}
        \label{orthogonal decomposition}
        \left\{
            \begin{aligned}
                \sum_{j}\left(\sum_{b} C^{ab} A_{ai} A_{bj}\right) \left(\sum_{d} C^{cd} A_{cj} A_{dk}\right) &= \delta_{ac} \sum_{b} C^{ab} A_{ai} A_{bk} \\
                \sum_{j} \left(\sum_{b} C^{ab} A_{ai} A_{bj}\right) \left(-\frac{B_{cj}B_{ck}}{k_{c}}\right) &= 0 \\
                \sum_{j} \left(-\frac{B_{ai}B_{aj}}{k_{a}} \right) \left(\sum_{c} C^{bc} A_{bj} A_{ck}\right) &= 0\\
                \sum_{j} \left(-\frac{B_{ai}B_{aj}}{k_{a}} \right) \left(-\frac{B_{bj}B_{bk}}{k_{b}} \right) &= \delta_{ab} \left(-\frac{B_{ai}B_{ak}}{k_a}\right)
            \end{aligned}
        \right.~.
    \end{align}
    In other words, the predicted relation (\ref{T BC 2}) is equivalent to the orthogonal decomposition by $\left\{\sum_{b}C^{ab} A_{ai} A_{bj} \right\}$ and $\left\{-\frac{B_{ai} B_{aj}}{k_a} \right\}$.
    Thus (\ref{T BC 2}) holds if and only if the following condition is needed.
    \begin{align}
        \label{T BC 3}
        &\quad \left(\text{the number of } U(1) \text{ gauge groups in the theory}\right) + \left(\text{the number of } U_i \text{ in the theory}\right) \notag \\
        &= \left(\text{the number of symplectic bosons } (X_i,Y_i) \text{ in the theory}\right)~.
    \end{align}
    When this condition (\ref{T BC 3}) holds,
    the equality (\ref{T BC}) holds and the stress tensor $T$ is also equivalent of $\sum_i T_i + T_{BC}$.
    In particular, when $U_i$ does not exist and (\ref{T BC 3}) is satisfied, $T \equiv \sum_{i} T_i$ holds.

    Similarly, when (\ref{T BC 3}) holds, let us set $T_{FC}$ satisfying the following equality.
    \begin{align}
        \label{T F}
        T_f +T_{FC}= \frac{1}{2}\sum_{i}\left(\partial \psi_i \chi_i - \psi_i \partial \chi_i\right)~,
    \end{align}
    where $T_{FC}$ consists of generators of dimension one using fermi multiplets only. 
    Since the calculation of (\ref{T BC}) is ordinary similar to one of (\ref{T F}), that $(X_i,Y_i)$ are respectively replacd for $(\psi_i, \chi_i)$ corresponds
    that $T_{FC}$ can be derived from $T_{BC}$ as follows.
    \begin{align}
        \label{T FC}
        T_{FC} \coloneqq \left.T_{BC}\right|_{(X_i,Y_i) \mapsto (\psi_i,\chi_i),k_a \mapsto -k_a}~.
    \end{align}

    In order to derive (\ref{T FC}), when we define $J_{f,i} = \sum_{k} A_{ik} \mathcal{A}_{f,k}$ which are used by $J_{BRST} = \sum_{i}c^i \left(J_{XY,i} + J_{f,i}\right)$,
    and $U_{F,i} = \sum_{k} B_{ik} \mathcal{A}_{f,i}$ which is the replacement of $(X_i,Y_i)$ with $(\psi_i,\chi_i)$ with respect to $U_i$ and satisfying\footnote{The difference of (\ref{BB}) is related to the following OPE regarding with $\mathcal{A}_{f,i}$,
    \begin{align*}
        \mathcal{A}_{f,i}(z) \mathcal{A}_{f,j} (0) \sim \frac{\delta_{ij}}{z^2}~.
    \end{align*}}
    \begin{align}
        U_{F,i}(z) U_{F,j} (0) \sim \frac{-k_i \delta_{ij}}{z^2}~.
    \end{align}
    When the condition (\ref{T BC 3}) holds,
    we can also derive the following equality by (\ref{T BC 2})
    \begin{align}
        T_f + \frac{1}{2} \sum_{a} \frac{U_{F,a} U_{F,a}}{-k_a} = \frac{1}{2} \sum_{i} \left(\partial \psi_i \chi_i - \psi_i \partial \chi_i\right)~.
    \end{align}
    Therefore, compared with (\ref{T F}), we can find out that $T_{FC}$ is satisfying the following relationship.
    \begin{align}
        T_F = \frac{1}{2} \sum_{a} \frac{U_{F,a} U_{F,a}}{-k_a} = \left.\frac{1}{2} \sum_{a} \frac{U_{a} U_{a}}{k_a}\right|_{\left(X_i,Y_i\right) \mapsto \left(\psi_i,\chi_i\right), k_a \mapsto -k_a}
        = \left. T_{BC} \right|_{\left(X_i,Y_i\right) \mapsto \left(\psi_i,\chi_i\right), k_a \mapsto -k_a}~,
    \end{align}
    where the above relationship is equal to (\ref{T FC}).

    Thus, we can prove that the stress tensor can be constructed by using $T_{BC}$ and $T_{FC}$.
    \begin{align}
        \label{definition of T ver.6}
        T &= T_{sb} +T_f + T_{bc} +T_{FC} \\
        \label{definition of T ver.7}
        &\equiv \sum_{i}T_i +T_{BC} +T_{FC}~.
    \end{align}

    In the previous subsection, we derived the equivalences (\ref{extension of T_i}), (\ref{extension of W_i}) and (\ref{transformation of XI, YI}). In order to derive the similar relations,
    let us consider the general equivalency of $X_i Y_i + \psi_i \chi_i$ like (\ref{equivalency of U+U_F}).
    We prove that we can set this formula $X_i Y_i +\psi_i \chi_i$ as follows.
    \begin{align}
        \label{special relation}
        \mathcal{A}_i + \mathcal{A}_{f,i} \equiv \check{U}_i~,
    \end{align}
    where $\check{U}_i$ is equal to $0$ or composed of linear sum of the currents of one dimension which consist of generators after fermionic extension.\footnote{For $T_{[n-1,1]}^{[1^n]}(SU(n))$,
    $\check{U}_i$ are equal to $U+U_F$ by (\ref{equivalency of U+U_F}). In another example, $1$-flavor SQED \cite{Costello:2018fnz}, $\check{U}$ is $0$.}
    By using $A_{ik}$ and $B_{ik}$, the following equivalency holds.
    \begin{align}
        \begin{bmatrix}
            A \\
            B
        \end{bmatrix}
        \begin{bmatrix}
            \mathcal{A}_{1} + \mathcal{A}_{f,1} \\
            \vdots
        \end{bmatrix}
        \equiv
        \begin{bmatrix}
            0 \\
            \vdots \\
            0 \\
            U_1 + U_{F,1} \\
            \vdots \\
            U_m + U_{F,m}
        \end{bmatrix}~.
    \end{align}
    We define $\tilde{A}$ (resp. $\tilde{B}$) by normalizing and orthogonalizing the row vectors of $A$ (resp. $B$). Then, 
    \begin{align}
        \begin{bmatrix}
            \tilde{A} \\
            \tilde{B}
        \end{bmatrix}
        \begin{bmatrix}
            \mathcal{A}_{1} + \mathcal{A}_{f,1} \\
            \vdots
        \end{bmatrix}
        \equiv
        \begin{bmatrix}
            0 \\
            \vdots \\
            0 \\
            \frac{1}{\sqrt{-k_1}} \left(U_1 + U_{F,1}\right) \\
            \vdots \\
            \frac{1}{\sqrt{-k_m}} \left(U_m + U_{F,m}\right)
        \end{bmatrix}~,
    \end{align}
    where $k_i < 0$ for all $i =1,\dots,m$. Since the matrix of the left-hand side is the orthogonal matrix, the following equivalency can be derived.
    \begin{align}
        \begin{bmatrix}
            \mathcal{A}_{1} + \mathcal{A}_{f,1} \\
            \vdots
        \end{bmatrix}
        \equiv
        \begin{bmatrix}
            \tilde{A}^T & \tilde{B}^T \\
        \end{bmatrix}
        \begin{bmatrix}
            0 \\
            \vdots \\
            0 \\
            \frac{1}{\sqrt{-k_1}} \left(U_1 + U_{F,1}\right) \\
            \vdots \\
            \frac{1}{\sqrt{-k_m}} \left(U_m + U_{F,m}\right)
        \end{bmatrix}~,
    \end{align}
    where $\tilde{A}^T$ (resp. $\tilde{B}^T$) is the transpose of $\tilde{A}$ (resp. $\tilde{B})$. 
    Therefore, when $m \geq 1$, the following relationships hold,
    \begin{align}
        \label{exactness}
        \mathcal{A}_{i} + \mathcal{A}_{f,i} \equiv \sum_{a} \frac{B_{ai}}{-k_a} \left(U_a + U_{F,a}\right)~,
    \end{align}
    and we respectively define $\check{U}_i$ as follows.
    \begin{align}
        \check{U}_i \coloneqq \sum_{a} \frac{B_{ai}}{-k_a} \left(U_a + U_{F,a}\right)~.
    \end{align}
    Similarly, when $m=0$, the following equivalences hold,
    \begin{align}
        \mathcal{A}_i + \mathcal{A}_{f,i} \equiv 0~,
    \end{align}
    and we set $\check{U}_i \coloneqq 0$.
    Thus, (\ref{special relation}) can be derived by the definition of the BRST current, generators $U_i$ and $U_{F,i}$, and generators regarding each flavor symmetry.
    
    When (\ref{special relation}) can be constructed, the relations (\ref{extension of T_i}), (\ref{extension of W_i}), (\ref{transformation of XI, YI}) can be written as follows.
    \begin{align}
        \label{summations}
        \left\{
            \begin{aligned}
                T_i &\equiv \frac{1}{2}\check{U}_i^2 -\frac{1}{2}[M^{+}_i,M^{-}_i]\\
                W_i &\equiv \frac{1}{3}\sqrt{\frac{2}{3}} \left\{\frac{3}{2}\left(M^{+}_{i} \partial M^{-}_{i} +M^{-}_{i} \partial M^{+}_{i}\right)
                -\frac{3}{2}\check{U}_i [M^{+}_{i},M^{-}_{i}] 
                +\check{U}_i^3 +\partial^2 \check{U}_i \right\}\\
                D_j \hat{\mathcal{X}}_{I,J \setminus \{j\}} &\equiv \check{U}_j \hat{\mathcal{X}}_{I,J \setminus \{j\}} + (-1)^{f_j} M_{j}^{+} \hat{\mathcal{X}}_{I \cup \{j\},J \setminus \{j\}} \\
                D_j \hat{\mathcal{Y}}_{I,J \setminus \{j\}} &\equiv \check{U}_j \hat{\mathcal{Y}}_{I,J \setminus \{j\}} - (-1)^{f_j} M_{j}^{-} \hat{\mathcal{Y}}_{I \cup \{j\},J \setminus \{j\}}
            \end{aligned}
        \right.~.
    \end{align}
    The stress tensor can be also given by 
    \begin{align}
        \label{definition of the stress tensor after f.e.}
        T &\equiv \frac{1}{2}\sum_{i} \tilde{U}_i^2 +T_{BC} +T_{FC} -\frac{1}{2}\sum_{i}[M^{+}_i,M^{-}_i]~.
    \end{align}
    Thus, we conjecture that the algebraic relations between a bosonic VOA and a boundary VOA are similar to (\ref{summations}).
    In particular, the stress tensor $T$ in boundary VOAs can be always constructed by the generators $M^{\pm}_i, U_i$ and $U_{F,i}$.\footnote{When $U_i$ and $U_{F,i}$ do not exist in boundary VOAs, the stress tensor $T$ is only consisting of $M^{\pm}_i$.}
    Thus, we understand that the operations related to $\mathcal{D}_i$ and $\hat{\mathcal{A}}_i$ can be replaced by $M^{\pm}_i$ and $\check{U}_i$.
\bibliography{reference}
\bibliographystyle{ytphys}
\end{document}